\documentclass[10pt]{article}
\usepackage[margin=.8in,left=.8in]{geometry}

\usepackage{amsthm}
\usepackage{amssymb}
\usepackage{amsmath}
\usepackage{mathtools}
\usepackage{adjustbox}

\usepackage[table]{xcolor}

\usepackage{stmaryrd}    
\usepackage{rotating}
\usepackage{xfrac}
\usepackage{enumitem}\setlist{nosep}
\usepackage{slashed}

\usepackage[safe]{tipa} 

\usepackage{pgffor}

\usepackage{hhline}

\usepackage{tensor}

\usepackage{accents}
\newlength{\dhatheight}

\usepackage{tikz}
\usetikzlibrary{
  backgrounds, 
  decorations.pathmorphing, 
  decorations.markings,
  decorations.text,
  snakes
}
\usetikzlibrary{cd}

\usepackage{mathrsfs} 
\DeclareMathAlphabet{\mathpzc}{OT1}{pzc}{m}{it} 

\usepackage{hyperref}
\hypersetup{
  colorlinks = true,
  linkcolor=darkblue, citecolor=darkblue,
  urlcolor=black
}

\definecolor{darkblue}{rgb}{0.05,0.25,0.65}
\definecolor{darkgreen}{RGB}{20,140,10}
\definecolor{lightgray}{rgb}{0.9,0.9,0.9}
\definecolor{darkorange}{RGB}{200,100,5}
\definecolor{darkyellow}{rgb}{.91,.91,0}

\newtheorem{theorem}{Theorem}[section]
\newtheorem{lemma}[theorem]{Lemma}

\newtheorem{proposition}[theorem]{Proposition}
\newtheorem{corollary}[theorem]{Corollary}
\theoremstyle{definition}
\newtheorem{definition}[theorem]{Definition}

\newtheorem{example}[theorem]{Example}

\newtheorem{remark}[theorem]{Remark}

\newcommand{\proofstep}[1]{\scalebox{.8}{#1}}

\newcommand{\HilbertSpace}[1]{\mathcal{#1}}

\newcommand{\ZTwo}{\mathbb{Z}_2}

\newcommand{\ten}{1\!0}

\newcommand{\dir}[2]{\overset{\scalebox{.54}{$\mathclap{\;#1}$}}{#2}}

\newcommand{\shift}{\mathrm{s}}

\newcommand{\Brn}{\mathfrak{Brn}}

\newcommand{\DFSpinor}{\phi}

\newcommand{\NN}{\mathbb{N}}

\newcommand{\IIA}{\mathrm{II}\mathfrak{A}}

\newcommand{\CE}{\mathrm{CE}}

\let\PLAINthebibliography\thebibliography
\renewcommand\thebibliography[1]{
  \PLAINthebibliography{#1}
  \setlength{\parskip}{0.5pt}
  \setlength{\itemsep}{0.5pt plus .3ex}
}

\newcommand{\yields}{\vdash}

\newcommand{\defneq}{\equiv}

\newcommand{\differential}{\mathrm{d}}

\newcommand\bosonic[1]{\mathstrut\mkern2.5mu#1\mkern-14mu\raise1.7ex%
  \hbox{$\scriptstyle\rightsquigarrow$}}

\newcommand{\grayunderbrace}[2]{\color{gray}\underbrace{\color{black}#1}_{\color{gray}#2}\color{black}}
\newcommand{\grayoverbrace}[2]{\color{gray}\overbrace{\color{black}#1}^{\color{gray}#2}\color{black}}

\newcommand{\FDGCA}{\mathbb{R}_\differential}

\newcommand{\frg}{\mathfrak{g}}

\newcommand{\FR}{\mathbb{R}}

\newcommand{\dd}{\mathrm{d}}

\newcommand{\bas}{\mathrm{bas}}

\begin{document}

\setlength{\abovedisplayskip}{3pt}
\setlength{\belowdisplayskip}{3pt}
\setlength{\abovedisplayshortskip}{-3pt}
\setlength{\belowdisplayshortskip}{3pt}

\title{Super-$\mathrm{Lie}_\infty$ T-Duality and M-Theory}

\author{
  Grigorios Giotopoulos${}^{\ast}$,
  \;\;
  Hisham Sati${}^{\ast \dagger}$,
  \;\;
  Urs Schreiber${}^{\ast}$
}

\maketitle

\begin{abstract}
  Super $L_\infty$-algebras unify extended super-symmetry with rational classifying spaces for higher flux densities: The super-invariant super-fluxes which control super $p$-branes and their supergravity target super-spaces are, together with their (non-linear) Bianchi identities, neatly encoded in (non-abelian) super-$L_\infty$ cocycles. These are the rational shadows of flux-quantization laws (in ordinary cohomology, K-theory, Cohomotopy, iterated K-theory, etc.).

\smallskip 
  We first review, in streamlined form while filling some previous gaps, double-dimensional reduction/oxidation and 10D superspace T-duality along higher-dimensional super-tori.
  We do so tangent super-space wise, by viewing it as an instance of adjunctions (dualities) between super-$L_\infty$-extensions and -cyclifications, applied to the avatar super-flux densities of 10D supergravity. 
  In particular,  this yields a derivation, at the rational level, of the traditional laws of ``topological T-duality'' from the super-$L_\infty$ structure of type II superspace. At this level, we also discuss a higher categorical analog of T-duality involving M-branes.

  \smallskip 
  Then, by considering super-space T-duality along all 1+9 spacetime dimensions while retaining the 11th dimension as in F-theory, we find the M-algebra appearing as the complete brane-charge extension of the fully T-doubled/correspondence super-spacetime.
  On this backdrop, we recognize the ``decomposed'' M-theory 3-form on the ``hidden M-algebra'' as an M-theoretic lift of the Poincar{\'e} super 2-form that controls superspace T-duality as the integral kernel of the super Fourier-Mukai transform. This provides the super-space structure of an M-theory
  lift of the doubled/correspondence space geometry, which controls T-duality.
\end{abstract}

\vspace{-.1cm}

\begin{center}
\begin{minipage}{8cm}
\tableofcontents
\end{minipage}
\end{center}

\vfill

\hrule
\vspace{.2cm}

{
\footnotesize
\noindent
\def\arraystretch{1}
\tabcolsep=0pt
\begin{tabular}{ll}
${}^*$\,
&
Mathematics, Division of Science; and
\\
&
Center for Quantum and Topological Systems,
\\
&
NYUAD Research Institute,
\\
&
New York University Abu Dhabi, UAE.  
\end{tabular}
\hfill
\adjustbox{raise=-15pt}{
\href{https://ncatlab.org/nlab/show/Center+for+Quantum+and+Topological+Systems}{
\includegraphics[width=3cm]{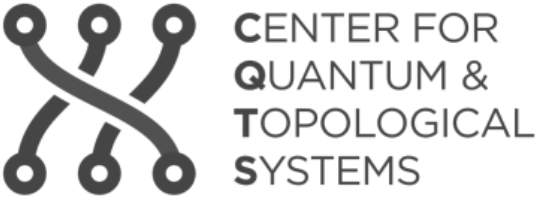}}
}

\vspace{1mm} 
\noindent ${}^\dagger$The Courant Institute for Mathematical Sciences, NYU, NY.

\vspace{.2cm}

\noindent
The authors acknowledge the support by {\it Tamkeen} under the 
{\it NYU Abu Dhabi Research Institute grant} {\tt CG008}.
}

\newpage

\section{Introduction \& Overview}
\label{IntroductionAndOverview}

\noindent
{\bf Open question of non-perturbative theory of strongly coupled/correlated quantum systems.}
In contrast to the oft-heard lament that contemporary experimental results in fundamental physics are all-too-well explained by existing theory (thus allegedly lacking desired hints of ``new physics'') it is a public secret that large swaths of phenomena exhibited by quantum systems are explained by existing theory at best in principle: Namely common perturbation/mean-field theory has little to say \cite{BakulevShirkov10} about the behavior of {\it strongly coupled} quantum systems \cite{FrishmanSonnenschein10}\cite{Strocchi13} (which includes no less than confined ordinary hadronic matter at non-excessive temperature \cite[p. 6]{Brambilla14}\cite{RobertsSchmidt20}) and hence of {\it strongly correlated} quantum systems \cite{Fulde12}\cite{FGSS20}, which notably includes  fractional quantum Hall systems \cite{Stormer99} and other quantum materials \cite{KeimerMoore17} expected to exhibit anyonic topological order arguably needed for scalable quantum computation (e.g. \cite{Sau17}\cite{SS23-ToplOrder}\cite{MySS24}\cite{SS25-Understanding}). 

\smallskip

\noindent
{\bf M-theory as a candidate solution.} Hints for a hypothetical but plausible non-perturbative theory of strongly-coupled quantum systems have emerged since the 1990s \cite{Witten95}\cite{Duff96}\cite{Duff99-MTheory}, originating in discussion of quantum gravity and ``grand unified'' quantum field theory. In its ``holographic'' guise (e.g. \cite{Natsuume15}) this approach models (possibly counter-intuitively but with remarkable success) strongly-coupled quantum systems by matching them onto the dynamics of ``branes'' (higher dimensional {\it m}embranes, whence ``M-theory'' \cite[p 1]{HoravaWitten96a}, introductions include \cite{West12}) whose quantum fluctuations inside an auxiliary higher dimensional spacetime turn out to be usefully reflected in the ambient gravitational field.

Notably, quantum critical superconductors have partially been understood holographically this way
\cite{HKSS07} \cite{GSW10a}\cite{GSW10b}\cite{GPR10}\cite{DGP13}  (review in \cite{Pires14}\cite{ZLSS15}\cite{Nastase17}\cite{HLS18}); but M-theory will be needed (cf. \cite[p. 60]{AGMOO00}\cite[p. 2]{CP18}\cite[p. iii]{Chester18}\cite[Fig. 4]{SS23-ToplOrder}) to complete the holographic description beyond the usual (but unrealistic) large-$N$ limit.

For example, with more M-theoretic aspects included, anyons as in fractional quantum Hall systems can potentially be understood from first principles by careful analysis (\cite{HellermanSusskind01}\cite{SS23-DefectBranes}\cite{SS23-ToplOrder}\cite{SS24-AbAnyons}\cite{SS25-Seifert}\cite{SS25-TQBit}\cite{SS25-Srni}) of $N=1$ five-dimensional such branes (``M5-branes'', review and pointers in \cite[\S 3]{Duff99-MTheory}\cite{GSS24-FluxOnM5}).

\smallskip

While the full formulation of M-theory remains an open and ambitious problem --- as the community periodically reminds itself of (e.g., \cite[\S 6]{Duff96}\cite[p. 43]{Moore14}\cite{CP18}\cite{Duff20} \cite{BLOY24}) --- there is a substantial web of hints as to its nature, mainly from

{\bf (i)} extended super-symmetry, 

{\bf (ii)} hidden duality symmetry,

\noindent
and our aim here and in the companion article \cite{GSS24-MGroup} is towards the combination/unification of these two aspects (which previously has found little attention):
 
\medskip

\noindent
{\bf Towards M-theory via extended super-symmetry.}
The characteristic local \footnote{
  Supergravity is {\it locally} supersymmetric in direct analogy to how ordinary Einstein gravity is {\it locally} Poincar{\'e}-symmetric (only), namely: on each (super-)tangent space (hence in the {\it infinitesimal} neighbourhood of any point) but not globally (in other words: on each ``super-Kleinian'' local model space but globally curved by ``super-Cartan geometry''). In contrast to global low-energy supersymmetry that has gotten so much attention in the last couple of decades, the phenomenolgy of local Planck-scale supersymmetry remains viable even if it has received much less attention, but see e.g. \cite{MeissnerNicolai18}\cite{BurgessQuevedo22}.
} 
spacetime super-symmetry of higher dimensional supergravity 
(cf. Ex. \ref{SupersymmetryAlgebras})
famously admits $\mathrm{Spin}$-equivariant central extensions (e.g. \cite{vanHoltenVanProeyen82} \cite{KrippendorfQuevedoSchlotterer10}, cf. Def. \ref{CentralExtension} \& Rem. \ref{ExtendedIIAAlgebraAndBraneCharges} below) by Noether charges of global symmetries of branes probing super-spacetime \cite{SS17-BPS}, suggesting these extensions as the local symmetry of completions of supergravity by brane dynamics.
The maximal such central extension for 11D SuGra, known as the {\it M-algebra} (\cite{DF82}\cite[(13)]{Townsend95}\cite[(1)]{Townsend98}\cite{Sezgin97}, cf. Def. \ref{BasicMAlgebra} below) is by charges of the very membrane (and fivebrane) that give M-theory its name, plausibly going some way towards elucidating its nature (cf. \cite{Townsend99}).

\medskip
\noindent
{\bf Towards M-theory via duality symmetry.}
But the founding observation of the field that came to be known as ``M-theory'' is that different-looking expected corners or limits of M-theory tend to be subtly equivalent to each other via ``dualities'' \cite{Schwarz97}\cite{ForsteLouis98}\cite{deWitLouis99}\cite{Polchinski17}, suggesting that the full M-theory could be revealed by making manifest a symmetry-principle of ``U-duality symmetries''. The archetypical example is {\it T-duality} (cf.  \cite{AAL95}\cite{Bugden19}\cite{Waldorf24}, we re-derive various ingredients in \S\ref{SuperspaceTDuality} below), whereby, remarkably spacetime dimensions transmute into charges of 1-dimensional branes (strings), and vice versa. 
Including also higher-dimensional brane charges into this picture reveals, at least super-tangent space wise, a much larger {\it brane rotating symmetry} 
(\cite{BaerwaldWest00}, cf. Ex. \ref{BraneRotatingSymmetry}) which in turn is argued to be the shadow of a humongous {\it U-duality} symmetry (\cite{HullTownsend95}\cite{ObersPiloine99} \footnote{
  The term ``U-duality'' was introduced in \cite{HullTownsend95} for restricted actions of integral subgroups $E_{n(n)}(\mathbb{Z}) \subset E_{n(n)}$ conjecturally enforced by charge quantization. But for the tangent-space wise Lie algebra discussion of concern here 
  (and since the precise nature of the charge quantization needs more attention anyway, cf. also \cite[p. 4]{deWitNicolai01})
  we will use ``U-duality'' more broadly (as is not uncommon, e.g. \cite{HohmSamtleben13a}, and in line with the common use of ``T-duality'') as shorthand for the ``hidden exceptional symmetry'' of toroidal compactifications of 11D SuGra. 
}) governed ultimately by the exceptional Kac-Moody Lie algebra  \footnote{
  Strictly all Lie algebras considered are split real forms, so that we omit the further notational decoration: $\mathfrak{sl}_n$ is short for $\mathfrak{sl}_n(\mathbb{R})$ and $\mathfrak{e}_{n}$ is short for $\mathfrak{e}_{n(n)}$.
} $\mathfrak{e}_{11}$ \cite{West01}, thought to exhibit much of the hidden structure of M-theory \cite{Nicolai99}\cite{deWitNicolai01}\cite{Nicolai24}. We will discuss general U-duality in the companion article \cite{GSS24-MGroup}, but here we first take a step back and revisit T-duality in view of M-theory.

\medskip

\noindent
{\bf Revisiting super-space T-duality.}
While T-duality is expected to be a symmetry already at the perturbative level, it carries  within the seed of the expected more general U-duality symmetries of M-theory: The expected $\mathrm{SL}(2,\mathbb{Z})$ S-duality symmetry of M/F-theory is the result of lifting T-duality on a single circle fiber to M-theory on a 2-torus fiber (e.g. \cite{Schwarz96}\cite{Johnson97}; we see the super-tangent space incarnation of this phenomenon below in Prop. \ref{SuperspaceSDuality}). At the same time, T-duality is by far the best-understood of the U-dualities: 

In its strong (rational-)topological formulation  (reviewed and in fact derived, rationally, in \S\ref{SuperspaceTDuality}, cf. Rem. \ref{TDualNSFlux}) 
T-duality is neatly  reflected in a twisted {\it Poincar{\'e} line bundle} with curvature {\it Poincar{\' e} 2-form} $P_2$ \eqref{PoincarFormInIntroduction}
on a ``doubled-'' (physics jargon) or ``correspondence-'' (math jargon) spacetime (cf. Rem. \ref{PoincareFormInLiterature}). For super-flux densities on super-tangent spaces, this statement was the main result of \cite{FSS18-TDualityA}, which below we review and extend in streamlined form. 

\medskip

\noindent
{\bf Superspace T-duality and M-theory.}
Our central observation here -- summarized in the big diagram \eqref{SummaryDiagram} below --
is that upon performing super-space T-duality on {\it all} 10 spacetime dimensions at once, then its correspondence-space super-geometry has a natural lift to M-theory, where
\begin{itemize}[
  leftmargin=1cm,
  topsep=2pt,
  itemsep=2pt
]

\item[\S\ref{LiftingToMTheory}] the doubled super-space is further extended 
\eqref{MAlgebraAsFiberProductOverDoubledSuperspace}
by the {\it M-algebra} (Def. \ref{BasicMAlgebra}) of extended 11D super-symmetry,

\item[\S\ref{ConclusionAndOutlook}] the Poincar{\'e} super 2-form is the reduction of a super 3-form $P_3$ 
\eqref{Poincare3Form}
on the M-algebra,
which may be identified with the ``decomposed'' M-theory 3-form \eqref{TheDecomposed3Form}
that had originally motivated the M-algebra, way back in \cite{DF82} (cf. \cite{AndrianopoliDAuria24}).
\end{itemize}

\noindent
We suggest this as further evidence that the M-algebra plays the role of the local model space for a duality-covariant completion of superspace supergravity, as argued in \cite{GSS24-E11} 
(following similar suggestions in \cite{West03}\cite{Vaula07}\cite{Bandos17} \cite[\S 4.6]{FSS20-HigherT}\cite{FSS20Exc}\cite{FSS21Exc}), further discussed in \cite{GSS24-MGroup}.

\smallskip 
Motivated by this result, we also revisit
(in \S\ref{HigherTduality})
the algebraic observation of rational higher T-duality on the M2-brane extended superspace \cite{FSS20-HigherT}\cite{SS18-HTDuality} and discuss its relation to the parity symmetry (the Ho{\v r}ava-Witten involution) of 11D supergravity.

\medskip

\noindent
{\bf Super-space T-duality redux.} But since the (higher) super-$L_\infty$ algebraic T-duality which supports these conclusions may not be widely appreciated, and since the previous reference \cite{FSS18-TDualityA} left some gaps and various issues untouched, the bulk of this article is a comprehensively completed and streamlined account of this topic, as we now survey.

\medskip

\noindent
{\bf Key role of super-flux avatars on super-tangent spaces.}
Remarkably, higher dimensional supergravity is an instance of \mbox{(super-)}Cartan geometry (cf. \cite[\S 7]{Sharpe97}\cite{Baaklini77-2}\cite[\S 2]{GSS24-SuGra}) in a strong higher sense: The ubiquitous {\it supergravity torsion constraints} (e.g., \cite{Lott90}\cite{Lott01}) say that the dynamics of supergravity super-fields on curved super-space $X^{1,d\,\vert\,\mathbf{N}}$ is largely controlled by the demand that the bifermionic components of higher super-flux densities restrict on each super-tangent space (the super-Kleinian model space, see Ex. \ref{SupersymmetryAlgebras})
\begin{equation}
  \label{LocalModelSpace}
  \underset{
    \mathclap{
      \adjustbox{
        scale=.7,
        raise=-2pt
      }{
        \color{gray}
        \def\arraystretch{.85}
        \begin{tabular}{c}
          super-tangent 
          \\
          super-space
        \end{tabular}
      }
    }
  }{
    T_x 
    X^{1,d\,\vert\,\mathbf{N}}
  }
    \;\;\;\;
    \simeq
    \;\;\;\; 
  \underset{
    \mathclap{
      \adjustbox{
        scale=.7,
        raise=-2pt
      }{
        \color{gray}
        \def\arraystretch{.85}
        \begin{tabular}{c}
          super-Kleinian
          \\
          model super-space
        \end{tabular}
      }
    }
  }{
  \mathbb{R}^{1,d\,\vert\,\mathbf{N}}
  }
  \;\;\;\;
  \simeq
  \;\;\;\;
  \underset{
    \mathclap{
      \adjustbox{
        scale=.7,
        raise=-2pt,
      }{
        \color{gray}
        \def\arraystretch{.85}
        \begin{tabular}{c}
          super-symmetry
          \\
          super-algebra
        \end{tabular}
      }
    }
  }{
  \mathfrak{iso}\big(
    \mathbb{R}^{1,d\,\vert\,\mathbf{N}}
  \big)
  }
  \,
  \big/
  \,
  \underset{
    \mathclap{
      \adjustbox{
        scale=.7,
        raise=-2pt
      }{
        \color{gray}
        \def\arraystretch{.85}
        \begin{tabular}{c}
          super-point
          \\
          algebra
        \end{tabular}
      }
    }
  }{
    \mathfrak{so}_{1,10}
  }
\end{equation}
to fixed super-invariant avatar forms. This is most pronounced for 11D SuGra (as was particularly highlighted by \cite{Howe97}):
Its typical super-tangent space $\mathbb{R}^{1,10\,\vert\,\mathbf{32}}$ with its canonical super-coframe field $\big((e^a)_{a=0}^{10},\,(\psi^\alpha)_{\alpha=1}^{32}\big)$
carries super-invariant avatars of the Hodge-duality symmetric C-field flux densities (\cite[(3.26)]{DF82}\cite[(2.27-28)]{NOF86}\cite[(6.6,10)]{CL94}\cite[(8.8)]{CdAIPB00}, cf. Ex. \ref{4SphereValuedSuperFlux}):
\begin{equation}
  \label{11dSuperFluxInIntroduction}
  \underset{
    \mathclap{
      \adjustbox{
        scale=.7,
        raise=-2pt
      }{
        \color{gray}
        Duality-symmetric
        avatar super-flux densities of
        \color{darkblue}
        11d SuGra
      }
    }
  }{
  \left.
  \def\arraystretch{1.5}
  \begin{array}{ccl}
    G_4
    &:=&
    \tfrac{1}{2}
    \big(\hspace{1pt}
      \overline{\psi}
        \,\Gamma_{a_1 a_2}\,
      \psi
    \big)
    e^{a_1} e^{a_2}
    \\
    G_7
    &:=&
    \tfrac{1}{5!}
    \big(\hspace{1pt}
      \overline{\psi}
        \,\Gamma_{a_1 \cdots a_5}\,
      \psi
    \big)
    e^{a_1} \cdots e^{a_5}
  \end{array}
 \!\! \right\}
  }
  \in
  \underset
  {
    \mathclap{
      \adjustbox{
        scale=.7,
        raise=-2pt
      }{
        \color{gray}
        \def\arraystretch{.9}
        \begin{tabular}{c}
          super left-invariant
          \\
          differential forms
        \end{tabular}
      }
    }
  }
  {
  \Omega^\bullet_{\mathrm{dR}}
  \big(
    \mathbb{R}^{1,10\,\vert\,\mathbf{32}}
  \big)^{\mathrm{li}}
  }
  \;\simeq\;
  \underset{
    \mathclap{
      \adjustbox{
        scale=.7,
        raise=-2pt
      }{
        \color{gray}
        \def\arraystretch{.9}
        \begin{tabular}{c}
          Chevalley-Eilenberg alg.
          \\
          of super-symmetry
        \end{tabular}
      }
    }
  }{
  \mathrm{CE}\big(
    \mathbb{R}^{1,10\,\vert\,\mathbf{32}}
  \big)
  }
  \,,
  \hspace{1cm}
  \underset{
    \mathclap{
      \adjustbox{
        scale=.7,
        raise=-2pt
      }{
        \color{gray}
        Bianchi identities
      }
    }
  }{
  \left\{\!\!
  \def\arraystretch{1.3}
  \begin{array}{ccl}
    \mathrm{d}
    \,
    G_4 &=& 0
    \\
    \mathrm{d}\,
    G_7 &=&
    \tfrac{1}{2}
    G_4\, G_4
    \mathrlap{\,,}
  \end{array}
  \right.
  }
\end{equation}
and the equations of motion of 11D SuGra on a supermanifold $X^{1,10\,\vert\,\mathbf{32}}$ are {\it equivalent} to the demand that this situation suitably globalizes to $X^{1,10\,\vert\,\mathbf{32}}$ (\cite[Thm. 3.1]{GSS24-SuGra} following \cite{BH80}\cite{CF80}\cite[\S III.8.5]{CDF91}).

\smallskip

Similarly for 10d type II supergravity, the NS\&RR-flux densities and their Bianchi identities  have super-invariant avatars on the respective super-tangent spaces $\mathbb{R}^{1,9\,\vert\,\mathbf{16} \oplus \overline{\mathbf{16}}}$ (type IIA, cf. Ex. \ref{11dAsExtensionFromIIA}) and $\mathbb{R}^{1,9\,\vert\,\mathbf{16} \oplus \mathbf{16}}$ (type IIB, cf. Ex. \ref{IIBasExtensionOf9d}), respectively, given \footnote{
  Our undecorated Clifford generators $(\Gamma^a)_{a = 0}^{\ten}$ are always those of $\mathrm{Pin}^+(1,10)$, reviewed in \S\ref{SpinorsIn11d}. 
  In particular, under the reduction $\mathrm{Spin}(1,9) \xhookrightarrow{\;} \mathrm{Pin}^+(1,10)$ the ``chirality operator'' often denoted ``$\Gamma_{\!1\!1}$'' (a reminiscence of the ancient tradition of writing ``$\gamma^5$'' for the chirality operator on Dirac spinors in 4d)
  is in our notation: $\Gamma_{\!\ten}$. This and further algebraic expressions of type II spinors in 10d in terms of Majorana spinors in 11D (immediate for type IIA and a little more subtle for type IIB) are discussed in \S\ref{SuperMinkowskiSpacetimes}.
}
for type IIA by (\cite[\S 6.1]{CdAIPB00}, cf. Prop. \ref{TheTypeIIACocyclesSummarized} below) by
\begin{equation}
  \label{IIASuperFluxInIntroduction}
  \underset{
    \mathclap{
      \adjustbox{
        scale=.7,
        raise=-2pt
      }{
        \color{gray}
        Duality-symmetric
        avatar
        super-flux densities of
        \color{darkblue}
        10d type IIA SuGra
      }
    }
  }{
  \left.
  \def\arraystretch{1.4}
  \begin{array}{ccl}
    H_3^A
    &:=&
    \big(\hspace{1pt}
      \overline{\psi}
      \,\Gamma_a\Gamma_{\!\ten}\,
      \psi
    \big)
    e^a
    \\
    F_{2\bullet}
    &:=&
    \tfrac{1}{(2k)!}
    \big(\hspace{1pt}
      \overline{\psi}
      \,\Gamma_{a_1 \cdots a_{2\bullet-2}}\,
      \psi
    \big)
    e^{a_1} \cdots e^{a_{2\bullet-2}}
  \end{array}
 \!\! \right\}
  }
  \,\in\,
  \mathrm{CE}\big(
    \mathbb{R}^{1,9\,\vert\,\mathbf{16} \oplus \overline{\mathbf{16}}}
  \big)
  \,,
  \hspace{.8cm}
  \underset{
    \mathclap{
      \adjustbox{
       scale=.7,
       raise=-3pt
      }{
        \color{gray}
        Bianchi identities
      }
    }  
  }{
  \left\{\!\!
  \def\arraystretch{1.3}
  \begin{array}{ccl}
    \mathrm{d}\, H_3^A
    &=&
    0
    \\
    \mathrm{d}\, F_{2k}
    &=&
    H_3^A\, F_{2\bullet-2}\,,
  \end{array}
  \right.
  }
\end{equation}
and for type IIB by (\cite[\S 2]{Sakaguchi00}, cf. Prop. \ref{TheTypeIIBCocyclesSummarized} below) by:
\begin{equation}
  \label{IIBSuperFluxInIntroduction}
  \underset{
    \mathclap{
      \adjustbox{
        raise=-2pt,
        scale=.7
      }{
        \color{gray}
        Duality-symmetric avatar super-flux densities of
        \color{darkblue}
        10d type IIB SuGra
      }
    }
  }{
  \left.
  \def\arraystretch{1.6}
  \def\arraycolsep{1pt}
  \begin{array}{ccl}
    H_3^B
    &:=&
    \big(\hspace{1pt}
      \overline{\psi}
      \,\Gamma_A^B\Gamma_{\!\ten}\,
      \psi
    \big)
    e^a
    \\
    F_{2\bullet+1}
    &=&
    \tfrac{1}{(2\bullet+1)!}
    \big(\hspace{1pt}
      \overline{\psi}
      \,\Gamma^B_{a_1} \cdots \Gamma^B_{a_{2\bullet-1}}
      \Gamma_9
      (\Gamma_{\!\ten})^{\bullet+1}
      \,
      \psi
    \big)
    e^{a_1} \cdots 
    e^{a_{2\bullet-1}}
  \end{array}
  \right\}
  }
  \in
  \mathrm{CE}\big(
    \mathbb{R}^{1,9\,\vert\,\mathbf{16} \oplus \overline{\mathbf{16}}}
  \big)
  \,,
  \hspace{5mm}
  \underset{
    \mathclap{
      \adjustbox{
       scale=.7,
       raise=-3pt
      }{
        \color{gray}
        Bianchi identities
      }
    }
  }{
  \left\{\!
  \def\arraystretch{1.3}
  \def\arraycolsep{1pt}
  \begin{array}{ccl}
    \mathrm{d}
    \,
    H_3^B
    &=&
    0
    \\
    \mathrm{d}
    \,
    F_{2\bullet+1}
    &=&
    H_3^B
    \,
    F_{2\bullet-1}\,.
  \end{array}
  \right.
  }
\end{equation}

This is a remarkable higher Cartan-geometric aspect of higher-dimensional supergravity (which is closely related to the point of view of \cite{DF82}\cite{CDF91}\cite{AndrianopoliDAuria24}): 
It means, as we discuss in a moment, that the global dynamics of global field configurations is tightly controlled by super-$L_\infty$ algebraic avatar structures on the typical super-tangent space \eqref{LocalModelSpace}. 
It is thus suggestive that also deeper structures of supergravity, and thereby of M-theory and its duality symmetries, are (rationally \cite{FSS19-QHigher}) reflected in and hence recognizable from the local super-space $L_\infty$-cocycle structure. \footnote{
  The issue of promoting $L_\infty$-algebra valued flux densities as in \eqref{11DFluxesAsCocyclesInIntroduction} to globally defined higher gauge fields is the topic of {\it flux quantization} \cite{SS24-Flux}\cite{GSS24-SuGra}\cite{GSS24-FluxOnM5} by which more subtle topological aspects of M-theory are resolved \cite{FSS20-H}\cite{GradySati21}\cite{FSS21Hopf}.
}

\smallskip

Among these is T-duality as such (surveyed now in \S\ref{DualityAsAdjunction}), and its geometric manifestation in ``doubled'' correspondence spaces (surveyed next in \S\ref{CorrespondencenSpaces}):

\subsection{Duality as Adjunction}
\label{DualityAsAdjunction}

At the heart of our discussion is the observation
(following \cite{FSS18-TDualityA}\cite{FSS18-TDualityB}\cite{BMSS19}) that T-duality between the type II super-flux avatars 
\eqref{TypeIIFluxesAsCocyclesInIntroduction}
is exhibited \eqref{TDualityDiagram}
by a fundamental ``adjunction'' (Prop. \ref{TheExtCycAdjunction}, the category theorist's term for ``duality'', cf. \cite{Corfield17}, gentle introduction in \cite{Sc16-Branes}\cite{Sc18-Toposes}) which neatly captures the phenomenon of {\it double-dimensional reduction} 
\footnote{
  The term ``double dimensional reduction'' was originally introduced, going back \cite{DHIS87}, 
  for the joint Kaluza-Klein reduction of a target spacetime and its brane worldvolume submanifolds  wrapping the fiber space.
 But together with the dimension of the branes also the flux densities that they couple to reduce in degree, by their integration over the fiber space, and it is the double dimensional reduction jointly of spacetime and its flux densities that we consider here.
}
\cite{DHIS87} and the reverse {\it oxidation} of flux densities.

\medskip

\noindent
{\bf Fluxes modulated by maps.}
What makes this work, in turn, is the observation that physical fields/fluxes are naturally represented by {\it maps} from spacetime into {\it classifying spaces} or generally {\it moduli stacks}. Schematically:
$$ 
\colorbox{lightgray}{$
 \left\{
  \adjustbox{}{\def\arraystretch{.9}\def\tabcolsep{1pt}\begin{tabular}{c}
      Configurations of
      \\
      field species $F$
      \\
      on spacetime $X$
    \end{tabular}
  }
  \right\}
  \hspace{.5cm}
  \simeq
  \hspace{.5cm}
  \left\{
  \adjustbox{}{\def\arraystretch{.9}\def\tabcolsep{1pt}\begin{tabular}{c}
     Maps from  $X$ to
      \\
      $F$-moduli stack 
    \end{tabular}
  }
  \!\right\}
  $}
$$
When understood properly, this holds in broad generality (cf. \cite{SS24-Flux}\cite{FSS23-Char}), but here we need it only for a comparatively simple case that is easily made concrete:
We take all (super-)spacetimes to be Minkowskian, $\mathbb{R}^{1,d\,\vert\,\mathbf{N}}$, but equipped with their translational super-symmetry algebra structure, and we also take the ``classifying spaces'' to be Lie algebras,
or rather: categorified symmetry Lie algebras called $L_\infty$-algebras (\S\ref{SuperLieAlgebras}, cf. \cite{FSS19-QHigher}\cite{Alfonsi24}).

\medskip

\noindent
{\bf Closed flux densities modulated by line Lie $n$-algebras.}
For instance, for each $n \in \mathbb{N}$ there is the {\it line Lie $n+1$-algebra} $b^n \mathbb{R}$ (Ex. \ref{LineLieNAlgebra}) which has a single generator in degree $n$ on which all brackets vanish. Simple as this is, it is readily seen to have the remarkable property that $L_\infty$-maps \eqref{LieHomomorphism} into it from any finite-dimensional super-Lie algebra $\mathfrak{g}$ are in natural bijection with the $n+1$-cocycles on $\mathfrak{g}$ in the classical sense of Lie algebra cohomology (cf. \cite{dAI95}), hence with the closed elements of the {\it Chevalley-Eilenberg algebra} $\mathrm{CE}(\mathfrak{g})$ \eqref{OrdinaryCEAlgebra}:
$$
  \Big\{\!\!
    \begin{tikzcd}
      \mathfrak{g}
      \ar[
        r,
        "{
          \omega_{n+1}
        }"
      ]
      &
      b^n \mathbb{R}
    \end{tikzcd}
 \!\! \Big\}
  \;\;
  \simeq
  \;\;
  \left\{\!\!
  \def\arraystretch{1.4}
  \begin{array}{c}
    \omega
    \,\in\,
    \wedge^{n+1}
    \mathfrak{g}^\ast
    \;\;\;
    \mbox{such that}
    \\
    \omega\big(
      [\mbox{-},\mbox{-}],
      \mbox{-},
      \cdots,
      \mbox{-}
    \big)
    -
    \omega\big(
      \mbox{-},
      [\mbox{-},\mbox{-}],
      \cdots,
      \mbox{-}
    \big)
    \,+\,
    \cdots
    \;=\;
    0
  \end{array}
  \!\!\right\}
  \;\;
  \simeq
  \;\;
  \left\{\!\!
    \def\arraystretch{1.2}
    \begin{array}{c}
      \omega \,\in\,
      \mathrm{CE}(\mathfrak{g})
      \;
      \mbox{s. t.}
      \\
      \mathrm{deg}(\omega) \,=\,
      (n+1,\mathrm{evn}),
      \\
      \mathrm{d}\, \omega_{n+1} \,=\, 0
    \end{array}
  \!\!\right\}
  .
$$
For instance, if $\mathfrak{g}$ is a semi-simple Lie algebra such as $\mathfrak{su}(2)$, equipped with a Killing form invariant polynomial $\langle\mbox{-},\mbox{-}\rangle$, then $\omega_3 \,:=\,\big\langle\mbox{-},[\mbox{-},\mbox{-}]\big\rangle$ is the 3-cocycle whose left-invariant extension to the corresponding Lie group $G$ is the B-field flux density for the WZW sigma-model describing the string propagating on $G$ (cf. \cite[\S 4.1]{FRS16}). 

\smallskip

The super-flux densities that we are concerned with here may be understood as higher-degree and super-Lie analogs of this WZW situation.
This was first understood (by \cite{Rabin87}\cite{dAT89}) for the closed flux densities like $G_4$ \eqref{11dSuperFluxInIntroduction} (constituting the WZ term of the M2-brane) and $H_3$ \eqref{IIASuperFluxInIntroduction} \eqref{IIBSuperFluxInIntroduction} (constituting the WZ term of the string), and then for general polynomial Bianchi identities (by \cite{FSS17}\cite{FSS15-M5WZW}) via generalization from ordinary {\it abelian} super-Lie cohomology to {\it twisted} and {\it non-abelian cohomology} (cf. \cite{FSS23-Char} and Def. \ref{RationalNonabelianTwistedCocycles} below):

\medskip

\noindent
{\bf Classifying $L_\infty$-algebras and characters.}
Namely, one finds that systems of species of flux densities $F^{(i)}$ satisfying  Bianchi identities whose right-hand sides are polynomials $P^{(i)}$ --- such as above in \eqref{11dSuperFluxInIntroduction}, \eqref{IIASuperFluxInIntroduction} and \eqref{IIBSuperFluxInIntroduction} --- equivalently define and are classified by $L_\infty$-algebras $\mathfrak{a}$ in that 
(cf. \cite{SS24-PhaseSpace})
$$
  \Big\{
  \def\arraystretch{1.7}
  \vec F
  \;\defneq\;
  \Big(
    F^{(i)}
    \,\in\,
    \Omega^{\mathrm{deg}_i}
    \big(
      \mathbb{R}^{1,d\,\vert\,\mathbf{N}}
    \big)
  \Big)_{i \in I}
  \;\Big\vert\;
  \mathrm{d}
  \, F^{(i)}
  \;=\;
  P^{(i)}\big(
    \vec F
  \big)
  \Big\}
  \hspace{.7cm}
  \Leftrightarrow
  \hspace{.7cm}
  \Big\{\!\!
  \begin{tikzcd}
    \mathbb{R}^{1,d\,\vert\,\mathbf{N}}
    \ar[
      r,
      "{
        \vec F
      }"
    ]
    &
    \mathfrak{a}
  \end{tikzcd}
 \!\! \Big\}
  \,.
$$
Here, the structure constants of the $n$-ary bracket of $\mathfrak{a}$ are identified with the coefficients of the $n$-ary monomials in $\vec P$.

Such  (classifying) $L_\infty$-algebras are not just analogous to (classifying) spaces: The fundamental theorem of dg-algebraic rational homotopy theory 
(reviewed in \cite[\S II]{FSS23-Char})
shows 
(under mild technical conditions \cite[Def. 5.1]{FSS23-Char} which need not concern us here)
that they are equivalent incarnations of the ``rational homotopy type'' underlying topological spaces. 

For us, this ``rationalization'' simply means (cf. \cite[p 19]{SS24-Flux}) {\it discarding torsion effects} in the would-be flux-quantized higher gauge fields, which are not reflected in their flux densities alone: As we regard any space $\mathcal{A}$ as being the classifying space for a generalized non-abelian cohomology theory \cite[Def. 6]{Lurie14}\cite[\S 2]{FSS23-Char}
$$
  H^1\big(
    X;\, \Omega\mathcal{A}
  \big)
  \;\;
  \defneq
  \;\;
  \pi_0
  \,
  \mathrm{Maps}\big(
    X
    ,\,
    \mathcal{A}
  \big)
$$
there is an $L_\infty$-algebra $\mathfrak{a} \,:=\,\mathfrak{l}\mathcal{A}$ ---  its {\it real Whitehead bracket $L_\infty$-algebra} (Ex. \ref{WhiteheadLInfinityAlgebra}) --- which reflects the non-torsion part in this cohomology, also called its {\it character}, generalizing the familiar Chern character on K-theory:

\smallskip

Namely, the cohomology theory known as complex topological K-theory $K(\mbox{-})$ has a classifying space $\mathrm{KU}$ (really: a spectrum of spaces), in that
$$
  K(X)
  \;\simeq\;
  H^1(X;\, \Omega \mathrm{KU})
  \;\defneq\;
  \pi_0
  \mathrm{Maps}\big(
    X,\,
    \mathrm{KU}
  \big)
$$
and the corresponding Whitehead $L_\infty$-algebra is the product of line Lie $n$-algebras as above, in every second degree
$$
  \mathfrak{l}
  (
  \mathrm{KU}
  )
  \;\;
  \simeq
  \;\;
  \prod_{k \in \mathbb{Z}}
  b^{2n+1} \mathbb{R}
  \,,
$$
reflecting the fact that the Chern character of a K-theory class is a sum of closed differential forms $F_{2\bullet}$ in every even degree (cf. \cite[Ex. 7.2]{FSS23-Char}).

This would describe the type IIA RR-flux densities from \eqref{IIASuperFluxInIntroduction} above if the NS B-field flux density $H_3^A$ were vanishing. That not being the case, there is a richer classifying space, which we may denote $\mathrm{KU} \sslash B\mathrm{U}(1)$, classifying 3-twisted K-theory (cf. \cite[Prop. 10.1]{FSS23-Char}), and its Whitehead $L_\infty$-algebra is the correct coefficient for classifying the type II flux densities (Prop. \ref{TheTypeIIACocyclesSummarized} and Prop. \ref{TheTypeIIBCocyclesSummarized} below, following \cite{FSS17}\cite{BMSS19}):
\begin{equation}
  \label{TypeIIFluxesAsCocyclesInIntroduction}
  \def\arraycolsep{11pt}
  \begin{array}{ccc}
  \mbox{
    \def\arraystretch{1}
    \def\tabcolsep{0pt}
    \begin{tabular}{c}
      $\big(H_3^{\mathcolor{purple}{A}}, (F_{2\bullet})\big)$
      \\
    \small  satisfying     Bianchi in \eqref{IIASuperFluxInIntroduction}
    \end{tabular}
  }
  &
  \Leftrightarrow
  &
  \begin{tikzcd}
  \mathbb{R}^{
    1,9\,\vert\,
    \mathbf{16}
    \oplus
    \overline{\mathbf{16}}
  }
  \ar[
    rr,
    "{
      \scalebox{.8}{$
        \big(
          H_3^{\mathcolor{purple}{A}},\,
          (F_{2\bullet})
        \big)
      $}
    }"
  ]
  &&
  \mathfrak{l}\big(
    \mathrm{KU}
    \!\sslash\!
    B\mathrm{U}(1)
  \big)
  \mathrlap{\,,}
  \end{tikzcd}
  \\[-5pt]
  \\
  \mbox{
    \def\arraystretch{1}
    \def\tabcolsep{0pt}
    \begin{tabular}{c}
      $\big(H_3^{\mathcolor{purple}{B}}, (F_{2\bullet \mathcolor{purple}{+1}})\big)$
         \\
       \small   satisfying  Bianchi in \eqref{IIBSuperFluxInIntroduction}
    \end{tabular}
  }
  &
  \Leftrightarrow
  &
  \begin{tikzcd}
  \mathbb{R}^{
    1,9\,\vert\,
    \mathbf{16}
    \oplus
    \mathbf{16}
  }
  \ar[
    rr,
    "{
      \scalebox{.8}{$
        \big(
          H_3^{\mathcolor{purple}{B}},\,
          (F_{2\bullet} \mathcolor{purple}{+1})
        \big)
      $}
    }"
  ]
  &&
  \mathfrak{l}\big(
    \mathcolor{purple}{\Sigma}
    \mathrm{KU}
    \!\sslash\!
    B\mathrm{U}(1)
  \big)
  \mathrlap{\,.}
  \end{tikzcd}
  \end{array}
\end{equation}
Similarly, for 11D SuGra the $L_\infty$-coefficients are those of 4-Cohomotopy (Ex. \ref{4SphereValuedSuperFlux} below) whose classifying space is the 4-sphere (\cite[\S 2.5]{Sati13}\cite{FSS17}):
\begin{equation}
  \label{11DFluxesAsCocyclesInIntroduction}
  \mbox{
    \def\arraystretch{1}
    \begin{tabular}{c}
      $(G_4, G_7)$
      \\
      \small
      satisfying  
      Bianchi in \eqref{11dSuperFluxInIntroduction}
    \end{tabular}
  }
  \hspace{1cm}
  \Leftrightarrow
  \hspace{1cm}
  \begin{tikzcd}
    \mathbb{R}^{1,10\,\vert\,\mathbf{32}}
    \ar[
      rr,
      "{
        (G_4, G_7)
      }"
    ]
    &&
    \mathfrak{l}S^4
    \mathrlap{\,.}
  \end{tikzcd}
\end{equation}

\medskip

\noindent
{\bf Double dimensional reduction via cyclification.}
In this language of classifying symmetric flux densities by $L_\infty$-maps, it now turns out that double dimensional reduction/oxidation is exhibited by the above-mentioned adjunction of Prop. \ref{TheExtCycAdjunction}.
If one super-Lie algebra $\widehat{\mathfrak{g}}$ is a central $\mathbb{R}$-extension of some $\mathfrak{g}$ via some 2-cocycle $c_1$  --- such as 10D super-spacetime is of 9D super-spacetime (Ex. \ref{IIAasExtensionOf9d}) --- then $\mathfrak{a}$-valued flux densities on the extension are in natural bijection with $\mathrm{cyc}(\mathfrak{a})$-valued flux densities on the base space respecting the extending cocycle
$$
  \Big\{\!
  \begin{tikzcd}
    \widehat{\mathfrak{g}}
    \ar[
      rr,
      "{
        \vec F
      }"
    ]
    &&
    \mathfrak{a}
  \end{tikzcd}
  \!\Big\}
  \begin{tikzcd}
    \ar[
      rr,
      shift left=6pt,
      "{
        \scalebox{.7}{
          \color{darkgreen}
          \bf
          reduction
        }
      }"
    ]
    \ar[
      rr,
      phantom,
      "{ \sim }"
    ]
    \ar[
      rr,
      <-,
      shift right=5pt,
      "{
        \scalebox{.7}{
          \color{darkgreen}
          \bf
          oxidation
        }
      }"{swap}
    ]
    &&
    {}
  \end{tikzcd}
  \Bigg\{\!\!
  \begin{tikzcd}[
    row sep=-2pt, column sep=large
  ]
    \mathfrak{g}
    \ar[
      rr,
      "{
        \mathrm{red}_{c_1}(\vec F)
      }"
    ]
    \ar[
      dr,
      "{ c_1 }"{swap}
    ]
    &&
    \mathrm{cyc}(\mathfrak{a})
    \ar[
      dl,
      "{ \omega_2 }"
    ]
    \\
    &
    b \mathbb{R}
  \end{tikzcd}
  \!\! \Bigg\}
  \,,
$$
where {\it cyclification} $\mathrm{cyc}(-)$ 
(Def. \ref{Cyclification})
is a universal operation that adjusts classifying $L_\infty$-algebras such as to classify not just the original flux densities but also the degree-reduction of their ``winding modes'' (here: infinitesimally), thereby retaining enough information to revert the process (whence ``oxidation''): the oxidized flux densities on the left above are ``in duality'' to the reduced flux densities on the right.

\medskip

\noindent
{\bf T-Duality via adjunction.}
But note that given just flux data in lower dimensions, the nature of the re-oxidation depends on what we take its coefficient $L_\infty$-algebra to be the cyclification of: Namely it may happen that $L_\infty$-algebras $\mathfrak{a}, \widetilde{\mathfrak{a}}$ have isomorphic cyclifications
$$
  \begin{tikzcd}
    \mathrm{cyc}(\mathfrak{a})
    \ar[
      r, 
      "{
        \color{purple}{T}
      }",
      "{ \sim }"{swap}
    ]
    &
    \mathrm{cyc}(\,\widetilde{\mathfrak{a}}\,)
    \,.
  \end{tikzcd}
$$
When that is the case, then $\mathfrak{a}$-valued flux densities on the $c_1$-extension of $\mathfrak{g}$ become dual to $\widetilde{\mathfrak{a}}$-valued flux densities on the $\widetilde{c}_1$-extension of $\mathfrak{g}$

\vspace{-.4cm}
$$
  \big\{\!\!
    \begin{tikzcd}
      \widehat{\mathfrak{g}}
      \ar[r]
      &
      \mathfrak{a}
    \end{tikzcd}
 \!\! \big\}
  \begin{tikzcd}[column sep=35pt]
    \ar[
      r,
      "{
        \scalebox{.7}{
          \color{darkgreen}
          \bf
          reduction
        }
      }"{yshift=2pt}
    ]
    &
    {}
  \end{tikzcd}
  \left\{\!\!\!\!
    \adjustbox{
      raise=5pt
    }{
    \begin{tikzcd}[
      column sep=12pt,
      row sep=3pt
    ]
      \mathfrak{g}
      \ar[
        dr,
        "{
          c_1
        }"{swap}
      ]
      \ar[
        rr
      ]
      &&
      \mathrm{cyc}(\mathfrak{a})
      \ar[
        dl,
        "{
          \omega_2
        }"
      ]
      \\
      &
      b\mathbb{R}
    \end{tikzcd}
    }
  \!\!\!\!\!\right\}
  \begin{tikzcd}
    \ar[
      r,
      "{
        \color{purple}{T}
      }",
      "{
        \sim
      }"{swap}
    ]
    &
    {}
  \end{tikzcd}
  \left\{\!\!\!\!
    \adjustbox{
      raise=5pt
    }{
    \begin{tikzcd}[
      column sep=12pt,
      row sep=3pt
    ]
      \mathfrak{g}
      \ar[
        dr,
        "{
          \widetilde{c}_1
        }"{swap}
      ]
      \ar[
        rr
      ]
      &&
      \mathrm{cyc}(\hspace{.5pt}\widetilde{\mathfrak{a}}\hspace{.5pt})
      \ar[
        dl,
        "{
          \widetilde{\omega}_2
        }"
      ]
      \\
      &
      b\mathbb{R}
    \end{tikzcd}
    }
  \!\!\!\!\!\right\}
  \begin{tikzcd}[column sep=45pt]
    \ar[
      r,
      "{
        \scalebox{.7}{
          \color{darkgreen}
          \bf
          re-oxidation
        }
      }"{yshift=2pt}
    ]
    &
    {}
  \end{tikzcd}
  \big\{\!
    \begin{tikzcd}
      \widetilde{\widehat{\mathfrak{g}}}
      \ar[r]
      &
      \widetilde{\mathfrak{a}}
    \end{tikzcd}
  \! \big\}
  \,.
$$
Remarkably, the avatar super-flux densities on type IIA- and type IIB super-spacetime \eqref{TypeIIFluxesAsCocyclesInIntroduction} are in correspondence by exactly such a super-$L_\infty$ algebraic T-duality! This result of \cite{FSS18-TDualityA}, which relies on a somewhat remarkable conspiracy of Clifford algebra coefficients \eqref{ObtainingTheIIBfluxes}, we discuss in streamlined form in \S\ref{SuperspaceTDualityI}.

\smallskip

With such a systematic understanding of T-duality of avatar super-flux densities in hand, it becomes now possible to investigate more complicated situation in a systematic manner. 

\smallskip 
First of all, we can now completely {\it classify} the choices of the above isomorphism ``$T$'' and find  (Lem. \ref{FurtherIsomorphismsOfCyclifiedTwistedKspectra}) that these correspond to choices of signs (Rem. \ref{RedundancyOfExtraIsomorphismsOfCyclifiedTwistedKSpectra}) in the dual RR-flux densities (a subtlety relevant for instance in the discussion of T-folds, where it may not have found due attention yet).

\medskip

\noindent
{\bf Higher-dimensional toroidal T-duality.}
Our main application here is the systematic investigation of the generalization of this super-$L_\infty$ algebraic duality, via adjunction, to T-duality along higher dimensional torus fibers (which was actually left open in \cite{FSS20-HigherT}, where only the corresponding Fourier-Mukai picture was considered). This is what occupies us in \S\ref{TorusReduction} in generality, and which we apply to the avatar super-flux densities in \S\ref{LiftingToMTheory} where we also recover the aforementioned Fourier-Mukai picture in Prop. \ref{ToroidalTDuality/FourierMukaiIsomorphism} (completing the details missing in the proof from \cite{FSS20-HigherT}).

\smallskip 
In particular, it is now readily possible to consider this kind of T-duality along {\it all} space-time dimensions at once, since the formalism applies to the temporal direction just as well (recalling here that the $L_\infty$-algebraic description we give applies to the local fluxes on each super-tangent space and hence does not actually require a global compactification of the time direction!).

\smallskip 
In this case it turns out that the T-dual ``type $\mathrm{II}\widetilde{\mathrm{A}}$''-spacetime (Def. \ref{TheFullyTDualSuperSpacetime}) consists, in a precise sense, entirely of the string-charges of the original IIA spacetime (which sounds just like what one expects from the physics picture but would  be hard to make precise by conventional means), in such a way that the corresponding ``doubled'' super-spacetime (Def. \ref{FullyDoubleSuperSpacetime}) is identified (Rem. \ref{FullyExtendedIIAAsFiberProduct}), as a super-Lie algebra, with the corresponding sub-algebra of the fully brane-extended IIA super-symmetry algebra (Def. \ref{FullyExtebdedIIASpacetime}), which is close to the {\it M-algebra} (Def. \ref{BasicMAlgebra}). This establishes a relation between T-dual spacetimes and brane-extended super-symmetry that may not previously have found attention -- and it opens the advertised connection to exceptional M-geometry, which we survey next in \S\ref{CorrespondencenSpaces}.

\medskip

\noindent
{\bf Higher-categorical T-duality.}
Another relation to T-duality seen within M-theory arises from considering higher categorical symmetry not just in the classifying $L_\infty$-algebras $\mathfrak{a}$, but already on the super-spacetime domains.
In one perspective or another, it is an old observation \cite{DF82}\cite{CDF91}\cite{CdAIPB00}\cite{Azcarraga05}\cite{FSS15-HigherWZW} that closed super-WZW-terms 
\footnote{
  Here by a ``WZW term'' we refer to a closed $p+2$-form $F_{p+2}$ on (super-)spacetime $X$ such that the Lagrangian density for a $p$-brane worldvolume $\Sigma^{p+1} \xrightarrow{\phi} U \xhookrightarrow{\iota_U} X$ has {\it locally}, on a chart $\iota_U$,
  a contribution proportional to $\phi^\ast A_{p+1}$, where $\mathrm{d} A_{p+1} = \iota_U F_{p+2}$. Alternative terminology is to refer to $A_{p+1}$ as the WZW term. When seen in {\it differential cohomology} these two perspectives unify: $F_{p+2}$ is the curvature and is $A_{p+1}$ a local connection datum of a single differential cocycle which would be the WZW term proper.
}
induce higher $L_\infty$-extensions of super-spacetime (see Def. \ref{HigherCentralExtension} below): 
\footnote{
  The reader with more tolerance for higher structures may understand this as follows: One choice of flux-quantization of a {\it closed} WZW term of degree $p+2$ is \cite[Ex. 3.10]{SS24-Flux} by a ``$\mathrm{U}(1)$-bundle $p$-gerbe'' over the underlying spacetime (in higher generalization of the familiar situation that the WZW-term for the electron, namely the Faraday tensor $F_2$, is quantized by an ordinary $U(1)$-bundle), and the corresponding higher-Lie extension of spacetime is the Whitehead $L_\infty$-algebra 
  (Ex. \ref{WhiteheadLInfinityAlgebra})
  of the {\it total space} of that bundle $p$-gerbe.
}

\smallskip 
\noindent For instance, the WZW term of the Green-Schwarz super-string induces a super-Lie 2-algebra extension 
$\mathfrak{string}_{\mathrm{IIA}} \! \twoheadrightarrow \mathbb{R}^{1,9\,\vert\,\mathbf{16} \oplus \overline{\mathbf{16}}}$ 
(Ex. \ref{StringExtendedSuperSpace})
on which in turn the WZW-terms of the super-D-branes are supported, and the WZW-term of the super M2-brane induces a super-Lie 3-algebra extension $\mathfrak{m}2\mathfrak{brane} \twoheadrightarrow \mathbb{R}^{1,10\,\vert\, \mathbf{32}}$ (Ex. \ref{TheM2braneExtensionOf11D}) on which in turn the WZW-term of the M5-brane is supported. Consecutive maximal invariant such higher extensions form a {\it bouquet} (see \cite[Fig. 1]{FSS19-QHigher}) which reflects much of the broad  structure of M-theory entirely in super-$L_\infty$ theory (namely: rationally):

\vspace{-4mm}
\begin{center}
\begin{tikzpicture}
\node at (0,0) {
\begin{tikzcd}[
  row sep=16.8pt
]
  &[-12pt]
  &[-24pt]
  &[-15pt]
  \mathfrak{m}5\mathfrak{brane}
  \ar[
    d,
    ->>
  ]
  &[-16pt]
  &[-20pt]
  &[-12pt]
  \\
  &
  &
  &
  \mathfrak{m}2\mathfrak{brane}
  \ar[
    dd,
    ->>
  ]
  \\
  \mathfrak{d}5\mathfrak{brane}
  \ar[
    ddr,
    ->>
  ]
  &
  \mathfrak{d}3\mathfrak{brane}
  \ar[
    dd,
    ->>
  ]
  &
  \mathfrak{d}1\mathfrak{brane}
  \ar[
    ddl,
    ->>
  ]
  &&
  \mathfrak{d}0\mathfrak{brane}
  \ar[
    ddr,
    ->>
  ]
  \ar[
    dl,
    dotted
  ]
  &
  \mathfrak{d}2\mathfrak{brane}
  \ar[
    dd,
    ->>
  ]
  &
  \mathfrak{d}4\mathfrak{brane}
  \ar[
    ddl,
    ->>
  ]
  \\
  \mathfrak{d}7\mathfrak{brane}
  \ar[
    dr,
    ->>
  ]
  &
  &&
  \mathbb{R}^{1,10\,\vert\,\mathbf{32}}
  \ar[
    ddr,
    ->>
  ]
  &&&
  \mathfrak{d}6\mathfrak{brane}
  \ar[
    dl,
    ->>
  ]
  \\
  \mathfrak{d}9\mathfrak{brane}
  \ar[
    r,
    ->>
  ]
  &
  \mathfrak{string}_{\mathrm{IIB}}
  \ar[
    dr,
    ->>
  ]
  &&&&
  \mathfrak{string}_{\mathrm{IIA}}
  \ar[
    dl,
    ->>
  ]
  &
  \mathfrak{d}8\mathfrak{brane}\
  \ar[
    l,
    ->>
  ]
  \\
  &&
  \mathbb{R}^{1,9\,\vert\,\mathbf{16} \oplus \mathbf{16}}
  &
  \mathbb{R}^{1,9\,\vert\,\mathbf{16}}
  \ar[
   r,
   shift left=3.7pt
  ]
  \ar[
   r,
   shift right=3.7pt
  ]
  \ar[
   l,
   shift left=3.7pt
  ]
  \ar[
   l,
   shift right=3.7pt
  ]
  \ar[
    dr,
    ->>
  ]
  &
  \mathbb{R}^{1,9\,\vert\,\mathbf{16} \oplus \overline{\mathbf{16}}}
  \\
  &&&
  \mathbb{R}^{1,5\,\vert\,\mathbf{8}}
  \ar[
    dl,
    ->>
  ]
  \ar[
   r,
   shift left=3.7pt
  ]
  \ar[
   r,
   shift right=3.7pt
  ]
  &
  \mathbb{R}^{1,5\,\vert\,\mathbf{8}\oplus\overline{\mathbf{8}}}
  \\
  &&
  \mathbb{R}^{1,3\,\vert\,\mathbf{4}\oplus\mathbf{4}}
  &
  \mathbb{R}^{1,3\,\vert\,\mathbf{4}}
  \ar[
    dl,
    ->>
  ]
  \ar[
   l,
   shift left=3.7pt
  ]
  \ar[
   l,
   shift right=3.7pt
  ]
  \\
  &&
  \mathbb{R}^{1,2\,\vert\,\mathbf{2}\oplus\mathbf{2}}
  &
  \mathbb{R}^{1,2\,\vert\,\mathbf{2}}
  \ar[
    dl, 
    ->>
  ]
  \ar[
   l,
   shift left=3.7pt
  ]
  \ar[
   l,
   shift right=3.7pt
  ]
  \\
  &&
  \mathbb{R}^{0\,\vert\,\mathbf{1}\oplus\mathbf{1}}
  &
  \mathbb{R}^{0\,\vert\,\mathbf{1}}
  \ar[
   l,
   shift left=3.7pt,
   gray
  ]
  \ar[
   l,
   shift right=3.7pt
  ]
  \ar[
    uuu,
    -Latex,
    shift right=33pt,
    shorten=13pt,
    gray,
    "{
      \scalebox{.7}{
        \color{darkgreen}
        \bf
        \def\arraystretch{.8}
        \begin{tabular}{c}
          (higher)
          \\
          invariant
          \\
          super-Lie
          \\
          extensions
        \end{tabular}
      }
    }"{swap, xshift=-8pt}
  ]
\end{tikzcd}
};
\begin{scope}[
  yscale=.88
]
\draw[
  draw opacity=0,
  fill=olive,
  fill opacity=.3
]
 (-1.1,6.3) rectangle 
 (+.8,-6.4);
\node
  at (.7,-5.8) {
      \adjustbox{
      raise=2pt,
      scale=.7
    }{
      \color{darkblue}
      \bf
      \def\arraycolsep{0pt}
      \def\arraystretch{.86}
      \begin{tabular}{c}
        super-
        \\
        point
      \end{tabular}
    }
  };
\node
  at (1.3,2.1) {
      \adjustbox{
      raise=2pt,
      scale=.7
    }{
      \color{darkblue}
      \bf
      \def\arraycolsep{0pt}
      \def\arraystretch{.86}
      \begin{tabular}{c}
        11D super-
        \\
        spacetime
      \end{tabular}
    }
  };
\node
  at (1.1,5.2) {
      \adjustbox{
      raise=2pt,
      scale=.7
    }{
      \color{darkblue}
      \bf
      \def\arraycolsep{0pt}
      \def\arraystretch{.86}
      \begin{tabular}{c}
        super
        \\
        M-branes
      \end{tabular}
    }
  };
\end{scope}
\end{tikzpicture}
\end{center}
\vspace{-.2cm}

Now, a surprising and tantalizing fact made visible by super-$L_\infty$ algebraic T-duality is that \cite[\S 4.3]{FSS20-HigherT}\cite{SS18-HTDuality}: 

\vspace{-.2cm}
\begin{center}
\fbox{\parbox{0.95\textwidth}
{\it  The 7-flux WZW term \eqref{PageCharge7Cocycle} for the M5-brane has relative to its base super-space $\mathfrak{m}2\mathfrak{brane}$ formally the same structure as the 3-flux WZW term \eqref{TDualityOnNSFlux}
for the string has relative to ordinary  10D super-spacetime, inducing a higher categorical duality of the former data, directly analogous in a precise sense to ordinary T-duality.}
}
\end{center}
\vspace{-.1cm}

In \S\ref{HigherExtensions}, we develop the adjuction point of view of this ``higher'' duality (which has not been previously considered), via the notion of a higher (odd) spherification (Def. \ref{HigherCyclification}), showing that this too stems from an isomorphism between the higher spherification of two (different) classifying $L_\infty$-algebras (Thm. \ref{HigherTwistedCocyclesRedIsoReOxi}). 

\smallskip 
We apply this to the case of the super-space being the $\mathfrak{m}2\mathfrak{brane}$ higher-extension of $\mathbb{R}^{1,10\vert \mathbf{32}}$ in \S\ref{HigherTduality}. What remains elusive in this analogy is the higher analog on $\mathfrak{m}2\mathfrak{brane}$ of the RR-flux densities \eqref{IIASuperFluxInIntroduction} \eqref{IIBSuperFluxInIntroduction} on 10D super-space, though a physics argument suggests \cite{Sa-Higher}\cite[Rem. 4.18]{FSS20-HigherT} that this role ought to be played by the (similarly elusive) M-theory lift of the gauge-flux density and its dual 10-form in heterotic string theory. While we do not fully resolve this puzzle here, we add the new observation (Rem. \ref{RelationToParity}) that the higher categorical T-duality on $\mathfrak{m}2\mathfrak{brane}$ involves exactly the ``parity symmetry'' (Ex. \ref{ParityIsomorphism}) of the C-field which controls the restriction of heterotic M-theory to the MO9 boundary on which the heterotic flux must appear. Along the way, we also show how this recovers the (equivalent) higher Fourier-Mukai picture of \cite[\S 4.3]{FSS20-HigherT}.

\subsection{Correspondence spaces}
\label{CorrespondencenSpaces}

While the above duality-by-adjunction is great for handling all manner of T-duality in 10D, rationally, its would-be lift to M-theory is at best subtle, due to the fact that the high-degree RR-flux densities $F_{> 6}$ in \eqref{IIASuperFluxInIntroduction} \eqref{IIBSuperFluxInIntroduction} do not and cannot (already by degree reasons) arise from double-dimensional reduction of the C-field flux densities \eqref{11dSuperFluxInIntroduction}. While we have shown in \cite{BMSS19}
how these ``missing'' flux densities do arise after {\it fiberwise stabilization} of the  C-field flux (which is a plausible algebro-topological reflection of a perturbative approximation of the flux structure in M-theory), we do not attempt here to understand how this stabilization meshes with T-duality. 

\smallskip 
The reason is the observation that super-$L_\infty$ algebraic T-duality in 10D is alternatively controlled by a ``Poincar{\'e} super 2-form'' $P_2$ on a doubled correspondence super-spacetime, and {\it this} data has a transparent M-theoretic lift, to a Poincar{\'e} 3-form on the afore-mentioned M-algebra (Def. \ref{BasicMAlgebra}). This is what we survey now.

\smallskip

(The passage from the above classifying maps for flux densities to structures on extended and/or doubled super-spacetime is implemented by passage to {\it homotopy fibers} of the former, cf. \cite[p 35]{FSS18-TDualityA}.)

\medskip

\noindent
{\bf Doubled tangent super-spacetime and Poincar{\'e} super-form.} 
To start with, the {\it fiber product} \eqref{CommutingPairOfExtensions} of the above IIA- and IIB- superspaces,
formed over the type II super-tangent space of 9d supergravity,
is the ``correspondence super-space'' or {\it doubled super-space} 
(Def. \ref{1DoubledSuperSpace}, as in \cite{HKS14}\cite{Bandos15}\cite{Cederwall16}\cite[\S 6]{FSS18-TDualityA}):
$$
  \begin{tikzcd}[
    row sep=-7pt,
    column sep=30pt
  ]
    &    
    \quad
    \mathbb{R}^{
      1,9\,\vert\,
      \mathbf{16}
      \oplus
      \overline{\mathbf{16}}
    }
    \underset{
      \mathclap{
        \adjustbox{
          raise=-3pt,
          scale=.7
        }{$
         \mathbb{R}^{
           1,8\vert
           \mathbf{16} \oplus \mathbf{16}
         }
        $}
      }
    }{\;\times\;}
    \mathbb{R}^{
      1,9\,\vert\,
      \mathbf{16}
      \oplus
      \mathbf{16}
    }
    \quad 
    \ar[
      dl,
      shorten <=-25pt,
      "{ \pi_A }"{swap, pos=-.1}
    ]
    \ar[
      dr,
      shorten <=-25pt,
      "{ \pi_B }"{pos=-.1}
    ]
    \\
    \mathllap{
      \scalebox{.7}{
        \color{darkblue}
        \bf
        \def\arraystretch{.9}
        \begin{tabular}{c}
          10D type IIA
          \\
          super-space
        \end{tabular}
      }
    }
    \mathbb{R}^{
      1,9\,\vert\,
      \mathbf{16}
      \oplus
      \overline{\mathbf{16}}
    }
    \ar[
     dr,
     ->>
    ]
    &&
    \mathbb{R}^{
      1,9\,\vert\,
      \mathbf{16}
      \oplus
      \mathbf{16}
    }
    \ar[
      dl,
      ->>
    ]
    \mathrlap{
      \scalebox{.7}{
        \color{darkblue}
        \bf
        \def\arraystretch{.9}
        \begin{tabular}{c}
          10D type IIB
          \\
          super-space
        \end{tabular}
      }
    }
    \\
    &
    \underset{
    \mathclap{
      \scalebox{.7}{
        \color{darkblue}
        \bf
        \def\arraystretch{.9}
        \begin{tabular}{c}
          9D type II
          \\
          super-space
        \end{tabular}
      }
    }      
    }{
    \mathbb{R}^{
      1,8\,\vert\,
      \mathbf{16}
      \oplus
      \mathbf{16}
    }
    }
  \end{tikzcd}
$$

\vspace{-2mm} 
\noindent This doubled super-spacetime carries a  further component $\tilde e^9$ to its coframe field, reflecting the string winding charges along the T-dualized coordinate axis (here: the 9th) and thus manifestly putting them on the same footing as spacetime dimensions. The wedge product of this with the original coframe field component in this direction, $e^9$, 
plays a significant role, as it may be identified with the local super-form version of what in topological T-duality is known as the twisted {\it Poincar{\'e} form} (Rem. \ref{PoincareFormInLiterature})
on the correspondence space. This is a coboundary for the difference of the NS super 3-flux densities of type IIA and IIB (Prop. \ref{PoincarFormTrivializesDifferenceOfPullbacksOfNSFluxes}):
\begin{equation}
  \label{PoincarFormInIntroduction}
  \underset{
    \adjustbox{
      scale=.7,
      raise=-2pt
    }
    {
      \color{gray}
          Twisted  
          Poincar{\'e} 
          super 2-form
          for
        \color{darkblue}
        doubled 10D type II SuGra
      }  
  }{
  P_2 
  \;:=\;
   e^9_B 
     \, 
   e^9_A 
  \;\;
  \in
  \mathrm{CE}\Big(
    \mathbb{R}^{1,9\,\vert\,
      \mathbf{16}
      \oplus
      \overline{\mathbf{16}}
    }
    \underset{
      \mathclap{
      \adjustbox{
        scale=.7,
        raise=-3pt
      }{$
      \mathbb{R}^{
       1,8\vert
       \mathbf{16} \oplus
       \mathbf{16}
      }
      $}
      }
    }{\;\times\;}
    \mathbb{R}^{1,9\,\vert\,
      \mathbf{16}
      \oplus
      \mathbf{16}
    }    
  \Big)
  }
  \,,
  \hspace{1cm}
  \underset{
    \adjustbox{
      scale=.7,
      raise=-9pt
    }{
     \color{gray}
      Bianchi identity
    }
  }{
  \mathrm{d}
  \,
  P_2
  \;=\;
  \pi_A^\ast H_3^A
  -
  \pi_B^\ast H_3^B
  }
  .
\end{equation}
In fact, the super-Poincar{\'e} form exhibits all of super-flux T-duality, in that it serves as the ``integral kernel'' of a Fourier-Mukai transformation (pp. \pageref{FourierMukaiTransform})
taking the RR super-flux densities of IIA and IIB  into each other (also known as {\it Hori's formula} \cite[(1.1)]{Hori99}, re-derived for the super-flux avatars in  \cite[\S 6]{FSS18-TDualityA} and Cor. \ref{PullPushViaAutomorphismOfCyclifiedCoefficients} below).

\medskip

\noindent
{\bf Fully doubled super-spacetime and full Poincar{\'e} form.} With the further super-$L_\infty$ algebraic machinery of {\it toroidification} (\S\ref{TorusReduction}, following \cite{SV23-Mysterious}\cite{SV24-Mysterious}) we may analyze the analogous super-space T-duality but for reduction along all 10 space-time dimensions down to the (super-)point (Ex. \ref{SuperMinkowskiAsToroidalExtensionOfSuperPoint}), where we find 
\eqref{10ToroidalTDualityDiagram}
type IIA self-duality exhibited on a fully (i.e., along all space-time axes) doubled super-spacetime:
$$
  \begin{tikzcd}[
    row sep=-7pt,
    column sep=25pt
  ]
    &    
    \qquad 
    \mathfrak{Dbl}
   {\;:=\;} \mathbb{R}^{
      1,9\,\vert\,
      \mathbf{16}
      \oplus
      \overline{\mathbf{16}}
    }
    \underset{
      \mathclap{
        \adjustbox{
          raise=-3pt,
          scale=.7
        }{$
         \mathbb{R}^{
           0
           \,\vert\,
           32
         }
        $}
      }
    }{\;\times\;}
    \widetilde{\mathbb{R}}^{
      1,9\,\vert\,
      \mathbf{16}
      \oplus
      \overline{
      \mathbf{16}
      }
    }
    \qquad 
    \ar[
      dl,
      shorten <=-25pt,
      "{ \pi_A }"{swap, pos=-.1}
    ]
    \ar[
      dr,
      shorten <=-25pt,
      "{ \pi_{\widetilde A} }"{pos=-.1}
    ]
    \\
    \mathllap{
      \scalebox{.7}{
        \color{darkblue}
        \bf
        \def\arraystretch{.9}
        \begin{tabular}{c}
          10D type IIA
          \\
          super-space
        \end{tabular}
      }
    }
    \mathbb{R}^{
      1,9\,\vert\,
      \mathbf{16}
      \oplus
      \overline{\mathbf{16}}
    }
    \ar[
     dr,
     ->>
    ]
    &&
    \widetilde{\mathbb{R}}^{
      1,9\,\vert\,
      \mathbf{16}
      \oplus
      \mathbf{16}
    }
    \ar[
      dl,
      ->>
    ]
    \mathrlap{
      \scalebox{.7}{
        \color{darkblue}
        \bf
        \def\arraystretch{.9}
        \begin{tabular}{c}
          10D type II$\widetilde{\mathrm{A}}$
          \\
          super-space
        \end{tabular}
      }
    }
    \\
    &
    \underset{
    \mathclap{
      \scalebox{.7}{
        \color{darkblue}
        \bf
        \def\arraystretch{.9}
        \begin{tabular}{c}
          super-point
        \end{tabular}
      }
    }      
    }{
    \mathbb{R}^{
      0\,\vert\,
      \mathbf{32}
    }
    }
  \end{tikzcd}
$$

\vspace{-1mm} 
\noindent
This fully doubled superspace $\mathfrak{Dbl}$ carries a dual coframe field $\tilde e_a$ for each of the coordinate axes, reflecting corresponding string winding charges,
and which thus supports an analogous Poincar{\'e} 2-form exhibiting 
(Prop. \ref{Poinre2FormForFullDualization})
the T-duality between type IIA and its full T-dual super-space II$\widetilde{\mathrm{A}}$ (essentially type IIA itself, cf. Rem. \ref{NatureOfFullyTDualSuperSpacetime}):
\begin{equation}
  \label{FullPoincare2FormInIntroduction}
  \underset{
    \adjustbox{
      raise=-5pt,
      scale=.7
    }{
      \color{gray}
      Twisted Poincar{\'e} super 2-form for
      \color{darkblue}
      fully doubled 10D type II SuGra
    }
  }{
  P_2
  \;:=\;
    - \tilde e_a
    \,
    e^a 
  \;\;
  \in
  \;\;
  \mathrm{CE}\big(
    \mathfrak{Dbl}
  \big)
  \,,
}
  \hspace{1.4cm}
  \underset{
    \adjustbox{
      raise=-5pt,
      scale=.7
    }{
      \color{gray}
      Bianchi identity
    }
  }{
  \mathrm{d}\, 
  P_2
  \;=\; 
  \pi_A^\ast 
  H_3^A 
  -
  \pi_{\widetilde A}^\ast
  H_3^{\widetilde{A}}
  }
  \,.
\end{equation}

\medskip

\noindent
{\bf Extended IIA super-symmetry algebra extends doubled super-space.}
Recalling that the ``doubling'' of spacetime happening here is equivalently the adjoining of string (winding) charges, we observe that the doubled super-spacetime is an intermediate stage in the fully extended type IIA super-symmetry algebra $\IIA$ (Def. \ref{FullyExtebdedIIASpacetime}), which adjoins furthermore the charges of the D/NS5-branes (Rem. \ref{ExtendedIIAAlgebraAndBraneCharges}). To make this more manifest, we consider the super-algebra $\Brn$ \eqref{PureBraneAlgebra} which extends the super-point purely by the type IIA brane charges (string, D-branes \& NS5-brane). We 
observe 
\eqref{FiberProductOfDoubledSpacetimeWitPureBraneAlgebra}
that the fully extended IIA super-symmetry algebra is the fiber product \eqref{CommutingPairOfExtensions}
over the fully T-dual IIA algebra
of
this pure brane charge algebra
with the fully doubled super-spacetime:
$$
  \begin{tikzcd}[
    row sep=-4pt, 
    column sep=65pt
  ]
    &
    \overset{
      \mathclap{
      \raisebox{4pt}{
         \scalebox{.7}{
           \color{darkblue}
           \bf
           \def\arraystretch{.9}
           \begin{tabular}{c}
             Fully extended IIA
             \\
             super-algebra
           \end{tabular}
         } 
      }
      }
    }{
      \IIA
    }
    \ar[
      dl,
      ->>
    ]
    \ar[
      dr,
      ->>
    ]
    \\
    \mathllap{
      \scalebox{.7}{
        \color{darkblue}
        \bf
        \def\arraystretch{.9}
        \begin{tabular}{c}
          Fully doubled
          \\
          super-spacetime
        \end{tabular}
      }
    }
    \mathfrak{Dbl}
    \ar[
      dr,
      ->>
    ]
    &&
    \Brn
    \ar[
      dl,
      ->>,
      "{
        \scalebox{.7}{
          \color{darkgreen}
          \bf
          \def\arraystretch{.9}
          \begin{tabular}{c}
            extension by
            \\
            brane charges
          \end{tabular}
        }
      }"{sloped, swap},
      "{
        p^{\mathrm{Brn}}
      }"{swap, xshift=3pt, yshift=-3pt}
    ]
    \mathrlap{
      \scalebox{.7}{
        \color{darkblue}
        \bf
        \def\arraystretch{.9}
        \begin{tabular}{c}
          Pure brane 
          \\
          charge algebra
        \end{tabular}
      }
    }
    \\
    &
    \underset{
      \raisebox{-3pt}{$
      \mathclap{
        \scalebox{.7}{
          \color{darkblue}
          \bf
          \begin{tabular}{c}
            Fully T-dual IIA
            \\
            super-spacetime
          \end{tabular}
        }
      }
      $}
    }{
    \widetilde{\mathbb{R}}^{
      1,9
      \,\vert\,
      \mathbf{16}
      \oplus
      \overline{\mathbf{16}}
    }
    }
  \end{tikzcd}
$$
This relation between extended super-symmetry and T-duality doubled super-spacetime may not have been addressed before --- but now we see that this brings out the M-theoretic lift of T-duality doubled super-space.

\medskip

\noindent
{\bf Extending doubled super-space to the M-algebra.} We then observe that finally extending all of the above setting also along the fibration of the 11D- over the 10D type IIA super-tangent space makes the $\IIA$-algebra extend to the basic {\it M-algebra} $\mathfrak{M}$ (recalled as Def. \ref{BasicMAlgebra}, and makes the $\mathfrak{Dbl}$-algebra extend to the analog of the F-theory super-tangent space \eqref{TheFTheorySuperSpacetime}
for reduction along all spacetime directions, which, therefore, we denote $\mathfrak{F}$). This concretely exhibits the M-algebra as the M-theoretic analog of the T-duality correspondence super-space
\vspace{-1mm} 
$$
  \begin{tikzcd}[
    row sep={20pt, between origins},
    column sep={90pt, between origins}
  ]
    &&
    \overset{
      \mathclap{
        \adjustbox{
          raise=3pt,
          scale=.7
        }{
          \color{darkblue}
          \bf
          \def\arraystretch{.9}
          \begin{tabular}{c}
            M-algebra
          \end{tabular}
        }
      }
    }{
      \mathfrak{M}
    }
    \ar[
      dl,
      ->>
    ]
    \ar[
      dr,
      ->>,
      "{
        p^M
      }"
    ]
    \\
    &
    \mathfrak{F}
    \ar[
      dr,
      ->>
    ]
    \ar[
      dl,
      ->>
    ]
    &&
    \IIA
    \ar[
      dl,
      ->>,
      "{
        p^{\mathrm{Brn}}
      }"{swap, xshift=3pt, yshift=-3pt}
    ]
    \\
    \mathbb{R}^{
      1,10\,\vert\,
      \mathbf{32}
    }
    \ar[
      dr,
      ->>,
      shorten <= -5pt
    ]
    \ar[
      dr,
      phantom,
      shift right=12pt,
      "{
        \scalebox{.7}{
          \color{darkgreen}
          \bf
          \def\arraystretch{.9}
          \begin{tabular}{c}
            extension by
            \\
            11th dimension
          \end{tabular}
        }
      }"{sloped, pos=.3}
    ]
    &&
    \mathfrak{Dbl}
    \ar[
      dl,
      ->>
    ]
    \ar[
      dr,
      ->>
    ]
    \\
    &
    \mathbb{R}^{
      1,9\,\vert\,
      \mathbf{16}
      \oplus
      \overline{\mathbf{16}}
    }
    &&
    \mathbb{R}^{
      1,9\,\vert\,
      \mathbf{16}
      \oplus
      \mathbf{16}
    }
    \\[-8pt]
    & 
    \mathclap{
      \scalebox{.7}{
        \color{gray}
        \def\arraystretch{.9}
        \begin{tabular}{c}
          Type IIA
          \\
          super-spacetime
        \end{tabular}
      }
    }
    &&
    \mathclap{
      \scalebox{.7}{
        \color{gray}
        \def\arraystretch{.9}
        \begin{tabular}{c}
          Type IIB
          \\
          super-spacetime
        \end{tabular}
      }
    }
  \end{tikzcd}
$$

In particular, this diagram exhibits the dual coframe field $\tilde e_a$ on (hence the doubled dimensions of) $\mathfrak{Dbl}$ as the {\it mem}brane charges $e_{a_1 a_2}$ that wrap the 11th dimension:
$
  \tilde e_a 
  \;\leftrightarrow\;
  \,e_{a\, \ten}
$
\eqref{MAlgebraExtensionOfFullyExtendedIIA}.

\medskip

\noindent
{\bf M-theoretic Poincar{\'e} 3-forms and the hidden decomposition of the M-theory 3-form.}
With this understood it becomes evident that there are super-invariant 3-forms on $\mathfrak{M}$
which dimensionally reduce
(Rem. \ref{LiftOfPoincare3Form})
to the Poincar{\'e} 2-form 
\eqref{FullPoincare2FormInIntroduction} controlling T-duality on 10D super-space:
\vspace{1mm} 
\begin{equation}
  \label{P3InIntroduction}
  P_3
  \;:=\;
  \tfrac{1}{2}
  e^{a_1}  
    e_{a_1 a_2}
  e^{a_2}
  \;, 
  \hspace{1.5cm}
  \grayunderbrace{
  {p^{\mathrm{Brn}}_{\mathrm{bas}}}\,
  p^M_\ast
  }{
    \mathclap{
      \scalebox{.7}{
        \def\arraystretch{.9}
        \begin{tabular}{c}
          Basic part of 
          \\
          11D fiber integration
        \end{tabular}
      }
    }
  }
  \;
  P_3
  \;\;
  =
  \;\;
  P_2
\end{equation}

\vspace{-2mm} 
Combining all this, we see that the doubled correspondence super-spacetime $\mathfrak{Dbl}$ for 10-toroidal T-duality is itself part of the following larger M-theoretic mesh of diagrams, from whose tip descends the Poincar{\'e} 2-form $P_2$ that controls T-duality:
\vspace{-3mm}
\begin{equation}
\label{SummaryDiagram}
\hspace{-2mm} 
  \begin{tikzcd}[
    row sep=14pt, 
    column sep=27pt
  ]
    &&
    &&
    \mathfrak{M}
    \ar[
      dddd,
      phantom,
      "{
        \adjustbox{
          raise=3pt,
          scale=.7
        }{
          \color{darkblue}
          \bf
          \def\arraystretch{.8}
          \begin{tabular}{c}
            {
              \color{gray}
              \eqref{CEOfBasicMAlgebra}
            }
            \\
            Basic
            \\
            M-algebra
        \end{tabular}
        }
      }"{pos=-.33}
    ]
    \ar[
      dddd,
      phantom,
      "{
        \begin{array}{c}
          \hspace{-9pt}
          \frac{1}{2}
          e^{a}  
          e_{a b}   
          e^b  
          \mathrlap{
            \scalebox{.8}{$
              \, =: \!P_3
            $}
          }
          \\
          \rotatebox[origin=c]{-90}{
            $\longmapsto$
          }
          \\
          e_a \, \tilde e^a
          \mathrlap{
            \scalebox{.8}{$
            \, =: \! P_2
            $}
          }
        \end{array}
      }",
      "{
        p^{\mathrm{Brn}}_{\mathrm{bas}} \,
p^{\mathrm{M}}_\ast
      }"{xshift=18pt, scale=.65},
      "{
        \scalebox{.67}{
          \color{olive}
          \bf
          \def\arraystretch{.85}
          \begin{tabular}{c}
            Poincar{\'e}
            \\
            super-forms
            \\
            \color{gray}
            \eqref{ReductionOfPoincare3Form}
          \end{tabular}
        }
      }"{xshift=-27pt, pos=.58}
    ]
    \ar[
      ddll,
      ->>,
      "{
        p^{\mathrm{Brn}}
      }"{swap}
    ]
    \ar[
      ddrr,
      ->>,
      "{
        p^{\mathrm{M}}
      }"{yshift=4pt, pos=.6},
      "{
        \scalebox{.7}{
          \color{gray}
          \eqref{MAlgebraExtensionOfFullyExtendedIIA}
        }
      }"{description, pos=.6}
    ]
    \\
    \\
    &&
    \mathfrak{F}
    \ar[
      dddd,
      phantom,
      "{
        \scalebox{.7}{
          \color{darkblue}
          \bf
          \def\arraystretch{.75}
          \begin{tabular}{c}
            full
            \\
            F-theory 
            \\
            super-
            \\
            spacetime
            \\
            \color{gray}
            \eqref{frakF}
          \end{tabular}
        }
      }"{pos=.30}      
    ]
    \ar[
      ddrr,
      ->>,
      "{
        \scalebox{.7}{
          \color{gray}
          \eqref{FasExtensionOfDbl}
        }
      }"{description}
    ]
    \ar[
      ddll,
      ->>
    ]
    &&
    &&
    \IIA
    \ar[
      dddd,
      phantom,
      "{
        \scalebox{.7}{
          \color{darkblue}
          \bf
          \def\arraystretch{.9}
          \begin{tabular}{c}
            fully 
            \\
            extended
            \\
            type IIA algebra
            \\
            \color{gray}
            \eqref{CEFullyExtendedIIA}
          \end{tabular}
        }
      }"{pos=.25}
    ]
    \ar[
      ddll,
      ->>,
      "{
        \scalebox{.7}{
          \color{gray}
          \eqref{FullYTDualSpacetimeAsExtensionOfFullyDoubledSpacetime}
        }
      }"{description}
    ]
    \ar[
      ddrr,
      ->>,
      "{
        \scalebox{.7}{
          \eqref{FiberProductOfDoubledSpacetimeWitPureBraneAlgebra}
        }
      }"{description}
    ]
    \\
    \\
    \mathbb{R}^{
      1,10\,\vert\, \mathbf{32}
    }
    \ar[
      dd,
      phantom,
      "{
        \scalebox{.7}{
          \color{darkblue}
          \bf
          \def\arraystretch{.75}
          \begin{tabular}{c}
            11D 
            \\
            super-
            \\
            spacetime
          \end{tabular}
        }
      }"{pos=.3}
    ]
    \ar[
      ddrr,
      ->>,
      "{
        \scalebox{.7}{
          \color{gray}
          \eqref{11dSuperspaceAsExtensionOfIIASuperspace}
        }
      }"{description}
    ]
    && &&
    \mathfrak{Dbl}
    \ar[
      dddd,
      phantom,
      "{
        \scalebox{.7}{
          \color{darkblue}
          \bf
          \def\arraystretch{.75}
          \begin{tabular}{c}
            fully
            \\
            doubled 
            \\
            super-
            spacetime
            \\
            \color{gray}
            \eqref{CEOfFullyDoubledSuperSpacetime}
          \end{tabular}
        }
      }"{pos=.25}
    ]
    \ar[
      ddll,
      ->>
    ]
    \ar[
      ddrr,
      ->>
    ]
    &&&&
    \Brn
    \ar[
      ddll,
      ->>
    ]
    \ar[
      dd,
      phantom,
      "{
        \scalebox{.7}{
          \color{darkblue}
          \bf
          \def\arraystretch{.9}
          \begin{tabular}{c}
            pure brane
            \\
            charge algebra
            \\
            \color{gray}
            \eqref{PureBraneAlgebra}
          \end{tabular}
        }
      }"{pos=.35, xshift=4pt}
    ]
    \ar[
      ddddllll,
      Latex-,
      shorten=20pt,
      loosely dashed,
      shift left=50pt,
      "{
        \scalebox{.7}{
          \color{darkgreen}
          \bf
        extend by brane charges
        }
      }"{sloped, swap}
    ]
    \\
    \\
    {}
    &&
    \mathbb{R}^{
      1,9\,\vert\,
      \mathbf{16}
      \oplus
      \overline{\mathbf{16}}
    }
    \ar[
      dd,
      phantom,
      "{
        \scalebox{.7}{
          \color{darkblue}
          \bf
          \def\arraystretch{.75}
          \begin{tabular}{c}
            type IIA
            \\
            super-
            \\
            spacetime
            \\
            \color{gray}
            \eqref{CEOfIIASpacetime}
          \end{tabular}
        }
      }"{pos=.4, xshift=-3pt}
    ]
    \ar[
      ddrr,
      ->>,
      "{
        \scalebox{.7}{
          \color{gray}
          \eqref{SuperspacetimeExtensionOfSuperpoint}
        }
      }"{description}
    ]
    && &&
    \widetilde{\mathbb{R}}^{
      1,9\,\vert\,
      \mathbf{16}
      \oplus
      \overline{\mathbf{16}}
    }
    \ar[
      dd,
      phantom,
      "{
        \scalebox{.7}{
          \color{darkblue}
          \bf
          \def\arraystretch{.75}
          \begin{tabular}{c}
            fully T-dual
            \\
            type IIA
            \\
            super-spacetime
            \\
            (string charges)
            \\
            \color{gray}
            \eqref{DualSpacetime}
          \end{tabular}
        }
      }"{pos=.45, xshift=11pt}
    ]
    \ar[
      ddll,
      ->>,
      "{
        \scalebox{.7}{
          \color{gray}
          \eqref{SuperspacetimeExtensionOfSuperpoint}
        }
      }"{description}
    ]
    &&
    {}
    \\
    \\
    && 
    {}
    &&
    \underset{
      \mathclap{
      \adjustbox{
        raise=-16pt,
        scale=.7
      }{
        \color{darkblue}
        \bf
        \def\arraystretch{.75}
        \begin{tabular}{c}
          super-
          \\
          point
          \\
          \color{gray}
          \eqref{CEOfSuperpoint}
        \end{tabular}    
      }
      }
    }{
    \mathbb{R}^{
      0\,\vert\,
     \mathbf{32}
    }
    }
    \ar[
      uuuullll,
      -Latex,
      shorten=20pt,
      loosely dashed,
      shift left=46pt,
      "{
        \scalebox{.7}{
          \color{darkgreen}
          \bf
          extend by spacetime dimensions
        }
      }"{sloped, swap}
    ]
    &&
    {}
  \end{tikzcd}
\end{equation}

\vspace{1mm} 
Remarkably, super-invariants with the leading term \eqref{P3InIntroduction}
have been discussed before
\eqref{TheDecomposed3Form}
from a rather different point of view, under the name of ``decomposed'' M-theory 3-forms (\cite{DF82}\cite{BDIPV04}\cite{AndrianopoliDAuria24}\cite{GSS24-MGroup}, see \eqref{TheDecomposed3Form} below)  satisfying in addition
\vspace{-1mm} 
\begin{equation}
  \label{Poincare3FormAsCoboundaryInIntroduction}
  \underset{
    \adjustbox{
      raise=-4pt,
      scale=.7
    }{
      Poincar{\'e} 3-form
      in 
      \color{darkblue}
      M-extended 11D SuGra
    }
  }{
  \widehat{P}_3 \;\in\;
  \mathrm{CE}\big(
    \widehat{\,\mathfrak{M}\,}
  \big),
  }
  \hspace{1cm}
  \hspace{1cm}
  \underset{
    \mathclap{
      \adjustbox{
        scale=.7,
        raise=-4pt
      }{
        \color{gray}
        Bianchi identity
      }
    }
  }{
  \mathrm{d}\, \widehat{P}_3
  \;=\;
  G_4\;.
  }
\end{equation}

\vspace{1mm} 
\noindent In the concluding section \S\ref{ConclusionAndOutlook}, we summarize the M-theoretic T-duality picture that we establish here and give an outlook on the Poincar{\'e} super 3-form $P_3$ as in \eqref{Poincare3FormAsCoboundaryInIntroduction} as exhibiting the rational-topological enhancement of aspects of U-duality symmetry for super-space supergravity.

\medskip

We will further discuss in \cite{GSS24-MGroup} the hidden M-algebra with its Poincar{\'e} 3-form as an M-theoretic candidate for U-duality covariant super-space supergravity. This article lays the groundwork by a comprehensive discussion of the underlying super-$L_\infty$ algebraic T-duality mechanisms, completing and extending previous such work in \cite{FSS18-TDualityA}\cite{FSS18-TDualityB}\cite{FSS20-HigherT}\cite{SS18-HTDuality}.

\medskip

\noindent 
{\bf The role of super $L_\infty$-algebras.}
While we {\it speak} of (finite-type) super $L_\infty$-algebras, we always {\it compute} with their dual Chevalley-Eilenberg super dg-algebras (the corresponding ``FDA''s, see \S\ref{SuperLTheory}): The super dg-algebras greatly ease the computations while their super $L_\infty$-algebra incarnation makes transparent the conceptual role of these algebras. This is an instance of the general {\it duality between geometry and algebra} (cf. \cite[Tbl 4]{JurcoEtAl19} for survey in the context of M-theory)
and directly analogous to how, in algebraic geometry, affine schemes are dually equivalent to commutative algebras, in which case it is familiar common practice that one {\it speaks of schemes} but {\it computes with their rings} of functions. This analogy is precise: In {\it derived} algebraic geometry $L_\infty$-algebras are the ``infinitesimal/formal higher stacks'' of which their CE-algebras are the derived algebras of functions (\cite{Hinich01}\cite[\S 4.5.1]{Sc13-cohesive}\cite{CalaqueGrivaux21}). While we do not need here the details of the theory of higher/derived  geometry, it should be(come) clear by the ample examples we provide that it is most fruitful to think in terms of super $L_\infty$-algebra but compute in terms of super dg-algebras. 

\bigskip

\noindent
{\bf Outline.}

\vspace{1mm} 
\hspace{-7mm}
\def\arraystretch{1.1}
\begin{tabular}{p{17cm}}
\S\ref{SuperLTheory} discusses super-$L_\infty$ algebraic T-duality in abstract generality,
\\
\S\ref{SuperspaceTDuality} realizes this on the avatar super-flux densities on super-tangent spacetimes.
\\
\S\ref{LiftingToMTheory} sees the M-algebra appear in extension of the fully doubled super-spacetime,
\\
\S\ref{ConclusionAndOutlook} sums up the curious picture thus obtained and gives an outlook on M-theoretic lessons.
\end{tabular}
 
\medskip

\newpage 
\section{Super-$L_\infty$ theory}
\label{SuperLTheory}

In this section, we discuss in abstract generality the super-$L_\infty$-algebraic structures and phenomena which, when applied to super-flux densities on super-spacetime, in \S\ref{SuperspaceTDuality} below, exhibit super-space T-duality. 
While the super-$L_\infty$ perspective makes various constructions nicely transparent, all our computations take place in the dual Chevalley-Eilenberg dgc-algebra picture that is familiar in the supergravity literature (``FDA''s, cf. Rem. \ref{CEAlgebrasAsQuotientsOfFDGCAs} below).

\subsection{Super-$L_\infty$ algebra}
\label{SuperLieAlgebras}

We recall (from \cite[\S 2]{FSS15-HigherWZW}\cite[\S 2]{FSS18-TDualityA}\cite[(21)]{FSS19-QHigher}\cite[\S 3.2]{HSS19}\cite[p 33, 48]{S21-HPG}) the notion of higher (meaning: categorified symmetry) super-Lie algebras (of finite type) and their identification with the ``FDA''s from the supergravity literature (\cite{vanNieuwenhuizen83}\cite[\S III.6]{CDF91}, cf. \cite{AndrianopoliDAuria24}).
Our ground field is the real numbers $\mathbb{R}$, and all super-vector spaces are assumed to be finite-dimensional.

\medskip

Given a finite-dimensional super-Lie algebra $\mathfrak{g} \,\simeq\, \mathfrak{g}_{\mathrm{evn}} \oplus \mathfrak{g}_{\mathrm{odd}}$,
the linear dual of the super-Lie bracket map
$$
  [\mbox{-},\mbox{-}]
  \;:\;
  \begin{tikzcd}
    \mathfrak{g}
    \vee 
    \mathfrak{g}
    \ar[r]
    &
    \mathfrak{g}
  \end{tikzcd}
$$
may be understood to map the first to the second exterior power of the underlying dual super-vector space, and as such it extends uniquely to a $\mathbb{Z}\!\times\!\ZTwo$-graded derivation $\mathrm{d}$ of degree=$(1,\mathrm{evn})$ on the exterior super-algebra (where the minus sign is just a convention)
$$
  \begin{tikzcd}[row sep=small]
    \wedge^1
    \mathfrak{g}^{\ast}
    \ar[
      rr,
      "{ 
        -[\mbox{-},\mbox{-}]^\ast
      }"
    ]
    \ar[
      d,
      hook
    ]
    &&
    \wedge^2 \mathfrak{g}^{\ast}
    \ar[
      d,
      hook
    ]
    \\
    \wedge^\bullet 
    \mathfrak{g}^\ast
    \ar[
      rr,
      "{ \mathrm{d} }"
    ]
    &&
    \wedge^\bullet 
    \mathfrak{g}^\ast    
  \end{tikzcd}
$$
With this, the condition $d \circ d = 0$ is equivalently the super-Jacobi identity on $[\mbox{-},\mbox{-}]$, and the
resulting differential graded super-commutative algebra is know as the {\it Chevalley-Eilenberg algebra} of $\mathfrak{g}$:
$$
  \mathrm{CE}\big(
    \mathfrak{g}
    ,\,
    [\mbox{-},\mbox{-}]
  \big)
  \;\;
  :=
  \;\;
  \big(
    \wedge^\bullet \mathfrak{g}^\ast
    ,\,
    \mathrm{d}
  \big)
  \,.
$$
This construction is a {\it fully faithful formal duality}
\begin{equation}
  \label{OrdinaryCEAlgebra}
  \begin{tikzcd}[row sep =-2pt, column sep=large]
    \mathrm{sLieAlg}
    \ar[
      rr,
      hook,
      "{ \mathrm{CE} }"
    ]
    &&
    \mathrm{sDGCAlg}
      ^{\mathrm{op}}
    \\
    \big(
      \grayunderbrace{V}{
        \mathclap{
          \scalebox{.7}{
            \def\arraystretch{.9}
            \begin{tabular}{c}
              super-
              \\
              vector space
            \end{tabular}
          }
        }
      }, 
      [\mbox{-},\mbox{-}]
    \big)
    &\longmapsto&
    \big(
      \wedge^\bullet V^\ast
      ,\,
      \mathrm{d} = -[\mbox{-},\mbox{-}]^\ast
    \big)
    \,,
  \end{tikzcd}
\end{equation}

\vspace{-2mm} 
\noindent in that 

\begin{itemize}
\item[{\bf (i)}] for every super-vector space $V$ a choice of such differential $\mathrm{d}$ on $\wedge^\bullet V^\ast$ uniquely comes from a super-Lie bracket $[\mbox{-},\mbox{-}]$ on $V$ this way, and 

\item[{\bf (ii)}] 
super-Lie homomorphisms $\phi : \mathfrak{g} \xrightarrow{\;} \mathfrak{g}'$ 
are in bijection with super-dg-algebra homomorphisms $\phi^\ast : \mathrm{CE}(\mathfrak{g}') \xrightarrow{} \mathrm{CE}(\mathfrak{g})$.
\end{itemize} 
More concretely, given $(T_i)_{i =1}^n$ a linear basis for $\mathfrak{g}$ with corresponding structure constants $\big(f^k_{i j} \in \mathbb{R}\big)_{i,j,k = 1}^n$, then the Chevalley-Eilenberg algebra is equivalently the graded-commutative polynomial algebra
$$
  \mathrm{CE}\big(
    \mathfrak{g}, 
    [\mbox{-},\mbox{-}]
  \big)
  \;\simeq\;
  \big(
    \mathbb{R}\big[
      t^1, \cdots, t^1
    \big]
    ,\,
    \mathrm{d}
  \big)
$$
on generators of degree $(1,\sigma_i)$ with corresponding structure constants for its differential:
\begin{equation}
\label{RelationBetweenStructureConstants}
\mbox{\small 
\def\arraystretch{1.7}
\begin{tabular}{|c||c|c|}
  \hline
  &
  \bf 
  \def\arraystretch{.95}
  \begin{tabular}{c}
    Super
    Lie algebra
  \end{tabular}
  &
  \bf 
  \def\arraystretch{.95}
  \begin{tabular}{c}
    Super
    dgc-algebra
  \end{tabular}
  \\
  \hline
  \hline
 Generators
  &
  $
    \big(
      \underbrace{
        T_i
      }_{ \color{gray} 
        \mathclap{
          \mathrm{deg} \,=\, (0,\sigma_i) 
        }
      }
    \big)_{i = 1}^n
  $
  &
  $
    \big(
      \underbrace{
        t^i
      }_{  \color{gray}
        \mathclap{
          \mathrm{deg} \,=\, (1,\sigma_i) 
        }
      }
    \big)_{i = 1}^n
  $
  \\[-16pt]
  &&
  \\
  \hline
  \rowcolor{lightgray}
  Relations
  &
  $
    [T_i, T_j] 
    \,=\, 
    f^k_{i j}
    \,
    T_k
  $
  &
  $
    \mathrm{d}
    \,
    t^k
    \;=\;
    -
    \tfrac{1}{2}
    f^k_{i j}
    \, t^i t^j
  $
  \\
  \hline
\end{tabular}
}
\end{equation}

\vspace{1mm} 
This dual perspective via the CE-algebra is most convenient for passing from super-Lie to {\it strong homotopy super-Lie algebras}, also known as {\it super Lie $\infty$-algebra} (subsuming Lie 2-algebras, Lie 3-algebras etc., hence infinitesimal ``categorified symmetry'' algebras), and also known as {\it super-$L_\infty$ algebras}, for short: These are obtained simply by dropping the assumption that the CE-generators are in degree 1:

\medskip

Namely for a {\it $\mathbb{Z}$-graded} super-vector space $V_\bullet$
(degree-wise finite-dimensional by our running assumption, hence ``of finite type''), a sequence of higher arity super-skew-commutative brackets is dually a map from the degreewise dual $V^\vee$ (with $V^\vee_n \,:=\, (V_n)^\ast$) to its graded Grassmann algebra:
$$
  \mathrm{d}
  \,:\,
  \wedge^1 V^{\vee}
  \xrightarrow{\quad}
  \wedge^\bullet V^\vee
$$
and the higher super-Jacobi identity is dually simply the statement that this map, extended uniquely as a super-graded derivation to all of $\wedge^\bullet V^\vee$, is a differential 
$$
  \mathrm{d}
  \;:\;
  \wedge^\bullet V^\vee
  \xrightarrow{\quad}
  \wedge^\bullet V^\vee
$$
in that it squares to zero: $\mathrm{d} \! \circ \! \mathrm{d} \,=\,0$. (This is the evident super-algebraic enhancement of the characterization of 
finite-type $L_\infty$-algebras in \cite[\S 6.1]{SSS09}.)

\smallskip 
This way, super-$L_\infty$ algebras (of finite type) are equivalently nothing but super dgc-algebras whose underlying super-graded algebra is of the form $\wedge^\bullet V^\vee$ for some $\mathbb{Z}$-graded super-vector space, with super $L_\infty$-homomorphisms identified as homomorphisms of these super dgc-algebras going in the {\it opposite} direction (``pullback''):
\begin{equation}
  \label{TheCategoryOfSuperLInfinityAlgebras}
  \begin{tikzcd}[
    row sep=-2pt, column sep=large
  ]
    \mathrm{sLieAlg}_\infty
    \ar[
      rr,
      "{ \mathrm{CE} }",
      hook
    ]
    &&
    \mathrm{sDGCAlg}^{\mathrm{op}}
    \\
    \big(
      \grayunderbrace{
        V
      }{
        \mathclap{
          \scalebox{.7}{
            \def\arraystretch{.9}
            \begin{tabular}{c}
              graded super-
              \\
              vector space
            \end{tabular}
          }
        }
      },
      [\mbox{-}],
      [\mbox{-},\mbox{-}],
      [\mbox{-},\mbox{-},\mbox{-}],
      \cdots
    \big)
    &\longmapsto&
    \big(
      \wedge^\bullet
      V^\vee
      ,\,
      \mathrm{d}
      = 
      -
      [\mbox{-}]^\ast - 
      [\mbox{-},\mbox{-}]^\ast -
      [\mbox{-},\mbox{-},\mbox{-}]^\ast
      -
      \cdots
    \big)
    \mathrlap{\,.}
  \end{tikzcd}
\end{equation}

\vspace{-1mm} 
\noindent More concretely, by a choice of linear basis $(T_i)_{i \in I}$ for its underlying graded super vector space $V$, the CE-algebra of a super-$L_\infty$-algebra may be written as:
\begin{equation}
  \label{CEAsQuotientOfFreeDGCA}
  \mathrm{CE}(
    \mathfrak{g}
  )
  \;\;
  \simeq
  \;\;
 \FDGCA \big[
    (
      \grayunderbrace{
        t^i
      }{
        \mathclap{
          \mathrm{deg}
          =
          (n_i, \sigma_i)
        }
      }
    )_{i \in I}
  \big]
  \big/
  \big(
    \mathrm{d}
    \,
    t^i
    \;=\;
    P^i(\vec t\;)
  \big)_{i \in I}
  \,,
\end{equation}
where 
\begin{itemize}[
  leftmargin=.6cm,
  topsep=2pt,
  itemsep=3pt
]
\item
$
  \mathrm{deg}(t^i)
  \;=\;
  \mathrm{deg}(T_i)
  +
  (1,\mathrm{evn})
$
\item $\FDGCA\big[(t^i)_{i \in I}\big]$  is the free differential $(\mathbb{Z}\times \ZTwo)$-graded symmetric algebra on these generators and their differentials \eqref{FreeDGCA}, whose product is subject only to the sign rule \eqref{Signs}.
\item $P^i(\vec t\;)$ are graded-symmetric polynomials in the generators,
\item $\mathrm{d}$ is extended from generators to polynomials as a super-graded derivation  of degree $(1,\mathrm{evn})$,
\item the consistency condition is (only) that $\mathrm{d} \circ \mathrm{d} = 0$.
\end{itemize}

Accordingly, a homomorphism of super $L_\infty$-algebras $f : \mathfrak{g} \xrightarrow{\;} \mathfrak{h}$ with dual linear basis $(e^i)_{i \in I}$ and $(t^j)_{j \in J}$ is dually given by an algebra homomorphism $f^\ast : \mathrm{CE}(\mathfrak{h}) \xrightarrow{\;} \mathrm{CE}(\mathfrak{g})$ pulling back the generators $t^j$ to polynomials $f^\ast(t^j) \in \wedge^\bullet(\mathfrak{g}^\vee)$ in the generators $e^i$ such that the differential is respected:
\begin{equation}
  \label{LieHomomorphism}
  \begin{tikzcd}[
    row sep=-2pt, column sep=3pt
  ]
    \mathfrak{g}
    \ar[
      rr,
      "{ f }"
    ]
    &&
    \mathfrak{h}
    \\
    \mathrm{CE}(\mathfrak{g})
    \ar[
      rr,
      <-,
      "{ f^\ast }"
    ]
    &&
    \mathrm{CE}(\mathfrak{h})
    \\
    f^\ast(t^j)
    &\longmapsfrom&
    t^j
  \end{tikzcd}
  \,,
  \hspace{.5cm}
  \mbox{such that}
  \hspace{.5cm}
  \underset{j \in J}{\forall}
  \;\;
  \mathrm{d} f^\ast(t^j)
  \;=\;
  f^\ast(\mathrm{d}t^j)
  \,.
\end{equation}

\begin{remark}[\bf CE-algebras are differential quotients of free differential graded-commutative algebras.]
\label{CEAlgebrasAsQuotientsOfFDGCAs}
As such, we may recognize the CE-algebras \eqref{CEAsQuotientOfFreeDGCA} as the ``free differential algebras'' of the supergravity literature \cite{vanNieuwenhuizen83}\cite[\S III.6]{CDF91}.
  The quotient notation in \eqref{CEAsQuotientOfFreeDGCA}, following \cite[\S 4]{FSS23-Char}, is justified by thinking of
  \begin{itemize}[
    leftmargin=.5cm,
    topsep=2pt,
    itemsep=2pt
  ]
  \item $\FDGCA\big[(e^i)_{i \in I}\big]$ as the (actual) {\it free differential} super-graded-commutative algebra, hence with each $\mathrm{d}e^i$ being a new generator subject to no relation (except super-graded commutativity), 

  \item $\left(\mathrm{d}\, e^i \,=\, P^i\big((e^j)_{j \in I}\big)\right)_{i \in I}$ as a differential ideal,
  \item the quotient hence enforcing these equations on the previously free differential.
  \end{itemize}
 \end{remark}

\begin{remark}[\bf $L_\infty$-jargon]
$\,$

\noindent {\bf (i)}  Another name for $L_\infty$-algebras is {\it strong homotopy Lie algebra} (which was more popular in the past), also abbreviated {\it sh-Lie algebra}, as in the original articles \cite{LadaStasheff93}\cite{LadaMarkl95}. Our formulation \eqref{TheCategoryOfSuperLInfinityAlgebras}
  of $L_\infty$-algebras via their CE-algebras (which brings out the equivalence of super $L_\infty$-algebras with the ``FDA''s in the supergravity literature, Rem. \ref{CEAlgebrasAsQuotientsOfFDGCAs}) is contained in these original articles, made explicit in \cite[Def. 13]{SSS09}. Similarly, our homomorphisms of $L_\infty$-algebras \eqref{LieHomomorphism} were also called {\it strong homotopy maps} or {\it sh-maps}, for short.

  \vspace{1mm} 
 \noindent {\bf (ii)} Or rather, our \eqref{LieHomomorphism} subsumes the slightly larger generality known as ``curved'' morphisms between non-curved(!) $L_\infty$-algebras (as in \cite[below (2)]{MehtaZambon12}):
  Namely, CE-algebras \eqref{CEAsQuotientOfFreeDGCA} carry a canonical  {\it augmentation} $\epsilon$ --- the homomorphism which projects out the scalar summand $\mathbb{R} \,\simeq\, \wedge^0 V^\vee \xhookrightarrow{\;} \wedge^\bullet V^\vee $ (dual to a canonical base-point, see Ex. \ref{RealWhiteheadAlgebraOfThePoint}):
  $$
    \begin{tikzcd}[
      row sep=0pt, column sep=large
    ]
      \mathfrak{l}(\ast)
      \ar[
        r, 
        hook,
        "{ 0 }"
      ]
      &
      \mathfrak{g}
      \\
      \mathbb{R} \;
      \ar[
        r,
        <<-,
        "{ \epsilon }"
      ]
      &
      \mathrm{CE}(\mathfrak{g})
      \simeq
      \mathbb{R}
      \oplus 
      \wedge^{\bullet\geq 1} V^\vee
    \end{tikzcd}
  $$
  and the ``non-curved'' morphisms $f : \mathfrak{g} \xrightarrow{\;} \mathfrak{g}'$ are those preserving these base-points, hence dually preserving these
  
  \newpage 
\noindent  augmentations: 
  \vspace{-2mm}
  $$
    \begin{tikzcd}[
      row sep=2pt, column sep=large
    ]
      & 
      \mathfrak{l}(\ast)
      \ar[
        dr, 
        "{ 0 }"
      ]
      \ar[
        dl,
        "{ 0 }"{swap}
      ]
      \\
       \mathfrak{g}
       \ar[
         rr,
         "{ f }"
       ]
       &&
       \mathfrak{g}'
       \\
       \mathrm{CE}(\mathfrak{g})
       \ar[
         rr,
         <-,
         "{ f^\ast }"{swap}
       ]
       &&
       \mathrm{CE}(\mathfrak{g}')
       \\
       &
       \mathbb{R}
       \ar[
         ur, 
         "{ \epsilon' }"{swap}
      ]
       \ar[
         ul, 
         "{ \epsilon }"
      ]
    \end{tikzcd}
  $$
  Hence ``non-curved homomorphisms between non-curved $L_\infty$''-algebras really means: base-point preserving homomorphisms. But we generically allow ``curved homomorphisms'', not required to respect the base-point.

  \vspace{1mm} 
\noindent {\bf (iii)}  Note that this is an issue if and only if $\wedge^1 V^\vee$ has elements in degree 0 (hence iff $V$ has elements in degree -1). In the interpretation of homomorphisms $\mathbb{R}^{1,d\,\vert\, \mathbf{N}} \xrightarrow{\;}  \mathfrak{g}$ as super-flux densities (in \S\ref{SuperspaceTDualityI} and \S\ref{LiftingToMTheory}) this is the case of ``axion fields''.
\end{remark}

\noindent
{\bf Examples of super-$L_\infty$ algebras.}
The base example in supergravity is the following Ex. \ref{SupersymmetryAlgebras}:
\begin{example}[\bf Supersymmetry algebras]
  \label{SupersymmetryAlgebras}
  For $d \in \mathbb{N}$ and $\mathbf{N} \,\in\, \mathrm{Rep}_{\mathbb{R}}(\mathrm{Spin}(1,d))$ a real spin-representation equipped with a $\mathrm{Spin}(1,d)$-equivariant linear map
  \begin{equation}
    \label{VectorialSpinorPairing}
    \big(\hspace{1pt}
      \overline{\mbox{(-)}}
      \Gamma
      (\mbox{-})
    \big)
    \;:\;
    \begin{tikzcd}
    \mathbf{N}
    \underset{\mathrm{sym}}{\otimes}
    \mathbf{N}
    \ar[r]
    &
    \mathbb{R}^{1,d}
    \,,
    \end{tikzcd}
  \end{equation}
  the corresponding super-translation super-Lie algebra $\mathbb{R}^{1,d\,\vert\,\mathbf{N}}$ is given by
  \begin{equation}
    \label{SuperMinkowskiCE}
    \mathrm{CE}\big(
      \mathbb{R}^{1,d\,\vert\,\mathbf{N}}
    \big)
    \;\simeq\;
    \FDGCA
    \Big[\;
      (\,
      \grayunderbrace
        { \psi^\alpha }
        {
          \mathclap{
          \mathrm{deg}
          \,=\,
          (1,\mathrm{odd})
          }
        }
       \, )_{\alpha=1}^N
        ,\;\;\;
        (
        \grayunderbrace
        { e^a }
        {
          \mathclap{
          \mathrm{deg}
          \,=\,
          (1,\mathrm{evn})
          }
        }
        )_{a=0}^d
    \Big]
    \Big/\!
    \left(\!
    \def\arraycolsep{2pt}
    \begin{array}{ccl}
      \mathrm{d}\,
      \psi^\alpha
      &=&
      0
      \\
      \mathrm{d}\,
      e^a
      &=&
      \big(\hspace{1pt}
        \overline{\psi}
        \,\Gamma^a\,
        \psi
      \big)
    \end{array}
   \!\!\right)
   .
  \end{equation}
  Specific examples of this kind are the topic of   \S\ref{SuperMinkowskiSpacetimes} below.

Dually, this means that the super-Lie algebra itself is
\begin{equation}
  \label{SuperMinkowskiLinearBasis}
  \mathbb{R}^{1,10\,\vert\,\mathbf{32}}
  \;\simeq\;
  \mathbb{R}\Big\langle \!
     \grayunderbrace{
       (Q_\alpha)_{\alpha=1}^{32}
     }{ 
       \mathrm{deg}
       \,=\, 
       (0,\mathrm{odd}) 
    }
    \,,
    \grayunderbrace{
      (P_a)_{a = 0}^{10}
    }{ 
      \mathrm{deg} \,=\, (0,\mathrm{evn}) 
    }
  \!\!\Big\rangle
\end{equation}
with the only non-trivial super-Lie brackets on basis elements being the usual
\footnote{
  \label{PrefactorConvention}
  Our prefactor convention in \eqref{TheBifermionicSuperBracket} -- ultimately enforced via the translation 
  \eqref{RelationBetweenStructureConstants}
  by our convention for the super-torsion tensor in \cite{GSS24-SuGra} and \cite{GSS24-SuGra} --
  coincides with that in \cite[(1.16)]{DeligneFreed99}\cite[p. 52]{Freed99}.
}
\begin{equation}
  \label{TheBifermionicSuperBracket}
  \big[
    Q_\alpha
    ,\,
    Q_\beta
  \big]
  \;=\;
  -
  2\,
  \Gamma^a_{\alpha \beta}
  P_a
  \,.
\end{equation}

  \smallskip

  The assumed $\mathrm{Spin}(1,d)$-equivariance implies that the ordinary Lorentz Lie algebra $\mathfrak{so}_{1,d}$ has a super-Lie action on $\mathbb{R}^{1,d\,\vert\,\mathbf{N}}$. The corresponding semidirect product super-Lie algebra is the {\it super-Poincar{\'e} Lie algebra}, the full ``supersymmetry algebra'' in these dimensions:
  $$
    \mathrm{CE}\big(
      \mathbb{R}^{
        1,d\,\vert\,\mathbf{N}}
        \,\rtimes\,
        \mathfrak{so}_{1,d}
    \big)
    \;\simeq\;
    \FDGCA
    \Big[\;
      (\,
      \grayunderbrace
        { \psi^\alpha }
        {
          \mathclap{
          \mathrm{deg}
          \,=\,
          (1,\mathrm{odd})
          }
        }
        \,)_{\alpha=1}^N
        ,\;\;\;
        (
        \grayunderbrace
        { e^a }
        {
          \mathclap{
          \mathrm{deg}
          \,=\,
          (1,\mathrm{evn})
          }
        }
        )_{a=0}^d
        ,\,
        (\,
        \grayunderbrace
        { \omega^{ab} = -\omega^{ba} }
        {
          \mathclap{
          \mathrm{deg}
          \,=\,
          (1,\mathrm{evn})
          }
        }
        \,)_{a,b=0}^d
    \Big]
    \Big/\!
    \left(
    \def\arraycolsep{2pt}
    \def\arraystretch{1}
    \begin{array}{ccl}
      \mathrm{d}\,
      \psi^\alpha
      &=&
      0
      \\
      \mathrm{d}\,
      e^a
      &=&
      \big(\hspace{1pt}
        \overline{\psi}
        \,\Gamma^a\,
        \psi
      \big)
      \,-\,
      \tensor
        {\omega}
        {^a_b}
      \, e^b
      \\
      \mathrm{d}
      \,
      \tensor
        {\omega}
        {^{ab}}
       &=& - \, 
       \tensor
         {\omega}
         {^a_c}
        \tensor
          {\omega}
          {^{cb}}
    \end{array}
   \!\!\right).
  $$
\end{example}

\vspace{-1mm} 
\noindent Important examples among higher Lie-algebras come from 

{\bf (I)} topological spaces, 

{\bf (II)} spectra of spaces, 

\noindent
and generally, unifying these two cases:

{\bf (III)} bundles of spectra over topological spaces.

\medskip

\noindent
{\bf (I) Whitehead $L_\infty$-algebras of spaces.}

\begin{example}[{\bf Real Whitehead $L_\infty$-algebras of topological spaces} {cf. \cite[Prop. 5.11]{FSS23-Char}}]
\label{WhiteheadLInfinityAlgebra}

Given a topological space $X$ --- which is (a) connected, (b) nilpotent, e.g., in that its  fundamental group is trivial, and (c) whose $\mathbb{R}$-cohomology $H^\bullet(X;\mathbb{R})$ is degreewise finite-dimensional --- there is an $L_\infty$-algebra, $\mathfrak{l}X$, characterized by the following two properties:
\begin{itemize}[
  leftmargin=.7cm,
  topsep=2pt,
  itemsep=2pt
]
\item[{\bf (i)}] The underlying graded vector space is the $\mathbb{R}$-rationalization of the homotopy groups $\pi_\bullet(X)$ of the based loop space $\Omega X$:
$$
  \mathfrak{l}X
  \;\simeq\;
  \big(
  \grayunderbrace{
  \pi_\bullet(\Omega X) \otimes_{{}_{\mathbb{Z}}}
  \mathbb{R}
  }{
    \mathrm{deg}\,=\,
    (\bullet, \mathrm{evn})
  }
  ,\; 
  [\mbox{-},\mbox{-}], 
  [\mbox{-}, \mbox{-}, \mbox{-}],
  \cdots
  \big)
  \,,
  \hspace{1.3cm}
  \mathrm{CE}\big(
    \mathfrak{l}X
  \big)
  \;\;
  \simeq
  \;\;
  \big(
    \wedge^\bullet 
      \big(
        \pi_\bullet(\Omega X)
        \otimes_{{}_{\mathbb{Z}}}
        \mathbb{R}
      \big)^\vee
    ,\;
    \mathrm{d}
  \big)
  \,,
$$
(which means, cf. below \eqref{CEAsQuotientOfFreeDGCA}, that the generators of $\mathrm{CE}(\mathfrak{l}X)$ are in the degrees of the homotopy groups of $X$).

\item[\bf (ii)] The cochain cohomology of its CE-algebra reproduces the ordinary cohomology of $X$:
$$
  H^\bullet\big(
    \mathrm{CE}(
      \mathfrak{l}X
    )
    ,\,
    \mathrm{d}
  \big)
  \;
  \simeq
  \;
  H^\bullet(X;\mathbb{R})
  \,.
$$
\end{itemize}
In rational homotopy theory the dg-algebra $\mathrm{CE}\big(\mathfrak{l}X\big)$ is known (reviewed in \cite[\S 5]{FSS23-Char}) as the {\it minimal Sullivan model} of the topological space $X$, retaining exactly the information of its rational homotopy type.
\end{example}

A trivial but useful example is the following:
\begin{example}[\bf The point]
  \label{RealWhiteheadAlgebraOfThePoint}
  The real Whitehead $L_\infty$-algebra (Ex. \ref{WhiteheadLInfinityAlgebra}) of the point space $\ast$ is the 0-object in super-$L_\infty$-algebras
  $$
    0 
      \,\simeq\,
    \mathfrak{l}(\ast)
  $$
  given by
  $$
    \mathrm{CE}\big(
      0
    \big)
    \;\simeq\;
    \big(
      \wedge^\bullet 0
      ,\,
      \mathrm{d}= 0
    \big)
    \;=\;
    \big(
      \mathbb{R}
      ,\,
      \mathrm{d}= 0
    \big)
    \,.
  $$
  Of course, this is also the real Whitehead $L_\infty$-algebra of every {\it contractible} topological space.
\end{example}
\begin{example}[\bf Line Lie $n$-algebra.]
  \label{LineLieNAlgebra}
  For $n \in \mathbb{N}$ and $X$ an integral Eilenberg-MacLane space 
  $$
    X
    \underset{\color{darkorange} \mathrm{hmtp}}{\;\simeq\;}
    B^n U(1)
    \underset{\color{darkorange}  \mathrm{hmtp}}{\;\simeq\;}
    K(\mathbb{Z}, n+1)
  $$
  (classifying ordinary integral cohomology in degree $n+1$ and equivalently 
  classifying complex line bundles, for $n = 1$, line bundle gerbes, for $n = 2$, and generally principal circle $n$-bundles, see \cite[Ex. 2.1]{FSS23-Char}) its real Whitehead $L_\infty$-algebra (Ex. \ref{WhiteheadLInfinityAlgebra}) 
  \begin{equation}
    \label{TheLineLieNAlgebra}
    b^n\mathbb{R}
    \;:=\;
    \mathfrak{l}
    \big(
      B^n \mathrm{U}(1)
    \big)
  \end{equation}
  is given by
  \begin{equation}
    \label{CEOfLineLieNAlgebra}
    \mathrm{CE}\big(
      \mathfrak{l}
      B^n \mathrm{U}(1)
    \big)
    \;\;
    \simeq
    \;\;
    \FDGCA\big[
      \grayunderbrace{
        \omega_{n+1}
      }{
        \mathclap{
          \mathrm{deg}\, 
          = (n+1,\mathrm{evn})
        }
      }
    \big]
    \Big/
    \big(
      \mathrm{d}\, 
      \omega_{n+1}
      \,=\,
      0
    \big)
    \,.
  \end{equation}
  This means that super-$L_\infty$ homomorphisms
  \eqref{LieHomomorphism}
  into these higher Lie algebras are equivalently $(n+1)$-cocycles:
  \begin{equation}
    \label{ModulatingOrdinaryCocycles}
    \begin{tikzcd}[
      sep=0pt
    ]
      \mathfrak{g}
      \ar[
        rr,
        "{ \alpha_{n+1} }"
      ]
      &&
      \mathfrak{l}\big(
        B^n \mathrm{U}(1)
      \big)
      \\
      \alpha_{n+1}
      &\longmapsfrom&
      \omega_{n+1}
    \end{tikzcd}
    \hspace{1cm}
    \Leftrightarrow
    \hspace{1cm}
    \left\{\!\!
    \def\arraystretch{1.2}
    \begin{array}{l}
      \alpha_{n+1}
      \,\in\,
      \mathrm{CE}(\mathfrak{g})\,,
      \\
      \mathrm{deg}(\alpha_{n+1})
      \,=\,
      (n+1,\,\mathrm{evn}),\,
      \\
      \mathrm{d}\,
      \alpha_{n+1} \,=\, 0\,.
    \end{array}
    \right.
  \end{equation}
    As an aside: For $n \in \mathbb{N}$ the classifying space $B^n \mathrm{U}(1) \,\simeq\, K(\mathbb{Z}, n+1)$ carries the structure of a higher (categorical symmetry-)group, equivalent to the based loop $\infty$-group of the next space in the sequence:
  $$
    B^n\mathrm{U}(1)
    \underset{
        \adjustbox{
          raise=-3pt,
          scale=.7
        }{\color{darkorange}
          as $\infty$-groups
        }
    }{\,\simeq\,}
    \Omega
    B^{n+1} \mathrm{U}(1)
    \,.
  $$
  (The underlying homotopy equivalences make the $\big(B^n \mathrm{U}(1)\big)_{n \in \mathbb{N}}$ a {\it spectrum} of spaces, cf. Ex. \ref{RealWhiteheadAlgebraOfSpectra}.)

  For this reason, the operation $B(-)$ is also called {\it de-looping}. After passage to Whitehead $L_\infty$-algebras 
  $$
    b \circ \mathfrak{l}\,(\mbox{-})
    \;\simeq\;
    \mathfrak{l} \circ B(\mbox{-})
  $$
  this is given, via \eqref{CEOfLineLieNAlgebra}, by shifting the degree of the single generator, 
\end{example}

\begin{example}[{\bf Real Whitehead $L_\infty$-algebra of the 4-sphere}, {cf. \cite[Ex. 5.3]{FSS23-Char}}]
\label{WhiteheadAlgebraOfS4}
The Whitehead $L_\infty$-algebra (Ex. \ref{WhiteheadLInfinityAlgebra}) of the 4-sphere, $\mathfrak{l}S^4$, is given by
\begin{equation}
  \label{CEOf4Sphere}
  \mathrm{CE}\big(
    \mathfrak{l}S^4
  \big)
  \;\;
  \simeq
  \;\;
  \FDGCA
  \left[\!\!
  \def\arraystretch{1}
  \begin{array}{c}
    g_4
    \\
    g_7
  \end{array}
  \!\!\right]
  \Big/\!
  \left(\!\!
  \def\arraystretch{1}
  \begin{array}{ccl}
    \mathrm{d}\, g_4
    &=&
    0
    \\
    \mathrm{d}\, g_7
    &=&
    \tfrac{1}{2}
    \, 
    g_4 \, g_4
  \end{array}
  \!\!\right).
\end{equation}
Namely, the generators in degree 4 and 7 reflect the fact that $S^4$ has such generators for non-torsion homotopy groups in these degrees, while the differential in \eqref{CEOf4Sphere} cuts down the resulting cohomology ring from $\mathbb{R}[g_4] \,\simeq\,H^4\big(K(\mathbb{Z},4); \mathbb{R}\big)$ to the correct  $\mathbb{R}[g_4]/(g_2^2) \,\simeq\,H^4\big(S^4;\, \mathbb{R}\big)$. The prefactor of $1/2$ in \eqref{CEOf4Sphere} is not fixed up to isomorphism of $L_\infty$-algebras, but is the natural choice for capturing the Bianchi identity of the C-field in the next Ex. \ref{4SphereValuedSuperFlux}.

Note the homomorphism  \eqref{LieHomomorphism} to the line Lie 4-algebra (Ex. \ref{LineLieNAlgebra})
\begin{equation}
  \label{4SphereFiberedOverKZ4}
  \begin{tikzcd}[row sep=small, 
    column sep=3pt
  ]
    \mathfrak{l}S^4
    \ar[
      d,
      ->>
    ]
    &
    g_4
    \\
    b^3 \mathbb{R}
    &
    \omega_4
    \ar[
      u,
      |->
    ]
  \end{tikzcd}
  \hspace{.7cm}
  \simeq
  \hspace{.7cm}
  \mathfrak{l}
  \left(\!\!\!
  \adjustbox{
    raise=2pt
  }{$
  \begin{tikzcd}[row sep=small]
    S^4 
    \ar[
      d,
      ->>
    ]
    \\
    B^3 \mathrm{U}(1)
  \end{tikzcd}
  $}
  \!\!\!\right)
\end{equation}
which rationally reflects the ``1st Postnikov stage'' of the 4-sphere (cf. \cite{GradySati21}).

\newpage 
On the other hand, rationally the Eilenberg-MacLane space $B^3 \mathrm{U}(1) \,\simeq\, K(\mathbb{Z},4)$ is indistinguishable from the classifying space for $\mathrm{SU}(2)$-principal bundles, $\mathfrak{l}B^3 \mathrm{U}(1) \,\simeq\, \mathfrak{l} B \mathrm{SU}(2)$, so that up to choices of flux quantization laws (\cite{SS24-Flux}) the above map may also be thought of as the map classifying the quaternionic Hopf fibration
\begin{equation}
  \label{HHopfFiberationAsSU2Bundle}
  \begin{tikzcd}[sep=small]
    S^7
    \ar[r]
    &
    S^4
    \ar[d]
    \\
    & 
    B \mathrm{SU}(2)
    \mathrlap{\,.}
  \end{tikzcd}
\end{equation}
\end{example}
The Lie 7-algebra $\mathfrak{l}S^4$ of Ex. \ref{WhiteheadAlgebraOfS4} (a 6-fold ``categorified symmetry'' algebra) is noteworthy because it provides the correct coefficients for the duality-symmetric C-field super-flux densities in 11D supergravity (for more on this see \cite{GSS24-SuGra}):
\begin{example}[{\bf 4-Sphere valued super-flux of 11D SuGra} {\cite[p 5]{FSS15-M5WZW}\cite[Cor. 2.3]{FSS17}\cite[Ex. 2.30]{GSS24-SuGra} following \cite[\S 2.5]{Sati13}}]
  \label{4SphereValuedSuperFlux}
  On the 11D super-Minkowski algebra $\mathbb{R}^{1,10\,\vert\,\mathbf{32}}$ (Ex. \ref{SupersymmetryAlgebras}) the super-invariants \eqref{11dSuperFluxInIntroduction}
  $$
    \left.
    \def\arraystretch{1.5}
    \begin{array}{ccl}
      G_4
      &:=&
      \tfrac{1}{2}
      \big(\,
        \overline{\psi}
        \,\Gamma_{a_1 a_2}\,
        \psi
      \big)
      e^{a_1} e^{a_2}
      \\
      G_7
      &:=&
      \tfrac{1}{5!}
      \big(\,
        \overline{\psi}
        \,\Gamma_{a_1 \cdots a_5}\,
        \psi
      \big)
      e^{a_1} \cdots e^{a_5}
    \end{array}
   \!\! \right\}
    \;\;
    \in
    \;
    \mathrm{CE}\big(
      \mathbb{R}^{
        1,10\,\vert\,\mathbf{32}
      }
    \big)
    \,,
    \hspace{1cm}
    \def\arraystretch{1.4}
    \begin{array}{lcl}
      \mathrm{d}\,
      G_4 &=&
      0
      \\
      \mathrm{d}\, G_7 
        &=&
      \tfrac{1}{2}
      G_4\, G_4
      \,,
    \end{array}
  $$
  are identified with a homomorphism \eqref{LieHomomorphism}
  of super-$L_\infty$-algebras from super-Minkowski space to $\mathfrak{l}S^4$ (Ex. \ref{WhiteheadAlgebraOfS4}):
 \begin{equation}
   \label{CFieldSuperFluxAsLieHom}
   \begin{tikzcd}[
     row sep=-2pt,
     column sep=20pt
   ]
     \mathbb{R}^{
       1,10\,\vert\,\mathbf{32}
    }
    \ar[
      rr,
      "{ (G_4,\,G_7) }"
    ]
    &&
    \mathfrak{l}S^4
    \\
    G_4 &\longmapsfrom& g_4
    \\
    G_7 &\longmapsfrom& g_7
    \, .
   \end{tikzcd}
 \end{equation}
\end{example}

\bigskip 
\noindent
{\bf (II) Whitehead $L_\infty$-algebras of spectra of spaces.}

\begin{example}[{\bf Real Whitehead $L_\infty$-algebras of spectra} {\cite[Lem. 2.25]{BMSS19}\cite[Ex. 5.7]{FSS23-Char}}]
\label{RealWhiteheadAlgebraOfSpectra}
The real Whitehead $L_\infty$-algebra of a spectrum $E$ of topological spaces has underlying it the $\mathbb{R}$-rationalization of the stable homotopy groups of $\Omega E$, equipped with trivial brackets / trivial differential:
$$
  \mathfrak{l}E
  \;\simeq\;
  \Big(
  \grayunderbrace{
  \pi_\bullet(\Omega E)
  \otimes_{{}_{\mathbb{Z}}}\!
  \mathbb{R}
  }{
    \mathrm{deg}
    \,=\,
    (\bullet, \mathrm{evn})
  }
  ,
  [\mbox{-}, \cdots, \mbox{-}]
  =
  0
  \Big)
  ,
  \hspace{1.3cm}
  \mathrm{CE}(\mathfrak{l}E)
  \;\simeq\;
  \Big(
    \wedge^\bullet
    \big(
      \pi_\bullet(\Omega E)
    \otimes_{{}_{\mathbb{Z}}}
    \mathbb{R}
    \big)^{\vee}
    ,\,
    \mathrm{d} = 0
  \Big).
$$
While the differential is trivial, the crucial difference here to the Whitehead $L_\infty$-algebras of topological spaces (Ex. \ref{WhiteheadLInfinityAlgebra}) is that there may be elements in non-positive degree.
\end{example}

\begin{example}[\bf Real Whitehead $L_\infty$-algebra of complex topological K-theory]\label{RationalKtheorySpectra}
The spectrum $\mathrm{KU}$ of complex topological K-theory has stable homotopy groups in every even degree, hence its suspension $\Sigma \mathrm{KU}$ in every odd degree
$$
  \pi_k(\mathrm{KU})
  \;\simeq\;
  \left\{\!\!
  \def\arraystretch{1}
  \begin{array}{cl}
    \mathbb{Z} & \mbox{for even $k$}
    \\
    \ast & \mbox{otherwise}
  \end{array}
  \right.
  \,,
  \hspace{.7cm}
  \pi_k(\Sigma\mathrm{KU})
  \;\simeq\;
  \left\{\!\!
  \def\arraystretch{1}
  \begin{array}{cl}
    \mathbb{Z} & \mbox{for odd $k$}
    \\
    \ast & \mbox{otherwise}
  \end{array}
  \right.
  \,,
  \hspace{.7cm}
  \pi_k(\Sigma^n\mathrm{KU})
  \;\simeq\;
  \left\{\!\!
  \def\arraystretch{1}
  \begin{array}{cl}
    \mathbb{Z} & \mbox{for even $k+n$}
    \\
    \ast & \mbox{otherwise}
  \end{array}
  \right.
$$

\vspace{2mm} 
\noindent and hence 
its real Whitehead $L_\infty$-algebra 
(Ex. \ref{RealWhiteheadAlgebraOfSpectra})
is given by
\vspace{1mm} 
\begin{equation}
  \label{CEOfWhiteheadOfKU}
  \mathrm{CE}\big(
    \mathfrak{l}(\mathrm{KU})
  \big)
  \;
  \simeq
  \;
  \FDGCA
  \big[
    \smash{
    (\,
      \grayunderbrace{
      f_{2k}
      }{
        \mathclap{
          \mathrm{deg}
          \,=\,
          (2k,\mathrm{evn})
        }
      }
    \,)_{k \in \mathbb{Z}}
    }
  \big]
  \big/
  \big(
    \mathrm{d}\,
    f_{2\bullet}
    \,=\,
    0
  \big)
  \,,
  \hspace{1cm}
  \mathrm{CE}\big(
    \mathfrak{l}(\Sigma\mathrm{KU})
  \big)
  \;
  \simeq
  \;
  \FDGCA
  \big[
    \smash{
    (\,
      \grayunderbrace{
      f_{2k+1}
      }{
        \mathclap{
          \mathrm{deg}
          \,=\,
          (2k+1,\mathrm{evn})
        }
      }
    \,)_{k \in \mathbb{Z}}
    }
  \big]
  \big/
  \big(
    \mathrm{d}\,
    f_{2\bullet + 1}
    \,=\,
    0
  \big).
\end{equation}

\vspace{6mm} 
\noindent Analogously to Ex. \ref{LineLieNAlgebra}, this means that super-$L_\infty$ homomorphisms
  \eqref{LieHomomorphism}
  into these higher Lie algebras are equivalently sequences of cocycles in degrees $(2k)_{k\in \mathbb{Z}}$:
  $$
    \begin{tikzcd}[
      sep=0pt
    ]
      \mathfrak{g}
      \ar[
        rr,
        "{ (F_{2k})_{k\in \mathbb{Z}} }"
      ]
      &&
      \mathfrak{l}(\mathrm{KU})
      \\
      (F_{2k})_{k\in \mathbb{Z}}
      &\longmapsfrom&
      (f_{2k})_{k\in \mathbb{Z}}
    \end{tikzcd}
    \hspace{1cm}
    \Leftrightarrow
    \hspace{1cm}
    \left\{\!\!
    \def\arraystretch{1.2}
    \begin{array}{l}
      (F_{2k})_{k\in \mathbb{Z}}
      \, \subset \,
      \mathrm{CE}(\mathfrak{g})\,,
      \\
      \mathrm{deg}(F_{2k})
      \,=\,
      (2k,\,\mathrm{evn}),\,
      \\
      \mathrm{d}\,
      F_{2k} \,=\, 0\, , 
    \end{array}
    \right. 
  $$

  \vspace{2mm} 
\noindent 
or similarly,  sequences of cocycles in degrees $(2k+1)_{k\in \mathbb{Z}}$ when valued instead in $\mathfrak{l}(\Sigma \mathrm{KU})$. 
Notice such sequences may also be thought of as even ``\textit{periodic cocycles}'' of degree $0$ mod $2$ and $1$ mod $2$, respectively.
\end{example}

\newpage 

\noindent
{\bf (III) Whitehead $L_\infty$-algebras of bundles of spectra.}

\begin{example}[{\bf Real Whitehead $L_\infty$-algebra of bundles of spectra} {\cite[\S 2.1]{BMSS19}}]
\label{WhiteheadAlgebraOfBundlesOfSpectra}
Given $X$, $\mathfrak{l}X$ as in Ex. \ref{WhiteheadLInfinityAlgebra} and $E$, $\mathfrak{l}E$ as in Ex. \ref{RealWhiteheadAlgebraOfSpectra}, the Whitehead $L_\infty$-algebras of $E$-fiber $\infty$-bundles 
$E \!\sslash\! \Omega X$
over $X$  are characterized as having 
underlying graded vector space that of $\mathfrak{l}(E) \oplus \mathfrak{l}(X)$ 
with $L_\infty$-brackets
such that the corresponding split exact sequence of graded vector spaces makes a (necessarily homotopy fiber-)sequence of $L_\infty$-homomorphisms:
$$
\small 
  \begin{tikzcd}[
    row sep=3pt,
    column sep=7pt
  ]
    \mathfrak{l}E
    \ar[
      rr,
      "{
        \mathrm{hofib}(p)
      }"
    ]
    \ar[
      d,
      equals,
      shorten >= -3pt
    ]
    &&
    \mathfrak{l}(
    E\!\sslash\!\Omega X)
    \ar[
      rr,
      ->>,
      "{ p }"
    ]
    \ar[
      d,
      equals,
      shorten >= -4pt
    ]
    &&
    \mathfrak{l}X
    \ar[
      d,
      equals,
      shorten >= -3pt
    ]
    \\
    \Big(
    \pi_\bullet(\Omega E)
    \!\otimes_{{}_{\mathbb{Z}}}\!
    \mathbb{R},
    [\mbox{-}, \mbox{-}]_E,
    [\mbox{-}, \mbox{-},\mbox{-}]_E,
    \cdots
    \Big)
    \ar[
      rr,
      hook
    ]
    &&
    \Big(
    \big(
      \pi_\bullet(\Omega E)
      \oplus
      \pi_\bullet(\Omega X)
    \big)
    \!\otimes_{{}_{\mathbb{Z}}}\!
    \mathbb{R}
    ,\,
    \mathcolor{purple}{
    [\mbox{-}, \mbox{-}],
    [\mbox{-}, \mbox{-},\mbox{-}],
    \cdots
    }
    \Big)
    \ar[
      rr,
      ->>
    ]
    &&
    \Big(
    \pi_\bullet(\Omega X)
    \!\otimes_{{}_{\mathbb{Z}}}\!
    \mathbb{R},
    [\mbox{-}, \mbox{-}]_X,
    [\mbox{-}, \mbox{-},\mbox{-}]_X,
    \cdots
    \Big)
    \\
    \Big(
      \wedge^\bullet\big(
        \pi_\bullet(\Omega E)
      \big)^\vee
      ,\,
      \mathrm{d}_E = 0
    \Big)
    &&
    \Big(
      \wedge^\bullet\big(
        \pi_\bullet(\Omega E)
        \oplus 
        \pi_\bullet(\Omega X)
      \big)^\vee
      ,\,
      \mathcolor{purple}{\mathrm{d}}
    \Big)
    \ar[
      ll
    ]
    &&
    \Big(
      \wedge^\bullet(\Omega X)^\vee
      ,\,
      \mathrm{d}_X
    \Big)
    \ar[
      ll
    ]
    \\
    \mathrm{CE}\big(
      \mathfrak{l}E
    \big)
    \ar[
      u, 
      equals,
      shorten >= -3pt
    ]
    &&
    \mathrm{CE}\big(
      \mathfrak{l}(
        E \!\sslash\! \Omega X
      )
    \big)
    \ar[
      u, 
      equals,
      shorten >= -3pt
    ]
    \ar[
      ll
    ]
    &&
    \mathrm{CE}\big(
      \mathfrak{l}(
        X
      )
    \big)
    \ar[
      u, 
      equals,
      shorten >= -3pt
    ]
    \ar[
      ll
    ]
  \end{tikzcd}
$$

\vspace{-1mm} 
\noindent (The only choice is in the shaded brackets/differential in the middle).
\end{example}

\begin{example}[{\bf Real Whitehead $L_\infty$-algebra of twisted K-theory spectrum} {cf. \cite[\S 4]{FSS17}\cite[Ex. 5.7, 6.6]{FSS23-Char}} ]
\label{WhiteheadLInfinityOfTwistedKTheorySpectrum}
  The real Whitehead $L_\infty$-algebras (Ex. \ref{WhiteheadLInfinityAlgebra}) of the classifying spectra $\Sigma^0\mathrm{KU}$ and $\Sigma^1\mathrm{KU}$ for complex topological K-theory canonically homotopy-quotiented by $\mathrm{PU}(\HilbertSpace{H}) \simeq B\mathrm{U}(1)$ have a generator in degree 3 together with generators in every even (every odd) degree, with differential of the form known from 3-twisted de Rham cohomology
 \begin{equation}
   \label{WhiteheadAlgOfTwistedKU}
  \def\arraystretch{1.2}
  \begin{array}{lll}
  \mathrm{CE}\Big(
    \mathfrak{l}
    \big(
      \,
      \Sigma^0\mathrm{KU}
       \!\sslash\!\!
      B\mathrm{U}(1)
      \,
    \big)
  \Big)
  &\simeq&
 \FDGCA
  \Big[
    \grayoverbrace{
      h_3
    }{
      \mathclap{
        \mathrm{deg} \,=\,
        (3,\mathrm{evn})
      }
    }
    ,
    \,
    \big(
    \grayunderbrace{
      f_{2k}
    }{
      \mathclap{
        \mathrm{deg}\,=\,
        (2k, \mathrm{evn})
      }
    }
    \big)_{k \in \mathbb{Z}}
  \Big]
  \Big/\!
  \left(\!
  \def\arraystretch{1}
  \def\arraycolsep{2pt}
  \begin{array}{ccl}
    \mathrm{d}\, h_3
    &=&
    0
    \\
    \mathrm{d}
    \,
    f_{2k+2}
    &=&
    h_3\, f_{2k}
  \end{array}
  \!\right)
  \\[20pt]
  \mathrm{CE}\Big(
    \mathfrak{l}
    \big(
      \,
      \Sigma^1\mathrm{KU}
       \!\sslash\!\!
      B\mathrm{U}(1)
      \,
    \big)
  \Big)
  &\simeq&
  \FDGCA
  \Big[
    \grayoverbrace{
      h_3
    }{
      \mathclap{
        \mathrm{deg} \,=\,
        (3,\mathrm{evn})
      }
    }
    ,
    \,
    \big(
    \grayunderbrace{
      f_{2k+1}
    }{
      \mathclap{
        \mathrm{deg}\,=\,
        (2k+1, \mathrm{evn})
      }
    }
    \big)_{k \in \mathbb{Z}}
  \Big]
  \Big/\!
  \left(\!
  \def\arraystretch{1}
  \def\arraycolsep{2pt}
  \begin{array}{ccl}
    \mathrm{d}\, h_3
    &=&
    0
    \\
    \mathrm{d}
    \,
    f_{2k+3}
    &=&
    h_3\, f_{2k+1}
  \end{array}
  \!\right)
  .
  \end{array}
\end{equation}  
Since the general $h_3$ here is closed, these $L_\infty$-algebras are canonically fibered over the line Lie 2-algebra (Ex. \ref{LineLieNAlgebra}) with the fiber being the Whitehead $L_\infty$-algebra \eqref{CEOfWhiteheadOfKU} of the plain K-theory spectrum:
\begin{equation}
  \label{TwistedKTheoryFibration}
  \begin{tikzcd}[
    row sep=-6pt,
    column sep=5pt
  ]
    \mathfrak{l}\big(
      \Sigma^m\mathrm{KU}
    \big)
    \ar[
      r
    ]
    &[10pt]
    \mathfrak{l}\big(
      \Sigma^m\mathrm{KU}
      \!\sslash\!\!
      B\mathrm{U}(1)
    \big)
    \ar[
      dd,
      ->>
    ]
    &
    h_3
    \\
    &
    &
    \rotatebox[origin=c]{90}{$\longmapsto$}
    \\
    &
    \mathfrak{l}
    \, B^2 \mathrm{U}(1)
    &
    \omega_3
    \mathrlap{\,.}
  \end{tikzcd}
\end{equation}
In rational homotopy theory this is the model for the fibration classifying 3-twisted complex-topological K-theory (cf. \cite[Ex. 3.4, Prop. 6.11, Prop. 10.1]{FSS23-Char}).
\end{example}

A key application of this Ex. \ref{WhiteheadLInfinityOfTwistedKTheorySpectrum} is as the classifying object for 3-twisted cohomology in the familiar sense of \cite[(23) \& appndx]{RohmWitten86}; in fact this is just the first example of a much more general concept of twisted real cohomology \cite[pp 120]{FSS23-Char} as we briefly recall now:

\medskip

\noindent
{\bf Twisted rational cohomology.}
We have seen in Ex. \ref{LineLieNAlgebra}  and Ex. \ref{RationalKtheorySpectra} that the $L_\infty$-algebras $b^n\mathbb{R}$ and $\Sigma^m \mathrm{KU}$ classify, respectively, ordinary $(n+1)$-cocycles and cocycles in 2-periodic degrees. Accordingly, the rational twisted K-theory spectra from Ex. \ref{WhiteheadLInfinityOfTwistedKTheorySpectrum} classify ``$3$-twisted periodic cocycles'' in the sense of \cite[(23) \& appndx]{RohmWitten86}\cite[\S 9.3]{BCMMS02}:

\begin{definition}[\bf 3-Twisted periodic Chevalley-Eilenberg complex]
\label{3TwistedPeriodicChevalleyEilenberg}
Let $\frg$ be a super-$L_\infty$ algebra and  $H_3\in \CE(\frg)$ a closed element in degree $(3,\mathrm{evn})$, to be called the  ``twisting 3-cocycle''. The \textit{$3$-twisted Chevalley--Eilenberg complex} of $\frg$ with respect to $H_3$ is  the $\mathbb{Z}_2\times \mathbb{Z}_2$-graded (periodic mod $2$ and super, respectively) dgca
$$\CE^{\bullet + H_3}(\frg) \, := \, \big(\,\CE(\frg),\, \dd_{H_3} := \dd_{\CE} - H_3 \, \big)\, , $$
where on the right-hand side we abusively write $\CE(\frg)$ for the  graded commutative super-algebra underlying the original Chevalley-Eilenberg \eqref{TheCategoryOfSuperLInfinityAlgebras} dgca of $\frg$.
\end{definition}
It follows immediately (e.g. \cite[Ex. 6.6]{FSS23-Char}) that $\mathfrak{l}\big(\Sigma^0\mathrm{KU}\!\sslash\!\! B\mathrm{U}(1)\big)$ (Ex. \ref{WhiteheadLInfinityOfTwistedKTheorySpectrum}) serves as a classifying object for  $H_3$-twisted (even) cocycles of degree $0$ mod $2$ on $\frg$. Indeed, maps of super-$L_\infty$ algebras from $\frg$ into $\mathfrak{l}\big(
      \Sigma^0\mathrm{KU}
      \!\sslash\!\!
      B\mathrm{U}(1)
    \big)$, which respect the corresponding fiberings over $b^2 \mathbb{R} \cong \mathfrak{l} B^2 U(1) $,  
correspond precisely to sequences $(F_{2k})_{k\in \mathbb{Z}}$ satisfying the $H_3$-twisted closure condition:
\begin{equation}
\label{TwistedKTheorySpectraClassifyTwistedCocycles}
  \begin{tikzcd}[row sep=-3pt, column sep=large]
      \mathfrak{g}
      \ar[
        rr,
        dashed,
        "{ {\big(H_3, \, (F_{2k})_{k\in \mathbb{Z}}\big)} }"
      ]
      \ar[
        dr,
        "{ H_3 }"{swap}
      ]
      &&
      \mathfrak{l}\big(
      \Sigma^0\mathrm{KU}
      \!\sslash\!\!
      B\mathrm{U}(1)
    \big)
      \ar[
        dl,
        "{ h_3 }"
      ]
      \\
      &
      b^2 \mathbb{R}
    \end{tikzcd}
    \hspace{1cm}
    \Leftrightarrow
    \hspace{1cm}
    \left\{\!\!
    \def\arraystretch{1.2}
    \begin{array}{l}
      (F_{2k})_{k\in \mathbb{Z}}
      \, \subset \,
      \mathrm{CE}(\mathfrak{g})\,,
      \\
      \mathrm{deg}(F_{2k})
      \,=\,
      (2k,\,\mathrm{evn}),\,
      \\
      \mathrm{d}\,
      F_{2k} \,=\,  H_3 \cdot F_{2k-2}\, . 
    \end{array}
    \right. 
\end{equation}

Viewing equivalently such a sequence as a ($0 $ mod $2$, evn)  cochain yields precisely an $H_3$-twisted cocycle
$$
(F_{2k})_{k\in\mathbb{Z}} \, \in \, \CE^{0+H_3}(\frg) \qquad \qquad \mathrm{s.t.} \qquad \qquad \dd_{H_3} (F_{2k})_{k\in\mathbb{Z}} \, = \, 0\, .
$$

\vspace{2mm} 
\noindent 
Analogously,  $\mathfrak{l}\big(
      \Sigma^1\mathrm{KU}
      \!\sslash\!\!
      B\mathrm{U}(1)
    \big)$ classifies $H_3$-twisted cocycles of degree ($1$ mod $2$, evn). In view of this canonical identification, we shall refer to $3$-twisted cocycles equivalently as  \textit{rational twisted $K$-theory cocycles} -- they are of the form of images of twisted K-theory classes under the twisted Chern character \cite[Prop. 10.1]{FSS23-Char}.

\smallskip

This situation has an evident generalization to higher degree twists:
Naturally, we may consider twisting the Chevalley-Eilenberg cochain complex of a super-$L_\infty$ algebra by any ordinary cocycle of degree $(2n+1, \mathrm{evn})$, instead, and hence also the corresponding $(2n+1)$-twisted cohomology for any $n\in \NN$.
\begin{definition}[\bf (2n+1)-Twisted periodic Chevalley-Eilenberg complex] \label{TwistedPeriodicChevalleyEilenberg}
Let $\frg$ be a super-$L_\infty$ algebra and  $H_{2n+1}\in \CE(\frg)$ a ``twisting $(2n+1)$-cocycle''. The \textit{$(2n+1)$-twisted Chevalley--Eilenberg complex} of $\frg$ with respect to $H_{2n+1}$ is  the $\mathbb{Z}_{2n}\times \mathbb{Z}_2$-graded (periodic mod $2n$ and super respectively) dgca
\vspace{1mm} 
$$
\CE^{\bullet + H_{2n+1}}(\frg) \, := \, \big(\,\CE(\frg),\, \dd_{H_{2n+1}} := \dd_{\CE} - H_{2n+1} \, \big)\, , 
$$

\vspace{1mm} 
\noindent 
where on the right-hand side we abusively write $\CE(\frg)$ for the  graded commutative super-algebra underlying the original Chevalley-Eilenberg \eqref{TheCategoryOfSuperLInfinityAlgebras} dgca of $\frg$.
\end{definition}

In a similar fashion to the 3-twisted case from Eq. \eqref{TwistedKTheorySpectraClassifyTwistedCocycles}, $(2n+1)$-twisted cocycles correspond precisely to maps into certain classifying super-$L_\infty$ algebras generalizing those of the rational twisted K-theory spectra from Ex. \ref{WhiteheadLInfinityOfTwistedKTheorySpectrum}.

\begin{example}[{\bf $L_\infty$-algebras classifying $(2n+1)$-twisted cocycles} {\cite[Ex. 6.7, Rem. 10.1]{FSS23-Char}}]\label{TwistedPeriodicCocycleClassifyingLinftyAlgebras} For any two positive integers $m,n\in \NN$ with $m<2n$, the classifying super-$L_\infty$ algebra for $(2n+1)$-twisted cocycles in degree $m$ mod $2n$ is defined by
 \begin{equation}
\label{TwistedPeriodicClassifyingCEAlgebras}
  \def\arraystretch{1.2}
  \begin{array}{lll}
  \mathrm{CE}\Big(
    \mathfrak{l}
    \big(
      \,
      \Sigma^m \mathrm{K}^{n}\mathrm{U}
       \!\sslash\!\!
      B^{2n-1}\mathrm{U}(1)
      \,
    \big)
  \Big)
  &\simeq&
 \FDGCA
  \Big[
    \grayoverbrace{
      h_{2n+1}
    }{
      \mathclap{
        \mathrm{deg} \,=\,
        (2n+1,\mathrm{evn})
      }
    }
    ,
    \,
    \big(
    \grayunderbrace{
      f_{2nk+m}
    }{
      \mathclap{
        \mathrm{deg}\,=\,
        (2nk+m, \mathrm{evn})
      }
    }
    \big)_{k \in \mathbb{Z}}
  \Big]
  \Big/\!
  \left(\!
  \def\arraystretch{1}
  \def\arraycolsep{2pt}
  \begin{array}{ccl}
    \mathrm{d}\, h_{2n+1}
    &=&
    0
    \\
    \mathrm{d}
    \,
    f_{2(n+1)k+m}
    &=&
    h_{2n+1}\, f_{2nk}
  \end{array}
  \!\right)
 
  \end{array}
\end{equation}  
In analogy with \eqref{TwistedKTheoryFibration}, since $h_{2n+1}$ is closed, these $L_\infty$-algebras are canonically fibered over the line Lie $2n$-algebra (Ex. \ref{LineLieNAlgebra}):
\begin{equation}\label{HigherTwistedCocycleSpectraFibration}
  \begin{tikzcd}[
    row sep=-6pt,
    column sep=5pt
  ]
    \mathfrak{l}\big(
      \Sigma^m\mathrm{K}^{n}\mathrm{U}
    \big)
    \ar[r]
    &[10pt]
    \mathfrak{l}\big(
      \Sigma^m\mathrm{K}^{n}\mathrm{U}
      \!\sslash\!\!
      B^{2n-1}\mathrm{U}(1)
    \big)
    \ar[
      dd,
      ->>
    ]
    &
    h_{2n+1}
    \\
    &
    &
    \rotatebox[origin=c]{90}{$\longmapsto$}
    \\
    &
    b^{2n}\mathbb{R}
    &
    \omega_{2n+1}
    \mathrlap{\,.}
  \end{tikzcd}
\end{equation}
\end{example}

\smallskip 
Evidently, $\mathfrak{l}\big(\Sigma^m\mathrm{K}^{n}\mathrm{U}\!\sslash\!\! B^{2n-1}\mathrm{U}(1)\big)$ classifies  $H_{2n+1}$-twisted cocycles on $\frg$
in the sense of Def. \ref{TwistedPeriodicChevalleyEilenberg},
since maps of super-$L_\infty$ algebras between the two, which respect the corresponding fiberings over $b^{2n} \mathbb{R}$,   
correspond precisely to sequences $(F_{2nk})_{k\in \mathbb{Z}}$ satisfying the $H_{2n+1}$-twisted closure condition
\begin{equation}\label{HigherTwistedKTheorySpectraClassifyTwistedCocycles}
    \begin{tikzcd}[row sep=-3pt, column sep=large]
      \mathfrak{g}
      \ar[
        rr,
        dashed,
        "{ {\big(H_{2n+1}, \, (F_{2nk})_{k\in \mathbb{Z}}\big)} }"
      ]
      \ar[
        dr,
        "{ H_{2n+1} }"{swap, pos=.8}
      ]
      &&
      \mathfrak{l}\big(
      \Sigma^m \mathrm{K}^{n}\mathrm{U}
      \!\sslash\!\!
      B^{2n-1}\mathrm{U}(1)
    \big)
      \ar[
        dl,
        "{ h_{2n+1} }"{pos=.5}
      ]
      \\
      &
      b^{2n} \mathbb{R}
    \end{tikzcd}
    \hspace{1cm}
    \Leftrightarrow
    \hspace{1cm}
    \left\{\!\!
    \def\arraystretch{1.2}
    \begin{array}{l}
      (F_{2nk+m})_{k\in \mathbb{Z}}
      \, \subset \,
      \mathrm{CE}(\mathfrak{g})\,,
      \\
      \mathrm{deg}(F_{2nk+m})
      \,=\,
      (2kn+m,\,\mathrm{evn}),\,
      \\
      \mathrm{d}\,
      F_{2(n+1)k+m} \,=\,  H_{2n+1} \cdot F_{2nk}\, . 
    \end{array}
    \right. 
\end{equation}

\vspace{1mm} 
\noindent Viewing equivalently such a sequence as a ($m $ mod $2n$, evn)  cochain yields precisely an $H_{2n+1}$-twisted cocycle
$$
(F_{2nk+m})_{k\in\mathbb{Z}} \, \in \, \CE^{m+H_{2n+1}}(\frg) \qquad \qquad \mathrm{s.t.} \qquad \qquad \dd_{H_{2n+1}} (F_{2nk+m})_{k\in\mathbb{Z}} \, = \, 0\, .
$$

\bigskip

\noindent
{\bf Twisted  Nonabelian cohomology.}
The above examples of classifying cocycles in \textit{abelian} $H_{2n+1}$-twisted CE algebras exhibit a clear pattern, namely: A $H_{2n+1}$-twisted periodic cocycle is precisely a lift of the twisting cocycle map $H_{2n+1}: \frg \longrightarrow b^{2n}\mathbb{R}$ along the twisting fibration of the twisted classifying space. This suggests the following general definition of what rational twisted ``non-abelian'' cocycles should be.

\newpage 
\begin{definition}[{\bf Rational Nonabelian Twisted Cocycles} {\cite[Def. 6.7]{FSS23-Char}}]
\label{RationalNonabelianTwistedCocycles}
Let $\frg$ and $ \mathfrak{c}$ be two super-$L_\infty$ algebras, where we think of $\mathfrak{c}$ as a \textit{classifying space}.
\begin{itemize}
\item[\bf (i)]
We call  \textit{(rational nonabelian) $\mathfrak{c}$-cocycles on $\frg$} simply the set of maps of super-$L_\infty$ algebras
$$
  \frg\longrightarrow \mathfrak{c}
  \,.
$$
 
\item[\bf (ii)] Given a fibration
$$
\mathfrak{c} \longrightarrow \widehat{\mathfrak{c}}\xlongrightarrow{h} \mathfrak{b}
$$
and a fixed \textit{``twisting''} $\mathfrak{b}$-cocycle 
$$
  H
    \,:\,
  \frg \longrightarrow \mathfrak{b} 
  \,,
$$ 
we call  \textit{(rational nonabelian) $H$-twisted $\mathfrak{c}$-cocycles} the set of lifts along the fibration $h:\widehat{\mathfrak{c}}\longrightarrow \mathfrak{b}$
\begin{equation*}\label{TwistedCocyclesAsLifts}
    \begin{tikzcd}[
      row sep=5pt, 
      column sep=50pt]
      \mathfrak{g}
      \ar[
        rr,
        dashed,
        shift left=1pt,
        "{
          \scalebox{.7}{
            \color{gray}
            twisted cocycle
          }
        }"
      ]
      \ar[
        dr,
        "{ H }"{swap},
        "{
          \scalebox{.7}{
            \color{gray}
            twist
          }
        }"{sloped}
      ]
      &&
      \widehat{\mathfrak{c}}
      \ar[
        dl,
        "{ h }"
      ]
      \\
      &
     \mathfrak{b}
    \end{tikzcd}
\end{equation*}
\end{itemize}
\end{definition}
\begin{remark}[\bf Twisted rational cohomology]  There is a notion of coboundaries between these twisted non-abelian $L_\infty$-cocycles (given by concordance), thus yielding a corresponding notion of twisted $\mathbb{R}$-rational cohomology \cite[\S 6]{FSS23-Char}, which in particular subsumes the notion of twisted abelian cohomology from Def. \ref{TwistedPeriodicChevalleyEilenberg} \cite[Prop. 6.13]{FSS23-Char}. 
Here we need not further dwell on this.
\end{remark}

Examples of this more general notion of twisting includes the following seemingly simple but important one:

\begin{example}[\bf Relative cocycles]
\label{WhiteheadOfUniversalPrincipal}
Elements $c_n \in \CE(\frg)$ that are ``closed relative to'' a fixed cocycle $k_{n+1}$,  in that $\dd c_n = k_{n+1}$,
are classified by $eb^{n-1}\mathbb{R}$ given by
\begin{equation}
  \label{BasicContractibleCEAlgebra}
  \mathrm{CE}\big(
    eb^{n-1}\mathbb{R}
  \big)
  \;
  \simeq
  \;
  \FDGCA
  \left[
  \def\arraystretch{1}
  \def\arraycolsep{1pt}
  \begin{array}{c}
    c_n
    \\
    k_{n+1}
  \end{array}
  \right]
  \big/ 
  \left(
    \def\arraystretch{1}
    \def\arraycolsep{2pt}
    \begin{array}{ccl}
      \dd\,c_n &=& k_{n+1}, 
      \\
      \dd\, k_{n+1} &=& 0
    \end{array}
  \right)
  \,,
\end{equation}
fibered as follows, this being the image under $\mathfrak{l}$ (Def. \ref{WhiteheadLInfinityAlgebra}, cf. Ex. \ref{LineLieNAlgebra}) of the universal $B^{n-1}\mathrm{U}(1)$-principal $\infty$-bundle (cf. \cite{SS21-EBund}):
$$
  \begin{tikzcd}[column sep=1pt]
    b^{n-1} \mathbb{R}
    \ar[
      r
    ]
    &[15pt]
    eb^{n-1} \mathbb{R}
    \ar[
      d,
      ->>
    ]
    &
    k_{n+1}
    \\
    & 
    b^n \mathbb{R}
    &
    \omega_{n+1}
    \ar[
      u,
      |->,
      shorten=4pt
    ]
  \end{tikzcd}
  \;\;
  \simeq
  \;\;
  \mathfrak{l}
  \left(\!\!\!
    \begin{tikzcd}[row sep=small]
      B^{n-1}\mathrm{U}(1)
      \ar[r]
      &[-5pt]
      E B^{n-1} \mathrm{U}(1)
      \ar[
        d,
        ->>
      ]
      \\
      &
      B^n \mathrm{U}(1)
    \end{tikzcd}
  \!\!\!\right).
$$
in that given $K_{n+1} \,\in\, \mathrm{CE}(\mathfrak{g})$ with $\mathrm{d}\, K_{n+1} \,=\, 0$ then
\begin{equation}
  \begin{tikzcd}
    &&
    e b^{n-1}\mathbb{R}
    \ar[
      d
    ]
    \\
    \mathfrak{g}
    \ar[
      rr,
      "{
        K_{n+1}
      }"{yshift=-1pt, pos=.6}
    ]
    \ar[
      urr,
      dashed,
      "{
        C_n
      }"
    ]
    &&
    b^n \mathbb{R}
  \end{tikzcd}
  \hspace{.8cm}
  \Leftrightarrow
  \hspace{.8cm}
  \mathrm{d}\, C_n
  \;=\;
  K_{n+1}
  \,.
\end{equation}
\end{example}

An example where the twisting cocycle appears as a relative closure but instead via a higher polynomial twisting condition is the following case of $\mathfrak{l}S^4$-cocycles:
\begin{example}[\bf $\mathfrak{l}S^4$-cocycles as twisted $\mathfrak{l}S^7$-cocycles]\label{S4cocyclesAsTwistedS7}
Rational 4-cohomotopy cocycles on a super-$L_\infty$ algebra $\mathfrak{g}$ may equivalently be regarded as  4-twisted (in the sense of Def. \ref{RationalNonabelianTwistedCocycles}) $b^6 \mathbb{R}$-cocycles via the fibration 
\begin{equation}
  \label{b6RtoS4tob3RFibration}
  \begin{tikzcd}[
    row sep=-6pt,
    column sep=5pt
  ]
    b^6\mathbb{R}
    \ar[
      r
    ]
    &[10pt]
   \mathfrak{l}S^4
    \ar[
      dd,
      ->>
    ]
    &
    g_4
    \\
    &
    &
    \rotatebox[origin=c]{90}{$\longmapsto$}
    \\
    &
   b^3 \mathbb{R}
    &
    \omega_4
  \end{tikzcd}
\end{equation}
or, rationally equivalently, as twisted cohomology classified by the fibration \eqref{HHopfFiberationAsSU2Bundle}
\begin{equation}
  \label{S7toS4tobS3Fibration}
  \begin{tikzcd}[
    row sep=-6pt,
    column sep=5pt
  ]
    \mathfrak{l}S^7
    \ar[
      r
    ]
    &[10pt]
   \mathfrak{l}S^4
    \ar[
      dd,
      ->>
    ]
    &
    g_4
    \\
    &
    &
    \rotatebox[origin=c]{90}{$\longmapsto$}
    \\
    &
   \mathfrak{l}BS^3
    &
    \omega_4
  \end{tikzcd}
\end{equation}
in that given $G_4$ with $\mathrm{d}G_4 \,=\, 0$ then
\begin{equation*}
  \begin{tikzcd}
    &&
    \mathfrak{l}S^4
    \ar[d]
    \\
    \mathfrak{g}
    \ar[
      rr,
      "{ G_4 }"
    ]
    \ar[
      urr,
      dashed,
      "{ G_7 }"
    ]
    &&
    b^3\mathbb{R} 
  \end{tikzcd}
  \hspace{.8cm}
  \Leftrightarrow
  \hspace{.8cm}
  \begin{tikzcd}
    &&
    \mathfrak{l}S^4
    \ar[d]
    \\
    \mathfrak{g}
    \ar[
      rr,
      "{ G_4 }"
    ]
    \ar[
      urr,
      dashed,
      "{ G_7 }"
    ]
    &&
    \mathfrak{l} B S^3
  \end{tikzcd}
  \hspace{.8cm}
  \Leftrightarrow
  \hspace{.8cm}
  \mathrm{d}\, G_7
  \,=\,
  \tfrac{1}{2} G_4 \, G_4
  \,.
\end{equation*}
This exhibits $G_7$ as being a 7-cocycle twisted by $G_4$, cf. \eqref{11dSuperFluxInIntroduction} -- in somewhat subtle variation of the familiar situation of 3-twisted cohomology in \eqref{TwistedKTheorySpectraClassifyTwistedCocycles} -- as suggested in \cite[\S 3]{Sati06-Duality}.
\end{example}

Another more subtle variant of Ex. \ref{WhiteheadOfUniversalPrincipal} is:

\begin{example}[{\bf Real Whitehead $L_\infty$-algebra of quaternionic Hopf fibration} {\cite[Prop. 3.20]{FSS20-H}}]  
\label{WhiteheadAlgebraOfQuaternionicHopf}
The relative real Whitehead $L_\infty$-algebra of the quaternionic Hopf fibration over the 4-sphere is $\mathfrak{l}_{S^4} S^7$ given by
$$
  \mathrm{CE}\big(
    \mathfrak{l}_{S^4} S^7
  \big)
  \;\simeq\;
  \FDGCA
  \left[
    \def\arraystretch{1.1}
    \def\arraycolsep{2pt}
    \begin{array}{c}
      g_4
      \\
      g_7
      \\
      h_3
    \end{array}
  \right]
  \Big/
  \left(\!
    \def\arraystretch{1.1}
    \begin{array}{ccl}
      \mathrm{d}\, 
      g_4
      &=&
      0
      \\
      \mathrm{d}\, 
      g_7
      &=&
      \tfrac{1}{2}
      \,
      g_4
      \,
      g_4
      \\
      \mathrm{d}\, h_3
      &=&
      g_4
    \end{array}
  \!\!\! \right)
$$
and fibered over the 4-sphere (Ex. \ref{WhiteheadAlgebraOfS4}) as
$$
  \begin{tikzcd}[
    column sep=1pt
  ]
    \mathfrak{l}S^3
    \ar[r]
    &[15pt]
    \mathfrak{l}_{S^4}S^7
    \ar[d]
    &
    g_4
    &
    g_7
    \\
    &
    \mathfrak{l}S^4
    &
    g_4
    \ar[
      u,
      |->,
      shorten=5pt
    ]
    &
    g_7
    \ar[
      u,
      |->,
      shorten=5pt
    ]
  \end{tikzcd}
  \;\;\;\simeq\;\;\;
  \mathfrak{l}
  \left(\!\!
    \begin{tikzcd}[row sep=12pt]
      S^3
      \ar[r]
      &
      S^7
      \ar[
        d,
        "~h_{\mathbb{H}}"
      ]
      \\
      &
      S^4
    \end{tikzcd}
  \!\!\right)
  .
$$
This means that the twisted non-abelian cocycles (Def. \ref{RationalNonabelianTwistedCocycles}) classified by the quaternionic Hopf fibration (a twisted form of 3-Cohomotopy) are rationally given by 3-coboundaries of the 4-form datum (only): Given a twisting cocycle in rational 4-Cohomotopy
$$
  \begin{tikzcd}
    \mathfrak{g}
    \ar[rr, "{(G_4, G_7)}"]
    &&
    \mathfrak{l}S^4
  \end{tikzcd}
  \;\;\;\leftrightsquigarrow\;\;\;
  \left(
    G_4, \, G_7 \,\in\,
    \mathrm{CE}(\mathfrak{g})
    \;\middle\vert\;
    \def\arraystretch{1.2}
    \begin{array}{ccl}
      \mathrm{d}\, G_4 &=& 0
      \\
      \mathrm{d}\, G_7 &=&
      \tfrac{1}{2}G_4 \, G_4
    \end{array}
  \!\! \right),
$$
then the corresponding twisted 3-Cohomotopy cocycles are
$$
  \begin{tikzcd}[row sep=-3pt, column sep=huge]
    \mathfrak{g}
    \ar[
      dr,
      "{
        (G_4, G_7)
      }"{sloped, swap}
    ]
    \ar[
      rr,
      dashed,
      "{ H_3 }"
    ]
    &&
    \mathfrak{l}_{S^4}S^7
    \ar[dl]
    \\
    &
    \mathfrak{l}S^4
  \end{tikzcd}
  \hspace{.7cm}
    \leftrightsquigarrow
  \hspace{.7cm}
  \big(
    H_3
    \,\in\,
    \mathrm{CE}(\mathfrak{g})
    \,\big\vert\,
    \mathrm{d}\, H_3 \,=\, G_4
  \big)
  \,.
$$
\end{example}

\subsection{Ext/Cyc adjunction}
\label{ExtCycAdj}

With (extended) super-spacetimes understood --- via their translational super-symmetry (Ex. \ref{SupersymmetryAlgebras}) --- as (higher) super-Lie algebras, fundamental constructions of super-Lie theory have (rational/infinitesimal) geometric significance. Notably the process of {\it central extension} (Def. \ref{CentralExtension}) of super-$L_\infty$ algebras by 2-cocycles corresponds in the super-geometric interpretation to the emergence of extra dimensions by 0-brane condensation (\cite[\S 2]{CdAIPB00}\cite[Rem. 3.11]{FSS15-HigherWZW}\cite{HS18}, see Ex. \ref{11dAsExtensionFromIIA}, \ref{IIAasExtensionOf9d} below).

\smallskip

One may hence ask for the (higher super) Lie-theoretic incarnation of the geometrically expected process of double \footnote{
  The term ``double dimensional reduction'' originates with \cite{DHIS87}, referring to the 
  fact that for Kaluza-Klein reduction of target spaces for $p$-branes both the target spacetime as well as the worldvolume of {\it wrapping} branes reduces in dimension -- or, essentially equivalently, that also the corresponding flux densities decrease in degree upon integration over the fiber spaces. This is, of course, the very mechanism that underlies the emergence of fields with enhanced/exceptional symmetry in lower dimensions.
} dimensional Kaluza-Klein reduction --- and conversely: oxidation --- along such extensions. 
Remarkably, this is given by the process of {\it cyclification} (passage to loop spaces homotopy-quotiented by loop rotation, as known from cyclic cohomology and from the geometric motivation for the Witten genus).
On the rational-homotopy level of  super-$L_\infty$-algebras this is due to \cite[\S 3]{FSS18-TDualityA}\cite[\S 2.6]{FSS18-TDualityB}, recalled as Def. \ref{Cyclification} and Prop. \ref{TheExtCycAdjunction} below (for exposition see \cite[\S 4]{Sc16-Branes}, for more in the context of Mysterious Triality and 
U-duality within a bosonic CDGA algebraic approach see \cite{SV23-Mysterious}\cite{SV23-Mysterious},  for the topological globalization see \cite[\S 2.2]{BMSS19}\cite{SS24-Cyclic} and for its application to double-field theory see \cite{Alfonsi20}\cite{Alfonsi21}).

\begin{definition}[\bf Central extension of super-$L_\infty$ algebra by 2-cocycle]
\label{CentralExtension}
Given $\mathfrak{g} \,\in\, \mathrm{sLieAlg}_{\infty}$ and a 2-cocycle
$$
  \omega_2
  \,\in\,
  \mathrm{CE}(
    \mathfrak{g}
  )
  \,,
  \;\;
  \mathrm{deg}(\omega_2)
  \,=\,
  (2,\mathrm{evn})
  \,,
  \;\;
  \mathrm{d}\,\omega_2
  \,=\,
  0
  \hspace{1cm}
  \Leftrightarrow
  \hspace{1cm}
  \begin{tikzcd}
    \mathfrak{g}
    \ar[
      rr,
      "{ \omega_2 }"
    ]
    &&
    b\mathbb{R}
  \end{tikzcd}
$$
then the corresponding {\it central extension} $\widehat{\mathfrak{g}} \,\in\, \mathrm{sLieAlg}_{\mathbb{R}}$ 
is that super-Lie algebra whose CE-algebra is that of $\mathfrak{g}$ with one more generator $e'$ adjoined whose differential is $\omega_2$:
$$
  \mathrm{CE}\big(
    \widehat{\mathfrak{g}}
  \big)
  \;\;
  =
  \;\;
  \mathrm{CE}(
    \mathfrak{g}
  )\big[
    \grayunderbrace{
      e'
    }{
      \mathclap{
      \mathrm{deg}
      =
      (1,\mathrm{evn})
      }
    }
  \big]
  \big/
  \big(
    \mathrm{d}\, e'
    \;=\;
    \omega_2
  \big)
  \hspace{1cm}
  \Leftrightarrow
  \hspace{1cm}
  \begin{tikzcd}
    \widehat{\mathfrak{g}}
    \ar[
      d,
      ->>,
      "{
        p\,:=\,
        \mathrm{hofib}(\omega_2)
      }"{pos=.36}
    ]
    \\
    \mathfrak{g}
    \ar[
      rr,
      "{
        \omega_2
      }"
    ]
    &&
    b\mathbb{R}
    \,.
  \end{tikzcd}
$$
\end{definition}

\begin{remark}[\bf Basic and fiber forms on a centrally extended super-$L_\infty$ algebra]
\label{BasicAndFiberForms}
$\,$

\noindent {\bf (i)}  
Given a central extension  as in Def. \ref{CentralExtension}, every element in its CE-algebra decomposes uniquely as the sum \footnote{
    Beware that \cite[(1)]{FSS18-TDualityA} and \cite[(21)]{FSS20-HigherT} have a minus sign on the second summand in \eqref{BasicAndFiberFormDecomposition}. This is, of course, a possible convention in itself, but breaks the desirable property of $p_\ast$ being a graded derivation \eqref{FiberIntegration}, that we want to retain here. With the plus sign in \eqref{BasicAndFiberFormDecomposition} we get the corresponding minus sign in \eqref{CyclificationHomBijection} below, correspondingly differing from the sign in \cite[(3)]{FSS18-TDualityA}.
  }
  \begin{equation}
    \label{BasicAndFiberFormDecomposition}
    \alpha
    \;=\;
    \alpha_{\mathrm{bas}}
    \,+\,
    e' \, p_\ast(\alpha)
  \end{equation}
  of a {\it basic form} (not involving the generator $e$, hence in the image of the pullback $p^\ast$)
  $$
    \alpha_{\mathrm{bas}}
    \,\in\,
    p^\ast
    \big(
    \mathrm{CE}(\mathfrak{g})
    \big)
  $$ 
  and the product of the generator $e'$ with the image of $\alpha$ under {\it fiber integration} $p_\ast$, which is a super-graded derivation of degree $(-1,\mathrm{evn})$:
  \begin{equation}
    \label{FiberIntegration}
    \begin{tikzcd}[row sep=-3pt, column sep=0pt]
    \mathrm{CE}(\, \widehat{\mathfrak{g}}\,)
    \ar[
      rr,
      "{
        p_\ast
      }"
    ]
    &&
    \mathrm{CE}(\mathfrak{g})
    \\
    e' &\longmapsto& 1
    \\
    e^i &\longmapsto& 0
    \end{tikzcd}
  \end{equation}
  (where in the last line $(e^i)_{\i \in I}$ denote generators for $\mathrm{CE}(\mathfrak{g})$).

\noindent {\bf (ii)}   The differential of a general element is given in this decomposition in terms of
(the image under $p^\ast$ of) the differential $\mathrm{d}_{\mathfrak{g}}$ by:
  \begin{equation}
    \label{DifferentialOnBasicDecomposition}
    \def\arraystretch{1.4}
    \begin{array}{ll}
      \mathrm{d}_{\widehat{\mathfrak{g}}}
      \big(
        \alpha_{\mathrm{bas}}
        \,+\,
        e'
        \,
        p_\ast \alpha
      \big)
      &
      =\;
      \mathrm{d}_{\widehat{\mathfrak{g}}}
      \,
      \alpha_{\mathrm{bas}}
      \,+\,
      \big(
        \mathrm{d}_{\widehat{g}}
        \,
        e'
      \big)
      \,
      p_\ast \alpha
      \,-\,
      e'
      \,
      \mathrm{d}_{\widehat{\mathfrak{g}}}
      \,
      p_\ast \alpha
      \\
      &=\;
      \big(
      \mathrm{d}_{\mathfrak{g}}
      \alpha_{\mathrm{bas}}
      \,+\,
      \omega_2\, 
      p_\ast \alpha
      \big)
      \,-\,
      e'
      \,
      \mathrm{d}_{\mathfrak{g}}
      p_\ast \alpha
      \,,
    \end{array}
\end{equation}
which further implies that the fiber integration is a morphism of cochain complexes
$$
 p_\ast \circ \mathrm{d}_{\widehat{\mathfrak{g}}}\,  = \, - \, \mathrm{d}_\frg \circ p_\ast \, .
$$
\end{remark}

\begin{definition}[{\bf Cyclification of super $L_\infty$-algebras}, {cf. \cite[Prop. 3.2]{FSS17}\cite[Def. 3.3]{FSS18-TDualityA}}]
  \label{Cyclification}
  Given $\mathfrak{h} \,\in\, \mathrm{sLieAlg}_\infty^{\mathrm{fin}}$ with presentation
  $
    \mathrm{CE}(\mathfrak{h})
    \;\simeq\;
    \FDGCA \big[
      (e^i)_{i \in I}
    \big]
    \big/
    \big(
      \mathrm{d}\, e^i
      =
      P^i(\vec e\,)
    \big)_{i \in I}
    $,
its {\it cyclification} $\mathrm{cyc}(\mathfrak{h}) \,\in\, \mathrm{sLieAlg}_\infty$ is given by
\begin{equation}
    \label{cykCE}
    \mathrm{CE}\big(
      \mathrm{cyc}(\mathfrak{h})
    \big)
    \;\;
    :=
    \;\;
    \FDGCA
    \Bigg[
      \adjustbox{
        raise=-4pt
      }{$
      \def\arraystretch{1.5}
      \begin{array}{l}
        \big(e^i\big)_{i \in I},
        \grayoverbrace{
          \omega_2,
        }{
          \mathclap{
          \mathrm{deg} 
          \,=\,
          (2,\mathrm{evn})
          }
        }
        \\
        \big(
        \grayunderbrace{
          \mathrm{s} e^i
        }{
         \mathclap{
           \scalebox{.7}{$
           \def\arraystretch{.9}
           \begin{array}{c}
           \mathrm{deg}
           \,=\,
           \\
           \mathrm{deg}(e^i) - (1,\mathrm{evn})
           \end{array}
           $}
          }
        }
        \big)_{i \in I}
      \end{array}
      $}
    \Bigg]
    \Big/
    \left(
    \def\arraystretch{1}
    \def\arraycolsep{2pt}
    \begin{array}{lcl}
      \mathrm{d}\, 
      \omega_2
      &=&
      0
      \\
      \mathrm{d}\, e^i
      &=&
      \mathrm{d}_{\mathfrak{h}}
      \, e^i
      +
      \omega_2\, 
      \mathrm{s} e^i
      \\
      \mathrm{d}\, \mathrm{s} e^i
      &=&
      -
      \mathrm{s}\big(
        \mathrm{d}_{\mathfrak{h}}
        \,
        e^i
      \big)
    \end{array}
   \!\! \right)
    ,
  \end{equation}
  where in the last line on the right the shift is understood as uniquely extended to a super-graded derivation of degree $(-1, \mathrm{evn})$:
\vspace{-1mm} 
  $$
    \begin{tikzcd}[row sep=-3pt,
      column sep=0pt
    ]
      \mathllap{
        \mathrm{s}
        :
        \;
      }
      \mathrm{CE}\big(
        \mathrm{cyc}(\mathfrak{h})
      \big)
      \ar[
        rr
      ]
      &&
      \mathrm{CE}\big(
        \mathrm{cyc}(\mathfrak{h})
      \big)
      \\
      \omega_2
      &\longmapsto&
      0
      \,,
      \\
      e^i &\longmapsto&
      \mathrm{s}e^i
      \,,
      \\
      \mathrm{s}e^i 
      &\longmapsto&
      0
      \,.
    \end{tikzcd}
  $$
\end{definition}
\ifdefined\JournalVersion
It is straightforward to check that this is well-defined:
\else
To check that this is well-defined:
\fi
\begin{lemma}[\bf Differential and shift in cyclification]
\label{DifferentialAndShiftInCyclification}
In Def. \ref{Cyclification}
the differential $\mathrm{d}$ and shift $\mathrm{s}$
square to zero and anti-commute with each other:
\begin{equation}
    \label{ShiftAnticommutesWithDifferential}
    \mathrm{d}\, \mathrm{d}
    \,=\, 0
    \,,
    \qquad
    \mathrm{s}\, \mathrm{s}
    \,=\, 0
    \,,
    \qquad 
    \mathrm{s}\, \mathrm{d}
    +
    \mathrm{d}\, \mathrm{s}
    \;=\;
    0
    \,.
\end{equation}
\end{lemma}
\begin{proof}
Since we are dealing with (graded) derivations and their (graded) commutator, it is sufficient to check all these statements on generators.
\ifdefined\JournalVersion
This is straightforward.
\else
Here it is immediate from the definition that $\mathrm{s}$ squares to zero,
while the mixed anticommutativity is seen as:
  $$
    \def\arraystretch{1.4}
    \begin{array}{ll}
    \mathrm{s}\, \mathrm{d} \, \omega_2
    +
    \mathrm{d}\, \mathrm{s} \, \omega_2
    &
    =\; 0 + 0
   \;=\; 0\;,
    \,
    \\
      \mathrm{s}\, \mathrm{d}\,
      e^i
      \,+\,
      \mathrm{d}\, \mathrm{s} \,  e^i
      &
      =\;
      \mathrm{s}
      \big(
        \mathrm{d}_{\mathfrak{h}}
        e^i
        +
        \omega_2 \,
        \mathrm{s} 
        e^i
      \big)
      -
      \mathrm{s} \, 
      \mathrm{d}_{\mathfrak{h}}
      e^i
      \;=\;
      0\;,
    \\
      \mathrm{s}\,
      \mathrm{d}\,
      \mathrm{s}e^i
      +
      \mathrm{d}\,
      \mathrm{s}\,
      \mathrm{s}e^i
      &
      =\;
      -
      \mathrm{s}\,
      \mathrm{s}\,
      \mathrm{d}_{\mathfrak{h}}e^i
      \;=\;
      0
      \,.
    \end{array}
  $$
Finally, for nilpotency of $\mathrm{d}$ we first trivially have $\mathrm{d}\, \mathrm{d}\, \omega_2 \,=\, 0$, then
$$
  \def\arraystretch{1.3}
  \begin{array}{ll}
    \mathrm{d}\,\mathrm{d}\,e^i
    & \;=\;
    \mathrm{d}\big(
      \mathrm{d}_{\mathfrak{h}}e^i
      +
      \omega_2\, \mathrm{s}e^i
    \big)
    \\
    &
    \;=\;
    \grayunderbrace{
    \mathrm{d}_{\mathfrak{h}}
    \,
    \mathrm{d}_{\mathfrak{h}}
    }{
      \mathclap{
        = 0
      }
    }
    e^i
    \,+\,
    \omega_2 \, s\, 
    \mathrm{d}_{\mathfrak{h}}
    e^i
    \,+\,
    \omega_2 
    \,
    \mathrm{d}
    (
      \mathrm{s}
      \,
      e^i
    )
    +
    \omega_2 
    \,
    \omega_2
    \grayunderbrace{
      \mathrm{s} \, \mathrm{s}
    }{
      \mathclap{ = 0 }
    }
    \, e^i
    \\
    &\;=\;
    \omega_2\big(
      \mathrm{s}
        \, \mathrm{d}_{\mathfrak{h}}
      -
      \mathrm{s}
        \, \mathrm{d}_{\mathfrak{h}}
    \big)
    e^i
    \;=\;
    0
    \,.
  \end{array}
$$
From this, finally:
$$
  \def\arraystretch{1.4}
  \begin{array}{ll}
    \mathrm{d}
    \,
    \mathrm{d}\,
    \mathrm{s} \, e^i
    &
    =\;
    \mathrm{s}
    \,
    \mathrm{d}
    \,
    \mathrm{d}\,
    \, e^i
    \;=\;
    0
    \,.
  \end{array}
$$
\vspace{-4mm} 
\fi
\end{proof}

The following statement is due to \cite[Thm. 3.8]{FSS18-TDualityA}, but we give a streamlined proof with more details. 
\begin{proposition}[{\bf The Ext/Cyc-adjunction}]
  \label{TheExtCycAdjunction}
  Given $\mathfrak{g}, \mathfrak{h} \in \mathrm{sLieAlg}_\infty$
  with a 2-cocycle \footnote{
    \label{OnFirstChernClasses}
    We usually give all algebra generators a subscript indicative of their degree. But here we write ``$c_1$'' since this is the standard symbol for the 1st Chern class of a line bundle, namely here for the Lie-theoretic line bundle $\widehat{\mathfrak{g}} \twoheadrightarrow \mathfrak{g}$.
  } $c_1 \,\in\, \mathrm{CE}(\mathfrak{g})$, there is a bijection
  between:
  
 \noindent {\bf  (i)} maps into $\mathfrak{h}$ out of the central extension $\widehat{\mathfrak{g}}$ classified by the 2-cocycle
  {\rm (Def. \ref{CentralExtension})},
  
\noindent {\bf   (ii)} maps out of $\mathfrak{g}$ into the cyclification of $\mathfrak{h}$ 
  {\rm (Def. \ref{Cyclification})} that preserve the 2-cocycle:
  \begin{equation}
    \label{CyclificationHomIsomorphism}
    \Big\{\!
    \begin{tikzcd}
      \widehat{\mathfrak{g}}
      \ar[
        rr,
        "{ f }"
      ]
      &&
      \mathfrak{h}
    \end{tikzcd}
   \! \!\Big\}
    \begin{tikzcd}[
      column sep=85pt
    ]
      \ar[
        r,
        shift left=5pt,
        "{  \scalebox{.7}{\color{darkgreen}
            \bf
            reduction}\;\;
          \mathrm{rdc}_{c_1}
        }",
        "{ \sim }"{swap, yshift=-2pt}
      ]
      \ar[
        r,
        <-,
        shift right=5pt,
        "{ \scalebox{.7}{
            \color{darkgreen}
            \bf
            oxidation
          }\;\;
          \mathrm{oxd}_{c_1}
        }"{swap},
      ]
      &
      {}
    \end{tikzcd}
    \Bigg\{\!\!
    \begin{tikzcd}[row sep=-3pt, column sep=large]
      \mathfrak{g}
      \ar[
        rr,
        "{ \widetilde f }"
      ]
      \ar[
        dr,
        "{ c_1 }"{swap}
      ]
      &&
      \mathrm{cyc}(\mathfrak{h})
      \ar[
        dl,
        "{ \omega_2 }"
      ]
      \\
      &
      b \mathbb{R}
    \end{tikzcd}
    \!\!\!\Bigg\}
  \end{equation}
  given by
\begin{equation}
  \label{CyclificationHomBijection}
  \hspace{-2.7cm} 
  \begin{tikzcd}[sep=0pt]
    \widehat{\mathfrak{g}}
    \ar[
      rr,
      "{ f }"
    ]
    &&
    \mathfrak{h}
    \\
    \alpha^i_{\mathrm{bas}}
    +
    e' \, p_\ast \alpha^i
    &\longmapsfrom&
    e^i
  \end{tikzcd}
  \hspace{1cm}
  \leftrightsquigarrow
  \hspace{1cm}
  \begin{tikzcd}[
    row sep=-1pt, 
    column sep=0pt]
    \mathfrak{g}
    \ar[
      rr,
      "{ 
        \widetilde{f} 
      }"
    ]
    &&
    \mathrm{cyc}(\mathfrak{h})
    \\
    \alpha^i_{\mathrm{bas}}
    &\longmapsfrom&
    e^i
    \\
    -
    p_\ast \alpha^i
    &\longmapsfrom&
    \mathrm{s}e^i
    \\
    c_1
    &\longmapsfrom&
    \omega_2
    \mathrlap{\,.}
  \end{tikzcd}
\end{equation}
\end{proposition}

\begin{proof}
The assignment \eqref{CyclificationHomBijection}
is manifestly a bijection of maps of underlying graded super-algebras. Hence, it suffices to show that if one of these is, moreover, a homomorphism of dg-algebras (in that it preserves the differential), then so is its image.

To that end, first note that when the map on the left of \eqref{CyclificationHomBijection} is a dg-homomorphism then this implies that
\begin{equation}
  \label{PropertyOfHomOutOfExt}
  \def\arraystretch{1.5}
  \begin{array}{lll}
    f^\ast(\mathrm{d}_{\mathfrak{h}}e^i)
    &
    \;=\;
    \mathrm{d}_{\widehat{\mathfrak{g}}}
    f^\ast(e^i)
    &
    \proofstep{
      by homomorphy
    }
    \\
    &\;=\;
    \mathrm{d}_{\widehat{\mathfrak{g}}}
    (
      \alpha^i_{\mathrm{bas}}
      +
      e'\, p_\ast \alpha^i
    )
    &
    \proofstep{
      by 
      \eqref{CyclificationHomBijection}
    }
    \\
    &\;=\;
    (
      \mathrm{d}_{\mathfrak{g}}
      \alpha^i_{\mathrm{bas}}
      +
      \omega_2\, p_\ast \alpha^i
    )
    -
    e'\, 
    \mathrm{d}_{\mathfrak{g}}
    \, p_\ast \alpha^i
    &
    \proofstep{
      by \eqref{DifferentialOnBasicDecomposition}
      ,
    }
  \end{array}
\end{equation}
while the map on the right being an algebra homomorphism already implies (seen e.g. by expanding in generators): 
\begin{equation}
  \label{PropertyOfHomIntoCyk}
  \def\arraystretch{1.5}
  \begin{array}{rcl}
 \widetilde{f}{}^\ast\big(
   \mathrm{d}_{\mathfrak{h}}
   e^i
 \big)
 &=&
 \big(
   f^\ast(
     \mathrm{d}_{\mathfrak{h}}
     e^i
  )
 \big)_{\mathrm{bas}}
 \\
  \widetilde{f}{}^\ast
  (\mathrm{s}\,\mathrm{d}_{\mathfrak{h}}e^i)
  &=&
  -p_\ast\big(
    f^\ast(\mathrm{d}_{\mathfrak{h}}e^i)
  \big)
  \,.
  \end{array}
\end{equation}
If the map $\widetilde{f}$ on the right is moreover a dg-homomorphism then this implies that the map $f$ on the left is so, as follows:
\begin{equation}
  \label{FurtherPropertyOfHomIntoCyk}
  \def\arraystretch{1.5}
  \begin{array}{lll}
    f^\ast(\mathrm{d}_{\mathfrak{h}}e^i)
    &
    \;=\;
    \big(
      f^\ast(\mathrm{d}_{\mathfrak{h}}e^i)
    \big)_{\mathrm{bas}}
    +
    e'\, p_\ast f^\ast(\mathrm{d}_{\mathfrak{h}}e^i)
    &
    \proofstep{
      by
      \eqref{BasicAndFiberFormDecomposition}
    }
    \\
   & \;=\;
      \widetilde{f}{}^\ast(
        \mathrm{d}_{\mathfrak{h}}
        e^i
      )
    -
    e'\, 
    \widetilde{f}{}^\ast\big(
      \mathrm{s}
      \,
      \mathrm{d}_{\mathfrak{h}}e^i
    \big)
    &
    \proofstep{
      by
      \eqref{PropertyOfHomIntoCyk}
    }
    \\
    &\;=\;
    \widetilde{f}{}^\ast(
      \mathrm{d}_{
        \mathrm{cyc}(\mathfrak{h})
      }
      e^i
      -
      \omega_2 \, \mathrm{s}e^i
    )
    -
    e'\ 
    \widetilde{f}{}^\ast\big(
      \mathrm{s}\, 
      \mathrm{d}_{\mathrm{cyc}(\mathfrak{h})}
      e^i
    \big)
    &
    \proofstep{
      by
      \eqref{cykCE}
    }
    \\
    &\;=\;
    \widetilde{f}{}^\ast(
      \mathrm{d}_{\mathrm{cyc}(\mathfrak{h})}
      e^i
      -
      \omega_2 \, \mathrm{s}e^i
    )
    +
    e'\ 
    \widetilde{f}{}^\ast\big(
      \mathrm{d}_{\mathrm{cyc}(\mathfrak{h})}
      \, 
      \mathrm{s}e^i
    \big)
    &
    \proofstep{
      by
      \eqref{ShiftAnticommutesWithDifferential}
    }
    \\
    &\;=\;
    \mathrm{d}_{\mathfrak{g}}
    \widetilde{f}{}^\ast(e^i)
      -
    \widetilde{f}{}^\ast(
      \omega_2 \,
      \mathrm{s}e^i
    )
    +
    e'\ 
    \mathrm{d}_{\mathfrak{g}}
    \widetilde{f}{}^\ast\big(
      \mathrm{s}e^i
    \big)
    &
    \proofstep{
      by
      homomorphy
    }
    \\
    &\;=\;
    \mathrm{d}_{\mathfrak{g}}
    \alpha^i_{\mathrm{bas}}
    +
    \omega_2
    \,
    p_\ast \alpha^i
    -
    e'\, \mathrm{d}_{\mathfrak{g}}
    p_\ast \alpha^i
    &
    \proofstep{
      by
      \eqref{CyclificationHomBijection}
    }
    \\
   & \;=\;
    \mathrm{d}_{\widehat{g}}
    \big(
      \alpha^i_{\mathrm{bas}}
      +
      e'\, p_\ast \alpha_i
    \big)
    &
    \proofstep{
      by
      \eqref{DifferentialOnBasicDecomposition}
    }
    \\
    &\;=\;
    \mathrm{d}_{\widehat{g}}
    \,
    f^\ast(e^i)
    &
    \proofstep{
      by
      \eqref{CyclificationHomBijection}.
    }
  \end{array}
\end{equation}
Conversely, when $f$ on the left of \eqref{CyclificationHomBijection}
is a dg-homomorphism, respect for the differential on the right is implied:

 $$
   \begin{tikzcd}[
     row sep=4pt,
     column sep=20pt
   ]
     \mathfrak{g}
     \ar[
       rrr,
       "{ \widetilde{f} }"
     ]
     &&&
     \mathrm{cyc}(\mathfrak{h})
     \\[-8pt]
     \alpha^i_{\mathrm{bas}}
     \ar[
       rrr,
       <-|,
       shorten=5pt
     ]
     \ar[
       dd,
       |->,
       shorten=-1pt,
       "{ 
          \mathrm{d} 
        }"{pos=.45}
     ]
     &&&
     e^i
     \ar[
       ddd,
       |->,
       shorten=2pt,
       "{ \mathrm{d} }"
     ]
    \\
    \\
     \mathrm{d}_{\mathfrak{g}}
     \alpha^i_{\mathrm{bas}}
     \ar[
       dr,
       equals,
       "{
         \scalebox{.7}{
           \color{gray}
           \eqref{PropertyOfHomOutOfExt}
         }
       }"{swap}
     ]
     \\
     &
     \big(
     f^\ast(
     \mathrm{d}_{\mathfrak{h}}
     e^i)\big)_{\mathrm{bas}}
     -
     \omega_2 \, 
     p_\ast \alpha^i
     \ar[
       rr,
       <-|,
       shorten=2pt,
       "{
         \scalebox{.7}{
           \color{gray}
           \eqref{PropertyOfHomIntoCyk}
         }
       }"{swap, yshift=-1pt}
     ]
     &&
     \mathrm{d}_{\mathfrak{h}}
     e^i
     +
     \omega_2 \, 
     \mathrm{s}e^i
     \\[-2pt]
     - p_\ast \alpha^i
     \ar[
       rrr,
       <-|,
       shorten=5pt
     ]
     \ar[
       dd,
       |->,
       shorten=-1pt,
       "{ 
          \mathrm{d} 
        }"{pos=.45}
     ]
     &&& 
     \mathrm{s}e^i
     \ar[
       ddd,
       |->,
       shorten=5pt,
       "{ \mathrm{d} }"
     ]
     \\
     \\
     -\mathrm{d}_{\mathfrak{g}}
     \,
     p_\ast \alpha^i
     \ar[
       dr,
       equals,
       "{
         \scalebox{.7}{
           \color{gray}
           \eqref{PropertyOfHomOutOfExt}
         }
       }"{swap}
     ]
     \\
     &
     p_\ast
     f^\ast(
       \mathrm{d}_{\mathfrak{h}}
       e^i
     )
     \ar[
       rr,
       <-|,
       shorten=5pt,
       "{
         \scalebox{.7}{
           \color{gray}
           \eqref{PropertyOfHomIntoCyk}
         }
       }"{swap, yshift=-1pt}
     ]
     &&
     - 
     \mathrm{s}\big(
       \mathrm{d}_{\mathfrak{h}}
       e^i
     \big)
     \\[-2pt]
     \omega_2 
     \ar[rrr, <-|, shorten=5pt]
     \ar[
       d,
       |->,
       shorten=1pt,
       "{ 
          \mathrm{d} 
        }"{pos=.45}
     ]
     &&&
     \omega_2
     \ar[
       d,
       |->,
       shorten=1pt,
       "{\; \mathrm{d} }"
     ]
     \\[+10pt]
     0
     \ar[rrr, <-|, shorten=5pt]
     &&&
     0
     \mathrlap{\,.}
   \end{tikzcd}
 $$

 \vspace{-4mm} 
\end{proof}

\begin{example}[{\bf Cyclification of the 4-Sphere} {\cite[Ex. 3.3]{FSS17}}]
The cyclification (Def. \ref{Cyclification}) of the real Whitehead $L_\infty$-algebra of the 4-sphere (Ex. \ref{WhiteheadLInfinityAlgebra}) is given by:
\begin{equation}
  \mathrm{CE}\big(
    \mathrm{cyc}(
      \mathfrak{l}S^4
    )
  \big)
  \;\simeq\;
  \FDGCA
  \left[
  \def\arraystretch{1}
  \def\arraycolsep{2pt}
  \begin{array}{c}
    \omega_2
    \\
    g_4
    \\
    \mathrm{s}g_4
    \\
    g_7
    \\
    \mathrm{s}g_7
  \end{array}
  \right]
  \Big/
  \left(
  \def\arraystretch{1}
  \def\arraycolsep{2pt}
  \begin{array}{ccl}
    \mathrm{d}\,
    \omega_2 
    &=&
    0
    \\
    \mathrm{d}\,
    g_4
    &=&
    \omega_2\, \mathrm{s}g_4
    \\
    \mathrm{d}
    \,
    \mathrm{s}g_4
    &=&
    0
    \\
    \mathrm{d}
    \,
    g_7
    &=&
    \tfrac{1}{2}
    g_4\, g_4
    +
    \omega_2 \, \mathrm{s}g_7
    \\
    \mathrm{d}\,
    \mathrm{s}g_7
    &=&
    -g_4\, \mathrm{s}g_4
  \end{array}
  \!\right).
\end{equation}
Note that this is fibered over $b^2 \mathbb{R}$ --- in fact over $\mathrm{cyc} \,  b^3 \mathbb{R}$, via \eqref{4SphereFiberedOverKZ4} --- and  as such, remarkably, a rational approximation to the twisted K-theory spectrum \eqref{WhiteheadAlgOfTwistedKU}, via a comparison map to the 6-truncation $\tau_6$ of its underlying space $\Omega^\infty$ (where only generators with degrees in $\{0, \cdots, 6\}$ are kept, cf. \cite[Prop. 4.8]{FSS18-TDualityA}):
$$
  \begin{tikzcd}[
    row sep=-3pt, column sep=10pt
  ]
    &[-10pt] 
    0
    \ar[
      rr,
      <-|,
      shorten=12pt
    ]
    &&
    f_0
    &[-10pt]
    \\
    &[-10pt] 
    \omega_2
    \ar[
      rr,
      <-|,
      shorten=12pt
    ]
    &&
    f_2
    &[-10pt]
    \\
    & 
    g_4
    \ar[
      rr,
      <-|,
      shorten=12pt
    ]
    &&
    f_4
    \\
    & 
    -
    \mathrm{s}g_7
    \ar[
      rr,
      <-|,
      shorten=12pt
    ]
    &&
    f_6
    \\
    \mathrm{s}g_4
    &
    \mathrm{cyc}
    \,
    \mathfrak{l}S^4
    \ar[
      rr
    ]
    \ar[
      dr,
      ->>
    ]
    &&
    \mathfrak{l}\big(
      \tau_6
      \Omega^\infty
      \,
      \mathrm{KU}
      \!\sslash\!
      B\mathrm{U}(1)
    \big)
    \ar[
      dl,
      ->>
    ]
    &&&
    h_3
    \\[5pt]
    &
    &
    b^2 \mathbb{R}
    \\[5pt]
    &&
    \omega_3
    \ar[
      uull,
      |->,
      shorten=10pt
    ]
    \ar[
      uurrrr,
      |->,
      shorten=10pt
    ]
  \end{tikzcd}
$$
A rationale for completing $\mathrm{cyc}\, \mathfrak{l}S^4$ to all of $\mathfrak{l}\big(\mathrm{KU} \!\sslash\!B\mathrm{U}(1)\big)$ is discussed in \cite{BMSS19}, namely by fiberwise {\it stabilization} (i.e., homotopical {\it linearization}) of $\mathrm{cyc} \, \mathfrak{l}S^4$ over $b^2 \mathbb{R}$, as would befit a perturbative approximation to the dimensional reduction of the non-linear Bianchi identity \eqref{11dSuperFluxInIntroduction}. This step is  relevant for a deeper understanding of the lift of T-duality into M-theory indicated in \S\ref{ConclusionAndOutlook}; but its discussion needs an article of its own.
\end{example}

The main example of interest here is the $L_\infty$-algebraic cyclifications of twisted K-theory spectra (Ex. \ref{CyclificationOfTwistedKTTheorySpectra}) since their structure turns out to embody the rational-topological structure of T-duality (Thm. \ref{TwistedKTheoryRedIsoReOxi}, as made concrete in \S\ref{SuperspaceTDuality}), whence we may speak of {\it $L_\infty$-algebraic T-duality} \cite[\S 5]{FSS18-TDualityA}:

\begin{example}[\bf Cyclification of bundle gerbe classifying space]\label{TDualityClassifyingAlgebra}
The cyclifications (Def. \ref{Cyclification}) of the real Whitehead $L_\infty$-algebra of $B^2 \mathrm{U}(1)$ (Ex. \ref{LineLieNAlgebra}), $\mathrm{cyc} \, \mathfrak{l} B^2 \mathrm{U}(1)$, is given by
$$
  \mathrm{CE}\big(
    \mathrm{cyc}
    \,
    \mathfrak{l}
    B^2 \mathrm{U}(1)
  \big)
  \;\simeq\;
  \FDGCA
  \left[
  \def\arraystretch{1.1}
  \def\arraycolsep{1pt}
  \begin{array}{c}
    \omega_2
    \\
    \omega_3
    \\
    \widetilde \omega_2
    :=\mathrm{s}\omega_3
  \end{array}
  \right]
  \Big/
  \left(
  \def\arraystretch{1.1}
  \def\arraycolsep{1pt}
  \begin{array}{ccl}
    \mathrm{d}
    \,
    \omega_2
    &=&
    0
    \\
    \mathrm{d}\,
    \omega_3
    &=&
    \omega_2
    \,
    \widetilde\omega_2
    \\
    \mathrm{d}\,
    \widetilde\omega_2
    &=&
    0
  \end{array}
  \!\right)
$$
being equivalently the higher central extension (Def. \ref{HigherCentralExtension}) of $b\mathbb{R}^2$ by its canonical 4-cocycle
\begin{equation}
  \label{CEOfBOfTDuality2Group}
  \begin{tikzcd}
    b\mathcal{T}\, := \, \mathrm{cyc}
    \,
    \mathfrak{l}
    B^2\mathrm{U}(1)
    \ar[
      rrr,
      "{
        \mathrm{hofib}(
          \omega_2 
          \,
          \widetilde \omega_2
        )
      }"
    ]
    &&&
    b\mathbb{R}^2
    \ar[
      rr,
      "{
        \omega_2 
        \,
        \widetilde \omega_2
      }"
    ]
    &&
    b^3 \mathbb{R}
  \end{tikzcd}
\end{equation}
\ifdefined\JournalVersion
\else
\newpage 
\noindent 
\fi
and as such also known as (the Whitehead $L_\infty$-algebra of) the delooping of the {\it T-duality Lie 2-group} \cite[\S 3.2.1]{FSS13-ExtndedCS}\cite[Rem. 7.2]{FSS18-TDualityA}\cite[\S 3.2]{NikolausWaldorf20}:
$$
  \mathrm{cyc}
  \,
  \mathfrak{l}B^2 \mathrm{U}(1)
  \;\;
  \simeq
  \;\;
  \mathfrak{l}
  B\mathcal{T}
  \;\;
  \defneq
  \;\;
  \mathrm{hofib}
  \Big(\!\!\!
    \begin{tikzcd}
    B\mathrm{U}(1)
    \times
    B\mathrm{U}(1)
    \ar[
      rr,
      "{
        \pi_1 c_1 
          \,\cup\, 
        \pi_2 c_1
      }"
    ]
    &&
    B^3 \mathrm{U}(1)
    \end{tikzcd}
  \!\!\!\Big)
  \,.
$$
\end{example}
It is evident at a glance that \eqref{CEOfBOfTDuality2Group} has an automorphism symmetry given by exchanging the two degree=2 generators (we may as well include a minus sign, for compatibility further below \ref{CyclificationOfTwistedKTTheorySpectra}):
\begin{equation}
  \label{AutomorphismOfCycOfB2U1}
  \begin{tikzcd}[row sep=-3pt,
    column sep=0pt
  ]
   b\mathcal{T}
    \ar[
      rr,
      <->,
      "{ \sim }"
    ]
    &&
    b\mathcal{T}
    \\
    -
    \widetilde \omega_2
    &\longmapsfrom&
    \omega_2
    \\
    -
    \omega_2
    &\longmapsfrom&
    \widetilde \omega_2
    \mathrlap{\,.}
  \end{tikzcd}
\end{equation}
This simplistic example already carries in it the seed of T-duality: The next example, recalled from \cite[Prop. 5.1]{FSS18-TDualityA}, shows that this automorphism lifts to an equivalence between the cyclifications of the 3-twisted K-theory spectra. This serves here as a warmup for the higher-dimensional toroidal super-$L_\infty$ T-duality introduced in \S\ref{TorusReduction}.

\begin{example}[{\bf Cyclification of twisted K-spectra and T-duality 2-group}]
\label{CyclificationOfTwistedKTTheorySpectra}
The cyclifications (Def. \ref{Cyclification}) of the real Whitehead $L_\infty$-algebra of the twisted K-theory spectra (Ex. \ref{WhiteheadLInfinityOfTwistedKTheorySpectrum}) are identified by an isomorphism \eqref{LieHomomorphism}: 

\begin{equation}
  \label{IsoOfDoublyCyclifiedTwistedKSpectra}
  \begin{tikzcd}[row sep=6pt, 
    column sep=2pt,
    ampersand replacement=\&,
    ,/tikz/column 3/.append style={anchor=base west}
  ]
  \mathrm{CE}
  \Big(
    \mathrm{cyc}
    \,
    \mathfrak{l}
    \big(
      \,
      \Sigma^{\color{purple}0}
      \mathrm{KU}
      \!\sslash\!\!B\mathrm{U}(1)
      \,
    \big)
  \Big)
  \&
    \simeq
    \ar[
      d,
      <-,
      shift left=120pt,
      shorten=35pt,
      "{ \sim }"{sloped},
      "{
        \def\arraystretch{.8}
        \def\arraycolsep{2pt}
        \begin{array}{ccccc}
          {\color{purple}-} \omega_2
          &
          {\color{purple}-} \mathrm{s}h_3
          &
          h_3
          &
          f_{2k}
          &
          \mathrm{s}\!f_{2k+2}
          \\
          \rotatebox[origin=c]{90}
            {$\mapsto$}
          &
          \rotatebox[origin=c]{90}
            {$\mapsto$}
          &
          \rotatebox[origin=c]{90}
            {$\mapsto$}
          &
          \rotatebox[origin=c]{90}
            {$\mapsto$}
          &
          \rotatebox[origin=c]{90}
            {$\mapsto$}
          \\ 
          \mathrm{s}h_3
          &
          \omega_2
          &
          h_3
          &
          \mathrm{s}\!f_{2k+1}
          &
          f_{2k+1}
        \end{array}
      }"{swap}
    ]
  \&
 \FDGCA
  \left[\hspace{0pt}
  \def\arraystretch{1.2}
  \def\arraycolsep{2pt}
  \begin{array}{c}
    \omega_2
    \\
    h_3
    \\
    \mathrm{s}h_3
    \\
    (f_{2k})_{k \in \mathbb{Z}}
    \\
    (\mathrm{s}\!f_{2k})_{k \in \mathbb{Z}}
  \end{array}
  \hspace{0pt}
  \right]
  \Big/
  \left(
  \def\arraystretch{1.2}
  \def\arraycolsep{2pt}
  \begin{array}{ccl}
    \mathrm{d}\,\omega_2
    &=&
    0
    \\
    \mathrm{d}\, 
    h_3
    &=&
    \omega_2\, \mathrm{s}h_3
    \\
    \mathrm{d}\,
    \mathrm{s}h_3
    &=&
    0
    \\
    \mathrm{d}
    \,
    f_{2k+2}
    &=&
    h_3\, f_{2k}
    +
    \omega_2\,\mathrm{s}\!f_{2k+2}
    \\
    \mathrm{d}
    \,
    \mathrm{s}\!f_{2k+2}
    &=&
    -
    (\mathrm{s}h_3)f_{2k}
    +
    h_3\, \mathrm{s}\!f_{2k}
    \\
  \end{array}
  \right)
  \\[40pt]
  \mathrm{CE}
  \Big(
    \mathrm{cyc}
    \,
    \mathfrak{l}
    \big(
      \,
      \Sigma^{\color{purple}1}
      \mathrm{KU}
        \!\sslash\!\!B\mathrm{U}(1)
      \,
    \big)
  \Big)
  \&\simeq\&
  \FDGCA
  \left[
  \def\arraystretch{1.2}
  \def\arraycolsep{2pt}
  \begin{array}{c}
    \omega_2
    \\
    h_3
    \\
    \mathrm{s}h_3
    \\
    (f_{2k+1})_{k \in \mathbb{Z}}
    \\
    (\mathrm{s}\!f_{2k+1})_{k \in \mathbb{Z}}
  \end{array}
  \right]
  \Big/
  \left(
  \def\arraystretch{1.2}
  \def\arraycolsep{2pt}
  \begin{array}{ccl}
    \mathrm{d}\,\omega_2
    &=&
    0
    \\
    \mathrm{d}\, 
    h_3
    &=&
    \omega_2\, \mathrm{s}h_3
    \\
    \mathrm{d}\,
    \mathrm{s}h_3
    &=&
    0
    \\
    \mathrm{d}
    \,
    f_{2k+1}
    &=&
    h_3\, f_{2k-1}
    +
    \omega_2\,\mathrm{s}\!f_{2k+1}
    \\
    \mathrm{d}
    \,
    \mathrm{s}\!f_{2k+3}
    &=&
    -
    (\mathrm{s}h_3)f_{2k+1}
    +
    h_3\, \mathrm{s}\!f_{2k+1}
    \\
  \end{array}
  \right).
  \end{tikzcd}
\end{equation}
compatible with their fibration \eqref{TwistedKTheoryFibration}
over $\mathfrak{l}B^2 \mathrm{U}(1)$
via its automorphisms \eqref{AutomorphismOfCycOfB2U1}, where the homotopy fiber of the cyclified fibration is now the direct sum of K-theory spectra in degrees 0 and 1, respectively, with the automorphism acting by swapping them (\cite[Prop. 7.3]{FSS18-TDualityA}):
\begin{equation}
  \label{TDualityOnCyclifiedKFibration}
  \begin{tikzcd}[row sep=-5pt]
    \mathfrak{l}
    \Sigma^{\mathcolor{purple}{0}}\mathrm{KU}
    \ar[
      ddrr
    ]
    &&
    \mathfrak{l}
    \Sigma^{\mathcolor{purple}{1}}\mathrm{KU}
    \\
    \times
    &&
    \times
    \\
    \mathfrak{l}
    \Sigma^{\mathcolor{purple}{1}}\mathrm{KU}
    \ar[
      uurr, 
      crossing over
    ]
    \ar[
      d
    ]
    &&
    \mathfrak{l}
    \Sigma^{\mathcolor{purple}{0}}\mathrm{KU}
    \ar[
      d
    ]
    \\[22pt]
    \mathrm{cyc}
    \,
    \mathfrak{l}
    \big(
      \Sigma^{\mathcolor{purple}{0}}\mathrm{KU}
      \!\sslash\!
      B\mathrm{U}(1)
    \big)
    \ar[
      rr,
      <->,
      "{ \sim }"
    ]
    \ar[
      d,
      ->>
    ]
    &&
    \mathrm{cyc}
    \,
    \mathfrak{l}
    \big(
      \Sigma^{\mathcolor{purple}{1}}\mathrm{KU}
      \!\sslash\!
      B\mathrm{U}(1)
    \big)
    \ar[
      d,
      ->>
    ]
    \\[24pt]
    \mathrm{cyc}\,
    \mathfrak{l}
    B^2 \mathrm{U}(1)
    \ar[
      rr,
      <->,
      "{ \sim }"
    ]
    &&
    \mathrm{cyc}\,
    \mathfrak{l}
    B^2 \mathrm{U}(1)
    \mathrlap{\,.}
  \end{tikzcd}
\end{equation}
\end{example}

\medskip

For comparison in the following Rem. \ref{RedundancyOfExtraIsomorphismsOfCyclifiedTwistedKSpectra}, the next Lemma \ref{FurtherIsomorphismsOfCyclifiedTwistedKspectra} records the simple but noteworthy fact that there exist three further isomorphisms with the property of swapping the 2-cocycles $\mathrm{s}h_3 \leftrightsquigarrow \omega_2$ and the `fluxes' $f_{2k+m}\leftrightsquigarrow \mathrm{s}f_{2k+m}$, up to a consistent choice of signs. This seems to have been previously unnoticed.
\begin{lemma}
[\bf All isomorphisms of cyclified twisted K-spectra]
\label{FurtherIsomorphismsOfCyclifiedTwistedKspectra} 
There are in total 4 isomorphisms
\begin{equation*}
  \begin{tikzcd}[
    row sep=-3pt,
    column sep=15pt
  ]
   \mathrm{cyc}
    \,
    \mathfrak{l}
    \big(
      \,
      \Sigma^{\color{purple}0}
      \mathrm{KU}
      \!\sslash\!\!B\mathrm{U}(1)
      \,
    \big)
    \ar[
      rr,
      <->,
      "{ \sim }"
    ]
    &&
    \mathrm{cyc}
    \,
    \mathfrak{l}
    \big(
      \,
      \Sigma^{\color{purple}1}
      \mathrm{KU}
      \!\sslash\!\!B\mathrm{U}(1)
      \,
    \big)\\
    &&
  \end{tikzcd}
\end{equation*}

\noindent with the property of swapping $\mathrm{s}h_3 \leftrightsquigarrow \pm \omega_2$ and $f_{2k+m}\leftrightsquigarrow \pm \mathrm{s}f_{2k+m+1}$, while mapping $h_3$ to $h_3$, up to relative sign prefactors. 
Explicitly, in addition to \eqref{IsoOfDoublyCyclifiedTwistedKSpectra} one has
\begin{equation}\label{FirstExtraIsoOfDoublyCyclifiedTwistedKSpectra}
  \begin{tikzcd}[
    row sep=-1pt, 
    column sep=0pt]
    \mathrm{cyc}
    \,
    \mathfrak{l}
    \big(
      \,
      \Sigma^{\color{purple}0}
      \mathrm{KU}
      \!\sslash\!\!B\mathrm{U}(1)
      \,
    \big)
       \ar[
      rr,
      <->,
      "{ \sim }"
    ]
    &&
    \mathrm{cyc}
    \,
    \mathfrak{l}
    \big(
      \,
      \Sigma^{\color{purple}1}
      \mathrm{KU}
      \!\sslash\!\!B\mathrm{U}(1)
      \,
    \big)
    \\
    h_3
&\longmapsfrom&
    h_3  
    \\
   - \mathrm{s}h_3
    &\longmapsfrom&
    \omega_2
       \\
   -\omega_2
    &\longmapsfrom&
    \mathrm{s}h_3
    \\
    -\mathrm{s} f_{2k+2}
    &\longmapsfrom&
    f_{2k+1}
 \\
    - f_{2k}
    &\longmapsfrom&
    \mathrm{s}f_{2k+1} 
    \mathrlap{\,,}
  \end{tikzcd}
\end{equation}
\medskip
\begin{equation}\label{SecondExtraIsoOfDoublyCyclifiedTwistedKSpectra}
  \begin{tikzcd}[
    row sep=-1pt, 
    column sep=0pt]
    \mathrm{cyc}
    \,
    \mathfrak{l}
    \big(
      \,
      \Sigma^{\color{purple}0}
      \mathrm{KU}
      \!\sslash\!\!B\mathrm{U}(1)
      \,
    \big)
       \ar[
      rr,
      <->,
      "{ \sim }"
    ]
    &&
    \mathrm{cyc}
    \,
    \mathfrak{l}
    \big(
      \,
      \Sigma^{\color{purple}1}
      \mathrm{KU}
      \!\sslash\!\!B\mathrm{U}(1)
      \,
    \big)
    \\
    h_3
    &\longmapsfrom&
    h_3  
    \\
    \mathrm{s}h_3
    &\longmapsfrom&
    \omega_2
       \\
   \omega_2
    &\longmapsfrom&
    \mathrm{s}h_3
    \\
    -\mathrm{s} f_{2k+2}
    &\longmapsfrom&
    f_{2k+1}
 \\
    f_{2k}
    &\longmapsfrom&
    \mathrm{s}f_{2k+1} 
    \mathrlap{\,,}
  \end{tikzcd}
\end{equation}
and
\begin{equation}\label{ThirdExtraIsoOfDoublyCyclifiedTwistedKSpectra}
  \begin{tikzcd}[
    row sep=-1pt, 
    column sep=0pt]
    \mathrm{cyc}
    \,
    \mathfrak{l}
    \big(
      \,
      \Sigma^{\color{purple}0}
      \mathrm{KU}
      \!\sslash\!\!B\mathrm{U}(1)
      \,
    \big)
       \ar[
      rr,
      <->,
      "{ \sim }"
    ]
    &&
    \mathrm{cyc}
    \,
    \mathfrak{l}
    \big(
      \,
      \Sigma^{\color{purple}1}
      \mathrm{KU}
      \!\sslash\!\!B\mathrm{U}(1)
      \,
    \big)
    \\
    h_3
&\longmapsfrom&
    h_3  
    \\
    \mathrm{s}h_3
    &\longmapsfrom&
    \omega_2
       \\
    \omega_2
    &\longmapsfrom&
    \mathrm{s}h_3
    \\
     \mathrm{s} f_{2k+2}
    &\longmapsfrom&
    f_{2k+1}
 \\
    - f_{2k}
    &\longmapsfrom&
    \mathrm{s}f_{2k+1} 
    \mathrlap{\,.}
  \end{tikzcd}
\end{equation}
Evidently, the original isomorphism \eqref{IsoOfDoublyCyclifiedTwistedKSpectra} and the first above are the two possible extensions of the automorphism \eqref{AutomorphismOfCycOfB2U1} of $b\mathcal{T}$, while the latter two isomorphisms are the two possible extensions of the ``opposite'' automorphism of $b\mathcal{T}$ 
\begin{equation*}
  \begin{tikzcd}[row sep=-3pt,
    column sep=0pt
  ]
   b\mathcal{T}
    \ar[
      rr,
      <->,
      "{ \sim }"
    ]
    &&
    b\mathcal{T}
    \\
    \mathrm{s} h_3
    &\longmapsfrom&
    \omega_2
    \\
    \omega_2
    &\longmapsfrom&\mathrm{s}h_3
    \mathrlap{\,.}
  \end{tikzcd}
  \end{equation*}
\end{lemma}
\begin{proof}This follows by direct inspection. Explicitly, starting (for instance) from the first isomorphism \eqref{IsoOfDoublyCyclifiedTwistedKSpectra} one may ask which possible extra set of signs one can insert in the image of the generators, such that it still commutes with the differentials. The relation $\dd h_3 = \omega_2 \mathrm{s}h_3$ restricts the underlying map of graded commutative algebras to be of the form
\begin{equation*}
  \begin{tikzcd}[
    row sep=-1pt, 
    column sep=0pt]
    \mathrm{cyc}
    \,
    \mathfrak{l}
    \big(
      \,
      \Sigma^{\color{purple}0}
      \mathrm{KU}
      \!\sslash\!\!B\mathrm{U}(1)
      \,
    \big)
       \ar[
      rr,
      <->,
      "{ \sim }"
    ]
    &&
    \mathrm{cyc}
    \,
    \mathfrak{l}
    \big(
      \,
      \Sigma^{\color{purple}1}
      \mathrm{KU}
      \!\sslash\!\!B\mathrm{U}(1)
      \,
    \big)
    \\
    h_3
&\longmapsfrom&
    h_3  
    \\
   - (-1)^{q} \mathrm{s}h_3
    &\longmapsfrom&
    \omega_2
       \\
   - (-1)^{q}\omega_2
    &\longmapsfrom&
    \mathrm{s}h_3
    \\
    (-1)^{x_0}\mathrm{s} f_{2k+2}
    &\longmapsfrom&
    f_{2k+1}
 \\
    (-1)^{x_1} f_{2k}
    &\longmapsfrom&
    \mathrm{s}f_{2k+1} 
    \mathrlap{\, ,}
  \end{tikzcd}
\end{equation*}
for some $q,x_0,x_1 \in \NN$. Demanding that it further commutes with the corresponding differentials on $f_{2k+1}$ (and $\mathrm{s} f_{2k+3}$) yields the condition
$$
(-1)^{q+x_0+x_1}= 1 \, ,
$$
whose set of solutions corresponds to the 3 extra isomorphisms above.
\end{proof}

With cyclification and with these isomorphisms in hand, we already obtain the following general construction, which turns out to be the $L_\infty$-algebraic template of super-space T-duality in \S\ref{SuperspaceTDuality} below, see \eqref{TDualityDiagram} there, for illustration:
\begin{theorem}[\bf Twisted K--theory cocycles under reduction--isomorphism--reoxidation]
\label{TwistedKTheoryRedIsoReOxi}
$\,$

\noindent The composite operation of 

\begin{itemize}

 \item[\bf (a)] reducing \eqref{CyclificationHomIsomorphism}
  twisted $\mathrm{KU}_{\mathcolor{purple}0}$ cocycles on a centrally extended  super $L_\infty$-algebra $\frg_A$
\begin{equation*}
  \begin{tikzcd}[
    row sep=-3pt, 
    column sep=0pt]
    \widehat{\mathfrak{g}}_A
    \ar[
      rr
    ]
    &&
    \mathfrak{l}
    \big(
      \,
      \Sigma^{\color{purple}0}
      \mathrm{KU}
      \!\sslash\!\!B\mathrm{U}(1) \, \big)
    \\
    H_{A}^3
    &\longmapsfrom&
    h_3
    \\
    (F_{2k})_{k\in \mathbb{Z}}
    &\longmapsfrom&
    (f_{2k})_{k\in \mathbb{Z}}
  \end{tikzcd}
\end{equation*}
along its fibration 
$$\widehat{\frg}_A \xlongrightarrow{\quad p_A\quad } \frg \xlongrightarrow{\quad c_A^1 \quad} b\mathbb{R} \, ,$$

\item[\bf (b)] applying the isomorphism \eqref{IsoOfDoublyCyclifiedTwistedKSpectra} on the target cyclification of twisted $\mathrm{KU}_{\mathcolor{purple}0}$, hence viewing them instead as valued in the cyclification of  twisted $\mathrm{KU}_{\mathcolor{purple}1}$, while noticing that this swaps the ``Chern class'' $c_1^A$ from that classifying the $\widehat{\frg}_A$-extension to that classifying a different extension
$$\widehat{\mathfrak{g}}_B \longrightarrow \frg \, ,$$
i.e., via
$$c_B^1 \, := \, {p_A}_* (H_A^3) \quad : \quad  \frg \longrightarrow b \mathbb{R} \, ,
$$
  \item[\bf (c)] re-oxidizing \eqref{CyclificationHomIsomorphism} the result, but now along the new fibration
  $$\widehat{\frg}_B \xlongrightarrow{\quad p_B\quad } \frg \xlongrightarrow{\quad c_B^1 \quad} b\mathbb{R} \, ,$$
\end{itemize}
  results in the twisted $\mathrm{KU}_{\mathcolor{purple}1}$ cocycles given precisely by
\begin{equation}
  \label{Red-Iso-ReoxiAction}
  \begin{tikzcd}[
    row sep=-3pt, 
    column sep=0pt]
    \widehat{\mathfrak{g}}_B
    \ar[
      rr
    ]
    &&
    \mathfrak{l}
    \big(
      \,
      \Sigma^{\color{purple}1}
      \mathrm{KU}
      \!\sslash\!\!B\mathrm{U}(1) \, \big)
    \\
    H^3_{A\, \bas} + e'_B \cdot c_1^A
    &\longmapsfrom&
    h_3
    \\
 \big(- {p_A}_*F_{2k+2} - e'_B \cdot F_{2k\, \bas} \, \big)_{k\in \mathbb{Z}}
    &\longmapsfrom&
    (f_{2k+1})_{k\in \mathbb{Z}} \, .
  \end{tikzcd}
\end{equation}
\end{theorem}
\begin{proof}
This is a matter of carefully tracking through the (bijective) operations on the corresponding sets of $L_\infty$-algebra morphisms. 
Explicitly, under the reduction \eqref{CyclificationHomIsomorphism} from Prop. \ref{TheExtCycAdjunction} the first step yields the map of super $L_\infty$-algebras
\begin{equation*}
  \begin{tikzcd}[
    row sep=-3pt, 
    column sep=0pt]
    \mathfrak{g}
    \ar[
      rr
    ]
    &&
    \mathrm{cyc}
    \,
    \mathfrak{l}
    \big(
      \,
      \Sigma^{\color{purple}0}
      \mathrm{KU}
      \!\sslash\!\!B\mathrm{U}(1)
      \,
    \big)
    \\
    H_{A\, \bas}^3
    &\longmapsfrom&
    h_3
    \\
    c_1^A
    &\longmapsfrom&
    \omega_2
    \\
    -
    {p_A}_\ast H_A^3
    &\longmapsfrom&
    \mathrm{s}h_3
 \\
    F_{2k\, \bas}
    &\longmapsfrom&
    f_{2k}
      \\
   - {p_A}_\ast F_{2k}
    &\longmapsfrom&
    \mathrm{s}f_{2k} 
    \mathrlap{\,.}
  \end{tikzcd}
\end{equation*}
In the second step, postcomposition of the above morphism with the first isomorphism in \eqref{IsoOfDoublyCyclifiedTwistedKSpectra} $$\mathrm{cyc}
    \,
    \mathfrak{l}
    \big(
      \,
      \Sigma^{\color{purple}0}
      \mathrm{KU}
      \!\sslash\!\!B\mathrm{U}(1)
      \,
    \big) \xlongrightarrow{\quad \sim \quad} \mathrm{cyc}
    \,
    \mathfrak{l}
    \big(
      \,
      \Sigma^{\color{purple}1}
      \mathrm{KU}
      \!\sslash\!\!B\mathrm{U}(1)
      \,
    \big)
$$ 
yields
\begin{equation*}
  \begin{tikzcd}[
    row sep=-2pt, 
    column sep=0pt]
    \mathfrak{g}
    \ar[
      rr
    ]
    &&
    \mathrm{cyc}
    \,
    \mathfrak{l}
    \big(
      \,
      \Sigma^{\color{purple}1}
      \mathrm{KU}
      \!\sslash\!\!B\mathrm{U}(1)
      \,
    \big)
    \\
    H_{A\, \bas}^3
    &\longmapsfrom&
    h_3
       \\
   c_B^1 \, := \,  
    {p_A}_\ast H_A^3
    &\longmapsfrom&
    \omega_2
    \\
   - c_1^A
    &\longmapsfrom&
    \mathrm{s}h_3
    \\
    - {p_A}_* F_{2k+2}
    &\longmapsfrom&
    f_{2k+1}
 \\
     F_{2k\, \bas}
    &\longmapsfrom&
    \mathrm{s}f_{2k+1} 
    \mathrlap{\,.}
  \end{tikzcd}
\end{equation*}
Lastly, in the third step oxidizing \eqref{CyclificationHomIsomorphism} via the new 2-cocycle
$$c_B^1 \, := \, {p_A}_* (H_A^3) \quad : \quad  \frg \longrightarrow b \mathbb{R} \, ,
$$
immediately yields precisely the morphism of super $L_\infty$-algebras out of the corresponding central extension
\begin{equation*}
  \begin{tikzcd}[
    row sep=-1pt, 
    column sep=0pt]
    \widehat{\mathfrak{g}}_B
    \ar[
      rr
    ]
    &&
    \mathfrak{l}
    \big(
      \,
      \Sigma^{\color{purple}1}
      \mathrm{KU}
      \!\sslash\!\!B\mathrm{U}(1) \, \big)
    \\
    H^3_{A\, \bas} + e'_B \cdot c_1^A
    &\longmapsfrom&
    h_3
    \\
 \big(- {p_A}_*F_{2k+2} - e'_B \cdot F_{2k\, \bas} \, \big)_{k\in \mathbb{Z}}
    &\longmapsfrom&
    (f_{2k+1})_{k\in \mathbb{Z}} \, .
  \end{tikzcd}
\end{equation*}

\vspace{-6mm}
\end{proof}

\begin{remark}[\bf Towards T-duality]
Since the isomorphism 
\eqref{IsoOfDoublyCyclifiedTwistedKSpectra}
swaps the Chern class $\omega_2$ with the dimensional reduction $\widetilde{\omega}_2 \,\defneq\,\mathrm{s}h_3$ of the 3-form, when applied over super-space this results in  ``swapping the spacetime extension''. At the same time, the same operation swaps the ``winding and non-winding modes'' of the corresponding ``fluxes'', sending $(F_{2k\, \bas}+e_A'  \cdot {p_{A}}_* F_{2k})$ to $(- {p_A}_*F_{2k+2} - e'_B \cdot F_{2k\, \bas} )$, up to an overall sign. These effects may be seen as abstract incarnations of the analogous phenomena in superspace T-duality, shown below in \S\ref{SuperspaceTDualityI}.
\end{remark}

\begin{remark}[\bf Effect of extra isomorphisms of cyclified twisted K-Spectra]\label{RedundancyOfExtraIsomorphismsOfCyclifiedTwistedKSpectra}
Applying instead one of the isomorphisms from Lem. \ref{FurtherIsomorphismsOfCyclifiedTwistedKspectra} in step {\bf (b)} of Thm. \ref{TwistedKTheoryRedIsoReOxi} yields different, but closely related bijections between \textit{rational} twisted $K$-theory cocycles of extensions over $\frg$. More explicitly:

\vspace{1mm} 
\noindent{\bf (i)}
 Using the first isomorphism \eqref{FirstExtraIsoOfDoublyCyclifiedTwistedKSpectra} from Lem. \ref{FurtherIsomorphismsOfCyclifiedTwistedKspectra}, the composite operation's action on twisted cocycles $\big(H_A^3\, , \, (F_{2k})_{k\in \mathbb{Z}}\big) : \widehat{\mathfrak{g}}_A \longrightarrow \mathfrak{l}\big(\,\Sigma^{\color{purple}0}
      \mathrm{KU}
      \!\sslash\!\!B\mathrm{U}(1) \, \big)  $ results instead in the twisted cocycles $\mathrm{KU}_{\mathcolor{purple}1}$
\begin{equation*}
  \begin{tikzcd}[
    row sep=-1pt, 
    column sep=0pt]
    \widehat{\mathfrak{g}}_B
    \ar[
      rr
    ]
    &&
    \mathfrak{l}
    \big(
      \,
      \Sigma^{\color{purple}1}
      \mathrm{KU}
      \!\sslash\!\!B\mathrm{U}(1) \, \big)
    \\
    H^3_{A\, \bas} + e'_B \cdot c_1^A
    &\longmapsfrom&
    h_3
    \\
 \big(+ {p_A}_*F_{2k+2} + e'_B \cdot F_{2k\, \bas} \, \big)_{k\in \mathbb{Z}}
    &\longmapsfrom&
    (f_{2k+1})_{k\in \mathbb{Z}} \, .
  \end{tikzcd}
\end{equation*}
Of course, these are essentially the same twisted cocycles as that of \eqref{Red-Iso-ReoxiAction}, up to  \textit{an overall sign} on the ``fluxes''.

\vspace{1mm} 
\noindent{\bf (ii)} Using instead the second isomorphism \eqref{SecondExtraIsoOfDoublyCyclifiedTwistedKSpectra} results in the twisted  $\mathrm{KU}_{\mathcolor{purple}1}$ cocycles given by
\begin{equation*}
  \begin{tikzcd}[
    row sep=-2pt, 
    column sep=0pt]
    \widehat{\mathfrak{g}}_{B'}
    \ar[
      rr
    ]
    &&
    \mathfrak{l}
    \big(
      \,
      \Sigma^{\color{purple}1}
      \mathrm{KU}
      \!\sslash\!\!B\mathrm{U}(1) \, \big)
    \\
    H^3_{A\, \bas} {\color{purple}-}  e'_{B'} \cdot c_1^A
    &\longmapsfrom&
    h_3
    \\
 \big({\color{purple}+}  {p_A}_*F_{2k+2} - e'_{B'} \cdot F_{2k\, \bas} \, \big)_{k\in \mathbb{Z}}
    &\longmapsfrom&
    (f_{2k+1})_{k\in \mathbb{Z}}
  \end{tikzcd}
\end{equation*}
where now the extension 
$$
\widehat{\mathfrak{g}}_{B'} \longrightarrow \frg \, ,
$$
is instead via the opposite 2-cocycle
$$
c_{B'}^1 \, := \, {\color{purple}-} {p_A}_* (H_A^3) \, .
$$
However, these can be identified as exactly the same twisted cocycles as in {\bf (i)}, 
via a precomposition with the (tautological) \textit{reflection} isomorphism of the two nominally different extensions 
\begin{equation}\label{TautologicalReflectionAlongOppositeCircleExtension}
  \begin{tikzcd}[
    row sep=-2pt, 
    column sep=0pt]
    \widehat{\mathfrak{g}}_{B}
    \ar[
      rr,
      "{ \sim }"
    ]
    &&
    \widehat{\frg}_{B'}
    \\
{\color{purple}-}  e'_{B} 
    &\longmapsfrom&
    e'_{B'}
  \end{tikzcd}
\end{equation}

\noindent{\bf (iii)} Finally, and analogously, using instead the third isomorphism \eqref{ThirdExtraIsoOfDoublyCyclifiedTwistedKSpectra} results in the twisted $\mathrm{KU}_{\mathcolor{purple}1}$ cocycles
\begin{equation*}
  \begin{tikzcd}[
    row sep=-2pt, 
    column sep=0pt]
    \widehat{\mathfrak{g}}_{B'}
    \ar[
      rr
    ]
    &&
    \mathfrak{l}
    \big(
      \,
      \Sigma^{\color{purple}1}
      \mathrm{KU}
      \!\sslash\!\!B\mathrm{U}(1) \, \big)
    \\
    H^3_{A\, \bas} {\color{purple}-}  e'_{B'} \cdot c_1^A
    &\longmapsfrom&
    h_3
    \\
 \big({\color{purple}-}  {p_A}_*F_{2k+2} {\color{purple}+} e'_{B'} \cdot F_{2k\, \bas} \, \big)_{k\in \mathbb{Z}}
    &\longmapsfrom&
    (f_{2k+1})_{k\in \mathbb{Z}} \, ,
  \end{tikzcd}
\end{equation*}
which are, in turn, exactly those of {\bf (ii)}, again up to a sign on the fluxes.

\vspace{1mm}
\noindent {\bf (iv)}
Despite the above rational (i.e., local/super-tangent space-wise) identifications of the 4 composite operations on rational twisted K-theory cocycles, we stress that these 4 isomorphisms may play a non-trivial role at the non-rational (i.e., global/topologically non-trivial) setting. Firstly, the `reflection' isomorphism of the extending dimension should correspond, non-rationally, to identifying the T-dual circle bundle to its opposite circle bundle (over a local patch), i.e., that classified by the opposite of the T-dual Chern class. Secondly, for the physically relevant applications (cf. \S\ref{SuperspaceTDualityI})  the sign inversion on the fluxes may be identified both as 

\begin{itemize}
 \item a discrete (charge) symmetry of the (dynamical) IIA/IIB 10D Supergravities (similar to the case of pure Yang-Mills), 
\item and as a freedom in the action of T-duality on the (background) RR-fluxes for the corresponding embedding string theory (see, e.g., \cite[(35)-(36)]{Hassan}).
\end{itemize} 
These observations suggest that all (non-rational lifts of the) 4-isomorphisms detailed above are \textit{equally admissible} ``T-duality operations''. In particular, this implies they should all be included as \textit{distinct} patching possibilities in the construction of  so called ``T-folds'' (see, e.g., \cite{KimSaemann22}).
\end{remark}

We now generalize dimensional reduction from 1-dimensional fibers to $n$-copies of $1$-dimensional fibers.

\newpage 
\subsection{Higher extensions}
\label{HigherExtensions}

We have discussed central extensions classified by 2-cocycles (Def. \ref{CentralExtension}, Def. \ref{CentralTorusExtension}). Traditionally,
these are the only central extensions considered in ordinary Lie theory. However, in $L_\infty$-theory we may also extend by higher cocycles (cf.
\cite[Prop. 3.5]{FSS15-HigherWZW}, the application to super-spacetimes goes back to \cite{DF82}\cite{CdAIPB00}\cite{Azcarraga05}):

\begin{definition}[{\bf Higher central extensions} ]
\label{HigherCentralExtension}
For $\mathfrak{g} \,\in\, \mathrm{sLie}_\infty$ equipped with an ordinary $(n+1)$-cocycle \eqref{ModulatingOrdinaryCocycles}
$$
  \begin{tikzcd}
    \mathfrak{g}
    \ar[
      r,
      "{
       \omega_{n+1} 
      }"
    ]
    &
    b^n \mathbb{R}
  \end{tikzcd}
  \hspace{1cm}
  \leftrightsquigarrow
  \hspace{1cm}
  \left\{\!\!
  \def\arraystretch{1.1}
  \begin{array}{l}
    \omega_{n+1}
    \,\in\,
    \mathrm{CE}(\mathfrak{g})
    \\
    \mathrm{deg}(\omega_{n+1})
    \,=\,
    (n+1, \mathrm{evn})
    \\
    \mathrm{d}\, \omega_{n+1} \,=\, 0
  \end{array}
  \right.
$$
the {\it higher central extension} it classifies is its homotopy fiber $\widehat{\mathfrak{g}}$ given by
\vspace{1mm} 
\begin{equation}
  \label{CEOfHigherCentralExtension}
  \mathrm{CE}(
    \mathfrak{g})
  \big[
    \grayunderbrace{
      b_n
    }{
      \mathclap{
        \mathrm{deg} \,=\,
        (n,\mathrm{evn})
      }
    }
  \,\big] 
  /
  (
    \mathrm{d}\,b_n
    \;=\;
    \omega_{n+1}
  )
\end{equation}
fitting into a fiber sequence
$$
  \begin{tikzcd}
    \widehat{\mathfrak{g}}
    \ar[
      rrr,
      ->>,
      "{p:= 
        \mathrm{hofib}(\omega_{n+1})
      }"
    ]
    &&&
    \mathfrak{g}
    \ar[
      r,
      "{
        \omega_{n+1}
      }"
    ]
    &
    b^n\mathbb{R}
    \,.
  \end{tikzcd}
$$
\end{definition}

\medskip

\noindent
{\bf Examples of higher central extension.} 

\noindent The base examples of relevance for super $p$-branes are the following:

\begin{example}[\bf The string-extensions of II super-space]
\label{StringExtendedSuperSpace}
The higher central extension (Def. \ref{HigherCentralExtension}) of 10D type II superspace (Ex. \ref{11dAsExtensionFromIIA}, Ex. \ref{TypeIIBSuperMinkowski}) by the NS super 3-flux densities $H_3^A$ \eqref{DDReducedMBraneCocycles} and $H_3^B$ \eqref{IIBCocycles}, respectively,
are super-space analogs of the {\it string Lie 2-algebra} (cf. \cite[Apnd]{FSS14-7d}) and as such called $\mathfrak{string}_{\mathrm{IIA/B}}$ or similar \cite[Thm. 21]{BH11}\cite[\S 3.1.3]{JH12}\cite[Def. 4.2]{FSS15-HigherWZW}\cite[pp 13]{HSS19}, cf. \cite[(6.12)]{CdAIPB00}.
$$
  \begin{tikzcd}[
   row sep=-1pt
  ]
    \mathfrak{string}_{\mathrm{IIA}}
    \ar[
      rr,
      "{
        \mathrm{hofib}
      }"
    ]
    &&
    \mathbb{R}^{
      1,9\,\vert\,
      \mathbf{16}
      \oplus
      \overline{\mathbf{16}}
    }
    \ar[
      rr,
      "{
        H_3^A
      }"
    ]
    &&
    b^2 \mathbb{R}
    \\
    \mathfrak{string}_{\mathrm{IIB}}
    \ar[
      rr,
      "{
        \mathrm{hofib}
      }"
    ]
    &&
    \mathbb{R}^{
      1,9\,\vert\,
      \mathbf{16}
      \oplus
      \mathbf{16}
    }
    \ar[
      rr,
      "{
        H_3^B
      }"
    ]
    &&
    b^2 \mathbb{R}
  \end{tikzcd}
$$
given by
\begin{equation}
  \label{CEOfStringLie2Algebra}
  \mathrm{CE}\big(
    \mathfrak{string}_{\mathrm{IIA/B}}
  \big)
  \;\simeq\;
  \FDGCA
  \left[
  \def\arraystretch{1.4}
  \def\arraycolsep{1pt}
  \begin{array}{c}
    (\psi^\alpha)_{\alpha=1}^{32}
    \\
    (e^a)_{a=0}^{9}
    \\
    b_2
  \end{array}
  \right]
  \Big/
  \left(
  \def\arraystretch{1.4}
  \def\arraycolsep{1pt}
  \begin{array}{ccl}
    \mathrm{d}
    \,\psi
    &=&
    0
    \\
    \mathrm{d}\,
    e^a
    &=&
    \big(\hspace{1pt}
      \overline{\psi}
      \,\Gamma^a_{A/B}\,
      \psi
    \big)
    \\
    \mathrm{d}\,
    b_2
    &=&
    \smash{
    \grayunderbrace{
    \big(\hspace{1pt}
      \overline{\psi}
      \,\Gamma_a^{A/B}\Gamma_{\!\ten}\,
      \psi
    \big)
    e^a
    }{
      H_3^{A/B}
    }
    }
  \end{array}  
  \right)
  .
\end{equation}
\end{example}

Analogously:
\begin{example}[\bf The M2-brane extension of 11D superspace]
\label{TheM2braneExtensionOf11D}
  The higher central extension (Def. \ref{HigherCentralExtension}) of the 11D super-spacetime (Ex. \ref{SupersymmetryAlgebras}) by the super 4-flux density $G_4$ \eqref{11dSuperFluxInIntroduction} is the higher analogue of the super-string Lie 2-algebra from  Ex. \ref{StringExtendedSuperSpace} and as such called the {\it super membrane algebra} $\mathfrak{m}2\mathfrak{brane}$ or similar \cite[Thm. 22]{BH11}\cite[\S 4.4]{FSS15-HigherWZW}\cite{JH17}. Its CE-algebra is the one originally considered in \cite[(3.15)]{DF82}, rediscovered in \cite[(105)]{CdAIPB00}:
  \vspace{1mm} 
  $$
    \mathrm{CE}\big(
      \mathfrak{m}2\mathfrak{brane}
    \big)
    \;\simeq\;
    \FDGCA
    \left[
    \def\arraystretch{1.4}
    \def\arraycolsep{1pt}
    \begin{array}{c}
      (\psi^\alpha)_{\alpha=1}^{32}
      \\
      (e^a)_{a=1}^{\ten}
      \\
      c_3
    \end{array}
    \right]
    \Big/
    \left(
    \def\arraystretch{1.4}
    \def\arraycolsep{1pt}
    \begin{array}{ccl}
      \mathrm{d}\,
      \psi &=&
      0
      \\
      \mathrm{d}\,
      e^a
      &=&
      \big(\hspace{1pt}
        \overline{\psi}
        \,\Gamma^a\,
        \psi
      \big)
      \\
      \mathrm{d}\,c_3
      &=&
      \smash{
      \grayunderbrace{
      \tfrac{1}{2}
      \big(\hspace{1pt}
        \overline{\psi}
        \Gamma_{a b }
        \psi
      \big)
      e^a e^b
      }{
        G_4
      }
      }
    \end{array}
    \right)
    , 
  $$
so that 
\vspace{0.3cm}
 \begin{equation*}
   \begin{tikzcd}[
   row sep=3pt
  ]
   \mathfrak{m}2\mathfrak{brane}
    \ar[
      rr,
      "{
        \mathrm{hofib}
      }"
    ]
    &&
    \mathbb{R}^{
      1,10\,\vert\,
      \mathbf{32}}   
    \ar[
      rr,
      "{
        G_4
      }"
    ]
    &&
    b^3 \mathbb{R} \, .
    \end{tikzcd}
  \end{equation*}
\end{example}

\begin{example}[\bf Parity isomorphism of 11D SuGra]
  \label{ParityIsomorphism}
  The canonical reflection action of $\Gamma_{\ten} \in \mathrm{Pin}^+(1,10)$ on $\mathbb{R}^{1,10\,\vert\, \mathbf{32}}$ lifts to
  its M2-brane extension of Ex. \ref{TheM2braneExtensionOf11D} by flipping the sign of the generator $c_3$:
  \begin{equation}
    \label{ComponentsOfParityIso}
    \begin{tikzcd}[row sep=-2pt, column sep=0pt]
      \mathfrak{m}2\mathfrak{brane}
      \ar[
        rr,
        "{ \mathrm{par} }",
        "{ \sim }"{swap}
      ]
      &&
      \mathfrak{m}2\mathfrak{brane}
      \\
      \Gamma_{\ten}\psi
      &\longmapsfrom&
      \psi
      \\
      +e^a &\longmapsfrom& e^a
      \mathrlap{\;\; \scalebox{.8}{$(a < 10)$}}
      \\
      -e^{\ten} &\longmapsfrom&
      e^{\ten}
      \\[+3pt]
      -c_3 &\longmapsfrom&
      c_3
      \mathrlap{\,.}
    \end{tikzcd}
  \end{equation}
  This is because $G_4$ \eqref{11dSuperFluxInIntroduction} changes sign under a super-reflection:
  $$
    \def\arraystretch{1.6}
    \begin{array}{ccll}
      \mathrm{par}^\ast G_4 &=&
      \mathrm{par}^\ast
      \Big(
        \tfrac{1}{2}
        \big(\,
          \overline{\psi}
            \Gamma_{a b}
          \psi
        \big)
        e^a e^b
      \Big)
      &
      \proofstep{
        by
        \eqref{11dSuperFluxInIntroduction}
      }
      \\
      &=&
      \sum_{a, b < 10}
      \tfrac{1}{2}\big(\,
        \overline{\Gamma_{\ten}\psi}
        \,\Gamma_{a b}\,
        \Gamma_{\ten}\psi
      \big)
      e^{a} e^b
      \,-\,
      \sum_{a < 10}
      \big(\,
        \overline{\Gamma_{\ten}\psi}
        \,\Gamma_{a \ten}\,
        \Gamma_{\ten}\psi
      \big)
      e^{a} e^{\ten}
      &
      \proofstep{
        by 
        \eqref{ComponentsOfParityIso}
      }
      \\
      &=&
      -
      \sum_{a, b < 10}
      \tfrac{1}{2}\big(\,
        \overline{\psi}
        \Gamma_{\ten}
        \,\Gamma_{a b}\,
        \Gamma_{\ten}\psi
      \big)
      e^{a} e^b
      \,+\,
      \sum_{a < 10}
      \big(\,
        \overline{\psi}
        \Gamma_{\ten}
        \,\Gamma_{a \ten}\,
        \Gamma_{\ten}\psi
      \big)
      e^{a} e^{\ten}
      &
      \proofstep{
        by
        \eqref{SkewSelfAdjointnessOfCliffordGenerators}
      }
      \\
      &=&
      -
      \sum_{a, b < 10}
      \tfrac{1}{2}\big(\,
        \overline{\psi}
        \,\Gamma_{a b}\,
        \Gamma_{\ten}
        \Gamma_{\ten}\psi
      \big)
      e^{a} e^b
      \,-\,
      \sum_{a < 10}
      \big(\,
        \overline{\psi}
        \,\Gamma_{a \ten}\,
        \Gamma_{\ten}
        \Gamma_{\ten}\psi
      \big)
      e^{a} e^{\ten}
      &
      \proofstep{
        by 
        \eqref{CliffordDefiningRelation}
      }
      \\
      &=&
      -
      \sum_{a, b < 10}
      \tfrac{1}{2}\big(\,
        \overline{\psi}
        \,\Gamma_{a b}\,
        \psi
      \big)
      e^{a} e^b
      \,-\,
      \sum_{a < 10}
      \big(\,
        \overline{\psi}
        \,\Gamma_{a \ten}\,
        \psi
      \big)
      e^{a} e^{\ten}
      &
      \proofstep{
        by
        \eqref{CliffordDefiningRelation}
      }
      \\
      &=&
      - \tfrac{1}{2}\big(\,
        \overline{\psi}
        \,\Gamma_{ab}\,
        \psi
      \big)
      e^{a}e^b
      \;=\;
      - G_4
      \,.
    \end{array}
  $$
  Note, by the $\mathfrak{l}S^4$-cocycle condition of $(G_4,G_7)$ (Ex. \ref{4SphereValuedSuperFlux}), this immediately implies that $G_7$ is instead preserved
  $$
  \mathrm{par}^* G_7 \, \, = \, \, G_7\, .
  $$
  This transformation \eqref{ComponentsOfParityIso}, of super-spacetime reflection together with sign-inversion on $c_3$,
  is known \cite[(2.2.29)]{DuffNilssonPope86}
  as a symmetry of the Lagrangian density for 11D SuGra, and it controls (e.g., \cite[(3.1)]{Falkowski99}\cite[p. 94]{Carlevaro06}) the behavior of the M-theory 3-form near a Ho{\v r}ava-Witten MO9-brane (i.e., near the fixed locus of the super-reflection).
\end{example}

\begin{remark}[\bf Basic and fiber forms on a higher centrally extended super-$L_\infty$ algebra]
\label{HigherBasicAndFiberForms}
$\,$
The decomposition in terms of basic and fiber forms follows for higher central extensions in complete analogy to the case of standard central extensions (Rem. \ref{BasicAndFiberForms}): 

\noindent {\bf (i)}  
Given a higher central extension  as in Def. \ref{HigherCentralExtension}, every element in its CE-algebra decomposes uniquely as the sum 
\begin{equation}
    \label{BasicAndFiberFormHigherDecomposition}
    \alpha
    \;=\;
    \alpha_{\mathrm{bas}}
    \,+\,
    b_n \, p_\ast(\alpha)
  \end{equation}
  of a {\it basic form}
  $$
    \alpha_{\mathrm{bas}}
    \,\in\,
    p^\ast
    \big(
    \mathrm{CE}(\mathfrak{g})
    \big)
  $$ 
  and the product of the generator $b_n$ with the image of $\alpha$ under {\it fiber integration} $p_\ast$, which is a super-graded derivation of degree $(-n,\mathrm{evn})$:
  \begin{equation}
    \label{HigherFiberIntegration}
    \begin{tikzcd}[row sep=-3pt, column sep=0pt]
    \mathrm{CE}(\, \widehat{\mathfrak{g}}\,)
    \ar[
      rr,
      "{
        p_\ast
      }"
    ]
    &&
    \mathrm{CE}(\mathfrak{g})
    \\
    b_n &\longmapsto& 1
    \\
    e^i &\longmapsto& 0 \, .
    \end{tikzcd}
  \end{equation}
\noindent {\bf (ii)}   The differential of a general element is given in this decomposition in terms of the differential $\mathrm{d}_{\mathfrak{g}}$ by:
  \begin{equation}
    \label{DifferentialOnHigherBasicDecomposition}
    \def\arraystretch{1.4}
    \begin{array}{ll}
      \mathrm{d}_{\widehat{\mathfrak{g}}}
      \big(
        \alpha_{\mathrm{bas}}
        \,+\,
        b_n
        \,
        p_\ast \alpha
      \big)
      &
      =\;
      \mathrm{d}_{\widehat{\mathfrak{g}}}
      \,
      \alpha_{\mathrm{bas}}
      \,+\,
      \big(
        \mathrm{d}_{\widehat{g}}
        \,
        b_n
      \big)
      \,
      p_\ast \alpha
      \,+ (-1)^{n}\,
      b_n
      \,
      \mathrm{d}_{\widehat{\mathfrak{g}}}
      \,
      p_\ast \alpha
      \\
      &=\;
      \big(
      \mathrm{d}_{\mathfrak{g}}
      \alpha_{\mathrm{bas}}
      \,+\,
      \omega_{n+1}\, 
      p_\ast \alpha
      \big)
      \,+ (-)^{n}\,
      b_n
      \,
      \mathrm{d}_{\mathfrak{g}}
      p_\ast \alpha
      \,, 
    \end{array}
\end{equation}
which further implies that the fiber integration is a morphism of cochain complexes
$$
 p_\ast \circ \mathrm{d}_{\widehat{\mathfrak{g}}}\,  = \, (-1)^n \, \mathrm{d}_\frg \circ p_\ast \, .
$$
\end{remark}

While aspects of T-duality for such higher extensions were discussed in \cite{FSS20-HigherT} (just) in terms of the higher analog Fourier-Mukai transform, next here we develop the full story of the higher Ext/Cyc-adjunction and the automorphisms of the higher cyclified twisted higher cocycles:

\medskip

\noindent
{\bf Higher Cyclification along rational odd spheres.} Hereon we focus on central extensions where $n=2t-1\in \NN$ is an odd positive integer, hence on central extensions along even higher cocycles $\omega_{2t} : \frg \longrightarrow b^{2t-1}\mathbb{R}$. Notice the target $L_\infty$-algebra here may be thought of equivalently as: 

{\bf (i)} the rationalization of an odd iterated delooping of the circle $B^{2t-1} S^1$, or 

 {\bf (ii)} perhaps more suggestively the rationalization of the an odd sphere $S^{2t-1}$.

\smallskip 
\noindent In this scenario, there is a straightforward higher generalization of the cyclification functor, which in view of {\bf (ii)} may also be referred to as an \textit{odd spherification} functor.

\begin{definition}[{\bf Higher Cyclification/Odd Spherification of super $L_\infty$-algebras}]
  \label{HigherCyclification}
  Given $\mathfrak{h} \,\in\, \mathrm{sLieAlg}_\infty^{\mathrm{fin}}$ with presentation
$$
    \mathrm{CE}(\mathfrak{h})
    \;\simeq\;
    \FDGCA \big[
      (e^i)_{i \in I}
    \big]
    \big/
    \big(
      \mathrm{d}\, e^i
      =
      P^i(\vec e\,)
    \big)_{i \in I}
    \,,
$$
and $t\in \NN$ its {\it$(2t-1)$-cyclification} $\mathrm{cyc}_{2t-1}(\mathfrak{h}) \,\in\, \mathrm{sLieAlg}_\infty$ is given by
\begin{equation}
    \label{HighercykCE}
    \mathrm{CE}\big(
      \mathrm{cyc}_{2t-1}(\mathfrak{h})
    \big)
    \;\;
    :=
    \;\;
    \FDGCA
    \Bigg[
      \adjustbox{
        raise=-4pt
      }{$
      \def\arraystretch{1.5}
      \begin{array}{l}
        \big(e^i\big)_{i \in I},
        \grayoverbrace{
          \omega_{2t},
        }{
          \mathclap{
          \mathrm{deg} 
          \,=\,
          (2t,\mathrm{evn})
          }
        }
        \\
        \big(
        \grayunderbrace{
          \mathrm{s}_{t} e^i
        }{
         \mathclap{
           \scalebox{.7}{$
           \def\arraystretch{.9}
           \begin{array}{c}
           \mathrm{deg}
           \,=\,
           \\
           \mathrm{deg}(e^i) - (2t-1,\mathrm{evn})
           \end{array}
           $}
          }
        }
        \big)_{i \in I}
      \end{array}
      $}
    \Bigg]
    \Big/
    \left(
    \def\arraystretch{1}
    \def\arraycolsep{2pt}
    \begin{array}{lcl}
      \mathrm{d}\, 
      \omega_{2t}
      &=&
      0
      \\
      \mathrm{d}\, e^i
      &=&
      \mathrm{d}_{\mathfrak{h}}
      \, e^i
      +
      \omega_{2t}\, 
      \mathrm{s}_t e^i
      \\
      \mathrm{d}\, \mathrm{s}_t e^i
      &=&
      -
      \mathrm{s}_t\big(
        \mathrm{d}_{\mathfrak{h}}
        \,
        e^i
      \big)
    \end{array}
   \!\! \right)
    ,
  \end{equation}
  where in the last line on the right the shift is understood as extended to a super-graded derivation of degree $(-2t+1, \mathrm{evn})$:
  $$
    \begin{tikzcd}[row sep=-3pt,
      column sep=0pt
    ]
      \mathllap{
        \mathrm{s}_t
        :
        \;
      }
      \mathrm{CE}\big(
        \mathrm{cyc}_{2t-1}(\mathfrak{h})
      \big)
      \ar[
        rr
      ]
      &&
      \mathrm{CE}\big(
        \mathrm{cyc}_{2t-1}(\mathfrak{h})
      \big)
      \\
      \omega_{2t}
      &\longmapsto&
      0
      \,,
      \\
      e^i &\longmapsto&
      \mathrm{s}_te^i
      \,,
      \\
      \mathrm{s}_t e^i 
      &\longmapsto&
      0
      \,.
    \end{tikzcd}
  $$
\end{definition}
The fact that this is well-defined, namely 
\begin{equation}
    \label{HigherShiftAnticommutesWithDifferential}
    \mathrm{d}\, \mathrm{d}
    \,=\, 0
    \,,
    \qquad
    \mathrm{s}_t\, \mathrm{s}_t
    \,=\, 0
    \,,
    \qquad 
    \mathrm{s}_t\, \mathrm{d}
    +
    \mathrm{d}\, \mathrm{s}_t
    \;=\;
    0
    \,.
\end{equation}
follows precisely as in Lem. \ref{DifferentialAndShiftInCyclification}.

\begin{remark}[\bf Rationalizing topological spherification]\label{RationalizingTopologicalSpherification}
The case $t=1$ in Def. \ref{HigherCentralExtension} recovers the standard cyclification (Def. \ref{Cyclification}), which may be seen \cite{VPB85}\cite{BMSS19}\cite{SV23-Mysterious}\cite{SV24-Mysterious}
 as the rationalization of a loop space homotopy-quotiented by loop rotation, namely for $\mathfrak{h}\cong \mathfrak{l} X$ 
$$ \mathrm{cyc}(\mathfrak{h}) \, \cong  \, \mathfrak{l} \big( \, [S^1,X] \, /\!/ S^1\,  \big)\, . 
$$
We expect the proof 
 of \cite{BMSS19} should hold for the $t=2$ case with minor modifications, where the odd 3-sphere $S^3\cong SU(2)$ still has proper group-structure, so that
$$ \mathrm{cyc}_{3}(\mathfrak{h}) \, \cong  \, \mathfrak{l} \big( \, [S^3,X] \, /\!/ S^3\,  \big)\, . 
$$
For $t\geq 3$, this pattern obviously breaks down since the higher odd topological spheres from $S^7$ onwards admit no group structure. Nevertheless, the rational higher cyclification (Def. \ref{HigherCyclification}) still makes complete sense, and so does the corresponding higher version of the Ext/Cyc adjunction of Prop. \ref{TheExtCycAdjunction}.
\end{remark}

\begin{proposition}[{\bf The Higher Ext/Cyc-adjunction}]
  \label{TheHigherExtCycAdjunction}
  Given $\mathfrak{g}, \mathfrak{h} \in \mathrm{sLieAlg}_\infty$
  with a 2t-cocycle  $\widehat{\omega}_{2t} \,\in\, \mathrm{CE}(\mathfrak{g})$, there is a bijection
  between:
  
 \noindent {\bf  (i)} maps into $\mathfrak{h}$ out of the higher central extension $\widehat{\mathfrak{g}}$ classified by the 2t-cocycle
  {\rm (Def. \ref{HigherCentralExtension})},
  
\noindent {\bf   (ii)} maps out of $\mathfrak{g}$ into the higher cyclification of $\mathfrak{h}$ 
  {\rm (Def. \ref{HigherCyclification})} that preserve the 2t-cocycle:
  \begin{equation}
    \label{HigherCyclificationHomIsomorphism}
    \Big\{\!
    \begin{tikzcd}
      \widehat{\mathfrak{g}}
      \ar[
        rr,
        "{ f }"
      ]
      &&
      \mathfrak{h}
    \end{tikzcd}
   \! \!\Big\}
    \begin{tikzcd}[
      column sep=85pt
    ]
      \ar[
        r,
        shift left=5pt,
        "{  \scalebox{.7}{\color{darkgreen}
            \bf
            reduction}\;\;
          \mathrm{rdc}_{\widehat{\omega}_{2t}}
        }",
        "{ \sim }"{swap, yshift=-2pt}
      ]
      \ar[
        r,
        <-,
        shift right=5pt,
        "{ \scalebox{.7}{
            \color{darkgreen}
            \bf
            oxidation
          }\;\;
          \mathrm{oxd}_{\widehat{\omega}_{2t}}
        }"{swap},
      ]
      &
      {}
    \end{tikzcd}
    \Bigg\{\!\!
    \begin{tikzcd}[row sep=-3pt, column sep=large]
      \mathfrak{g}
      \ar[
        rr,
        "{ \widetilde f }"
      ]
      \ar[
        dr,
        "{ \widehat{\omega}_{2t} }"{swap}
      ]
      &&
      \mathrm{cyc}_{2t-1}(\mathfrak{h})
      \ar[
        dl,
        "{ \omega_{2t} }"
      ]
      \\
      &
      b^{2t-1} \mathbb{R}
    \end{tikzcd}
    \!\!\!\Bigg\}
  \end{equation}
  given by
\begin{equation}
  \label{HigherCyclificationHomBijection}
  \hspace{-3cm} 
  \begin{tikzcd}[sep=0pt]
    \widehat{\mathfrak{g}}
    \ar[
      rr,
      "{ f }"
    ]
    &&
    \mathfrak{h}
    \\
    \alpha^i_{\mathrm{bas}}
    +
    b_{2t-1} \, p_\ast \alpha^i
    &\longmapsfrom&
    e^i
  \end{tikzcd}
  \hspace{1.5cm}
  \leftrightsquigarrow
  \hspace{1.5cm}
  \begin{tikzcd}[
    row sep=-1pt, 
    column sep=0pt]
    \mathfrak{g}
    \ar[
      rr,
      "{ 
        \widetilde{f} 
      }"
    ]
    &&
    \mathrm{cyc}(\mathfrak{h})
    \\
    \alpha^i_{\mathrm{bas}}
    &\longmapsfrom&
    e^i
    \\
    -
    p_\ast \alpha^i
    &\longmapsfrom&
    \mathrm{s}_t e^i
    \\
    \widehat{\omega}_{2t}
    &\longmapsfrom&
    \omega_{2t}
    \mathrlap{\,.}
  \end{tikzcd}
\end{equation}
\end{proposition}
\begin{proof}
The proof follows essentially verbatim as in that of Prop. \ref{TheExtCycAdjunction}, by adjusting the degrees of the even cocycle and extra odd generator accordingly.
\end{proof}

\begin{example}[\bf Higher cyclification of higher bundle gerbe classifying space]\label{HigherTDualityClassifyingAlgebra} 
The higher cyclifications (Def. \ref{HigherCyclification}) of the real Whitehead $L_\infty$-algebra of $B^{4t-2} \mathrm{U}(1)$ (Ex. \ref{LineLieNAlgebra}), $\mathrm{cyc}_{2t-1} b^{4t-2}\mathbb{R}$, is given by
$$
  \mathrm{CE}\big(
    \mathrm{cyc}_{2t-1} b^{4t-2}\mathbb{R}
  \big)
  \;\simeq\;
  \FDGCA
  \left[
  \def\arraystretch{1.1}
  \def\arraycolsep{1pt}
  \begin{array}{c}
    \omega_{2t}
    \\
   h_{4t-1}
    \\
    \widetilde \omega_{2t}
    :=
    \mathrm{s}_t h_{4t-1}
  \end{array}
  \right]
  \Big/
  \left(
  \def\arraystretch{1.1}
  \def\arraycolsep{1pt}
  \begin{array}{ccl}
    \mathrm{d}
    \,
    \omega_{2t}
    &=&
    0
    \\
    \mathrm{d}\,
    h_{4t-1}
    &=&
    \omega_{2t}
    \,
    \widetilde\omega_{2t}
    \\
    \mathrm{d}\,
    \widetilde\omega_{2t}
    &=&
    0
  \end{array}
\!  \right)
$$

\vspace{1mm} 
\noindent being equivalently the higher central extension (Def. \ref{HigherCentralExtension}) of $b^{2t-1}\mathbb{R}^{2}$ by its canonical $4t$-cocycle
\begin{equation}
  \label{CEOfBOfHigherTDuality2Group}
  \begin{tikzcd}
   b\mathcal{T}_{t} \, :=  \, \mathrm{cyc}_{2t-1} b^{4t-2}\mathbb{R}
    \ar[
      rrr,
      "{
        \mathrm{hofib}(
          \omega_{2t} 
          \,
          \widetilde \omega_{2t}
        )
      }"
    ]
    &&&
    b^{2t-1}\mathbb{R}^2
    \ar[
      rr,
      "{
        \omega_{2t} 
        \,
        \widetilde \omega_{2t}
      }"
    ]
    &&
    b^{4t-1} \mathbb{R}
  \end{tikzcd}
\end{equation}
and as such may be called (the Whitehead $L_\infty$-algebra of) the delooping of the {\it higher T-duality Lie group} (cf. \cite[Def. 3.14]{FSS18-TDualityA}):
$$
  \mathrm{cyc}
  \,
  \mathfrak{l}B^{4t-2} \mathrm{U}(1)
  \;\;
  \simeq
  \;\; \mathfrak{l} \, \bigg( 
  \mathrm{hofib}
  \big(\!\!
    \begin{tikzcd}
    B^{2t-1}\mathrm{U}(1)
    \times
    B^{2t-1}\mathrm{U}(1)
    \ar[
      rrr,
      "{
        \pi_1 \omega_{2t} 
          \,\cup\, 
        \pi_2 \omega_{2t}
      }"
    ]
    &&&
    B^{4t-1} \mathrm{U}(1)
    \end{tikzcd}
  \!\!\big) \! \bigg)
  \,.
$$
\end{example}

\smallskip 
Similarly to the $t=1$ case, it is evident that \eqref{CEOfBOfHigherTDuality2Group} has an automorphism symmetry given by exchanging the two degree$=2t$ generators (we again include a minus sign, for compatibility below in Ex. \ref{CyclificationOfHigherTwistedCocyclesClassifyingSpaces}):
\begin{equation}
  \label{AutomorphismOfCycOfB2tU1}
  \begin{tikzcd}[row sep=-2pt,
    column sep=0pt
  ]
    b\mathcal{T}_t
    \ar[
      rr,
      <->,
      "{ \sim }"
    ]
    &&
   b\mathcal{T}_t
    \\
    -
    \widetilde \omega_{2t}
    &\longmapsfrom&
    \omega_{2t}
    \\
    -
    \omega_{2t}
    &\longmapsfrom&
    \widetilde \omega_{2t}
    \mathrlap{\,.}
  \end{tikzcd}
\end{equation}
This already carries in it the seed of \textit{higher} T-duality, with the next example lifting this automorphism to an equivalence between the higher $(2t-1)$-cyclifications of the $(4t-1)$-twisted cocycle classifying $L_\infty$-algebras from Ex. \ref{TwistedPeriodicCocycleClassifyingLinftyAlgebras}, generalizing the situation of Ex. \ref{CyclificationOfTwistedKTTheorySpectra}.

\begin{example}[{\bf Higher cyclification of higher twisted cocycle classifying $L_\infty$-algebras and higher T-duality group}]
\label{CyclificationOfHigherTwistedCocyclesClassifyingSpaces}
The $(2t-1)$-cyclifications (Def. \ref{HigherCyclification}) of the $(4t-1)$-twisted, $(4t-2)$-periodic cocycle classifying $L_\infty$-algebras from Ex. \ref{TwistedPeriodicCocycleClassifyingLinftyAlgebras} are identified by an isomorphism \eqref{LieHomomorphism}. In order to ease the notation below, we abbreviate $n_t := 2t-1$: 
\begin{equation}
  \label{IsoOfDoublyCyclifiedHigherTwistedCocycleSpectra}
  \hspace{-3mm} 
  \begin{tikzcd}[row sep=6pt, 
    column sep=2pt,
    ampersand replacement=\&,
    ,/tikz/column 3/.append style={anchor=base west}
  ]
  \mathrlap{
  \mathrm{CE}
  \Big(
    \mathrm{cyc}_{n_t}
    \,
    \mathfrak{l}
    \big(
      \,
      \Sigma^{m}
      \mathrm{K}^{n_t}\mathrm{U}
      \!\sslash\!\!B^{2n_t-1}\mathrm{U}(1)
      \,
    \big)
  \Big)
  }
  \\
  \&
  \simeq
    \ar[
      d,
      <-,
      shift left=250pt,
      shorten=35pt,
      "{ \sim }"{sloped},
      "{
        \def\arraystretch{.8}
        \def\arraycolsep{2pt}
        \begin{array}{ccccc}
          \omega_{n_t+1}
          &
          \mathrm{s}_t h_{2n_t +1}
          &
          h_{2 n_t + 1}
          &
          f_{2 k n_t + m}
          &
          \mathrm{s}_t\!f_{(2k+1)n_t +m}
          \\
          \rotatebox[origin=c]{90}
            {$\mapsto$}
          &
          \rotatebox[origin=c]{90}
            {$\mapsto$}
          &
          \rotatebox[origin=c]{90}
            {$\mapsto$}
          &
          \rotatebox[origin=c]{90}
            {$\mapsto$}
          &
          \rotatebox[origin=c]{90}
            {$\mapsto$}
          \\
          {\color{purple}-}
          \mathrm{s}_t h_{2n_t+1}
          &
          {\color{purple}-}\omega_{n_t+1}
          &
          h_{2n_t+1}
          &
          \mathrm{s}_t\!f_{2(k+1)n_t +m}
          &
          f_{2kn_t+m}
        \end{array}
      }"{swap}
    ]
  \&
 \FDGCA
  \left[\hspace{0pt}
  \def\arraystretch{1.2}
  \def\arraycolsep{2pt}
  \begin{array}{c}
    \omega_{n_t +1}
    \\
    h_{2n_t +1}
    \\
    \mathrm{s}_t h_{2n_t+1}
    \\
    (f_{2k n_t+m})_{k \in \mathbb{Z}}
    \\
    (\mathrm{s}_t\!f_{2kn_t+m})_{k \in \mathbb{Z}}
  \end{array}
  \hspace{0pt}
  \right]
  \Big/
  \left(
  \def\arraystretch{1.2}
  \def\arraycolsep{2pt}
  \begin{array}{ccl}
    \mathrm{d}\,\omega_{n_t+1}
    &=&
    0
    \\
    \mathrm{d}\, 
    h_{2n_t +1}
    &=&
    \omega_{n_t+1}\, \mathrm{s}_t h_{2n_t +1}
    \\
    \mathrm{d}\,
    \mathrm{s}_t h_{2n_t +1}
    &=&
    0
    \\
    \mathrm{d}
    \,
    f_{2(k+1)n_t+m}
    &=&
    h_{2n_t+1}\, f_{2kn_t +m}
    +
    \omega_{n_t+1}\,\mathrm{s}_t\!f_{2(k+1)n_t +m }
    \\
    \mathrm{d}
    \,
    \mathrm{s}_t \!f_{2(k+1)n_t +m}
    &=&
    -
    (\mathrm{s}_t h_{2n_t+1})f_{2kn_t +m }
    +
    h_{2n_t+1}\, \mathrm{s}_t \!f_{2kn_t+m}
    \\
  \end{array}
  \right)
  \\[40pt]
  \& \, \&
  \FDGCA
  \left[
  \def\arraystretch{1.2}
  \def\arraycolsep{2pt}
  \begin{array}{c}
    \omega_{n_t+1}
    \\
    h_{2n_t+1}
    \\
    \mathrm{s}_t h_{2n_t+1}
    \\
    (f_{(2k-1)n_t+m})_{k \in \mathbb{Z}}
    \\
    (\mathrm{s}_t\!f_{(2k-1)n_t+m})_{k \in \mathbb{Z}}
  \end{array}
  \right]
  \Big/
  \left(
  \def\arraystretch{1.2}
  \def\arraycolsep{2pt}
  \begin{array}{ccl}
    \mathrm{d}\,\omega_{n_t+1}
    &=&
    0
    \\
    \mathrm{d}\, 
    h_{2n_t+1}
    &=&
    \omega_{n_t+1}\, \mathrm{s}_t h_{2n_t+1}
    \\
    \mathrm{d}\,
    \mathrm{s}_t h_{2n_t+1}
    &=&
    0
    \\
    \mathrm{d}
    \,
    f_{(2k+1)n_t+m}
    &=&
    h_{2n_t+1}\, f_{(2k-1)n_t +m}
    +
    \omega_{n_t+1}\,\mathrm{s}_t\!f_{(2k+1) n_t + m}
    \\
    \mathrm{d}
    \,
    \mathrm{s}_t\!f_{(2k+1)n_t +m}
    &=&
    -
    (\mathrm{s}_t h_{2n_t+1})f_{(2k-1)n_t +m}
    +
    h_{2n_t+1}\, \mathrm{s}_t\!f_{(2k-1)n_t +m }
    \\
  \end{array}
  \right)
  \\
    \& \simeq \& \mathrlap{
  \mathrm{CE}
  \Big(
    \mathrm{cyc}_{n_t}
    \,
    \mathfrak{l}
    \big(
      \,
      \Sigma^{\mathcolor{purple}{m-n_t}}
      \mathrm{K}^{n_t}\mathrm{U}
        \!\sslash\!\!B^{2n_t-1}\mathrm{U}(1)
      \,
    \big)
  \Big)    
    }  
  \end{tikzcd}
\end{equation}
compatible with their fibration \eqref{HigherTwistedCocycleSpectraFibration}
over $  b\mathcal{T}_{t} \cong \mathrm{cyc}_{n_t} b^{2n_t} \mathbb{R}$
via its automorphisms \eqref{AutomorphismOfCycOfB2tU1}, where the homotopy fiber of the cyclified fibration is now the direct sum of $2n_t$-periodic cocycle spectra in degrees $m$ and $(m-n_t)$, respectively, with the isomorphism acting by swapping their order:
\begin{equation}
  \label{HigherTDualityOnCyclifiedTwistedCocycleFibration}
  \begin{tikzcd}[row sep=-5pt]
    \mathfrak{l}
    \Sigma^{\mathcolor{purple}{m}}\mathrm{K}^{n_t}\mathrm{U}
    \ar[
      ddrr
    ]
    &&
    \mathfrak{l}
    \Sigma^{\mathcolor{purple}{m - n_t}}\mathrm{K}^{n_t}\mathrm{U}
    \\
    \times
    &&
    \times
    \\
    \mathfrak{l}
    \Sigma^{\mathcolor{purple}{m-n_t}}\mathrm{K}^{n_t}\mathrm{U}
    \ar[
      uurr, 
      crossing over
    ]
    \ar[
      d
    ]
    &&
    \mathfrak{l}
    \Sigma^{\mathcolor{purple}{m}}\mathrm{K}^{n_t}\mathrm{U}
    \ar[
      d
    ]
    \\[22pt]
    \mathrm{cyc}
    \,
    \mathfrak{l}
    \big(
      \Sigma^{\mathcolor{purple}{m}}\mathrm{K}^{n_t}\mathrm{U}
      \!\sslash\!
      B^{2n_t-1}\mathrm{U}(1)
    \big)
    \ar[
      rr,
      <->,
      "{ \sim }"
    ]
    \ar[
      d,
      ->>
    ]
    &&
    \mathrm{cyc}
    \,
    \mathfrak{l}
    \big(
      \Sigma^{\mathcolor{purple}{m-n_t}}\mathrm{K}^{n_t}\mathrm{U}
      \!\sslash\!
      B^{2n_t-1}\mathrm{U}(1)
    \big)
    \ar[
      d,
      ->>
    ]
    \\[24pt]
b\mathcal{T}_t
    \ar[
      rr,
      <->,
      "{ \sim }"
    ]
    &&
b\mathcal{T}_t
    \mathrlap{\,.}
  \end{tikzcd}
\end{equation}
\end{example}
\begin{lemma}
[\bf All isomorphisms of higher cyclified twisted cocycle classifying $L_\infty$-algebras]\label{FurtherIsomorphismsOfHigherCyclifiedTwistedCocycleAlgebras}
In analogy to Lem. \ref{FurtherIsomorphismsOfCyclifiedTwistedKspectra}, there exist in total 4 isomorphisms \begin{equation*}
  \begin{tikzcd}[
    row sep=-3pt,
    column sep=15pt
  ]
   \mathrm{cyc}_{n_t}
    \,
    \mathfrak{l}
    \big(
      \,
      \Sigma^{\color{purple}m}
      \mathrm{K}^{n_t}\mathrm{U}
      \!\sslash\!\!B^{2n_t-1}\mathrm{U}(1)
      \,
    \big)
    \ar[
      rr,
      <->,
      "{ \sim }"
    ]
    &&
    \mathrm{cyc}_{n_t}
    \,
    \mathfrak{l}
    \big(
      \,
      \Sigma^{\color{purple}m-n_t}
      \mathrm{K}^{n_t}\mathrm{U}
      \!\sslash\!\!B^{2n_t-1}\mathrm{U}(1)
      \,
    \big)\\
    &&
  \end{tikzcd}
\end{equation*}

\noindent with the property of swapping $\mathrm{s}h_{2n_t+1} \leftrightsquigarrow \omega_{2n_t}$ and $f_{2kn_t+m}\leftrightsquigarrow \mathrm{s}f_{2(k+1)n_t+m}$, while mapping $h_{2n_t+1}$ to $h_{2n_t +1}$, up to relative sign prefactors.  
Explicitly, the extra 3 isomorphisms are given by
\begin{equation*}
  \begin{tikzcd}[
    row sep=-1pt, 
    column sep=0pt]
   \mathrm{cyc}_{n_t}
    \,
    \mathfrak{l}
    \big(
      \,
      \Sigma^{m}
      \mathrm{K}^{n_t}\mathrm{U}
      \!\sslash\!\!B^{2n_t-1}\mathrm{U}(1)
      \,
    \big)
       \ar[
      rr,
      <->,
      "{ \sim }"
    ]
    &&
    \mathrm{cyc}_{n_t}
    \,
    \mathfrak{l}
    \big(
      \,
      \Sigma^{\mathcolor{purple}{m-n_t}}
      \mathrm{K}^{n_t}\mathrm{U}
      \!\sslash\!\!B^{2n_t-1}\mathrm{U}(1)
      \,
    \big)
    \\
    h_{2n_t+1}
    &\longmapsfrom&
    h_{2n_t+1}  
    \\
    -\mathrm{s}_t h_{2n_t+1}
    &\longmapsfrom&
    \omega_{n_t+1}
       \\
   -\omega_{n_t+1}
    &\longmapsfrom&
    \mathrm{s}_t h_{2n_t+1}
    \\
    -\mathrm{s}_t f_{(2k+1)n_t+m}
    &\longmapsfrom&
    f_{2kn_t+m}
 \\
   - f_{2kn_t+m}
    &\longmapsfrom&
    \mathrm{s}_t f_{2(k+1)n_t +m} 
    \mathrlap{\,,}
  \end{tikzcd}
\end{equation*}
\medskip
\begin{equation}\label{SecondIsoOfDoublyCyclifiedHigherTwistedCocycleSpectra}
  \begin{tikzcd}[
    row sep=-1pt, 
    column sep=0pt]
   \mathrm{cyc}_{n_t}
    \,
    \mathfrak{l}
    \big(
      \,
      \Sigma^{m}
      \mathrm{K}^{n_t}\mathrm{U}
      \!\sslash\!\!B^{2n_t-1}\mathrm{U}(1)
      \,
    \big)
       \ar[
      rr,
      <->,
      "{ \sim }"
    ]
    &&
    \mathrm{cyc}_{n_t}
    \,
    \mathfrak{l}
    \big(
      \,
      \Sigma^{\mathcolor{purple}{m-n_t}}
      \mathrm{K}^{n_t}\mathrm{U}
      \!\sslash\!\!B^{2n_t-1}\mathrm{U}(1)
      \,
    \big)
    \\
    h_{2n_t+1}
    &\longmapsfrom&
    h_{2n_t+1}  
    \\
    \mathrm{s}_t h_{2n_t+1}
    &\longmapsfrom&
    \omega_{n_t+1}
       \\
   \omega_{n_t+1}
    &\longmapsfrom&
    \mathrm{s}_t h_{2n_t+1}
    \\
    -\mathrm{s}_t f_{(2k+1)n_t+m}
    &\longmapsfrom&
    f_{2kn_t+m}
 \\
    f_{2kn_t+m}
    &\longmapsfrom&
    \mathrm{s}_t f_{2(k+1)n_t +m} 
    \mathrlap{\,,}
  \end{tikzcd}
\end{equation}
and
\begin{equation*}
  \begin{tikzcd}[
    row sep=-1pt, 
    column sep=0pt]
   \mathrm{cyc}_{n_t}
    \,
    \mathfrak{l}
    \big(
      \,
      \Sigma^{m}
      \mathrm{K}^{n_t}\mathrm{U}
      \!\sslash\!\!B^{2n_t-1}\mathrm{U}(1)
      \,
    \big)
       \ar[
      rr,
      <->,
      "{ \sim }"
    ]
    &&
    \mathrm{cyc}_{n_t}
    \,
    \mathfrak{l}
    \big(
      \,
      \Sigma^{\mathcolor{purple}{m-n_t}}
      \mathrm{K}^{n_t}\mathrm{U}
      \!\sslash\!\!B^{2n_t-1}\mathrm{U}(1)
      \,
    \big)
    \\
    h_{2n_t+1}
    &\longmapsfrom&
    h_{2n_t+1}  
    \\
    \mathrm{s}_t h_{2n_t+1}
    &\longmapsfrom&
    \omega_{n_t+1}
       \\
   \omega_{n_t+1}
    &\longmapsfrom&
    \mathrm{s}_t h_{2n_t+1}
    \\
    \mathrm{s}_t f_{(2k+1)n_t+m}
    &\longmapsfrom&
    f_{2kn_t+m}
 \\
   - f_{2kn_t+m}
    &\longmapsfrom&
    \mathrm{s}_t f_{2(k+1)n_t +m} 
    \mathrlap{\,.}
  \end{tikzcd}
\end{equation*}
Evidently, the original isomorphism \eqref{IsoOfDoublyCyclifiedHigherTwistedCocycleSpectra} and the first above are the two possible extensions of the automorphism \eqref{AutomorphismOfCycOfB2tU1} of $b\mathcal{T}_t$ , while the latter two isomorphisms are the two possible extensions of the ``opposite'' automorphism of $b\mathcal{T}_t$ 
\begin{equation*}
  \begin{tikzcd}[row sep=-3pt,
    column sep=0pt
  ]
   b\mathcal{T}_t
    \ar[
      rr,
      <->,
      "{ \sim }"
    ]
    &&
    b\mathcal{T}_t
    \\
    \mathrm{s}_t h_{2n_t+1}
    &\longmapsfrom&
    \omega_{2t}
    \\
    \omega_{2t}
    &\longmapsfrom&\mathrm{s}_t h_{2n_t+1}
    \mathrlap{\,.}
  \end{tikzcd}
  \end{equation*}
\end{lemma}
\begin{proof}By direct inspection, completely analogously to that of Lem. \ref{FurtherIsomorphismsOfCyclifiedTwistedKspectra}.
\end{proof}

\smallskip 
With the higher cyclification and with the above isomorphism in hand, we obtain a higher generalized template for T-dualization along higher-dimensional odd (rational) spheres, directly generalizing the standard case from Thm. \ref{TwistedKTheoryRedIsoReOxi}.

\begin{theorem}[\bf Higher twisted cocycles under reduction--isomorphism--reoxidation]
\label{HigherTwistedCocyclesRedIsoReOxi}
\label{HigherLInfinityAlgebraicTDuality}
$\,$

\noindent {\bf (i)} The composite operation of 

\begin{itemize}

 \item[\bf (a)] reducing \eqref{HigherCyclificationHomIsomorphism}
  twisted $\mathfrak{l}
    \Sigma^{\mathcolor{purple}{m}}\mathrm{K}^{n_t}\mathrm{U}$ cocycles on a higher centrally extended  super $L_\infty$-algebra $\frg_A$
\begin{equation*}
  \begin{tikzcd}[
    row sep=-2pt, 
    column sep=0pt]
    \widehat{\mathfrak{g}}_A
    \ar[
      rr
    ]
    &&
    \mathfrak{l}
    \big(
      \,
      \Sigma^{\color{purple}m}
      \mathrm{K}^{n_t}{U}
      \!\sslash\!\!B^{2n_t-1}\mathrm{U}(1) \, \big)
    \\
    H_{A}^{2n_t+1}
    &\longmapsfrom&
    h_{2n_t+1}
    \\
    (F_{2kn_t +m})_{k\in \mathbb{Z}}
    &\longmapsfrom&
    (f_{2kn_t +m})_{k\in \mathbb{Z}}
  \end{tikzcd}
\end{equation*}
along its vibration 
\vspace{-1mm} 
$$\widehat{\frg}_A \xlongrightarrow{\quad p_A\quad } \frg \xlongrightarrow{\quad \omega^A_{n_t+1} \quad} b^{n_t}\mathbb{R} \, ,$$

\item[\bf (b)] applying the isomorphism \eqref{IsoOfDoublyCyclifiedHigherTwistedCocycleSpectra} on the target cyclification of $\mathfrak{l}
    \big(
      \,
      \Sigma^{\color{purple}m}
      \mathrm{K}^{n_t}{U}
      \!\sslash\!\!B^{2n_t-1}\mathrm{U}(1) \, \big)$, hence viewing them instead as valued in the cyclification of $\mathfrak{l}
    \big(
      \,
      \Sigma^{\color{purple}m-n_t}
      \mathrm{K}^{n_t}{U}
      \!\sslash\!\!B^{2n_t-1}\mathrm{U}(1) \, \big)$, while noticing that this swaps the higher cocycle $\omega^A_{n_t+1}$ from that classifying the $\widehat{\frg}_A$-extension to that classifying a different extension
$$\widehat{\mathfrak{g}}_B \longrightarrow \frg \, ,$$
i.e., via
$$\omega^B_{n_t+1} \, := \, {p_A}_* (H_A^{2n_t+1}) \quad : \quad  \frg \longrightarrow b^{n_t} \mathbb{R} \, ,
$$
  \item[\bf (c)] re-oxidizing \eqref{HigherCyclificationHomIsomorphism} the result, but now along the new fibration
  $$\widehat{\frg}_B \xlongrightarrow{\quad p_A\quad } \frg \xlongrightarrow{\quad \omega^B_{n_t+1} \quad} b^{n_t}\mathbb{R} \, ,$$
\end{itemize}
  results in the twisted $\mathfrak{l}
    \Sigma^{\mathcolor{purple}{m-n_t}}\mathrm{K}^{n_t}\mathrm{U}$ cocycles given precisely by
\begin{equation}\label{HigherRed-Iso-ReoxiAction}
  \begin{tikzcd}[
    row sep=-2pt, 
    column sep=12pt]
    \widehat{\mathfrak{g}}_B
    \ar[
      rr
    ]
    &&  
     \mathfrak{l}
    \big(
      \,
      \Sigma^{\color{purple}m-n_t}
      \mathrm{K}^{n_t}{U}
      \!\sslash\!\!B^{2n_t-1}\mathrm{U}(1) \, \big).
    \\
    H^{2n_t+1}_{A\, \bas} + b_B^{n_t} \cdot \omega^A_{n_t+1}
    &\longmapsfrom&
    h_{2n_t+1}
    \\
 \big(- {p_A}_*F_{2kn_t+m} - b^{n_t}_B \cdot (F_{2(k-1)n_t+m})_{\bas} \, \big)_{k\in \mathbb{Z}}
    &\longmapsfrom&
    (f_{(2k-1)n_t +m })_{k\in \mathbb{Z}}
  \end{tikzcd}
\end{equation}

\noindent {\bf (ii)} Applying instead one of the isomorphisms from Lem. \ref{SecondIsoOfDoublyCyclifiedHigherTwistedCocycleSpectra} in step {\bf (b)} yields similar, but essentially different maps between higher twisted cocycles of different extensions over $\frg$. For instance, using the isomorphism \eqref{SecondIsoOfDoublyCyclifiedHigherTwistedCocycleSpectra} results in the twisted  $\mathfrak{l}
    \Sigma^{\mathcolor{purple}{m-n_t}}\mathrm{K}^{n_t}\mathrm{U}$ cocycles given by
\begin{equation*}
  \begin{tikzcd}[
    row sep=-2pt, 
    column sep=12pt]
    \widehat{\mathfrak{g}}_B
    \ar[
      rr
    ]
    &&  
     \mathfrak{l}
    \big(
      \,
      \Sigma^{\color{purple}m-n_t}
      \mathrm{K}^{n_t}{U}
      \!\sslash\!\!B^{2n_t-1}\mathrm{U}(1) \, \big)
    \\
    H^{2n_t+1}_{A\, \bas} {\color{purple}-} b_{B'}^{n_t} \cdot \omega^A_{n_t+1}
    &\longmapsfrom&
    h_{2n_t+1}
    \\
 \big({\color{purple}+} {p_A}_*F_{2kn_t+m} - b^{n_t}_{B'} \cdot (F_{2(k-1)n_t+m})_{\bas} \, \big)_{k\in \mathbb{Z}}
    &\longmapsfrom&
    (f_{(2k-1)n_t +m })_{k\in \mathbb{Z}}
  \end{tikzcd}
\end{equation*}
where now the extension 
$$
\widehat{\mathfrak{g}}_{B'} \longrightarrow \frg \, ,
$$
is instead via the opposite $(n_t+1)$-cocycle
$$
\omega^{B'}_{n_t+1} \, := \, {\color{purple}-} {p_A}_* (H_A^{2n_t+1}) \, .
$$
\end{theorem}
\begin{proof}
The formulas have been set up such that the proof follows essentially verbatim with the proof of Thm. \ref{TwistedKTheoryRedIsoReOxi}, by adjusting the degrees accordingly.  
\end{proof}

\begin{remark}[\bf Towards Higher T-duality.] In direct generalization of Thm. \ref{TwistedKTheoryRedIsoReOxi}, the composite operation of Thm. \ref{HigherTwistedCocyclesRedIsoReOxi} swaps the extending cocycle $\omega_{n_t+1}^A$ with the $n_t$-dimensional reduction $\omega^{B}_{n_t+1} = {p_{A}}_* H_{A}^{2 n_t+1}$ of the original twisting cocycle. At the same time, it swaps ``wrapped and non-wrapped modes'' of the corresponding would be ``higher'' fluxes, now along the corresponding higher odd $n_t$-extensions. This should plausibly capture rational 
aspects of the generalized cohomology perspective in \cite{LSW16}. 
\end{remark}

\begin{remark}
[\bf Effect of extra isomorphisms of higher cyclified twisted cocycle classifying $L_\infty$-algebras]\label{RedundancyOfExtraIsomorphismsOfHigherCyclifedTwistedKSpectra}
In complete analogy to Rem. \ref{RedundancyOfExtraIsomorphismsOfCyclifiedTwistedKSpectra}, applying instead any of the isomorphisms from Lem. \ref{FurtherIsomorphismsOfHigherCyclifiedTwistedCocycleAlgebras} in step {\bf (b)} of Thm. \ref{TwistedKTheoryToroidalRedIsoReOxi} yields three further distinct, but closely related, bijections between higher twisted cocycles of higher central extensions over $\frg$. However, their non-rational lift and the corresponding interepretation is currently not so clear, both due to lack of concrete physical examples, and also the lack of a non-rational higher spherification adjunction for $t>3$ (cf. Rem. \ref{RationalizingTopologicalSpherification}). But see \S\ref{HigherTduality} for some relevant phenomena.
\end{remark}

\subsection{Torus extensions}
\label{TorusReduction}

We discuss a higher dimensional analog of the Ext/Cyc-adjunction of \S\ref{ExtCycAdj} corresponding to double-dimensional reduction/oxidation along products of rational circles.
This {\it toroidification} construction may conceptually be understood via rational homotopy theory, see \cite[p. 10]{SV24-Mysterious}. 
Here, we give an analogous discussion without explicitly passing through algebraic topology.

\medskip

\noindent
{\bf Toroidal central extensions.}
In evident generalization of Def. \ref{CentralExtension} we may consider super-$L_\infty$ extension by a whole sequence of 2-cocycles:
\begin{definition}[{\bf Central $n$-torus extension} {\cite[\S 3.1]{FSS20-HigherT}}]
\label{CentralTorusExtension}
For $\mathfrak{g} \,\in \, \mathrm{sLieAlg}_\infty$ equipped with $n \in \mathbb{N}$ 2-cocycles
$$
  \grayunderbrace{
  \dir{1}{c}_1,
  \cdots,
  \dir{n}{c}_{1}
  }{
    \mathrm{deg}
    \,=\,
    (\mathrm{2,\mathrm{evn}})
  }
  \;\in\;
  \mathrm{CE}(
    \mathfrak{g}
  )
  \,,
  \hspace{.4cm}
  \underset{r}{\forall}
  \;\;
  \mathrm{d}\, \dir{r}{c}_1
  \,=\,
  0
$$

\vspace{-2mm} 
\noindent we say that the {\it $n$-toroidal} central extension classified thereby 
is $\dir{1\cdots n}{\widehat{\mathfrak{g}}} \in \mathrm{sLieAlg}_\infty$ given by
$$
  \mathrm{CE}\Big(
    \;\;
    \dir{1\cdots n}{
      \widehat{\mathfrak{g}}
    }
    \;\;
  \Big)
  \;\simeq\;
  \mathrm{CE}(\mathfrak{g})
  \big[\;
    \smash{
    \grayunderbrace{
    \dir{1}{e}
    ,\cdots
    \dir{n}{e}
    }{
      \mathclap{
      \mathrm{deg}
        \,=\,
      (1,\mathrm{evn})
      }
    }
    }
  \;\big]
  \big/
  \big(
    \mathrm{d}\, \dir{r}{e}
    \;=\;
    \dir{r}{c}_1
  \big)_{r=1}^n
  \,.
$$
\vspace{.2cm}
\end{definition}
  The terminology {\it $n$-torus extension} in Def. \ref{CentralTorusExtension} refers to  the following phenomenon:
\begin{example}[\bf The real Whitehead $L_\infty$-algebra of the $n$-torus]
\label{RealWhiteheadAlgebraOfTorus}
 Consider on $0 \,\in\, \mathrm{sLieAlg}_\infty$ (Ex. \ref{RealWhiteheadAlgebraOfThePoint}) $n$ copies of (necessarily) the vanishing 2-cocycle 
 $$
   \dir{r}{c}_1 \,=\, 0
   \;\in\;
   \mathrm{CE}(0)
   \,,
   \;\;\;\;
   r \in \{1, \cdots, n\}\,.
  $$
  The corresponding $n$-toroidal extension (Def. \ref{CentralTorusExtension}) is the real Whitehead $L_\infty$-algebra (Ex. \ref{WhiteheadLInfinityAlgebra}) of the actual $n$-torus
  \begin{equation}
    \label{PlainNTorusAlgebra}
    \dir{1\cdots n}{\widehat{0}}
    \;\simeq\;
    \mathbb{R}^n
    \;\simeq\;
    \mathfrak{l}\big(
      \mathbb{T}^n
    \big)
    \;\defneq\;
    \mathfrak{l}\big(
      (\mathbb{R}/\mathbb{Z})^n
    \big)
    \,,
  \end{equation}
  given by
  $$
    \mathrm{CE}\big(
      \mathfrak{l}T^n
    \big)
    \;\simeq\;
    \FDGCA
    \big[\!\!
    \begin{array}{c}
      (\dir{r}{e})_{r=1}^{n}
    \end{array}
    \!\!\big]
    \big/
    \big(
      \mathrm{d}
      \,
      \dir{r}{e}
      \;=\;
      0
    \big)_{r=1}^n
    \,.
  $$
  Incidentally this shows also that the delooping
  $$
    b(\mathbb{R}^n)
    \;\simeq\;
    b\big(
      \mathfrak{l}(\mathbb{R}^n)
    \big)
    \;\simeq\;
    \mathfrak{l}\big(
      B \mathbb{R}^n
    \big)
  $$
  given by
  $$
    \mathrm{CE}\big(
      b\mathbb{R}^n
    \big)
    \;\simeq\;
    \FDGCA
    \Big[
      \smash{
        (\,
        \grayunderbrace{
          \dir{r}{\omega}_2 
        }{
          \mathclap{
          \mathrm{deg}
          \,=\,
          (2,\mathrm{evn})
          }
        }
        \,)_{r=1}^n
      }
    \Big]
    \big/
    \big(
      \mathrm{d}\,
      \dir{r}{\omega}_2 \, = 0 \, 
    \big)_{r=1}^n
  $$
  \vspace{.2cm}

  \noindent
  is the classifying $L_\infty$-algebra for $n$-toroidal extensions, in that an $n$-tuple of 2-cocycles is equivalently an $L_\infty$-homomorphism \eqref{LieHomomorphism} into it, this being the image under passage to real Whitehead $L_\infty$-algebras of the classification of $n$-torus principal bundles $P$ by the classifying space $B \mathbb{T}^n$, via pullback of the universal $n$-torus bundle $E \mathbb{T}^n$:
  $$
    \begin{tikzcd}[
      row sep=10pt,
      column sep=15pt
    ]
      \dir{1\cdots n}{\widehat{\mathfrak{g}}}
      \ar[
        d,
        ->>,
        "{
          \mathrm{hofib}\big(
            \;
            \dir{1\cdots n}{c}_1
          \big)\;
        }"{swap}
      ]
      \\
      \mathfrak{g}
      \ar[
        rr,
        "{
          \dir{1\cdots n}{c}_1
        }"
      ]
      &&
      b\mathbb{R}^n
      \\[-13pt]
      \dir{r}{c}_1
      &\mathclap{\longmapsfrom}&
      \dir{r}{\omega}_2
    \end{tikzcd}
    \hspace{.8cm}
    \underset{
      \adjustbox{
        raise=-4pt,
        scale=.6}{
        if $\mathfrak{g}
        \simeq \mathfrak{l}(X)$
      }
    }{
      \simeq
    }
    \hspace{.8cm}
    \mathfrak{l}
    \left(\!
    \begin{tikzcd}[
      row sep=20pt,
      column sep=15pt
    ]
      P
      \ar[
        d,
        ->>
      ]
      \ar[
        rr,
      ]
      \ar[
        drr,
        phantom,
        "{
          \scalebox{.6}{
            \color{gray}
            (pb)
          }
        }"{pos=.3}
      ]
      &&
      E \mathbb{T}^n
      \ar[
        d,->>
      ]
      \\
      X
      \ar[
        rr,
        "{
          \vec c_1
        }"
      ]
      &&
      B \mathbb{T}^n
    \end{tikzcd}
    \!\!\!\right).
  $$
\end{example}

The following example of toroidal super-extensions is noteworthy (maybe first highlighted in \cite[\S 2.1]{CdAIPB00}, see also \cite{HS18}):
\begin{example}[\bf Super-Minkowski spacetime as toroidal extension of a super-point]
\label{SuperMinkowskiAsToroidalExtensionOfSuperPoint}
Every super-Minkowski spacetime super-Lie algebra $\mathbb{R}^{1,d\,\vert\,\mathbf{N}}$ \eqref{SuperMinkowskiCE} is a $(1+d)$-toroidal central super-$L_\infty$ extension (Def. \ref{CentralTorusExtension})
of a superpoint:
\begin{equation}
  \label{SuperspacetimeExtensionOfSuperpoint}
  \begin{tikzcd}
    \mathbb{R}^{1,d\,\vert\,\mathbf{N}}
    \ar[
      r,
      ->>,
      "{
      }"
    ]
    &
    \mathbb{R}^{0\,\vert\,\mathbf{N}}
    \ar[
      rr,
      "{
        \scalebox{.8}{$
        (\hspace{1pt}
          \overline{\psi}
          \,\Gamma\,
          \psi
        )
        $}
      }"
    ]
    &&
    b\mathbb{R}^{1+d}
    \,,
  \end{tikzcd}
\end{equation}
where the super-point super-Lie algebra $\mathbb{R}^{0\vert N}$ is given simply by
\begin{equation}
  \label{CEOfSuperpoint}
  \mathrm{CE}\big(
    \mathbb{R}^{0\,\vert\,\mathbf{N}}
  \big)
  \;\;
  \simeq
  \;\;
  \FDGCA\big[
    (\psi^\alpha)_{\alpha=1}^N
  \big]\big/
  \big(
    \mathrm{d}\,
    \psi^\alpha
    \,=\,
    0
  \big)_{\alpha=1}^N
  \,.
\end{equation}
\end{example}

\begin{remark}[\bf Decomposing $n$-torus extensions into circle-extensions]
\label{DecomposingNTorusExtensions}
$\,$

\noindent {\bf (i)}   A 1-torus extension in the sense of Def. \ref{CentralTorusExtension} is evidently the same as a central extension according to Def. \ref{CentralExtension}, namely a ``circle extension'' (cf. Ex. \ref{RealWhiteheadAlgebraOfTorus})
\vspace{-2mm} 
  $$
    \widehat{\mathfrak{g}}
    \;\simeq\;
    \dir{1}{\widehat{\mathfrak{g}}}
    \,.
  $$
 Any $n$-torus extension may equivalently be obtained as a sequence of $r_j$-torus extensions, for any partitioning $r_j \in \mathbb{N}$, $\sum_j r_j = n$, e.g.:
 \vspace{-4mm} 
  \begin{equation}
    \label{IteratedToroidalHofibers}
    \hspace{1cm} 
    \begin{tikzcd}[
      row sep=28pt,
      column sep=15pt
    ]
    \dir{1\cdots 3}{\widehat{\mathfrak{g}}}
    \ar[
      dddd,
      "{
        \mathrm{hofib}\big(
          \;\,
          \dir{1\cdots 3}{c}_1
        \big)
      }"
    ]
    \\
    {}
    \\
    {}
    \\
    {}
    \\
    \mathfrak{g}
    \ar[
      rr,
      "{
        \dir{1\cdots 3}{c}_1
      }"
    ]
    &&
    b\mathbb{R}^3
    \end{tikzcd}
    \hspace{1.3cm}
    \begin{tikzcd}[
      row sep=24pt,
      column sep=15pt
    ]
    \dir{1\cdots 3}{\widehat{\mathfrak{g}}}
    \ar[
      d,
      ->>,
      "{
        \mathrm{hofib}(
          \dir{3}{c}_1
        )
      }"
    ]
    \\
    \dir{1\cdots 2}{\widehat{\mathfrak{g}}}
    \ar[
      ddd,
      "{
        \mathrm{hofib}(
          \;\,
          \dir{1\cdots 2}{c}_1
        )
      }"
    ]
    \ar[
      rr,
      "{
        \dir{3}{c}_1
      }"
    ]
    &&
    b \mathbb{R}
    \\
    {}
    \\
    {}
    \\
    \mathfrak{g}
    \ar[
      rr,
      "{
        \dir{1\cdots 2}{c}_1
      }"
    ]
    &&
    b\mathbb{R}^2
    \end{tikzcd}
    \hspace{1.3cm}
    \begin{tikzcd}[
      row sep=28pt,
      column sep=15pt
    ]
      \dir{1\cdots 3}{\widehat{\mathfrak{g}}}
      \ar[
        d,
        ->>,
        "{
          \mathrm{hofib}(
            \dir{3}{c}_1
          )
        }"
      ]
      \\
      \dir{1\cdots 2}{\widehat{\mathfrak{g}}}
      \ar[
        rr,
        "{
          \dir{3}{c}_1
        }"
      ]
      \ar[
        d,
        ->>,
        "{
          \mathrm{hofib}(
            \dir{2}{c}_1
          )
        }"
      ]
      &&
      b\mathbb{R}
      \\
      \dir{1}{\widehat{\mathfrak{g}}}
      \ar[
        d,
        ->>,
        "{
          \mathrm{hofib}(
            \dir{1}{c}_1
          )
        }"{pos=0.4}
      ]
      \ar[
        rr,
        "{
          (
            \dir{2}{c}_1
          )
        }"
      ]
      &&
      b\mathbb{R}
      \\
      \mathfrak{g}
      \ar[
        rr,
        "{
          \dir{1}{c}_1
        }"
      ]
      &&
      b \mathbb{R}
      \,.
    \end{tikzcd}
  \end{equation}
\noindent {\bf (ii)}   Here the order of the extensions does not matter, up to isomorphism, in that the following diagram commutes:
  \begin{equation}
    \label{CommutingPairOfExtensions}
    \begin{tikzcd}[row sep=small, column sep=70pt]
      &
      &
      \dir{1\cdots 2}{\widehat{g}}
      \ar[
        dr,
        ->>,
        "{
          \mathrm{hofib}\big(
            \dir{2}{p}{}^\ast
            (\dir{1}{c}_1)
          \big)
        }"
      ]
      \ar[
        dl,
        ->>,
        "{
          \mathrm{hofib}\big(
            \dir{1}{p}{}^\ast
            (\dir{2}{c}_1)
          \big)
        }"{swap}
      ]
      \\
      b\mathbb{R}
      \ar[
        r,
        <-,
        "{
          \dir{1}{p}{}^\ast(
            \dir{2}{c}_1
          )
        }"
      ]
      &
      \dir{1}{\widehat{\mathfrak{g}}}
      \ar[
        dr,
        ->>,
        "{
          \dir{1}{p}
          \,=\,
          \mathrm{hofib}(
            \dir{1}{c}_1
          )
        }"{swap}
      ]
      &&
      \dir{2}{\widehat{\mathfrak{g}}}
      \ar[
        dl,
        ->>,
        "{
          \dir{2}{p}
          \,=\,
          \mathrm{hofib}(
            \dir{2}{c}_1
          )
        }"
      ]
      \ar[
        r,
        "{
          \dir{2}{p}{}^\ast(
            \dir{1}{c}_1
          )
        }"
      ]
      &
      b
      \mathbb{R}
      \\
      &
      &
      \mathfrak{g}
      \ar[
        dr,
        "{
          \dir{1}{c}_1
        }"{swap}
      ]
      \ar[
        dl,
        "{
          \dir{2}{c}_1
        }"
      ]
      \\
      &
      b\mathbb{R}
      &&
      b\mathbb{R}
      \mathrlap{\,.}
    \end{tikzcd}
  \end{equation}
  This is ultimately because all cocycles are (by assumption) defined down on $\mathfrak{g}$, hence each independent of the extension classified by the others, which we may express as the statement that the fiber integration \eqref{FiberIntegration} vanishes
  of the next cocycle over the extension fiber brought about by the previous cocycle:
  \begin{equation}
    \label{TrivialFiberIntegrationofTorusCocycles}
    \dir{1}{p}{}_\ast
    \dir{2}{p}{}^\ast 
    \dir{r}{c}_1
    \;=\;
    0
    \,,
    \;\;\;\;
    \dir{2}{p}{}_\ast
    \dir{1}{p}{}^\ast 
    \dir{r}{c}_1
    \;=\;
    0
    \,.
  \end{equation}
  This relation becomes crucial below in specializing iterated {\it cyclification} along $n$-torus fibrations to {\it toroidifictionation}.

 \noindent {\bf (iii)}  Stated more abstractly, the commuting square in \eqref{CommutingPairOfExtensions} is ``Cartesian'', exhibiting $\dir{1\cdots 2}{\widehat{\mathfrak{g}}}$ as the {\bf fiber product} (in fact as the {\it homotopy fiber product}) of $\dir{1}{\widehat{\mathfrak{g}}}$ with $\dir{2}{\widehat{\mathfrak{g}}}$ over their common base $\mathfrak{g}$. 
  
\noindent {\bf (iv)}   Note that the property of commuting squares to exhibit fiber products is closed under ``pasting'' these squares together: In a commuting diagram of the form
\vspace{-2mm} 
  \begin{equation}
    \label{PastingLaw}
    \begin{tikzcd}[
      row sep={20pt, between origins},
      column sep={60pt, between origins}
    ]
      & &
      \dir{1\cdots 3}{\widehat{\mathfrak{g}}}
      \ar[
        dl,
        ->>
      ]
      \ar[
        dr,
        ->>
      ]
      \\
      &
      \dir{1 \cdots 2}{\widehat{\mathfrak{g}}}
      \ar[
        dr,
        ->>
      ]
      \ar[
        dl,
        ->>
      ]
      &&
      \dir{3}{\widehat{\mathfrak{g}}}\;.
      \ar[
        dl,
        ->>
      ]
      \\
      \dir{1}{\widehat{\mathfrak{g}}}
      \ar[
        dr,
        ->>
      ]
      &&
      \dir{2}{\widehat{\mathfrak{g}}}
      \ar[
        dl,
        ->>
      ]
      \\
      &
      \mathfrak{g}
    \end{tikzcd}
  \end{equation}
  if the bottom diamond exhibits a fiber product, then the total diamond does, too, iff the top diamond does. (This is a general abstract fact known as the ``pasting law'', but here in our context of $n$-toroidal central extensions of super-$L_\infty$ algebras it is also readily checked by inspecting generators.)

\noindent {\bf (v)}   This property clearly iterates over consecutive pastings of fiber product diamonds. As such it controls the final picture in \eqref{SummaryDiagram} below.
\end{remark}

\noindent
{\bf 2-Toroidification.}
We need the following simple observation:
\begin{lemma}[\bf Quotient of $L_\infty$-algebra by abelian anti-ideal]
\label{QuotientByAbelianAntiIdeal}
  Given an $L_\infty$-algebra $\mathfrak{g}$ whose CE-algebra has a closed generator $\omega$, then discarding that generator yields the CE-algebra of a sub-$L_\infty$-algebra:
  \vspace{-1mm} 
  $$
    \begin{tikzcd}[
      ampersand replacement=\&
    ]
    \FDGCA
    \big[
    \def\arraystretch{1.3}
    \def\arraycolsep{2pt}
    \begin{array}{c}
      (e^i)_{i \in I}
    \end{array}
    \big]
    \big/
    \big(\!
    \def\arraycolsep{2pt}
    \def\arraystretch{1.3}
    \begin{array}{ccl}
      \mathrm{d}
      \,
      e^i
      &=&
      P_{(0)}^i(\vec e\,)
    \end{array}
    \!\big)
    \ar[
      rr,
      <<-
    ]
    \&\&
    \FDGCA
    \left[
    \def\arraystretch{1.1}
    \def\arraycolsep{2pt}
    \begin{array}{c}
      \omega
      \\
      (e^i)_{i \in I}
    \end{array}
    \right]
    \Big/
    \left(\!
    \def\arraycolsep{2pt}
    \def\arraystretch{1.1}
    \begin{array}{ccl}
      \mathrm{d}\, \omega
      &=&
      0
      \\
      \mathrm{d}
      \,
      e^i
      &=&
      \sum_k
      P_{(k)}^i(\vec e\,)
      \,
      \smash{
     \underbrace{\omega \cdots \omega}
       _{\scalebox{.7}{\color{gray} \rm $k$ factors}}
      }
    \end{array}
    \right).
    \end{tikzcd}
  $$

\end{lemma}
\begin{proof}
The claim is that the operator $\mathrm{d}$ on the left still satisfies $\mathrm{d}\circ \mathrm{d} = 0$ if that on the right does. But on the right, the condition is 
$$
  \mathrm{d}^2\, e^i \,=\, 0
  \qquad 
    \Leftrightarrow
  \qquad 
  \underset{k \in \mathbb{N}}{\forall}
  \big(
  \mathrm{d}
  P^i_{(k)}(
    \vec e
  \,)
  =
  0
  \big),
$$
hence immediately implies the claim for $k = 0$.
\end{proof}

\begin{example}[\bf Double cyclification]
\label{DoubleCyclification}
For $\mathfrak{h} \in \mathrm{sLieAlg}_\infty$ with presentation as in Def. \ref{Cyclification}, applying cyclification (Def. \ref{Cyclification}) twice yields the $L_\infty$-algebra $\mathrm{cyc}^2 \mathfrak{h}$ given by
\begin{equation}
  \label{CEOfDoubleCyclification}
  \mathrm{CE}\big(
    \mathrm{cyc}^2(
      \mathfrak{h}
    )
  \big)
  \;\simeq\;
  \FDGCA
  \left[\!
  \def\arraystretch{1.2}
  \begin{array}{c}
    \dir{2}{\omega}_2,
    \\
    \dir{1}{\omega}_2,
    \\
    \dir{2}{\mathrm{s}}
    \dir{1}{\omega}_2,
    \\
    (e_i)_{i \in I},
    \\
    (\dir{2}{\mathrm{s}}e^i)_{i \in I}
    \\
    (\dir{1}{\mathrm{s}}e^i)_{i \in I},
    \\
    (\dir{2}{\mathrm{s}}\dir{1}{\mathrm{s}}e^i)_{i \in I}
  \end{array}
  \!\right]
  \Big/\ \!
  \left(\!\!
  \def\arraystretch{1.2}
  \begin{array}{ccc}
    \mathrm{d}
    \,
    \dir{2}{\omega}_2
    &=&
    0
    \\
    \mathrm{d}
    \,
    \dir{1}{\omega}_2
    &=&
    \dir{2}{\omega}_2 
    \,
    \dir{2}{\mathrm{s}}
    \dir{1}{\omega}_2
    \\
    \mathrm{d}
    \,
    \dir{2}{\mathrm{s}}
    \dir{1}{\omega}_2
    &=&
    0
    \\
    \mathrm{d}
    \,
    e^i
    &=&
    \mathrm{d}_{\mathfrak{h}}e^i
    +
    \dir{1}{\omega}_2 
    \,
    \dir{1}{\mathrm{s}}e^i
    +
    \dir{2}{\omega}_2
    \,
    \dir{2}{\mathrm{s}} e^i
    \\
    \mathrm{d}
    \,
    \dir{2}{\mathrm{s}}e^i
    &=&
    -
    \dir{2}{\mathrm{s}}
    \,
    \mathrm{d}_{\mathfrak{h}}
    e^i
    -
    (\dir{2}{\mathrm{s}}
    \dir{1}{\omega}_2)
    (\dir{1}{\mathrm{s}}e^i)
    -
    \dir{1}{\omega}_2
    \,
    \dir{2}{\mathrm{s}}
    \dir{1}{\mathrm{s}}
    e^i
    \\
    \mathrm{d}
    \,
    \dir{1}{\mathrm{s}}e^i
    &=&
    -\dir{1}{\mathrm{s}}
     \,
      \mathrm{d}_{\mathfrak{h}}e^i
    + 
    \dir{2}{\omega}_2\, 
    \dir{2}{\mathrm{s}}
    \dir{1}{\mathrm{s}}
    e^i
    \\
    \mathrm{d}
    \,
    \dir{2}{\mathrm{s}}
    \dir{1}{\mathrm{s}}
    e^i
    &=&
    \dir{2}{\mathrm{s}}
    \dir{1}{\mathrm{s}}
    \,
    \mathrm{d}_{\mathfrak{h}}
    e^i
  \end{array}
  \!\! \right).
\end{equation}
\end{example}

\begin{definition}[\bf 2-Toroidification]
\label{Toroidification}
Since the generator $\dir{2}{\mathrm{s}}\dir{1}{\omega}_2$ in the CE-algebra \eqref{CEOfDoubleCyclification} of a double cyclification 
is closed, discarding this generator yields (by Lem. \ref{QuotientByAbelianAntiIdeal}) a sub-$L_\infty$-algebra, to be called the {\it toroidification} of the given $\mathfrak{h} \,\in\, \mathrm{sLieAlg}_\infty$
\vspace{-1mm} 
$$
  \begin{tikzcd}
  \mathrm{tor}^2(\mathfrak{h})
  \ar[r, hook]
  &
  \mathrm{cyc}^2(\mathfrak{h})
  \,,
  \end{tikzcd}
$$ 
and given by
\begin{equation}
  \label{CEOfToroidification}
  \mathrm{CE}\big(
    \mathrm{tor}^2 \mathfrak{h}
  \big)
  \;\simeq\;
  \FDGCA
  \left[\!
  \def\arraystretch{1.3}
  \begin{array}{c}
    \dir{2}{\omega}_2,
    \\
    \dir{1}{\omega}_2,
    \\
    (e_i)_{i \in I},
    \\
    (\dir{2}{\mathrm{s}}e^i)_{i \in I}
    \\
    (\dir{1}{\mathrm{s}}e^i)_{i \in I},
    \\
    (\dir{2}{\mathrm{s}}\dir{1}{\mathrm{s}}e^i)_{i \in I}
  \end{array}
  \!\right]
  \Big/\ \!
  \left(\!
  \def\arraystretch{1.3}
  \begin{array}{ccc}
    \mathrm{d}
    \,
    \dir{2}{\omega}_2
    &=&
    0
    \\
    \mathrm{d}
    \,
    \dir{1}{\omega}_2
    &=&
    0
    \\
    \mathrm{d}
    \,
    e^i
    &=&
    \mathrm{d}_{\mathfrak{h}}e^i
    +
    \dir{1}{\omega}_2 
    \,
    \dir{1}{\mathrm{s}}e^i
    +
    \dir{2}{\omega}_2
    \,
    \dir{2}{\mathrm{s}} e^i
    \\
    \mathrm{d}
    \,
    \dir{2}{\mathrm{s}}e^i
    &=&
    -
    \dir{2}{\mathrm{s}}
    \,
    \mathrm{d}_{\mathfrak{h}}
    e^i
    -
    \dir{1}{\omega}_2
    \,
    \dir{2}{\mathrm{s}}
    \dir{1}{\mathrm{s}}
    e^i
    \\
    \mathrm{d}
    \,
    \dir{1}{\mathrm{s}}e^i
    &=&
    -\dir{1}{\mathrm{s}}
      \, 
      \mathrm{d}_{\mathfrak{h}}e^i
    + 
    \dir{2}{\omega}_2\, 
    \dir{2}{\mathrm{s}}
    \dir{1}{\mathrm{s}}
    e^i
    \\
    \mathrm{d}
    \,
    \dir{2}{\mathrm{s}}
    \dir{1}{\mathrm{s}}e^i
    &=&
    \dir{2}{\mathrm{s}}
    \dir{1}{\mathrm{s}}
    \,
    \mathrm{d}_{\mathfrak{h}}
    e^i
  \end{array}
  \!\!\right)
  .
\end{equation}
We may regard this as fibered over $b\mathbb{R}^2$ via:
\vspace{-3mm} 
\begin{equation}
  \label{ToroidificationChernClasses}
  \begin{tikzcd}[
    row sep=8pt, 
    column sep=5pt
  ]
    \mathrm{tor}^2(\mathfrak{h})
    \ar[
      dd,
      "{\;\;
        \dir{1\cdots 2}{\omega}_2
        \,:=\
        (\dir{1}{\omega}_2, \dir{2}{\omega}_2)
      }"
    ]
    &
    \\
    \\
    b\mathbb{R}^2   
    \mathrlap{\,.}
  \end{tikzcd}
\end{equation}
\end{definition}

\begin{example}[{\bf Double cyclification and toroidification of the 4-sphere} {\cite[Ex. 2.6]{SV23-Mysterious}}]
The double cyclification (Ex. \ref{DoubleCyclification}) of (the real Whitehead $L_\infty$-algebra of) the 4-sphere (Ex. \ref{WhiteheadAlgebraOfS4}) is given by
\vspace{-1mm}
$$
  \mathrm{CE}\big(
    \mathrm{cyc}^2(
      \mathfrak{l}S^4
    )
  \big)
  \;\simeq\;
 \FDGCA
  \left[
  \def\arraystretch{1}
  \begin{array}{c}
    \dir{2}{\omega}_2
    \\
    \dir{1}{\omega}_2
    \\
    \dir{2}{\mathrm{s}}
    \dir{1}{\omega}_2
    \\
    g_4
    \\
    \dir{2}{\mathrm{s}}
    g_4
    \\
    \dir{1}{\mathrm{s}}g_4
    \\
    \dir{2}{\mathrm{s}}
    \dir{1}{\mathrm{s}}
    g_4
    \\
    g_7
    \\
    \dir{2}{\mathrm{s}}g_7
    \\
    \dir{1}{\mathrm{s}}
    g_7
    \\
    \dir{2}{\mathrm{s}}
    \dir{1}{\mathrm{s}}g_7
  \end{array}
  \right]
  \Big/
  \left(\!
  \def\arraystretch{1.3}
  \begin{array}{ccc}
    \mathrm{d}
    \,
    \dir{2}{\omega}_2
    &=&
    0
    \\
    \mathrm{d}
    \,
    \dir{1}{\omega}_2
    &=&
    \dir{2}{\omega}_2\, 
    \dir{2}{\mathrm{s}}
    \omega_2
    \\
    \mathrm{d}
    \,
    \dir{2}{\mathrm{s}}
    \omega_2
    &=&
    0
    \\
    \mathrm{d}
    \,
    g_4
    &=&
    \dir{1}{\omega}_2
    \,
    \dir{1}{\mathrm{s}}g_4
    +
    \dir{2}{\omega}_2
    \,
    \dir{2}{\mathrm{s}}g_4
    \\
    \mathrm{d}
    \,
    \dir{2}{\mathrm{s}}g_4
    &=&
    -
    (\dir{2}{\mathrm{s}}\dir{1}{\omega}_2) (\dir{1}{\mathrm{s}}g_4)
    -
    \dir{1}{\omega}_2
    \,
    \dir{2}{\mathrm{s}}
    \dir{1}{\mathrm{s}}
    g_4
    \\
    \mathrm{d}
    \,
    \dir{1}{\mathrm{s}}
    g_4
    &=&
    \dir{2}{\omega}_2\,
    \dir{2}{\mathrm{s}}
    \dir{1}{\mathrm{s}}
    g_4
    \\
    \mathrm{d}
    \,
    \dir{2}{\mathrm{s}}
    \dir{1}{\mathrm{s}}
    g_4
    &=&
    0
    \\
    \mathrm{d}
    \,
    g_7
    &=&
    \tfrac{1}{2}
    g_4 g_4
    +
    \dir{1}{\omega}_2
    \,
    \dir{1}{\mathrm{s}}g_7
    +
    \dir{2}{\omega}_2
    \,
    \dir{2}{\mathrm{s}}g_7
    \\
    \mathrm{d}\,
    \dir{2}{\mathrm{s}}
    g_7
    &=&
    -g_4 \, 
      \dir{2}{\mathrm{s}}g_4
    -(\dir{2}{\mathrm{s}}\omega_2)
    (\dir{1}{\mathrm{s}}g_7)
    -
    \dir{1}{\omega}_2\, 
    \dir{2}{\mathrm{s}}
    \dir{1}{\mathrm{s}}g_7
    \\
    \mathrm{d}
    \,
    \dir{1}{\mathrm{s}}
    g_7
    &=&
    -g_4
    \,
    \dir{1}{\mathrm{s}}g_4
    +
    \dir{2}{\omega}_2
    \,
    \dir{2}{\mathrm{s}}
    \dir{1}{\mathrm{s}}
    g_7
    \\
    \mathrm{d}
    \,
    \dir{2}{\mathrm{s}}
    \dir{1}{\mathrm{s}}
    g_7
    &=&
    (\dir{2}{\mathrm{s}}g_4)
    (\dir{1}{\mathrm{s}}g_4)
    +
    g_4
    \,
    \dir{2}{\mathrm{s}}
    \dir{1}{\mathrm{s}}g_4
  \end{array}
  \!\! \right)
$$
and its toroidification (Ex. \ref{Toroidification})
is given by
$$
  \mathrm{CE}\big(
    \mathrm{tor}^2(\mathfrak{l}S^4)
  \big)
  \;\;
  \simeq
  \;\;
  \FDGCA
  \left[
  \def\arraystretch{1.3}
  \begin{array}{c}
    \dir{2}{\omega}_2
    \\
    \dir{1}{\omega}_2
    \\
    g_4
    \\
    \dir{2}{\mathrm{s}}
    g_4
    \\
    \dir{1}{\mathrm{s}}g_4
    \\
    \dir{2}{\mathrm{s}}
    \dir{1}{\mathrm{s}}g_4
    \\
    g_7
    \\
    \dir{2}{\mathrm{s}}g_7
    \\
    \dir{1}{\mathrm{s}}g_7
    \\
    \dir{2}{\mathrm{s}}
    \dir{1}{\mathrm{s}}g_7
  \end{array}
  \right]
  \Big/
  \left(\!
  \def\arraystretch{1.3}
  \begin{array}{ccc}
    \mathrm{d}
    \,
    \dir{2}{\omega}_2
    &=&
    0
    \\
    \mathrm{d}
    \,
    \dir{1}{\omega}_2
    &=&
    0
    \\
    \mathrm{d}
    \,
    g_4
    &=&
    \dir{1}{\omega}_2
    \,
    \dir{1}{\mathrm{s}}g_4
    +
    \dir{2}{\omega}_2
    \,
    \dir{2}{\mathrm{s}}g_4
    \\
    \mathrm{d}
    \,
    \dir{2}{\mathrm{s}}g_4
    &=&
    -
    \dir{1}{\omega}_2
    \,
    \dir{2}{\mathrm{s}}
    \dir{1}{\mathrm{s}}
    g_4
    \\
    \mathrm{d}
    \,
    \dir{1}{\mathrm{s}}g_4
    &=&
    \dir{2}{\omega}_2\,
    \dir{2}{\mathrm{s}}
    \dir{1}{\mathrm{s}}g_4
    \\
    \mathrm{d}
    \,
    \dir{2}{\mathrm{s}}
    \dir{1}{\mathrm{s}}g_4
    &=&
    0
    \\
    \mathrm{d}
    \,
    g_7
    &=&
    \tfrac{1}{2}
    g_4 g_4
    +
    \dir{1}{\omega}_2
    \,
    \dir{1}{\mathrm{s}}g_7
    +
    \dir{2}{\omega}_2
    \,
    \dir{2}{\mathrm{s}}g_7
    \\
    \mathrm{d}\,
    \dir{2}{\mathrm{s}}
    g_7
    &=&
    -g_4 \, 
      \dir{2}{\mathrm{s}}g_4
    -\dir{1}{\omega}_2\, 
     \dir{2}{\mathrm{s}}
     \dir{1}{\mathrm{s}}g_7
    \\
    \mathrm{d}
    \,\dir{1}{\mathrm{s}}g_7
    &=&
    -g_4
    \,
    \dir{1}{\mathrm{s}}g_4
    +
    \dir{2}{\omega}_2
    \,
    \dir{2}{\mathrm{s}}
    \dir{1}{\mathrm{s}}g_7
    \\
    \mathrm{d}
    \,
    \dir{2}{\mathrm{s}}
    \dir{1}{\mathrm{s}}
    g_7
    &=&
    (\dir{2}{\mathrm{s}}g_4)
    (\dir{1}{\mathrm{s}}g_4)
    +
    g_4
    \,
    \dir{2}{\mathrm{s}}
    \dir{1}{\mathrm{s}}g_4
  \end{array}
  \!\! \right).
$$
\end{example}

\smallskip

In a generalization of Rem. \ref{BasicAndFiberForms} we have:
\begin{remark}[\bf Basic and fiber forms on a 2-toroidially extended super-$L_\infty$ algebra]\label{BasicAndFiberFormsOn2ToroidalExtension}
Given a 2-toroidal central extension as in Def. \ref{CentralTorusExtension}, every element in its CE-algebra decomposes uniquely as  a sum of the form
\begin{equation}
  \label{DecompositionOver2ToroidalExtension}
  \alpha
  \;=\;
  \alpha_{\mathrm{bas}}
  \,+\,
  \dir{1}{e}\, \alpha_1
  \,+\,
  \dir{2}{e}\, \alpha_2
  \,+\,
  \dir{2}{e}\, \dir{1}{e}\,
  \alpha_{2\, 1}
\end{equation}
with all the coefficients in the image of the total pullback operation 
$$
  \alpha_{\mathrm{bas}}, \alpha_1, \alpha_2, \alpha_{2\,1}
  \;\in\;
  \dir{2}{p}{}^\ast
  \,
  \dir{1}{p}{}^\ast
  \,
  \mathrm{CE}(\mathfrak{g})
  \,.
$$
The {\it total fiber integration} $\dir{1\,2}{p}_\ast$, being equivalently the composition of the two (partial) fiber integrations $\dir{2}{p}_\ast$ and $\dir{1}{p}_\ast$, is the super-graded derivation of degree $(-2,\mathrm{evn})$:
\vspace{-2mm} 
  \begin{equation}
\label{2ToroidalFiberIntegration}
    \begin{tikzcd}[row sep=-3pt, column sep=0pt]
    \mathrm{CE}(\, \dir{1\,2}{\widehat{\mathfrak{g}}}\,)
    \ar[
      rr,
      "{
        \dir{1\,2}{p}_\ast 
      }"
    ]
    &&
    \mathrm{CE}(\mathfrak{g})
    \\
    \dir{2}{e} \, \dir{1}{e} &\longmapsto& 1
    \\
    \dir{2}{e} &\longmapsto& 0
    \\
    \dir{1}{e} &\longmapsto& 0
    \\
    e^i &\longmapsto& 0 
    \mathrlap{\,,}
    \end{tikzcd}  
  \end{equation}
  hence extracting the coefficient $\alpha_{2\, 1}$, being the ``mode winding along both toroidal directions''
  (where in the last line $(e^i)_{\i \in I}$ denote generators for $\mathrm{CE}(\mathfrak{g})$). Furthermore, since each of the (partial) fiber integrations in the composition $\dir{1\, 2}{p}_\ast = \dir{1}{p}_\ast \circ \dir{2}{p}_\ast$ are maps of chain complexes themselves (cf. Rem. \ref{BasicAndFiberForms}), it follows that the total fiber integration also commutes with the differentials
$$
\dir{1\,2}{p}_\ast  \circ \mathrm{d}_{\, \dir{1\,2}{\widehat{{\mathfrak{g}}}}}\,  = \, + \, \mathrm{d}_\frg \circ \dir{1\,2}{p}_\ast \, .
$$
\end{remark}

\begin{proposition}[\bf Universal property of 2-toroidification]
\label{UniversalPropertyOfToroidification}
Consider a central extension {\rm (Ex. \ref{CentralExtension})}
\vspace{-1mm} 
$$
  \begin{tikzcd}
    \dir{1}{\widehat{\mathfrak{g}}}
    \ar[
      rrr,
      "{
        \scalebox{1}{$
        \dir{1}{p} \,=\,
        \mathrm{hofib}(
          \dir{1}{c}_1
        )
        $}
      }"
    ]
    &&&
    \mathfrak{g}
    \ar[
      rr,
      "{
        \dir{1}{c}_1
      }"
    ]
    &&
    b\mathbb{R}
  \end{tikzcd}
$$
carrying a $\mathrm{cyc}(\mathfrak{h})$-valued cocycle 
$\tilde f : \dir{1}{\widehat{\mathfrak{g}}} \xrightarrow{\;} \mathrm{cyc}(\mathfrak{h})$ whose Chern class 
$$
  \dir{2}{c}_1
  \;:=\;
  \tilde f^\ast \omega_2
$$
has trivial fiber integration along $\dir{1}{p}$.
Then the reduction 
$\tilde{\tilde f}$\eqref{CyclificationHomIsomorphism} of $\tilde f$ along $\dir{1}{p}$
has factors uniquely through the toroidification \eqref{CEOfToroidification}, which has an inverse 
oxidation when regarded as sliced over $b\mathbb{R}^2$ via \eqref{ToroidificationChernClasses}:
\vspace{-1mm} 
$$
  \begin{tikzcd}[
    row sep=8pt
  ]
    \dir{1\cdots 2}{\widehat{\mathfrak{g}}}
    \ar[
      dd,
      ->>,
      "{
        \dir{2}{p}
      }"{swap,pos=.3}
    ]
    \ar[
      rr,
      "{
        f
      }"
    ]
    &&
    \mathfrak{h}
    \\
    \\
    \dir{1}{\widehat{\mathfrak{g}}}
    \ar[
      dd,
      ->>,
      "{
        \dir{1}{p}
      }"{swap, pos=.4}
    ]
    \\
    \\
    \mathfrak{g}
  \end{tikzcd}
  \hspace{.0cm}
  \leftrightsquigarrow
  \hspace{.5cm}
  \begin{tikzcd}[
    row sep=10pt
  ]
    \dir{1\cdots 2}{\widehat{\mathfrak{g}}}
    \ar[
      dd,
      ->>,
      "{
        \dir{2}{p}
      }"{swap, pos=.3}
    ]
    \\
    {}
    \\
    \dir{1}{\widehat{\mathfrak{g}}}
    \ar[
      dd,
      ->>,
      "{
        \dir{1}{p}
      }"{swap, pos=.4}
    ]
    \ar[
      rr,
      "{ \widetilde f }"
    ]
    \ar[
      dr,
      "{
        \dir{2}{c}_1
      }"{swap, xshift=3pt, yshift=3pt}
    ]
    &&
    \mathrm{cyc}(\mathfrak{h})
    \ar[
      dl,
      "{
        \omega_2
      }"
    ]
    \\
    &
    b\mathbb{R}
    \\
    \mathfrak{g}
  \end{tikzcd}
  \hspace{.4cm}
  \underset{
    {
      \adjustbox{
        raise=-4pt,
        scale=.9
      }{$
      \dir{1}{p}_\ast
        f^\ast 
        \dir{1}{\omega}_2
      \;=\;
      0  
    $}
    }
  }{
    \leftrightsquigarrow
  }
  \hspace{.4cm}
  \adjustbox{
    raise=-.1cm
  }{
  \begin{tikzcd}[
    row sep=10pt
  ]
    \dir{1 \cdots 2}{\widehat{\mathfrak{g}}}
    \ar[
      dd,
      ->>,
      "{
        \dir{2}{p}
      }"{swap, pos=.3}
    ]
    &&
    &[-10pt]
    \\
    \\
    \dir{1}{\widehat{\mathfrak{g}}}
    \ar[
      dd,
      ->>,
      "{
        \dir{1}{p}
      }"{swap, pos=.3}
    ]
    \\
    \\
    \mathfrak{g}
    \ar[
      rr,
      dashed
    ]
    \ar[
      rr,
      dashed,
      "{
        \tilde{\tilde f}
      }"{description}
    ]
    \ar[
      dr,
      "{
        \dir{1\cdots 2}{c}_1
      }"{swap}
    ]
    &&
    \mathrm{tor}^2(\mathfrak{h})
    \ar[
      dl,
      "{ 
         \dir{1\cdots 2}{\omega}_1
       }"
    ]
    \ar[
      r,
      hook
    ]
    &
    \mathrm{cyc}^2(\mathfrak{h})
    \\
    &
    b\mathbb{R}^2
  \end{tikzcd}
  }
$$
given by
\begin{equation}
  \label{2ToroidificationHomBijection}
  \hspace{-2.2cm}
  \begin{tikzcd}[
    sep=0pt
  ]
    \dir{1\cdots 2}{\widehat{\mathfrak{g}}}
    \ar[
      rr,
      "{ f }"
    ]
    &&
    \mathfrak{h}
    \\
    \alpha^i_{\mathrm{bas}}
    +
    \dir{1}{e} 
    \, 
    \alpha^i_1
    \,+
    &\longmapsfrom&
    e^i
    \\[-2pt]
    \dir{2}{e} 
    \, 
    \alpha^i_2
    \,+\,
    \dir{2}{e} 
    \,
    \dir{1}{e}
    \;
    \alpha^i_{2\, 1}
  \end{tikzcd}
  \hspace{.5cm}
  \leftrightsquigarrow
  \hspace{.5cm}
  \begin{tikzcd}[
    row sep=-1pt, 
    column sep=0pt]
    \dir{1}{\widehat{\mathfrak{g}}}
    \ar[
      rr,
      "{ 
        \widetilde{f} 
      }"
    ]
    &&
    \mathrm{cyc}(\mathfrak{h})
    \\[+3pt]
    c_2
    &\longmapsfrom&
    \omega_2
    \\
    \alpha^i_{\mathrm{bas}}
    +
    \dir{1}{e} 
    \, 
    \alpha^i_1
    &\longmapsfrom&
    e^i
    \\
    -
    \alpha^i_2
    \,-\,
    \dir{1}{e} 
    \;
    \alpha^i_{2\, 1}
    &\longmapsfrom&
    \dir{2}{\mathrm{s}}e^i
    \mathrlap{\,,}
  \end{tikzcd}
    \hspace{.5cm}
    \leftrightsquigarrow
    \hspace{.5cm}
  \begin{tikzcd}[
    row sep=-6pt, 
    column sep=0pt]
    \dir{1\cdots2}{
      \widehat{\mathfrak{g}}
    }
    \ar[
      rr,
      "{
        \tilde{\tilde{f}}
      }"
    ]
    &&
    \mathrm{tor}^2(\mathfrak{h})
    \\
    c_2
    &\longmapsfrom&
    \dir{2}{\omega}_2
    \\
    c_1
    &\longmapsfrom&
    \dir{1}{\omega}_2
    \\
   \alpha^i_{\mathrm{bas}}
    &\longmapsfrom&
    e^i
    \\
    -
    \alpha^i_1
    &\longmapsfrom&
    \dir{1}{\mathrm{s}}e^i
    \\
    -
    \alpha^i_2
    &\longmapsfrom&
    \dir{2}{\mathrm{s}}e^i
    \\
    -
    \alpha^i_{2\,1}
    &\longmapsfrom&
    \dir{2}{\mathrm{s}}
    \dir{1}{\mathrm{s}}e^i.
  \end{tikzcd}
\end{equation}

\end{proposition}

\vspace{-2mm} 
\noindent 
In view of \eqref{CommutingPairOfExtensions}, this says equivalently that 
\begin{itemize}[
  leftmargin=.6cm,
  topsep=2pt,
  itemsep=2pt
]
\item while double-cyclification $\mathrm{cyc}^2$ classifies reductions along arbitrary 2-dimensional central extensions,

\item  toroidification $\mathrm{tor}^2$ classifies among these the reductions along 2-torus extensions.
\end{itemize}
\begin{proof}
  This follows from the universal property of cyclification \eqref{CyclificationHomIsomorphism} by observing that the datum needed for re-oxidation which is forgotten by factoring through $\mathrm{tor}^2 \xhookrightarrow{\;} \mathrm{cyc}^2$, namely $\dir{1}p_\ast \dir{1}{c}_1$, is exactly the datum which vanishes by the assumption that the 2-dimensional extension is a 2-torus extension, via \eqref{TrivialFiberIntegrationofTorusCocycles}.
\end{proof}

\begin{definition}[\bf Higher dimensional toroidification]
  \label{HigherDimensionalToroidification}
  We say that 
  \begin{itemize}[
    leftmargin=.8cm,
    topsep=2pt,
    itemsep=2pt
  ]
  \item[\bf (i)] {\it 1-toroidification} is the same as {\it cyclification} $\mathrm{cyc}$ from Def. \ref{Cyclification}:
  $
    \mathrm{tor}^1(\mathfrak{h})
    \;:=\;
    \mathrm{cyc}(\mathfrak{h})\;.
  $

  \item[\bf (ii)] {\it 2-toroidification} is the same as the toroidification $\mathrm{tor}^2$ from Ex. \ref{Toroidification}:
  $
    \mathrm{tor}^2(\mathfrak{h})
    \xhookrightarrow{\quad}
    \mathrm{cyc}^2(\mathfrak{h})\;.
  $

  \item[\bf (iii)] {\it $(n+1)$-toroidification} for $n \geq 2$ is the sub-$L_\infty$-algebra of the cyclification of, recursively, the $n$-toroidification obtained by discarding (via Lem. \ref{QuotientByAbelianAntiIdeal}) all generators 
  $\dir{n\hspace{-2pt}+\hspace{-2pt}1}{\shift}\;\,\dir{r}{\omega}_2$ for $r \in \{1,\cdots, n\}$:
  \begin{equation}
    \label{TorRecursion}
    \begin{tikzcd}
      \mathrm{tor}^{n+1}(\mathfrak{h})
      \,:=\,
      \mathrm{tor}
      \,
      \mathrm{tor}^n
      (\mathfrak{h})
      \ar[
        rr,
        hook
      ]
      &&
      \mathrm{cyc}\,\mathrm{tor}^n(\mathfrak{h})\;.
    \end{tikzcd}
  \end{equation}
  \end{itemize}
\end{definition}

The following statement is due to \cite[Thm. 2.6]{SV24-Mysterious}, there argued via rational homotopy theory. We give a direct proof. 
\begin{proposition}[{\bf Explicit $n$-Toroidification}]
  \label{ExplicitNToroidification}
  For $n \in \mathbb{N}$ the $n$-toroidification {\rm (Def. \ref{HigherDimensionalToroidification})} of $\mathfrak{h} \,\in\, \mathrm{sLieAlg}_\infty$, with presentation as in Def. \ref{Cyclification},
  is given by
  \vspace{0mm} 
  \begin{equation}
    \label{CEAlgebraOfNToroidification}
    \mathrm{CE}\big(
      \mathrm{tor}^n(\mathfrak{h})
    \big)
    \;\;
    \simeq
    \;\;
    \FDGCA
    \bigg[
      \smash{
      \big(
        \grayunderbrace{
        \dir{r}{\omega}_2
        }{
          \mathclap{
            \adjustbox{scale=.6}{$
            \def\arraystretch{1.1}
            \begin{array}{c}
              \mathrm{deg}
              \,=
              \\
              (2,\mathrm{evn})
            \end{array}
          $}
          }
        }
      \big)_{r =1}^n
      }
      ,\;
      \smash{
      \big(
        \;
        \grayunderbrace{
        \dir{i_r}{\mathrm{s}}
        \cdots
        \dir{i_2}{\mathrm{s}}
        \dir{\hspace{5pt}i_1}{\mathrm{s}}
        e^i
        }{
          \mathclap{
          \adjustbox{
            scale=.6
          }{$
          \def\arraystretch{.9}
          \def\arraycolsep{0pt}
          \begin{array}{c}
            \mathrm{deg}
            \,=\,
            \\
            \mathrm{deg}(e^i)
            -
            (r, \mathrm{evn})
           \end{array}
         $}
         }
        }
        \;
      \big)_{
        \adjustbox{
          scale=.6
        }{$
        \def\arraystretch{1.1}
        \def\arraycolsep{0pt}
        \begin{array}{l}
          i \in I,
          \;\;\;\;
          0 \leq r \leq n,
          \\
          n \geq 
            i_r > \cdots > i_2 > i_1
        \geq 1
        \end{array}
      $}
      }
      }
      \;
    \bigg]
    \Big/
    \left(
    \def\arraystretch{1.3}
    \def\arraycolsep{2pt}
    \begin{array}{l}
      \mathrm{d}\,
      \dir{r}{\omega}_2
      \;=\;
      0
      \\
      \mathrm{d}
      \,
      e^i
      \;=\;
      \mathrm{d}_{\mathfrak{h}}
      e^i
      \,+\,
      \sum_{r=1}^n
      \,
      \dir{r}{\omega}_2
      \, 
      \dir{r}{\mathrm{s}}
      e^i
      \\
      \mathrm{d}\circ 
      \dir{r}{\mathrm{s}}
      \,=\,
      -
      \dir{r}{\mathrm{s}}
      \circ
      \mathrm{d}
      ,\;\;
      \dir{r}{\mathrm{s}}
        \circ 
      \dir{r'}{\mathrm{s}}
      \,=\,
      -
      \dir{r'}{\mathrm{s}}
      \circ
      \dir{r}{\mathrm{s}}
    \end{array}
    \right)
    .
  \end{equation}
\end{proposition}
\noindent
Here on the right of \eqref{CEAlgebraOfNToroidification} we mean that the differential $\mathrm{d}$ is extended to the shifted generators by the rule that it graded-commutes with all the shift operators $\dir{r}{\mathrm{s}}$, which in turn are regarded as uniquely extended to graded derivations of degree $\mathrm{deg} = -1$, anti-commuting among each other --- in evident generalization of \eqref{ShiftAnticommutesWithDifferential}.

\begin{proof}
  This follows readily by induction on $r$:
  For $r = 0$ the statement is trivial, for $r =1$ the statement is that of
  Lem. \ref{DifferentialAndShiftInCyclification}.
  Hence for the induction step, assume that the statement is true for some $n \geq 1$.
  Then we find the CE-differential on $\mathrm{tor}\, \big(\mathrm{tor}^{n}(\mathfrak{h})\big)$ to be given on generators as claimed:
\vspace{-1mm} 
$$
  \def\arraystretch{1.5}
  \begin{array}{ccll}
    \mathrm{d}_{\mathrm{tor}\, \mathrm{tor}^n}
    \,
    \dir{i_r}{\mathrm{s}}
      \cdots 
    \dir{i_1}{\mathrm{s}}
    e^i
    &=&
    \mathrm{d}_{\mathrm{tor}^n}
    \,
    \dir{i_r}{\mathrm{s}}
      \cdots 
    \dir{i_1}{\mathrm{d}}
    e^i
    \;+\;
    \dir{n\!+\!\!1}{\omega}_2
    \;\;
    \dir{n\!+\!\!1}{\mathrm{s}}
    \;\;
    \dir{i_r}{\mathrm{s}}
      \cdots 
    \dir{i_1}{\mathrm{s}}
    e^i
    &
    \proofstep{
      by 
      \eqref{TorRecursion}
      \&
      \eqref{CEOfToroidification}
    }
    \\
    &=&
    (-1)^r
    \;
    \dir{i_r}{\mathrm{s}}
    \cdots
    \dir{i_1}{\mathrm{s}}
    \big(
      \mathrm{d}_{\mathfrak{h}}
      e^i
      \,+\,
      \sum_{r'=1}^n
      \,
      \dir{r'}{\omega}_2
      \,
      \dir{r'}{\mathrm{s}}e^i
    \big)
    \;+\;
    \dir{n\!+\!\!1}{\omega}_2
    \;\;
    \dir{n\!+\!\!1}{\mathrm{s}}
    \;\;
    \dir{i_r}{\mathrm{s}}
      \cdots 
    \dir{i_1}{\mathrm{s}}
    e^i
    &
    \proofstep{
      by indctn assmpt \& 
      \eqref{CEAlgebraOfNToroidification}
    }
    \\
    &=&
    (-1)^r
    \;
    \dir{i_r}{\mathrm{s}}
    \cdots
    \dir{i_1}{\mathrm{s}}
    \big(
      \mathrm{d}_{\mathfrak{h}}
      e^i
      \,+\,
      \sum_{r'=1}^{n \mathcolor{purple}{+ 1}}
      \,
      \dir{r'}{\omega}_2
      \,
      \dir{r'}{\mathrm{s}}e^i
    \big)
    \\
    &=&
    \mathrm{d}_{\mathrm{tor}^{n\mathcolor{purple}{+1}}}
    \,
    \dir{i_r}{\mathrm{s}}
      \cdots 
    \dir{i_1}{\mathrm{s}}
    e^i
    &
    \proofstep{
      according to 
      \eqref{CEAlgebraOfNToroidification}
    }
  \end{array}
$$
and trivially:
$$
  \def\arraystretch{1}
  \begin{array}{ccll}
    \mathrm{d}_{\mathrm{tor}\, \mathrm{tor}^n}
    \;
    \dir{n+1}{\omega}_2
    &=&
    0
    &
    \proofstep{
      by 
      \eqref{CEOfToroidification}
    }
    \\
    &=&
    \mathrm{d}_{\mathrm{tor}^{n\mathcolor{purple}{+1}}}
    \;
    \dir{n+1}{\omega}_2
    &
    \proofstep{
      by 
      \eqref{CEAlgebraOfNToroidification}
      \,.
    }
  \end{array}
$$

\vspace{-.45cm}
\end{proof}

\begin{example}[\bf 3-Toroidification]
\label{3Toroidification}
The 3-toroidification $\mathrm{tor}^3(\mathfrak{h})$ (Def. \ref{HigherDimensionalToroidification})
of a super-$L_\infty$-algebra with generators as in Def. \ref{Cyclification}
is given by:
$$
  \mathrm{CE}\big(
    \mathrm{tor}^3
    (\mathfrak{h})
  \big)
  \;\;
  \simeq
  \;\;
  \FDGCA
  \left[\!
  \def\arraystretch{1.3}
  \def\arraycolsep{2pt}
  \begin{array}{c}
    \dir{3}{\omega}_2
    \\
    \dir{2}{\omega}_2
    \\
    \dir{1}{\omega}_2
    \\
    (e^i)_{i \in I}
    \\
    (\dir{1}{\shift}e^i)_{i \in I})
    \\
    (\dir{2}{\shift}e^i)_{i \in I}
    \\
    (\dir{3}{\shift}e^i)_{i \in I}
    \\
    (\dir{2}{\shift}\dir{1}{\shift}e^i)_{i \in I}
    \\
    (\dir{3}{\shift}\dir{1}{\shift}e^i)_{i \in I})
    \\
    (\dir{3}{\shift}\dir{2}{\shift}e^i)_{i \in I}   
    \\
    (\dir{3}{\shift}\dir{2}{\shift}\dir{1}{\shift}e^i)_{i \in I}
  \end{array}
  \!\right]
  \Big/
  \left(\!
  \def\arraystretch{1.4}
  \def\arraycolsep{2pt}
  \begin{array}{ccc}
    \mathrm{d}
    \,
    \dir{3}{\omega}_2
    &=&
    0
    \\
    \mathrm{d}
    \,
    \dir{2}{\omega}_2
    &=&
    0
    \\
    \mathrm{d}
    \,
    \dir{1}{\omega}_2
    &=&
    0
    \\
    \mathrm{d}\,
    e^i
    &=&
    \mathrm{d}_{\mathfrak{h}}
    e^i
    +
    \dir{1}{\omega}_2
    \,
    \dir{1}{\mathrm{s}}e^i
    +
    \dir{2}{\omega}_2
    \,
    \dir{2}{\mathrm{s}}e^i
    +
    \dir{3}{\omega}_2
    \,
    \dir{3}{\mathrm{s}}e^i
    \\
    \mathrm{d}\,
    \dir{1}{\mathrm{s}}
    e^i
    &=&
    -
    \dir{1}{\shift}
    \, 
    \mathrm{d}_{\mathfrak{h}}
    e^i
    +
    \dir{2}{\omega}_2
    \,
    \dir{2}{\shift}
    \dir{1}{\shift}
    e^i
    +
    \dir{3}{\omega}_2
    \,
    \dir{3}{\shift}
    \dir{1}{\shift}
    e^i
    \\
    \mathrm{d}
    \,
    \dir{2}{\shift}e^i
    &=&
    -
    \dir{2}{\shift}
    \, 
    \mathrm{d}_{\mathfrak{h}}
    \,
    e^i
    -
    \dir{1}{\omega}_2
    \,
    \dir{2}{\shift}
    \dir{1}{\shift}
    e^i
    +
    \dir{3}{\omega}_2
    \,
    \dir{3}{\shift}
    \dir{2}{\shift}
    e^i
    \\
    \mathrm{d}\, 
    \dir{3}{\shift}e^i
    &=&
    -
    \dir{3}{\shift}
    \,
    \mathrm{d}_{\mathfrak{h}}
    e^i
    -
    \dir{1}{\omega}_2
    \,
    \dir{3}{\shift}
    \dir{1}{\shift}
    e^i
    -
    \dir{2}{\omega}_2
    \,
    \dir{3}{\shift}
    \dir{2}{\shift}e^i
    \\
    \mathrm{d}\,
    \dir{2}{\shift}
    \dir{1}{\shift}
    e^i
    &=&
    \dir{2}{\shift}
    \dir{1}{\shift}
    \,
    \mathrm{d}_{\mathfrak{h}}e^i
    +
    \dir{3}{\omega}_2
    \,
    \dir{3}{\shift}
    \dir{2}{\shift}
    \dir{1}{\shift}e^i
    \\
    \mathrm{d}\,
    \dir{3}{\shift}
    \dir{1}{\shift}
    e^i
    &=&
      \dir{3}{\shift}
      \dir{1}{\shift}
      \, 
      \mathrm{d}_{\mathfrak{h}}
      e^i
      -
      \dir{2}{\omega}_2
      \,
      \dir{3}{\shift}
      \dir{2}{\shift}
      \dir{1}{\shift}
      e^i
      \\
      \mathrm{d}
      \,
      \dir{3}{\shift}
      \dir{2}{\shift}
      e^i
      &=&
        \dir{3}{\shift}
        \dir{2}{\shift}\,
        \mathrm{d}_{\mathfrak{h}}
        e^i
        +
        \dir{1}{\omega}\,
        \dir{3}{\shift}
        \dir{2}{\shift}
        \dir{1}{\shift}
        e^i
    \\
    \mathrm{d}\,
    \dir{3}{\shift}
    \dir{2}{\shift}
    \dir{1}{\shift} e^i
    &=&
     -\dir{3}{\shift}
      \dir{2}{\shift}
      \dir{1}{\shift}
      \,
      \mathrm{d}_{\mathfrak{h}}
      e^i
    \end{array}
  \!\right).
$$
\end{example}

In generalization of Prop. \ref{UniversalPropertyOfToroidification}, we obtain the universal property of $n$-toroidification (cf.  {\cite[Thm. 2.13]{SV24-Mysterious}}).
\begin{proposition}[{\bf Universal property of $n$-toroidification}]\label{UniversalPropertyOfnToroidification}
The $n$-toroidification operation (Def. \ref{HigherDimensionalToroidification}) provides the $L_\infty$-coefficients for (double-dimensionally) reducing $L_\infty$-cocycles along $n$-toroidal extensions (Def. \ref{CentralTorusExtension}):

\vspace{-.5cm}
$$
  \mathfrak{g} \,\in\, \mathrm{sLieAlg}_\infty
  \,,
  \hspace{.4cm}
  \dir{1}{c}_1
  ,\,
  \cdots
  \dir{n}{c}_1
  \;\in\;
  \mathrm{CE}(\mathfrak{g})
  \,,
  \;\;\;
  \mathrm{d}\, 
  \dir{k}{c}_1
  \,=\,
  0
  \hspace{1cm}
  \begin{tikzcd}[row sep=small]
    \dir{1\cdots n}{\widehat{\mathfrak{g}}}
    \ar[
      d
    ]
    \\
    \mathfrak{g}
    \ar[
      r,
      "{
        \dir{1\cdots n}{c}_1
      }"
    ]
    &
    b\mathbb{R}^n
  \end{tikzcd}
$$
in that we have bijections of sets of $L_\infty$-homomorphisms, in generalization of \eqref{CyclificationHomIsomorphism} and \eqref{2ToroidificationHomBijection}:
\begin{equation}\label{nToroidificationHomIsomorphism}
  \Big\{\,
    \begin{tikzcd}
      \dir{1\cdots n}{\widehat{g}}
      \ar[
        rr,
        "{ f }"
      ]
      &&
      \mathfrak{h}
    \end{tikzcd}
  \!\!\Big\}
  \begin{tikzcd}[column sep=huge]
    \ar[
      rr,
      shift left=4pt,
      "{
          \mathclap{
            \adjustbox{
              scale=.7,
              raise=0pt
            }{
              \color{darkgreen}
              \bf
              reduction
            }
          }
        }{\qquad 
        \mathrm{rdc}_{
          \;\;
          \dir{1\cdots n}{c}_1
        }
      }",
      "{ \sim }"{swap, yshift=-2pt}
    ]
    \ar[
      rr,
      shift right=4pt,
      <-,
      "{
          \mathclap{
            \scalebox{.7}{
              \color{darkgreen}
              \bf
              oxidation
            }
          }
        }{\qquad 
        \mathrm{oxd}_{
          \;\;
          \dir{1\cdots n}{c}_1
        }
      }"{swap}
    ]
    &&
    {}
  \end{tikzcd}
  \Bigg\{\!\!
    \begin{tikzcd}[
      row sep=-3pt, column sep=large
    ]
      \mathfrak{g}
      \ar[
        rr,
        "{
          \widetilde f
        }"
      ]
      \ar[
        dr,
        "{
         \dir{1\cdots n}{c}_1
        }"{swap}
      ]
      &&
      \mathrm{tor}^n(\mathfrak{h})
      \ar[
        dl,
        "{
          \dir{1\cdots n}{\omega}_2
        }"
      ]
      \\
      &
      b \mathbb{R}^n
    \end{tikzcd}
 \!\!\! \Bigg\}.
\end{equation}
  given by
\begin{equation}
  \label{nToroidificationHomBijection}
  \hspace{-2.6cm} 
  \begin{tikzcd}[sep=0pt]
    \dir{1\cdots n}{\widehat{\mathfrak{g}}}
    \ar[
      rr,
      "{ f }"
    ]
    &&
    \mathfrak{h}
    \\
    \alpha^i_{\mathrm{bas}}
    + \!\!
    \sum\limits_{\substack{1\leq r \leq n \\ 1 \leq i_1< \cdots < i_r \leq n}} \!\!\dir{i_r}{e}\cdots \dir{i_1}{e} \cdot \alpha^i_{i_r \cdots i_1} 
    &\longmapsfrom&
    e^i
  \end{tikzcd}
  \hspace{1cm}
  \leftrightsquigarrow
  \hspace{1cm}
  \begin{tikzcd}[
    row sep=-1pt, 
    column sep=0pt]
    \mathfrak{g}
    \ar[
      rr,
      "{ 
        \widetilde{f} 
      }"
    ]
    &&
    \mathrm{tor}^n(\mathfrak{h})
    \\
    \alpha^i_{\mathrm{bas}}
    &\longmapsfrom&
    e^i
    \\
    (-1)^{r(r+1)/2} \cdot \alpha^i_{i_r \cdots i_1}
    &\longmapsfrom&
    \dir{i_r}{\mathrm{s}}\cdots \dir{i_1}{\mathrm{s}}e^i
    \\
    \dir{r}{c}_1
    &\longmapsfrom&
    \dir{r}{\omega}_2
    \mathrlap{\,.}
  \end{tikzcd}
\end{equation}
Moreover, along the lines of Rem. \ref{DecomposingNTorusExtensions}, these reductions/oxidations may be applied in any sequence of steps:
$$
  \left(
  \hspace{-4pt}
  \adjustbox{raise=4pt}{
  \begin{tikzcd}[
    row sep=-2pt,
    column sep=-10pt
  ]
    \widehat{\widehat{\widehat{\widehat{\mathfrak{g}}}}}
    \ar[dd, ->>]
    \ar[rr]
    &&
    \mathfrak{h}
    \\
    &
    \phantom{\mathrm{cyc}^n(b\mathbb{R})}
    \\
    \widehat{\widehat{\widehat{\mathfrak{g}}}}
    \ar[dd, ->>]
    \\
    &
    \phantom{\mathrm{cyc}^n(b\mathbb{R})}
    \\
    \widehat{\widehat{\mathfrak{g}}}
    \ar[dd, ->>]
    \\
    &
    \phantom{\mathrm{cyc}^n(b\mathbb{R})}
    \\
    \widehat{\mathfrak{g}}
    \ar[dd, ->>]
    \\
    &
    \phantom{\mathrm{cyc}^n(b\mathbb{R})}
    \\
    \mathfrak{g}
  \end{tikzcd}
  }
  \hspace{-5pt}
  \right)
  \quad 
  \leftrightsquigarrow
  \quad 
  \left(
  \hspace{-4pt}
  \adjustbox{raise=4pt}{
  \begin{tikzcd}[
    row sep=-2pt,
    column sep=-10pt
  ]
    \widehat{\widehat{\widehat{\widehat{\mathfrak{g}}}}}
    \ar[dd, ->>]
    \\
    &
    \phantom{\mathrm{cyc}^n(b\mathbb{R})}
    \\
    \widehat{\widehat{\widehat{\mathfrak{g}}}}
    \ar[dd, ->>]
    \ar[rr]
    \ar[dr, shorten=-3pt]
    &&
    \mathrm{tor}(\mathfrak{h})
    \ar[dl, shorten =-3pt]
    \\
    &
    b\mathbb{R} 
    \\
    \widehat{\widehat{\mathfrak{g}}}
    \ar[dd, ->>]
    \\
    &
    \phantom{\mathrm{cyc}^n(b\mathbb{R})}
    \\
    \widehat{\mathfrak{g}}
    \ar[dd, ->>]
    \\
    &
    \phantom{\mathrm{cyc}^n(b\mathbb{R})}
    \\
    \mathfrak{g}
  \end{tikzcd}
  }
  \hspace{-9pt}
  \right)
  \quad 
  \leftrightsquigarrow
  \quad 
  \left(
  \hspace{-4pt}
  \adjustbox{raise=4pt}{
  \begin{tikzcd}[
    row sep=-2pt,
    column sep=-10pt
  ]
    \widehat{\widehat{\widehat{\widehat{\mathfrak{g}}}}}
    \ar[dd, ->>]
    \\
    &
    \phantom{\mathrm{cyc}^n(b\mathbb{R})}
    \\
    \widehat{\widehat{\widehat{\mathfrak{g}}}}
    \ar[dd, ->>]
    \\
    &
    \phantom{\mathrm{cyc}^n(b\mathbb{R})}
    \\
    \widehat{\widehat{\mathfrak{g}}}
    \ar[dd, ->>]
    \ar[rr]
    \ar[dr, shorten=-3pt]
    &&
    \mathrm{tor}^2(\mathfrak{h})
    \ar[dl, shorten =-3pt]
    \\
    &
    b\mathbb{R}^2 
    \\
    \widehat{\mathfrak{g}}
    \ar[dd, ->>]
    \\
    &
    \phantom{\mathrm{cyc}^n(b\mathbb{R})}
    \\
    \mathfrak{g}
  \end{tikzcd}
  }
  \hspace{-9pt}
  \right)
  \quad 
  \leftrightsquigarrow
  \quad 
  \left(
  \hspace{-4pt}
  \adjustbox{raise=4pt}{
  \begin{tikzcd}[
    row sep=-2pt,
    column sep=-10pt
  ]
    \widehat{\widehat{\widehat{\widehat{\mathfrak{g}}}}}
    \ar[dd, ->>]
    \\
    &
    \phantom{\mathrm{cyc}^n(b\mathbb{R})}
    \\
    \widehat{\widehat{\widehat{\mathfrak{g}}}}
    \ar[dd, ->>]
    \\
    &
    \phantom{\mathrm{cyc}^n(b\mathbb{R})}
    \\
    \widehat{\widehat{\mathfrak{g}}}
    \ar[dd, ->>]
    \\
    &
    \phantom{\mathrm{cyc}^n(b\mathbb{R})}
    \\
    \widehat{\mathfrak{g}}
    \ar[dd, ->>]
    \ar[rr]
    \ar[dr, shorten=-3pt]
    &&
    \mathrm{tor}^3(\mathfrak{h})
    \ar[dl, shorten =-3pt]
    \\
    &
    b\mathbb{R}^3 
    \\
    \mathfrak{g}
  \end{tikzcd}
  }
  \hspace{-9pt}
  \right).
$$
\end{proposition}
\begin{proof}
This readily follows by induction on $n\in \NN$. The base cases of $n=2$ and $n=1$ hold by Prop. \ref{UniversalPropertyOfToroidification} and Prop. \ref{TheExtCycAdjunction}, respectively. Now assume that the case of a fixed $n\in \NN$ holds via  \eqref{nToroidificationHomBijection}, and consider any morphism of $L_\infty$ algebras on the further 1-toroidal extension  ($\equiv$ cyclification) along a basic 2-cocycle $\, \, \dir{n+1}{c}_1\in \CE(\frg)$
  \begin{equation*}
  \begin{tikzcd}[row sep=-2pt, column sep=large]
    \dir{1\cdots (n+1)}{\widehat{\mathfrak{g}}}
    \ar[
      rr,
      "{ f }"
    ]
    &&
    \mathfrak{h}
    \\
    \alpha^i_{\mathrm{bas}}
    +
    \sum\limits_{\substack{1\leq r \leq n+1 \\ 1 \leq i_1< \cdots < i_r \leq n+1}} \dir{i_r}{e}\cdots \dir{i_1}{e} \cdot \alpha^i_{i_r \cdots i_1} 
    &\longmapsfrom&
    e^i \, .
  \end{tikzcd}
  \end{equation*}
The sum in the image of the generators may be trivially expanded as
 \begin{equation*}
  \begin{tikzcd}[sep=0pt]
    \alpha^i_{\mathrm{bas}}
    +
    \sum\limits_{\substack{1\leq r \leq n \\ 1 \leq i_1< \cdots < i_r \leq n}} \dir{i_r}{e}\cdots \dir{i_1}{e} \cdot \alpha^i_{i_r \cdots i_1}  +
    \sum\limits_{\substack{1\leq r' \leq n \\ 1 \leq i_1< \cdots < i_{r'} \leq n}} \dir{n+1}{e}\cdot \dir{i_{r'}}{e}\cdots \dir{i_1}{e} \cdot \alpha^i_{i_{(n+1)} i_{r'} \cdots i_1} 
    &\;\; \longmapsfrom& \;\;
    e^i \, ,
  \end{tikzcd}
  \end{equation*}
whereby reducing \eqref{CyclificationHomBijection} along $\, \, \dir{n+1}{c}_1$ yields the cocycle 
\begin{equation*}
  \begin{tikzcd}[row sep=-2pt, column sep=large]
    \dir{1\cdots n}{\widehat{\mathfrak{g}}}
    \ar[
      rr,
      "{ f }"
    ]
    &&
    \mathrm{cyc}(\mathfrak{h})
    \\
    \alpha^i_{\mathrm{bas}}
    +
    \sum\limits_{\substack{1\leq r \leq n \\ 1 \leq i_1< \cdots < i_r \leq n}} \dir{i_r}{e}\cdots \dir{i_1}{e} \cdot \alpha^i_{i_r \cdots i_1} 
    &\longmapsfrom&
    e^i  
    \\
    -
    \sum\limits_{\substack{1\leq r \leq n \\ 1 \leq i_1< \cdots < i_{r} \leq n}} \dir{i_{r}}{e}\cdots \dir{i_1}{e} \cdot \alpha^i_{i_{(n+1)}i_r \cdots i_1} 
 &\; \longmapsfrom& \; \dir{n+1}{s} \, e^i 
  \end{tikzcd}
  \end{equation*}
Applying the induction assumption and hence further $n$-toroidally reducing via \eqref{nToroidificationHomBijection} yields 
   \begin{equation*}
  \begin{tikzcd}[row sep=-3pt, column sep=large]
\mathfrak{g}
    \ar[
      rr,
      "{ f }"
    ]
    &&
    \mathrm{tor}^{n+1}(\mathfrak{h})
    \\
    \alpha^i_{\mathrm{bas}}
    &\longmapsfrom&
    e^i  
      \\
   (-1)^{r(r+1)/2} \cdot \alpha^i_{i_r \cdots i_1}
    &\longmapsfrom&
    \dir{i_r}{\mathrm{s}}\cdots \dir{i_1}{\mathrm{s}}e^i
    \\
    - (-1)^{r}   (-1)^{r(r+1)/2}
    \alpha^i_{i_{(n+1)} i_r \cdots i_1} 
 &\longmapsfrom& \dir{n+1}{s} \,\, \, \dir{i_r}{\mathrm{s}}\cdots \dir{i_1}{\mathrm{s}}e^i
    \\
    \dir{n+1}{c}_1 &\longmapsfrom& \dir{n+1}{\omega}_2 
        \\
    \dir{r}{c}_1 &\longmapsfrom& \dir{r}{\omega}_2 \, ,
  \end{tikzcd}
  \end{equation*}
where in the third line the extra $(-1)^r$ sign arises by (anti-)commuting the shift $s^{n+1}$ passed the rest of the $r$-shifts. This completes the proof since 
$$
-(-1)^r (-1)^{r(r+1)/2} = (-1)^{(r+1)(r+1+1)/2} \, . 
$$

\vspace{-4mm} 
\end{proof}
\medskip

\noindent
{\bf Toroidification of twisted K-theory spectra.} We turn to the main example of interest here, the toroidifications of twisted K-theory spectra carrying archetypes of, as we will see in \S\ref{LiftingToMTheory}, toroidal T-duality of flux densities.
Or  rather, the T-duality must be carried by a sub-algebra of the toroidification, as brought out by the following variant of 
Ex. \ref{TDualityClassifyingAlgebra}.

\begin{example}[\bf Geometric 2-Toroidification of bundle gerbe classifying space]
\label{ToroidalTDualityClassifyingAlgebra}
It is readily checked that the 2-cyclification $\mathrm{cyc}^2 b^2 \mathbb{R}$ of the line Lie 3-algebra \eqref{TheLineLieNAlgebra}
no longer enjoys an automorphism of the kind \eqref{AutomorphismOfCycOfB2U1}
that the 1-cyclification did, namely by swapping $\dir{i}{\mathrm{s}}h_3 \leftrightsquigarrow \dir{i}{\omega}_2$ \textit{for both} $i=1,2$, as would befit a toroidal T-duality classifying algebra. In fact, nor does its toroidified subalgebra (Def. \ref{Toroidification}),
$$
  \begin{tikzcd}
  \mathrm{tor}^2(b^2 \mathbb{R})
  \ar[r, hook]
  &
  \mathrm{cyc}^2(b^2\mathbb{R})
  \,,
  \end{tikzcd}
$$ 
given by
\begin{equation}
  \label{CEOfToroidificationOfb2R}
  \mathrm{CE}\big(
    \mathrm{tor}^2 b^2\mathbb{R}
  \big)
  \;\simeq\;
  \FDGCA
  \left[\!
  \def\arraystretch{1.3}
  \begin{array}{c}
    \dir{2}{\omega}_2,
    \\
    \dir{1}{\omega}_2,
    \\
    h_3
    \\
    \dir{2}{\mathrm{s}}h_3
    \\
    \dir{1}{\mathrm{s}}h_3
    \\
    \dir{2}{\mathrm{s}}\dir{1}{\mathrm{s}}h_3
  \end{array}
  \!\right]
  \Big/\ \!
  \left(\!
  \def\arraystretch{1.3}
  \begin{array}{ccc}
    \mathrm{d}
    \,
    \dir{2}{\omega}_2
    &=&
    0
    \\
    \mathrm{d}
    \,
    \dir{1}{\omega}_2
    &=&
    0
    \\
    \mathrm{d}
    \,
    h_3
    &=&
    \dir{1}{\omega}_2 
    \,
    \dir{1}{\mathrm{s}}h_3
    +
    \dir{2}{\omega}_2
    \,
    \dir{2}{\mathrm{s}} h_3
    \\
    \mathrm{d}
    \,
    \dir{2}{\mathrm{s}}h_3
    &=&
       -
    \dir{1}{\omega}_2
    \,
    \dir{2}{\mathrm{s}}
    \dir{1}{\mathrm{s}}
    h_3
    \\
    \mathrm{d}
    \,
    \dir{1}{\mathrm{s}}h_3
    &=&
    + 
    \dir{2}{\omega}_2\, 
    \dir{2}{\mathrm{s}}
    \dir{1}{\mathrm{s}}
    h_3
    \\
    \mathrm{d}
    \,
    \dir{2}{\mathrm{s}}
    \dir{1}{\mathrm{s}}h_3
    &=&
    0
  \end{array}
  \!\!\right)
  ,
\end{equation}
via which one immediately identifies the obstruction for such a swapping automorphism to exist to be the double shift of the generator $h_3$
$$
    \dir{2}{\mathrm{s}}
    \dir{1}{\mathrm{s}}
    h_3 
    \;\;
    \in
    \;\;
    \mathrm{CE}\big(
      \mathrm{tor}^2
      \,b^2
      \,\mathbb{R}
    \big)
    \,.
$$ 
However, since this generator is closed,
this implies at once that there exists (via Lem. \ref{QuotientByAbelianAntiIdeal}) the $L_\infty$-subalgebra obtained by discarding this term 
$$
  \begin{tikzcd}
  \mathrm{tor}^{2\mathcolor{purple}{'}}(b^2 \mathbb{R})
  \ar[r, hook]&\mathrm{tor}^2(b^2 \mathbb{R})
  \ar[r, hook]
  &
  \mathrm{cyc}^2(b^2\mathbb{R})
  \,,
  \end{tikzcd}
$$
given by
\begin{equation}
  \label{CEOfReducedToroidificationOfb2R}
  \mathrm{CE}\big(
    \mathrm{tor}^{2'} b^2\mathbb{R}
  \big)
  \;\simeq\;
  \FDGCA
  \left[\!
  \def\arraystretch{1.3}
  \begin{array}{c}
    \dir{2}{\omega}_2,
    \\
    \dir{1}{\omega}_2,
    \\
    h_3
    \\
    \dir{2}{\mathrm{s}}h_3
    \\
    \dir{1}{\mathrm{s}}h_3
  \end{array}
  \!\right]
  \Big/\ \!
  \left(\!
  \def\arraystretch{1.3}
  \begin{array}{ccc}
    \mathrm{d}
    \,
    \dir{2}{\omega}_2
    &=&
    0
    \\
    \mathrm{d}
    \,
    \dir{1}{\omega}_2
    &=&
    0
    \\
    \mathrm{d}
    \,
    h_3
    &=&
    \dir{1}{\omega}_2 
    \,
    \dir{1}{\mathrm{s}}h_3
    +
    \dir{2}{\omega}_2
    \,
    \dir{2}{\mathrm{s}} h_3
    \\
    \mathrm{d}
    \,
    \dir{2}{\mathrm{s}}h_3
    &=& 0
    \\
    \mathrm{d}
    \,
    \dir{1}{\mathrm{s}}h_3
    &=&
   0
  \end{array}
  \!\!\right)
  ,
\end{equation}
being equivalently the higher central extension (Def. \ref{HigherCentralExtension}) of $b\mathbb{R}^2\times b\mathbb{R}^2$ by its canonical 4-cocycle
\begin{equation}
  \label{CEOfBOf2ToroidalTDuality2Group}
  \begin{tikzcd}[
    column sep=30pt
  ]
    b\mathcal{T}^2\, := \, \mathrm{tor}^{2'} b^2\mathbb{R}
    \ar[
      rrr,
      "{
        \mathrm{hofib}\big(
          \dir{1}{\omega}_2 
        \,
        \dir{1}{\widetilde \omega}_2 + \dir{2}{\omega}_2 
        \,
        \dir{2}{\widetilde \omega}_2
        \big)
      }"
    ]
    &&&
    b\mathbb{R}^2 \times b\mathbb{R}^2
    \ar[
      rr,
      "{
        \dir{1}{\omega}_2 
        \,
        \dir{1}{\widetilde \omega}_2 + \dir{2}{\omega}_2 
        \,
        \dir{2}{\widetilde \omega}_2
      }"
    ]
    &&
    b^3 \mathbb{R} , 
  \end{tikzcd}
\end{equation}
does have an automorphism symmetry given by swapping the two degree=2 generators as:
\begin{equation}
  \label{AutomorphismOfRedTorOfB2U1}
  \begin{tikzcd}[row sep=-3pt,
    column sep=0pt
  ]
   b\mathcal{T}^2
    \ar[
      rr,
      <->,
      "{ \sim }"
    ]
    &&
    b\mathcal{T}^2
    \\
    \dir{1}{\mathrm{s}}h_3
    &\longmapsfrom&
    \dir{1}{\omega}_2
    \\
    \dir{1}{\omega}_2
    &\longmapsfrom&
    \dir{1}{\mathrm{s}}h_3
       \\
    \dir{2}{\mathrm{s}}h_3
    &\longmapsfrom&
    \dir{2}{\omega}_2
    \\
    \dir{2}{\omega}_2
    &\longmapsfrom&
   \dir{2}{\mathrm{s}}h_3
  \end{tikzcd}
\end{equation}
\end{example}
\begin{remark}[\bf Geometric T-duality]
  \label{GeometricTDuality}
  {\bf (i)} 
  When it comes to toroidal T-duality in \S\ref{LiftingToMTheory},
  the vanishing of second contractions $\dir{2}{\mathrm{s}}\dir{1}{\mathrm{s}}h_3$ 
  of the 3-flux $H_3$ --- that is reflected by the restricted toroidifications \eqref{CEOfReducedToroidificationOfb2R} and  \eqref{Restricted2ToroidificationOfTwistedK}, 
  and the vanishing of higher contractions that is further reflected below in \eqref{RestrictedNToroidificationOfTwistedK} --- 
  is known in the literature as that corresponding to ``geometric-'' or ``$F^2$-'' T-duality backgrounds (e.g. \cite[p. 6]{KimSaemann22}).
  
 \noindent {\bf (ii)} 
 In our context of super-space T-duality, this is the case realized by the fixed form of the avatar super-flux densities $H_3^A$ \eqref{IIACocycles} and $H_3^B$ \eqref{IIBCocycles}, which turn out to have vanishing higher order contraction with bosonic vector fields. Aspects of T-duality beyond this ``geometric'' case have been discussed (as ``non-geometric-backgrounds'' such as ``T-folds'' or yet more exotic structures) but their possible relation to actual supergravity may not have found attention.
\end{remark}

Now the K-theoretic enhancement of Ex. \ref{ToroidalTDualityClassifyingAlgebra} and thus the 2-dimensional analog of Ex \ref{CyclificationOfTwistedKTTheorySpectra} is:
\begin{example}[\bf 2-Toroidification of twisted K-theory]\label{2ToroidificationOfTwistedKtheory}
The toroidification (Def. \ref{Toroidification}) of the Whitehead $L_\infty$-algebras of twisted K-theory spectra (Ex. \ref{WhiteheadLInfinityOfTwistedKTheorySpectrum}) is given by
\begin{equation}
  \label{2ToroidificationOfTwistedK}
  \def\arraystretch{1.5}
  \begin{array}{l}
  \mathrm{CE}\Big(
    \mathrm{tor}^2
    \,
    \mathfrak{l}
    \big(
      \Sigma^m
      \mathrm{KU}
      \!\sslash\!
      B\mathrm{U}(1)
    \big)
 \! \Big)
  \\
  \;\simeq\;
  \FDGCA
  \left[
  \def\arraystretch{1.3}
  \begin{array}{c}
    \dir{1}{\omega}_2
    \\
    \dir{2}{\omega}_2
    \\
    h_3
    \\
    \dir{1}{\mathrm{s}}h_3
    \\
    \dir{2}{\mathrm{s}}h_3
    \\
    \dir{2}{\mathrm{s}}
    \dir{1}{\mathrm{s}}h_3
    \\
    f_{2\bullet+m}
    \\
    \dir{1}{\mathrm{s}}\!f_{2\bullet+m}
    \\
    \dir{2}{\mathrm{s}}\!f_{2\bullet+m}
    \\
    \dir{2}{\mathrm{s}}
    \dir{1}{\mathrm{s}}\!f_{2\bullet+m}
  \end{array}
  \right]
  \Big/
  \left(
  \def\arraystretch{1.3}
  \def\arraycolsep{1pt}
  \begin{array}{ccl}
    \mathrm{d}\,
    \dir{1}{\omega}_2
    &=&
    0
    \\
    \mathrm{d}\,
    \dir{2}{\omega}_2
    &=&
    0
    \\
    \mathrm{d}\,
    h_3
    &=&
    0
    \\
    \mathrm{d}
    \,
    \dir{1}{\mathrm{s}}
    h_3
    &=&
    +
    \dir{2}{\omega}_2
    \,
    \dir{2}{\mathrm{s}}
    \dir{1}{\mathrm{s}}
    h_3
    \\
    \mathrm{d}
    \,
    \dir{2}{\mathrm{s}}
    h_3
    &=&
    -
    \dir{1}{\omega}_2
    \,
    \dir{2}{\mathrm{s}}
    \dir{1}{\mathrm{s}}
    h_3
    \\
    \mathrm{d}\,
    \dir{2}{\mathrm{s}}
    \dir{1}{\mathrm{s}}
    h_3
    &=&
    0
    \\
    \mathrm{d}\,
    f_{2k+m}
    &=&
    h_3\,
    f_{2(k-1)+m}
    \\
    \mathrm{d}
    \,
    \dir{1}{\mathrm{s}}\!f_{2k+m}
    &=&
    -\dir{1}{\mathrm{s}}h_3
    \, 
    f_{2(k-1)+m}
    +
    h_3\, \dir{1}{\mathrm{s}}\!f_{2(k-1)+m}
    +
    \dir{2}{\omega}_2
    \,
    \dir{2}{\mathrm{s}}
    \dir{1}{\mathrm{s}}\!f_{2k+m}
    \\
    \mathrm{d}
    \,
    \dir{2}{\mathrm{s}}\!f_{2k+m}
    &=&
    -\dir{2}{\mathrm{s}}h_3
    \, 
    f_{2(k-1)+m}
    +
    h_3\, \dir{2}{\mathrm{s}}\!f_{2(k-1)+m}
    -
    \dir{1}{\omega}_2
    \,
    \dir{2}{\mathrm{s}}
    \dir{1}{\mathrm{s}}\!f_{2k+m}
    \\
    \mathrm{d}
    \,
    \dir{2}{\mathrm{s}}
    \dir{1}{\mathrm{s}}
    \!f_{2k+m}
    &=&
    \dir{2}{\mathrm{s}}
    \dir{1}{\mathrm{s}}
    h_3
    \, f_{2(k-1)+m}
    -
    \dir{2}{\mathrm{s}}h_3
    \,
    \dir{1}{\mathrm{s}}\!f_{2(k-1)+m}
    +
    \dir{1}{\mathrm{s}}h_3
    \,
    \dir{2}{\mathrm{s}}\!f_{2(k-1)+m}
    +
    h_3
    \, 
    \dir{2}{\mathrm{s}}
    \dir{1}{\mathrm{s}}
    \!
    f_{2(k-1)+m}
  \end{array}
  \right)
  .
  \end{array}
\end{equation}

Since here the generator $\dir{2}{\mathrm{s}}\dir{1}{\mathrm{s}}h_3$ is closed, discarding it (cf. Rem. \ref{GeometricTDuality}) yields (by Lem. \ref{QuotientByAbelianAntiIdeal}) the \textit{geometric} 2-toroidification  sub-$L_\infty$-algebra 
\begin{equation}
  \label{Restricted2ToroidificationOfTwistedK}
  \begin{tikzcd}
    \mathrm{tor}^{2\mathcolor{purple}{'}}
    \,
    \mathfrak{l}\big(
      \Sigma^m\mathrm{KU}
      \!\sslash\!
      B\mathrm{U}(1)
    \big)
    \ar[
      r,
      hook
    ]
    &
    \mathrm{tor}^{2}
    \,
    \mathfrak{l}\big(
      \Sigma^m\mathrm{KU}
      \!\sslash\!
      B\mathrm{U}(1)
    \big)
  \end{tikzcd}
\end{equation}
given by
\begin{equation}
  \label{SubalgebrasOf2TorOfTwistedK}
  \def\arraystretch{1.7}
  \begin{array}{l}
  \mathrm{CE}\Big(
    \mathrm{tor}^{2\mathcolor{purple}{'}}
    \,
    \mathfrak{l}\big(
      \Sigma^m\mathrm{KU}
      \!\sslash\!
      B\mathrm{U}(1)
    \big)
  \!\Big)
  \\
  \;\simeq\;
  \FDGCA
  \left[
  \def\arraystretch{1.3}
  \def\arraycolsep{2pt}
  \begin{array}{c}
    \dir{1}{\omega}_2
    \\
    \dir{2}{\omega}_2
    \\
    h_3
    \\
    \dir{1}{\mathrm{s}}h_3
    \\
    \dir{2}{\mathrm{s}}h_3
    \\
    f_{2\bullet+m}
    \\
    \dir{1}{\mathrm{s}}\!
    f_{2\bullet+m}
    \\
    \dir{2}{\mathrm{s}}\!
    f_{2\bullet+m}
    \\
    \dir{2}{\mathrm{s}}
    \dir{1}{\mathrm{s}}
    \!
    f_{2\bullet+m}
  \end{array}
  \right]
  \Big/
  \left(
  \def\arraystretch{1.3}
  \def\arraycolsep{1pt}
  \begin{array}{ccl}
    \mathrm{d}\,
    \dir{1}{\omega}_2
    &=&
    0
    \\
    \mathrm{d}\,
    \dir{2}{\omega}_2
    &=&
    0
    \\
    \mathrm{d}\,
    h_3
    &=&
    \dir{1}{\omega}_2\, 
    \dir{1}{\mathrm{s}}h_3
    +
    \dir{2}{\omega}_2\, 
    \dir{2}{\mathrm{s}}h_3
    \\
    \mathrm{d}
    \,
    \dir{1}{\mathrm{s}}
    h_3
    &=&
    0
    \\
    \mathrm{d}
    \,
    \dir{2}{\mathrm{s}}
    h_3
    &=&
    0
    \\
    \mathrm{d}\, 
    f_{2k+m}
    &=&
    +
    h_3\, f_{2(k-1)+m}
    +
    \dir{1}{\omega}_2
    \,
    \dir{1}{\mathrm{s}}\!f_{2k+m}
    +
    \dir{2}{\omega}_2
    \,
    \dir{2}{\mathrm{s}}\!f_{2k+m}
    \\
    \mathrm{d}
    \,
    \dir{1}{\mathrm{s}}\!f_{2k+m}
    &=&
    -\dir{1}{\mathrm{s}}h_3
    \, 
    f_{2(k-1)+m}
    +
    h_3\, \dir{1}{\mathrm{s}}\!f_{2(k-1)+m}
    +
    \dir{2}{\omega}_2
    \,
    \dir{2}{\mathrm{s}}
    \dir{1}{\mathrm{s}}\!f_{2k+m}
    \\
    \mathrm{d}
    \,
    \dir{2}{\mathrm{s}}\!f_{2k+m}
    &=&
    -\dir{2}{\mathrm{s}}h_3
    \, 
    f_{2(k-1)+m}
    +
    h_3\, \dir{2}{\mathrm{s}}\!f_{2(k-1)+m}
    -
    \dir{1}{\omega}_2
    \,
    \dir{2}{\mathrm{s}}
    \dir{1}{\mathrm{s}}\!f_{2k+m}
    \\
    \mathrm{d}
    \,
    \dir{2}{\mathrm{s}}
    \dir{1}{\mathrm{s}}
    \!f_{2k+m}
    &=&
    -
    \dir{2}{\mathrm{s}}h_3
    \,
    \dir{1}{\mathrm{s}}\!f_{2(k-1)+m}
    +
    \dir{1}{\mathrm{s}}h_3
    \,
    \dir{2}{\mathrm{s}}\!f_{2(k-1)+m}
    +
    h_3
    \, 
    \dir{2}{\mathrm{s}}
    \dir{1}{\mathrm{s}}
    \!
    f_{2(k-1)+m}
  \end{array}
  \right)
  \,.
  \end{array}
\end{equation}
\end{example}

\begin{lemma}[\bf T-Automorphism of geometric 2-toroidified twisted K-theory]
\label{TAutomorphismOf2TorifiedTwistedKTheory}
The CE-algebra \eqref{SubalgebrasOf2TorOfTwistedK} of the geometric 2-toroidified twisted K-theory spectrum has an automorphism given by
\begin{equation}
  \label{TAutomorphismOfGeometric2ToroidifiedTwistedK}
  \begin{tikzcd}[row sep=-4pt, column sep=10pt]
    \mathrm{tor}^{2'}
    \,
    \mathfrak{l}\big(
      \Sigma^m
      \mathrm{KU}
      \!\sslash\!
      B\mathrm{U}(1)
    \big)
    \ar[
      rr,
      "{ T^2 }",
      "{
        \sim
      }"{swap}
    ]
    &&
    \mathrm{tor}^{2'}
    \,
    \mathfrak{l}\big(
      \Sigma^m
      \mathrm{KU}
      \!\sslash\!
      B\mathrm{U}(1)
    \big)
    \\
    h_3
    &\longmapsfrom&
    h_3
    \\
    \dir{1}{\omega}_2
    &\longmapsfrom&
    \dir{1}{\mathrm{s}}h_3
    \\
   \dir{2}{\omega}_2
    &\longmapsfrom&
    \dir{2}{\mathrm{s}}h_3
    \\
    \dir{1}{\mathrm{s}}h_3
    &\longmapsfrom&
    \dir{1}{\omega}_2
    \\ 
    \dir{2}{\mathrm{s}}h_3
    &\longmapsfrom&
    \dir{2}{\omega}_2
    \\
- \dir{2}{\mathrm{s}}
    \dir{1}{\mathrm{s}}
    \!f_{2(k+1) + m}
    &\longmapsfrom&
    f_{2k+m}
    \\
    -
    \dir{2}{\mathrm{s}}\!
    f_{2k+m}    
    &\longmapsfrom&
    \dir{1}{\mathrm{s}}\!
    f_{2k+m}
    \\
    +
    \dir{1}{\mathrm{s}}\!
    f_{2k+m}    
    &\longmapsfrom&
    \dir{2}{\mathrm{s}}\!
    f_{2k+m}
    \\
    +f_{2(k-1)+m}
    &\longmapsfrom&
    \dir{2}{\mathrm{s}}
    \dir{1}{\mathrm{s}}
    \!f_{2k+m}
    \mathrlap{\;.}
  \end{tikzcd}
\end{equation}
\end{lemma}
\begin{proof}
It is clear that the assignment uniquely extends to an automorphism of the underlying graded superalgebra. What remains to be seen is that this respects the differential \eqref{SubalgebrasOf2TorOfTwistedK}. 
\ifdefined\JournalVersion
This is straightforward to check.
\else
By unwinding the definitions, we check this explicitly on all generators:

$$
  \begin{tikzcd}[sep=10pt]
    -\dir{1}{\mathrm{s}}h_3
    \ar[
      rr,
      <-|
    ]
    \ar[
      dd,
      |->,
      "{ \mathrm{d} }"
    ]
    && 
    \dir{1}{\omega}_2
    \ar[
      dd,
      |->,
      "{ \mathrm{d} }"
    ]
    \\
    \\
    0
    \ar[
      rr,
      <-|
    ]
    &&
    0
  \end{tikzcd}
  \hspace{1cm}
  \begin{tikzcd}[sep=10pt]
    -\dir{2}{\mathrm{s}}h_3
    \ar[
      rr,
      <-|
    ]
    \ar[
      dd,
      |->,
      "{ \mathrm{d} }"
    ]
    && 
    \dir{2}{\omega}_2
    \ar[
      dd,
      |->,
      "{ \mathrm{d} }"
    ]
    \\
    \\
    0
    \ar[
      rr,
      <-|
    ]
    &&
    0
  \end{tikzcd}
  \hspace{1cm}
  \begin{tikzcd}[sep=10pt]
    -
    \dir{1}{\omega}_2
    \ar[
      dd,
      |->,
      "{ \mathrm{d} }"
    ]
    \ar[
      rr,
      <-|
    ]
    &&
    \dir{1}{\mathrm{s}}h_3
    \ar[
      dd,
      |->,
      "{ \mathrm{d} }"
    ]
    \\
    \\
    0
    \ar[
      rr,
      <-|
    ]
    &&
    0
  \end{tikzcd}
  \hspace{1cm}
  \begin{tikzcd}[sep=10pt]
    -
    \dir{2}{\omega}_2
    \ar[
      dd,
      |->,
      "{ \mathrm{d} }"
    ]
    \ar[
      rr,
      <-|
    ]
    &&
    \dir{2}{\mathrm{s}}h_3
    \ar[
      dd,
      |->,
      "{ \mathrm{d} }"
    ]
    \\
    \\
    0
    \ar[
      rr,
      <-|
    ]
    &&
    0
  \end{tikzcd}
$$

$$
\begin{tikzcd}[
  column sep=10pt,
  row sep=10pt
]
  h_3
  \ar[
    rrr,
    <-|
  ]
  \ar[
    dd,
    |->,
    "{ \mathrm{d} }"
  ]
  &[-50pt]
  &&
  h_3
  \ar[
    ddd,
    |->,
    "{ \mathrm{d} }"
  ]
  \\
  \\
  \dir{1}{\omega}_2
  \, 
  \dir{1}{\mathrm{s}}h_3
  +
  \dir{2}{\omega}_2
  \, 
  \dir{2}{\mathrm{s}}h_3
  \ar[
    dr,
    equals
  ]
  \\
  &
  +\dir{1}{\mathrm{s}}h_3
  \, 
  \dir{1}{\omega}_2
  +
  \dir{2}{\mathrm{s}}h_3
  \, 
  \dir{2}{\omega}_2
  \ar[
    rr,
    <-|
  ]
  &&
  \dir{1}{\omega}_2
  \, 
  \dir{1}{\mathrm{s}}h_3
  +
  \dir{2}{\omega}_2
  \, 
  \dir{2}{\mathrm{s}}h_3
\end{tikzcd}
$$

\smallskip

\smallskip

$$
  \begin{tikzcd}[
    column sep=10pt,
    row sep=10pt
  ]
    -\dir{2}{\mathrm{s}}
    \dir{1}{\mathrm{s}}
    \!
    f_{2(k+1) + m}
    \ar[rrr, <-|]
    \ar[
      dd, |->
    ]
    &[-90pt]
    &&
    f_{2k+m}
    \ar[
      ddd,
      |->
    ]
    \\
    \\
    -\dir{2}{\mathrm{s}}
    \dir{1}{\mathrm{s}}
    (h_3\, f_{2k + m})
    \ar[
      dr,
      equals
    ]
    \\
    &
   - h_3
    \,
    \dir{2}{\mathrm{s}}
    \dir{1}{\mathrm{s}}
    \!
    f_{2k + m}
    -
    \dir{1}{\mathrm{s}}h_3
    \,
    \dir{2}{\mathrm{s}}
    \!
    f_{2k+m}
    +
    \dir{2}{\mathrm{s}}h_3
    \,
    \dir{1}{\mathrm{s}}
    \!
    f_{2k+m}
    \ar[
      rr,
      <-|
    ]
    &&
    h_3
    \,
    f_{2(k-1) + m}
    +
    \dir{1}{\omega}_2
    \,
    \dir{1}{\mathrm{s}}
    \!
    f_{2k+m}
    +
    \dir{2}{\omega}_2
    \,
    \dir{2}{\mathrm{s}}
    \!
    f_{2k+m}
  \end{tikzcd}
$$

\smallskip

$$
  \begin{tikzcd}[
    column sep=10pt,
    row sep=10pt
  ]
    -\dir{2}{\mathrm{s}}\!f_{2k+m}
    \ar[
      rrr,
      <-|
    ]
    \ar[
      dd,
      |->
    ]
    &[-140pt]
    &&
    \dir{1}{\mathrm{s}}\!f_{2k+m}
    \ar[
      ddd,
      |->
    ]
    \\
    \\
    +\dir{2}{\mathrm{s}}
    (
      h_3 \, f_{2(k-1) + m}
      +
      \dir{1}{\omega}_2
      \,
      \dir{1}{\mathrm{s}}\!f_{2k+m}
    )
    \ar[
      dr,
      equals
    ]
    \\
    &
    +\dir{1}{\omega}_2
      \, 
      \dir{2}{\mathrm{s}}
      \dir{1}{\mathrm{s}}\!
      f_{2k + m}
    - h_3\,
     \dir{2}{\mathrm{s}}\!
     f_{2(k-1) + m}
      -
      \dir{2}{\mathrm{s}}h_3
      \,
      f_{2(k-1)+m} 
    \ar[
      rr,
      <-|
    ]
    &&
    -\dir{1}{\mathrm{s}}h_3\, 
      f_{2(k-1) + m}
    + h_3\,
     \dir{1}{\mathrm{s}}\!
     f_{2(k-1) + m}
      +
      \dir{2}{\omega}_2
      \,
      \dir{2}{\mathrm{s}}
      \dir{1}{\mathrm{s}}\!f_{2k+m}   
  \end{tikzcd}
$$


$$
  \begin{tikzcd}[
    column sep=10pt,
    row sep=10pt
  ]
    +
    \dir{1}{\mathrm{s}}\!f_{2k+m}
    \ar[
      rrr,
      <-|
    ]
    \ar[
      dd,
      |->
    ]
    &[-140pt]
    &&
    \dir{2}{\mathrm{s}}\!f_{2k+m}
    \ar[
      ddd,
      |->
    ]
    \\
    \\
    -\dir{1}{\mathrm{s}}
    (h_3 \, f_{2(k-1)+m} 
    +\dir{2}{\omega}_2\,
    \dir{2}{\mathrm{s}}\!
    f_{2k+m}
    )
    \ar[
      dr,
      equals
    ]
    \\
    &
    +
    \dir{2}{\omega}_2
    \,
    \dir{2}{\mathrm{s}}
    \dir{1}{\mathrm{s}}\!
    f_{2k+m}
    +
    h_3\, 
    \dir{1}{\mathrm{s}}\!
    f_{2(k-1)+m}
    -
    \dir{1}{\mathrm{s}}h_3\,
    f_{2(k-1)+m}
    \ar[
      rr,
      <-|
    ]
    &&
    -\dir{2}{\mathrm{s}}h_3
    \,
    f_{2(k-1)+m}
    +
    h_3\, 
    \dir{2}{\mathrm{s}}\!f_{2(k-1)+m}
    -
    \dir{1}{\omega}_2
    \,
    \dir{2}{\mathrm{s}}
    \dir{1}{\mathrm{s}}
    \!
    f_{2k+m}
  \end{tikzcd}
$$

\smallskip

$$
  \begin{tikzcd}[
    column sep=10pt,
    row sep=10pt
  ]  
    f_{2(k-1)+m}
    \ar[
      rrr,
      <-|
    ]
    \ar[
      dd,
      |->
    ]
    &[-200pt]
    &&
    \dir{2}{\mathrm{s}}
    \dir{1}{\mathrm{s}}\!
    f_{2k+m}
    \ar[
      ddd,
      |->
    ]
    \\
    \\
    +
    h_3 \, f_{2(k-2)+ m}
    +
    \dir{1}{\omega}_2
    \,
    \dir{1}{\mathrm{s}}\!
    f_{2(k-1) + m}
    +
    \dir{2}{\omega}_2
    \,
    \dir{2}{\mathrm{s}}\!
    f_{2(k-1) + m}
    \ar[
      dr,
      equals
    ]
    \\
    &
    +
    \dir{2}{\omega}_2
    \,
    \dir{2}{\mathrm{s}}\!
    f_{2(k-1)+m}
    +
    \dir{1}{\omega}_2\,
    \dir{1}{\mathrm{s}}\!
    f_{2(k-1)+m}
    +
    h_3\,
    f_{2(k-2) + m}
    \ar[
      rr,
      <-|
    ]
    &&
    -
    \dir{2}{\mathrm{s}}h_3
    \,
    \dir{1}{\mathrm{s}}\!
    f_{2(k-1)+ m}
    +
    \dir{1}{\mathrm{s}}h_3
    \,
    \dir{2}{\mathrm{s}}\!f_{2(k-1)+m}
    +
    h_3
    \,
    \dir{2}{\mathrm{s}}
    \dir{1}{\mathrm{s}}
    \!
    f_{2(k-1) + m}\;.
  \end{tikzcd}
$$
\vspace{-.3cm}
\fi
\end{proof}

\begin{remark}[\bf Automorphism of geometric 2-toroidified K-theory fibrations]\label{AutomorphismOfGeometric2ToroidifiedKtheoryFibrations}
In analogy with Ex. \ref{CyclificationOfTwistedKTTheorySpectra}, 
the automorphism \eqref{TAutomorphismOfGeometric2ToroidifiedTwistedK}
from Lem. \ref{TAutomorphismOf2TorifiedTwistedKTheory}
is compatible with the automorphism \eqref{AutomorphismOfRedTorOfB2U1} from Ex. \ref{ToroidalTDualityClassifyingAlgebra}, and as such constitutes an automorphism of geometric 2-toroidifed K-theory fibrations, 
with fibers given by 2 even copies and 2 odd copies respectively. The automorphism acts by swapping appropriately the even fluxes among themselves, and the odd fluxes among themselves, respectively. The diagram from Ex. \ref{CyclificationOfTwistedKTTheorySpectra} is modified appropriately to
\begin{equation*}
  \label{ToroidalTDualityOnToroidifiedKFibration}
  \begin{tikzcd}[row sep=-5pt]
     \mathfrak{l}
    \Sigma^{\mathcolor{purple}{0}}\mathrm{KU}
    \ar[
      ddrr
    ]
    &&
    \mathfrak{l}
    \Sigma^{\mathcolor{purple}{0}}\mathrm{KU}
    \\
    \times
    &&
    \times
    \\
    \mathfrak{l}
    \Sigma^{\mathcolor{purple}{0}}\mathrm{KU}
    \ar[
      uurr, 
      crossing over
    ]
    &&
    \mathfrak{l}
    \Sigma^{\mathcolor{purple}{0}}\mathrm{KU}
    \\
 \times
    &&
    \times
    \\
    \mathfrak{l}
    \Sigma^{\mathcolor{purple}{1}}\mathrm{KU}
    \ar[
      ddrr
    ]
    &&
    \mathfrak{l}
    \Sigma^{\mathcolor{purple}{1}}\mathrm{KU}
    \\
    \times
    &&
    \times
    \\
    \mathfrak{l}
    \Sigma^{\mathcolor{purple}{1}}\mathrm{KU}
    \ar[
      uurr, 
      crossing over
    ]
    \ar[
      d
    ]
    &&
    \mathfrak{l}
    \Sigma^{\mathcolor{purple}{1}}\mathrm{KU}
    \ar[
      d
    ]
    \\[22pt]
   \mathrm{tor}^{2'}
    \,
    \mathfrak{l}
    \big(
      \Sigma^{\mathcolor{purple}{0}}\mathrm{KU}
      \!\sslash\!
      B\mathrm{U}(1)
    \big)
    \ar[
      rr,
      <->,
      "{ \sim }"
    ]
    \ar[
      d,
      ->>
    ]
    &&
    \mathrm{tor}^{2'}
    \,
    \mathfrak{l}
    \big(
      \Sigma^{\mathcolor{purple}{0}}\mathrm{KU}
      \!\sslash\!
      B\mathrm{U}(1)
    \big)
    \ar[
      d,
      ->>
    ]
    \\[24pt]
     b\mathcal{T}^2
    \ar[
      rr,
      <->,
      "{ \sim }"
    ]
    &&
     b\mathcal{T}^2
    \mathrlap{\,.}
  \end{tikzcd}
\end{equation*}
Here the first cross over map corresponds to the exchange (up to signs) of generators with even number of shifts  
$$
  f_{2k} 
    \;\longleftrightarrow\; 
  \mathrm{s}_1 \mathrm{s}_2 f_{2k+2}
  \,,
$$ 
while the second one corresponds to the exchange  (up to signs) of generators with an odd number of shifts 
$$
  \mathrm{s}_2 f_{2k} 
    \;\longleftrightarrow\; 
  \mathrm{s}_{1} f_{2k}
\,.
$$ 
\end{remark}

The above observations generalize appropriately to the case of a toroidification along $3$ rational circles.

\begin{example}
[\bf Geometric 3-Toroidification of bundle gerbe classifying space]
\label{3ToroidalTDualityClassifyingAlgebra}
In completely analogous manner to Ex. \ref{ToroidalTDualityClassifyingAlgebra}, it can be seen that neither the 3-cyclification $\mathrm{cyc}^3( b^2 \mathbb{R})$ nor the 3-toroidification algebra 
\begin{equation}
    \label{3TorRecursion}
    \begin{tikzcd}
      \mathrm{tor}^{3}(b^2 \mathbb{R})
      \,:=\,
      \mathrm{tor}
      \,
      \mathrm{tor}^2
      (b^2 \mathbb{R})
      \ar[
        rr,
        hook
      ]
      &&
      \mathrm{cyc}\,\mathrm{tor}^2(b^2 \mathbb{R})\;.
    \end{tikzcd}
  \end{equation}

  \vspace{-1mm} 
\noindent enjoys an automorphism with the property that of swapping $\dir{i}{\mathrm{s}}h_3 \leftrightsquigarrow \dir{i}{\omega}_2$ \textit{for all} $i=1,2,3$, as would befit a 3-toroidal T-duality classifying algebra.  The latter 3-toroidification is given by the CE-algebra

\begin{equation}
  \label{CEOf3ToroidificationOfb2R}
  \mathrm{CE}\big(
    \mathrm{tor}^3 b^2\mathbb{R}
  \big)
  \;\simeq\;
  \FDGCA
  \left[\!
  \def\arraystretch{1.3}
  \begin{array}{c}
    \dir{2}{\omega}_2,
    \\
    \dir{1}{\omega}_2,
    \\
    h_3
    \\
    \dir{3}{\mathrm{s}}h_3
    \\
    \dir{2}{\mathrm{s}}h_3
    \\
    \dir{1}{\mathrm{s}}h_3
    \\
    \dir{2}{\mathrm{s}}\dir{1}{\mathrm{s}}h_3
    \\
     \dir{3}{\mathrm{s}}\dir{1}{\mathrm{s}}h_3
     \\
      \dir{3}{\mathrm{s}}\dir{2}{\mathrm{s}}h_3 
      \\
       \dir{3}{\mathrm{s}}\dir{2}{\mathrm{s}}\dir{1}{\mathrm{s}}h_3
  \end{array}
  \!\right]
  \Big/\ \!
  \left(\!
  \def\arraystretch{1.3}
  \begin{array}{ccc}
  \mathrm{d}
    \,
    \dir{3}{\omega}_2
    &=&
    0
    \\
    \mathrm{d}
    \,
    \dir{2}{\omega}_2
    &=&
    0
    \\
    \mathrm{d}
    \,
    \dir{1}{\omega}_2
    &=&
    0
    \\
    \mathrm{d}
    \,
    h_3
    &=&
    \dir{1}{\omega}_2 
    \,
    \dir{1}{\mathrm{s}}h_3
    +
    \dir{2}{\omega}_2
    \,
    \dir{2}{\mathrm{s}} h_3  +
    \dir{3}{\omega}_2
    \,
    \dir{3}{\mathrm{s}} h_3
    \\
    \mathrm{d}
    \,
    \dir{3}{\mathrm{s}}h_3
    &=&
       -
    \dir{1}{\omega}_2
    \,
    \dir{3}{\mathrm{s}}
    \dir{1}{\mathrm{s}}
    h_3 - \dir{3}{\omega}_2
    \,
    \dir{3}{\mathrm{s}}
    \dir{2}{\mathrm{s}}
    h_3 
    \\
    \mathrm{d}
    \,
    \dir{2}{\mathrm{s}}h_3
    &=&
       -
    \dir{1}{\omega}_2
    \,
    \dir{2}{\mathrm{s}}
    \dir{1}{\mathrm{s}}
    h_3       +
    \dir{3}{\omega}_2
    \,
    \dir{3}{\mathrm{s}}
    \dir{2}{\mathrm{s}}
    h_3
    \\
    \mathrm{d}
    \,
    \dir{1}{\mathrm{s}}h_3
    &=&
    + 
    \dir{2}{\omega}_2\, 
    \dir{2}{\mathrm{s}}
    \dir{1}{\mathrm{s}}
    h_3  + 
    \dir{3}{\omega}_2\, 
    \dir{3}{\mathrm{s}}
    \dir{1}{\mathrm{s}}
    h_3
    \\
    \mathrm{d}
    \,
    \dir{2}{\mathrm{s}}
    \dir{1}{\mathrm{s}}h_3
    &=&
    + \dir{3}{\omega}_2\,\dir{3}{\mathrm{s}}
    \dir{2}{\mathrm{s}}
    \dir{1}{\mathrm{s}}h_3
     \\
     \mathrm{d}
    \,
    \dir{3}{\mathrm{s}}
    \dir{1}{\mathrm{s}}h_3
    &=&
    - \dir{2}{\omega}_2\,\dir{3}{\mathrm{s}}
    \dir{2}{\mathrm{s}}
    \dir{1}{\mathrm{s}}h_3
    \\
      \mathrm{d}
    \,
    \dir{3}{\mathrm{s}}
    \dir{2}{\mathrm{s}}h_3
    &=&
    + \dir{1}{\omega}_2\,\dir{3}{\mathrm{s}}
    \dir{2}{\mathrm{s}}
    \dir{1}{\mathrm{s}}h_3
    \\
    \mathrm{d}
    \,\dir{3}{\mathrm{s}}
    \dir{2}{\mathrm{s}}
    \dir{1}{\mathrm{s}}h_3
    &=&
    0
  \end{array}
  \!\!\right)
  ,
\end{equation}
via which the obstructions for such a swapping automorphism to exist are seen to be the triply shifted $0$-cocycle and the doubly shifted (twisted) cocycle 1-generators
\vspace{-1mm} 
$$
      \dir{3}{\mathrm{s}} \dir{2}{\mathrm{s}}
    \dir{1}{\mathrm{s}}
    h_3, \quad \dir{3}{\mathrm{s}}
    \dir{2}{\mathrm{s}}
    h_3 , \quad \dir{3}{\mathrm{s}}
    \dir{1}{\mathrm{s}}
    h_3 , \quad    \dir{2}{\mathrm{s}}
    \dir{1}{\mathrm{s}}
    h_3  \, .
$$ 

This implies that the \textit{geometric} 3-toroidification $L_\infty$-subalgebra (cf. Rem. \ref{GeometricTDuality})  obtained by : {\bf 1)} firstly discarding the $3$-shifted 0-cocycle $ \dir{3}{\mathrm{s}} \dir{2}{\mathrm{s}}
    \dir{1}{\mathrm{s}}
    h_3$ and {\bf 2)} secondly discarding the (resulting untwisted) $2$-shifted 1-cocycles  $ \{ \dir{i}{\mathrm{s}}
    \dir{j}{\mathrm{s}}
    h_3 \}_{i,j=1,2,3}$ (Lem. \ref{QuotientByAbelianAntiIdeal}),
\vspace{-1mm} 
$$
  \begin{tikzcd}
  \mathrm{tor}^{3'}(b^2 \mathbb{R})
  \ar[r, hook]&\mathrm{tor}^3(b^2 \mathbb{R})
  \ar[r, hook]
  &
  \mathrm{cyc}\mathrm{tor}^2(b^2\mathbb{R})
  \,,
  \end{tikzcd}
$$
given by
\vspace{-1mm} 
\begin{equation}
  \label{CEOfReduced3ToroidificationOfb2R}
  \mathrm{CE}\big(
    \mathrm{tor}^{3'} b^2\mathbb{R}
  \big)
  \;\simeq\;
  \FDGCA
  \left[\!
  \def\arraystretch{1.3}
  \begin{array}{c}
     \dir{3}{\omega}_2,
    \\
    \dir{2}{\omega}_2,
    \\
    \dir{1}{\omega}_2,
    \\
    h_3
     \\
    \dir{3}{\mathrm{s}}h_3
    \\
    \dir{2}{\mathrm{s}}h_3
    \\
    \dir{1}{\mathrm{s}}h_3
  \end{array}
  \!\right]
  \Big/\ \!
  \left(\!
  \def\arraystretch{1.3}
  \begin{array}{ccc}
  \mathrm{d}
    \,
    \dir{3}{\omega}_2
    &=&
    0
    \\
    \mathrm{d}
    \,
    \dir{2}{\omega}_2
    &=&
    0
    \\
    \mathrm{d}
    \,
    \dir{1}{\omega}_2
    &=&
    0
    \\
    \mathrm{d}
    \,
    h_3
    &=&
    \dir{1}{\omega}_2 
    \,
    \dir{1}{\mathrm{s}}h_3
    +
    \dir{2}{\omega}_2
    \,
    \dir{2}{\mathrm{s}} h_3 +
    \dir{3}{\omega}_2
    \,
    \dir{3}{\mathrm{s}} h_3
    \\
    \mathrm{d}
    \,
    \dir{3}{\mathrm{s}}h_3
    &=& 0
    \\
    \mathrm{d}
    \,
    \dir{2}{\mathrm{s}}h_3
    &=& 0
    \\
    \mathrm{d}
    \,
    \dir{1}{\mathrm{s}}h_3
    &=&
   0
  \end{array}
  \!\!\right)
  ,
\end{equation}
being equivalently the higher central extension (Def. \ref{HigherCentralExtension}) of $b\mathbb{R}^2\times b\mathbb{R}^2 \times b\mathbb{R}^2$ by its canonical 4-cocycle
\begin{equation}
  \label{CEOfBOf3ToroidalTDuality2Group}
  \begin{tikzcd}[column sep=30pt]
    b\mathcal{T}^3\, := \, \mathrm{tor}^{3'} b^2\mathbb{R}
    \ar[
      rrrr,
      "{
        \mathrm{hofib}(
          \dir{1}{\omega}_2 
        \,
        \dir{1}{\widetilde \omega}_2 + \dir{2}{\omega}_2 
        \,
        \dir{2}{\widetilde \omega}_2
        + \dir{3}{\omega}_2 
        \,
        \dir{3}{\widetilde \omega}_2)
      }"
    ]
    &&&&
    b\mathbb{R}^2 \times b\mathbb{R}^2 \times b\mathbb{R}^2
    \ar[
      rrr,
      "{
        \dir{1}{\omega}_2 
        \,
        \dir{1}{\widetilde \omega}_2 + \dir{2}{\omega}_2 
        \,
        \dir{2}{\widetilde \omega}_2  + \dir{3}{\omega}_2 
        \,
        \dir{3}{\widetilde \omega}_2
      }"
    ]
    &&&
    b^3 \mathbb{R}\, ,
  \end{tikzcd}
\end{equation}
does have an automorphism symmetry given by swapping the two degree=2 generators as:
\begin{equation}
  \label{AutomorphismOfRed3TorOfB2U1}
  \begin{tikzcd}[row sep=-4pt,
    column sep=0pt
  ]
   b\mathcal{T}^3
    \ar[
      rr,
      <->,
      "{ \sim }"
    ]
    &&
    b\mathcal{T}^3
    \\
    -
    \dir{1}{\mathrm{s}}h_3
    &\longmapsfrom&
    \dir{1}{\omega}_2
    \\
    -
    \dir{1}{\omega}_2
    &\longmapsfrom&
    \dir{1}{\mathrm{s}}h_3
       \\
    -
    \dir{2}{\mathrm{s}}h_3
    &\longmapsfrom&
    \dir{2}{\omega}_2
    \\
    -
    \dir{2}{\omega}_2
    &\longmapsfrom&
    \dir{2}{\mathrm{s}}h_3
    \\
    -
   \dir{3}{\mathrm{s}}h_3
    &\longmapsfrom&
    \dir{3}{\omega}_2
    \\
    -
    \dir{3}{\omega}_2
    &\longmapsfrom&
    \dir{3}{\mathrm{s}}h_3
    \mathrlap{\;.}
  \end{tikzcd}
\end{equation}
\end{example}

This geometric 3-toroidification construction extends to the twisted $K$-theory spectra along the lines of Ex. \ref{2ToroidificationOfTwistedKtheory}, and so does the 
swapping automorphism \eqref{AutomorphismOfRed3TorOfB2U1} to an isomorphism of the geometric 3-toroidified $K$-theory spectra fibrations 
\begin{equation*}
  \label{ToroidalTDualityOn3ToroidifiedKFibration}
  \begin{tikzcd}[row sep=-5pt]
    \mathrm{tor}^{3'}
    \,
    \mathfrak{l}
    \big(
      \Sigma^{\mathcolor{purple}{0}}\mathrm{KU}
      \!\sslash\!
      B\mathrm{U}(1)
    \big)
    \ar[
      rr,
      <->,
      "{ \sim }"
    ]
    \ar[
      d,
      ->>
    ]
    &&
    \mathrm{tor}^{3'}
    \,
    \mathfrak{l}
    \big(
      \Sigma^{\mathcolor{purple}{1}}\mathrm{KU}
      \!\sslash\!
      B\mathrm{U}(1)
    \big)
    \ar[
      d,
      ->>
    ]
    \\[24pt]
     b\mathcal{T}^3
    \ar[
      rr,
      <->,
      "{ \sim }"
    ]
    &&
     b\mathcal{T}^3
    \mathrlap{,}
  \end{tikzcd}
\end{equation*}
in analogy to Rem. \ref{AutomorphismOfGeometric2ToroidifiedKtheoryFibrations},  which now acts by swapping appropriately the 4 even $\Sigma^0\mathrm{KU}$ and the 4 odd $\Sigma^{1}\mathrm{KU}$ fibers. The extended automorphism can be seen to act by 
 (up to signs) 
$$
\def\arraystretch{1.5}
\begin{array}{ccc}
\dir{3}{\mathrm{s}}\!f_{2k} 
  &\longleftrightarrow& 
\dir{1}{\mathrm{s}}\dir{2}{\mathrm{s}}\!f_{2k+1}
\\
\dir{1}{\mathrm{s}}\!f_{2k}   &\longleftrightarrow& 
\dir{2}{\mathrm{s}} \dir{3}{\mathrm{s}}\!f_{2k+1}
\\
\dir{2}{\mathrm{s}}\!f_{2k}   &\longleftrightarrow& 
\dir{3}{\mathrm{s}}\dir{1}{\mathrm{s}}\!f_{2k+1} 
\\
f_{2(k-1)} 
  &\longleftrightarrow& 
\dir{1}{\mathrm{s}}\dir{2}{\mathrm{s}}\dir{3}{\mathrm{s}}\!f_{2k+1} 
\mathrlap{\,,}
\end{array}
$$
hence swapping the even and odd copies.

Rather than writing out the analogous explicit formulas and proofs for this 3-toroidified case,  we do this more generally and concisely for the (geometric) $n$-toroidified twisted K-theory spectra.

\smallskip 

\noindent
{\bf $n$-Toroidification of twisted K-theory spectra.} With the case of 2- and 3-toroidification thus understood, we next give more abstract but general formulas for the situation of $n$-toroidification. 
\begin{example}[\bf Geometric $n$-Toroidification of bundle gerbe classifying space]
\label{nToroidificationOfBundleGerbeClassifyingSpace}
The geometric $n$-toroidification (cf. Rem. \ref{GeometricTDuality}) of the line Lie $3$-algebra 
\vspace{-2mm} 
$$
  \begin{tikzcd}
 b\mathcal{T}^{n} \, \equiv \,   \mathrm{tor}^{n'}(b^2 \mathbb{R})
  \ar[r, hook]&\mathrm{tor}^n(b^2 \mathbb{R})
  \ar[r, hook]
  &
  \mathrm{cyc}\mathrm{tor}^{n-1}(b^2\mathbb{R})
  \,,
  \end{tikzcd}
$$
is the $L_\infty$-algebra obtained by (successively) discarding all shifted generators $ \mathrm{s}_{i_1}\cdots \mathrm{s}_{i_k} h_3$ for $k\in \{2,\cdots, n\} $ and $i_{1},\cdots,i_k \in \{1,\cdots, n\}$ (those shifted more than once), hence given by
\begin{equation}
  \label{RestrictedNToroidificationOfb3R}
  \mathrm{CE}\big( b \mathcal{T}^n \big)
  \;\;
  \simeq
  \;\;
  \left(
  \def\arraystretch{1}
  \def\arraycolsep{2pt}
  \begin{array}{ccl}
    \mathrm{d}\,
    \dir{r}{\omega}_2
    &=&
    0
    \\
    \mathrm{d}
    \,
    h_3
    &=&
    \sum_{r}
    \dir{r}{\omega}_2
    \,
    \dir{r}{\mathrm{s}}
    h_3
    \\
    \mathrm{d}\, 
    \dir{r}{\mathrm{s}}h_3
    &=&
    0
  \end{array}
  \!\right).
\end{equation}
\end{example}
Crucially, the geometric n-toroidification of $b^2 \mathbb{R}$ enjoys the swapping automorphism 
\begin{equation}
\label{AutomorphismOfRednTorOfB2U1}
  \begin{tikzcd}[row sep=-4pt,
    column sep=10pt
  ]
   b\mathcal{T}^n
    \ar[
      rr,
      <->,
      "{ \sim }"
    ]
    &&
    b\mathcal{T}^n
    \\
    (-1)^{n}\, 
    \dir{r}{\mathrm{s}}h_3
    &\longmapsfrom&
    \dir{r}{\omega}_2
    \\
    (-1)^{n}\,
    \dir{r}{\omega}_2
    &\longmapsfrom&
    \dir{r}{\mathrm{s}}h_3 \, ,
  \end{tikzcd}
\end{equation}
which extends to a swapping isomorphism of the (geometric) $n$-toroidification of twisted K-theory spectra.
\begin{example}[\bf Geometric $n$-Toroidification of twisted K-theory]
\label{nToroidificationOfTwistedKTheory}
The geometric $n$-toroidification (cf. Rem. \ref{GeometricTDuality}) of (the Whitehead $L_\infty$-algebra of) the twisted K-theory spectra is:
\begin{equation}
  \label{RestrictedNToroidificationOfTwistedK}
  \mathrm{CE}\Big(
    \mathrm{tor}^{n'}
    \,
    \mathfrak{l}
    \big(
      \Sigma^m
      \mathrm{KU}
      \!\sslash\!
      B\mathrm{U}(1)
    \big)
  \!\Big)
  \;\;
  \simeq
  \;\;
  \left(
  \def\arraystretch{1}
  \def\arraycolsep{2pt}
  \begin{array}{ccl}
    \mathrm{d}\,
    \dir{r}{\omega}_2
    &=&
    0
    \\
    \mathrm{d}
    \,
    h_3
    &=&
    \sum_{r}
    \dir{r}{\omega}_2
    \,
    \dir{r}{\mathrm{s}}
    h_3
    \\
    \mathrm{d}\, 
    \dir{r}{\mathrm{s}}h_3
    &=&
    0
    \\
    \mathrm{d}
    \,
    \dir{i_r}{\mathrm{s}}
    \cdots
    \dir{i_1}{\mathrm{s}}
    f_{2k+m}
    &=&
    (-1)^r
    \,
    \dir{i_r}{\mathrm{s}}
    \cdots
    \dir{i_1}{\mathrm{s}} 
    \big(
      h_3 \, f_{2(k-1)+m}
      +
      \sum_{r'}
      \dir{r'}{\omega}_2
      \,
      \dir{r'}{\mathrm{s}}
      f_{2k+m}
    \big)
  \end{array}
  \!\right).
\end{equation}
\end{example}

For analyzing this, it will be convenient to make explicit the following:

\begin{definition}[\bf Fiberwise Hodge duality in $n$-toroidification]
\label{FiberwiseHodgeDuality}
Consider the Hodge duality-like operation on iterated shift operators  $  \dir{\mathcolor{darkblue}{i_r}}{\mathrm{s}}
  \cdots
  \dir{\mathcolor{darkblue}{i_2}\;}{\mathrm{s}}
  \dir{\;\,\mathcolor{darkblue}{i_1}}{\mathrm{s}}
$ in the CE-algebra of the $n$-toroidified twisted K-theory from Ex. \ref{nToroidificationOfTwistedKTheory}, given by:
\begin{equation}
  \label{TheFiberwiseHodgeDualityOperation}
  \big(
  \star\, 
  \dir{\mathcolor{darkblue}{i_r}}{\mathrm{s}}
  \cdots
  \dir{\mathcolor{darkblue}{i_2}\;}{\mathrm{s}}
  \dir{\;\,\mathcolor{darkblue}{i_1}}{\mathrm{s}}
  \,
  \big)
  \alpha
  \;\;
  :=
  \;\;
  \tfrac{1}{(n-r)!}
  \,
  \epsilon_{
    \mathcolor{purple}{
      i_{n} \cdots i_{r+1}
    }
    \mathcolor{darkblue}
      {i_r \cdots i_2 i_1}
  }
  \dir{\mathcolor{purple}{i_n}}{\mathrm{s}}
  \cdots
  \dir{\mathcolor{purple}{i_{r+1}}\;\;}{\mathrm{s}}
  \alpha
  \,,
  \;\;\;\;
  (\epsilon_{n \cdots 2 1}
  \,:=\, 1)
  \, ,
\end{equation}
with the understanding that $\star \, \alpha \, := \, \dir{n}{\mathrm{s}} \cdots \dir{1}{\mathrm{s}} \alpha$.
\end{definition}
 
\begin{lemma}[\bf Properties of Hodge operator on winding modes]
The operation \eqref{TheFiberwiseHodgeDualityOperation} satisfies the usual properties of a Hodge star operator:
\begin{align}
  \label{IdemFiberwiseHodge}
  (-1)^{ r(n-1) } \star \, \star 
  \,&=\,
   \mathrm{id}
   \\
  \label{HodgeDualShift}
  (-1)^{rn-1}
  \star
  \,
  \dir{k}{\mathrm{s}}
  \,
  \star
  \;&=\;
  \mbox{\rm derivation that removes $\dir{k}{\mathrm{s}}$}.
\end{align}


\end{lemma}
\begin{proof}
We have 
\vspace{-1mm} 
$$
  \def\arraystretch{1.4}
  \begin{array}{ll}
    \big(
    \star \star
    \dir{\mathcolor{darkblue}{i_r}}{\mathrm{s}}
    \cdots
    \dir{\mathcolor{darkblue}{i_1}}{\mathrm{s}}
    \big)
    \alpha
    &
    \;=\;
    \tfrac{1}{(n-r)!}
    \,
    \epsilon_{
      \mathcolor{purple}{
        i_{n} \cdots i_{r+1}
      }
      \mathcolor{darkblue}
        {i_r \cdots i_1}
    }
    \big(
    \star
    \dir{\mathcolor{purple}{i_n}}{\mathrm{s}}
    \cdots
    \dir{\mathcolor{purple}{i_{r+1}}\;\;}{\mathrm{s}}
    \,
    \big)
    \alpha
    \\
   & \;=\;
    \tfrac{1}{(n-r)!}
    \,
    \epsilon_{
      \mathcolor{purple}{
        i_{n} \cdots i_{r+1}
      }
      \mathcolor{darkblue}
        {i_r \cdots i_1}
    }
    \tfrac{1}{r!}
    \epsilon^{
      \mathcolor{darkgreen}
        {j_r \cdots j_1}
      \mathcolor{purple}
        {i_n \cdots i_{r+1}}
    }
    \;
    \dir{\mathcolor{darkgreen}{j_r}}{\mathrm{s}}
    \cdots
    \dir{\mathcolor{darkgreen}{j_{1}}\;\;}{\mathrm{s}}
    \alpha
    \\
    & \;=\;
    (-1)^{ r(n-r) }
    \;
    \tensor*{\delta}{
      ^{ 
        \mathcolor{darkgreen}{
          j_r \cdots j_1
        }
      }
      _{
        \mathcolor{darkblue}
        {i_r \cdots i_1}
      }
    }
    \dir{\mathcolor{darkblue}{i_r}}{\mathrm{s}}
    \cdots
    \dir{\mathcolor{darkblue}{i_1}}{\mathrm{s}}
    \alpha
    \\
    & 
    \;=\;
    (-1)^{ r(n-1) }
    \;
    \dir{\mathcolor{darkblue}{i_r}}{\mathrm{s}}
    \cdots
    \dir{\mathcolor{darkblue}{i_1}}{\mathrm{s}}
    \alpha
  \end{array}
$$

\vspace{-1mm} 
and
$$
  \def\arraystretch{1.5}
  \begin{array}{ll}
  \def\arraystretch{1.5}
  \big(
  \star
  \,
  \dir{k}{\mathrm{s}}
  \,
  \star
  \,
  \dir{\mathcolor{darkblue}{i_r}}{\mathrm{s}}
  \cdots
  \dir{\mathcolor{darkblue}{i_1}}{\mathrm{s}}
  \,
  \big)
  \alpha
  &
  \;=\;
  \tfrac{1}{(n-r)!}
  \,
  \epsilon_{
    \mathcolor{purple}{
      i_{n} \cdots i_{r+1}
    }
    \mathcolor{darkblue}
      {i_r \cdots i_2 i_1}
  }
  \big(
  \star\, 
  \dir{k}{\mathrm{s}}
  \,
  \dir{\mathcolor{purple}{i_n}}{\mathrm{s}}
  \cdots
  \dir{\mathcolor{purple}{i_{r+1}}\;\;}{\mathrm{s}}
  \,
  \big)
  \alpha
  \\
  &\;=\;
  \tfrac{1}{(n-r)!}
  \tfrac{1}{(r-1)!}
  \,
  \epsilon_{
    \mathcolor{purple}{
      i_{n} \cdots i_{r+1}
    }
    \mathcolor{darkblue}
      {i_r \cdots i_1}
  }
  \epsilon^{
    \mathcolor{darkgreen}{
      j_{r-1}\cdots j_{1}
    }
    k
    \mathcolor{purple}{i_n\cdots i_{r+1}} 
  }
  \;
  \dir{\mathcolor{darkgreen}{j_{r-1}}}{\mathrm{s}}
  \cdots
  \dir{\mathcolor{darkgreen}{j_{1}}}{\mathrm{s}}
  \alpha
  \\
  & \;=\;
  (-1)^{r(n-r) + r-1}
  \;
  r
  \;
  \tensor*{\delta}
    {
      ^{k\,\mathcolor{darkgreen}
        {i_{r-1} \cdots i_1}
      }
      _{
        \mathcolor{darkblue}
        {i_r \cdots i_1}
      }
    }
  \;\,
  \dir{\mathcolor{darkgreen}{i_{r-1}}}{\mathrm{s}}
  \cdots
  \dir{\mathcolor{darkgreen}{i_{1}}}{\mathrm{s}}
  \alpha
  \\
  & \;=\;
  (-1)^{rn-1}
  \;
  r
  \;
  \tensor*{\delta}
    {
      ^{k\,\mathcolor{darkgreen}
        {i_{r-1} \cdots i_1}
      }
      _{
        \mathcolor{darkblue}
        {i_r \cdots i_1}
      }
    }
  \;\,
  \dir{\mathcolor{darkgreen}{i_{r-1}}}{\mathrm{s}}
  \cdots
  \dir{\mathcolor{darkgreen}{i_{1}}}{\mathrm{s}}
  \alpha
  \\[+3pt]
  & \;=\;
  \left\{\!\!\!
  \begin{array}{cl}
  (-1)^{rn-1} \;
  \dir{\mathcolor{darkblue}{i_{r-1}}}{\mathrm{s}}
  \;\;\,
  \dir{\mathcolor{darkblue}{i_{r-2}}}{\mathrm{s}}
  \;\;\,
  \dir{\mathcolor{darkblue}{i_{r-3}}}{\mathrm{s}}
  \cdots
  \dir{\mathcolor{darkblue}{i_1}}{\mathrm{s}}
  \alpha
  &
  \mbox{if $k = i_r$}
  \\
  - (-1)^{rn-1}\;
  \dir{\mathcolor{darkblue}{i_{r}}}{\mathrm{s}}
  \;\;\,
  \dir{\mathcolor{darkblue}{i_{r-2}}}{\mathrm{s}}
  \;\;\,
  \dir{\mathcolor{darkblue}{i_{r-3}}}{\mathrm{s}}
  \cdots
  \dir{\mathcolor{darkblue}{i_1}}{\mathrm{s}}
  \alpha
  &
  \mbox{if $k = i_{r-1}$}
  \\[-4pt]
  \vdots
  &
  \;\;\;\;\;\;\vdots
   \\[-4pt]
   (-1)^{r-1} (-1)^{rn-1}\;
  \dir{\mathcolor{darkblue}{i_{r}}}{\mathrm{s}}
  \;\;\,
  \dir{\mathcolor{darkblue}{i_{r-1}}}{\mathrm{s}}
  \;\;\,
  \dir{\mathcolor{darkblue}{i_{r-2}}}{\mathrm{s}}
  \cdots
  \dir{\mathcolor{darkblue}{i_2}}{\mathrm{s}}
  \alpha
  &
  \mbox{if $k = i_{1}$} 
   \\[-5pt]
  0
  &
  \mbox{otherwise} 
  \mathrlap{\,.}
  \end{array}
  \right.
  \end{array}
$$

\vspace{-5mm} 
\end{proof}

Using the Hodge duality notation of Rem. \ref{FiberwiseHodgeDuality},
we may generalize Lem. \ref{TAutomorphismOf2TorifiedTwistedKTheory} as follows:

\begin{proposition}[\bf T-isomorphism of $n$-torodified twisted K-theory]
\label{TIsomorphismOf2kToroidifiedTwistedKTheory}
For any $n \in \mathbb{N}$, there is 
an isomor- phism between the geometric $n$-toroidified twisted spectra \eqref{RestrictedNToroidificationOfTwistedK} of degree $m$ and degree $m + (n \, \mathrm{mod} \, 2)$ given by
\begin{equation}
  \label{TnAutomorphism}
  \begin{tikzcd}[row sep=-4pt]
    \mathrm{tor}^{n'}
    \,
    \mathfrak{l}
    \big(
      \Sigma^m
      \mathrm{KU}
      \!\sslash\!
      B\mathrm{U}(1)
    \big)
    \ar[
      rr,
      "{ T^{n} }",
      "{ \sim }"{swap}
    ]
    &&
    \mathrm{tor}^{n'}
    \,
    \mathfrak{l}
    \big(
      \Sigma^{m+(n \, \mathrm{mod}\,  2)}
      \mathrm{KU}
      \!\sslash\!
      B\mathrm{U}(1)
    \big)
    \\
    {(-1)}^{n}\, 
    \dir{r}{\mathrm{s}}h_3
    &\longmapsfrom&
    \dir{r}{\omega}_2
    \\
    (-1)^{n}\,
    \dir{r}{\omega}_2    
    &\longmapsfrom&
    \dir{r}{\mathrm{s}}h_3
    \\
    (-1)^{(n+r+1+  r(r+1)/2)}
    \,
    \star
    \dir{i_r}{\mathrm{s}}
    \cdots
    \dir{i_1}{\mathrm{s}}\!
    f
    &\longmapsfrom&
    \dir{i_r}{\mathrm{s}}
    \cdots
    \dir{i_1}{\mathrm{s}}\!
    f
    \mathrlap{\,,}
  \end{tikzcd}
\end{equation}

\vspace{0mm} 
\noindent {\rm where in the last line of \eqref{TnAutomorphism} and in the following we suppress the indices on the generators just for notational brevity. Notably, this is an automorphism for even $n=2k$, $k\in \NN$.}

\end{proposition}

\begin{proof}
The map is manifestly invertible on generators, so we need to check that the differential is respected. On the generators corresponding to the $b\mathcal{T}^n$ $L_\infty$-subalgebra this is immediate (cf. \ref{RestrictedNToroidificationOfb3R}):
$$
  \begin{tikzcd}[
    sep=10pt
  ]
    (-1)^n\, \dir{r}{\mathrm{s}}h_3
    \ar[rr, <-|]
    \ar[
      dd,
      |->
    ]
    &&
    \dir{r}{\omega}_2
    \ar[
      dd,
      |->
    ]
    \\
    \\
    0
    \ar[rr, <-|]
    &&
    0
  \end{tikzcd}
  \hspace{1cm}
  \begin{tikzcd}[
    sep=10pt
  ]
   (-1)^n\,  \dir{r}{\omega}_2
    \ar[rr, <-|]
    \ar[
      dd,
      |->
    ]
    &&
    \dir{r}{\mathrm{s}}h_3
    \ar[
      dd,
      |->
    ]
    \\
    \\
    0
    \ar[rr, <-|]
    &&
    0
  \end{tikzcd}
  \hspace{1cm}
  \begin{tikzcd}[
    row sep=6pt, column sep=15pt
  ]
    h_3
    \ar[rrr, <-|]
    \ar[
      dd,
      |->
    ]
    &[-35pt]
    &&
    h_3
    \ar[
      ddd,
      |->
    ]
    \\
    \\
    \sum_{r}
    \dir{r}{\omega}_2
    \,
    \dir{2}{\mathrm{s}}h_3
    \ar[
      dr,
      equals
    ]
    \\
    &
    \sum_{r}
    \dir{2}{\mathrm{s}}h_3
    \,
    \dir{r}{\omega}_2
    \ar[
      rr,
      <-|
    ]
    &&
    \sum_{r}
    \dir{r}{\omega}_2
    \,
    \dir{2}{\mathrm{s}}h_3
  \end{tikzcd}
$$
To see it on the remaining generators, abbreviate
$$
  \sigma(r)
  :=
  (-1)^{r(r+1)/2}
$$
and note that
$$
  \def\arraystretch{1.4}
  \begin{array}{ccr}
  \sigma(r+1)
  &=&
  -
  (-1)^r
  \cdot
  \sigma(r)
  \\
  \sigma(r-1)
  &=&
  (-1)^r
  \cdot
  \sigma(r) 
  \,.
  \end{array}
$$
With this, we compute as follows: 
\smallskip 
$$
\hspace{-5mm} 
  \begin{tikzcd}[row sep=20pt, column sep=40pt]
  (-1)^{n+r+1}  \sigma(r)
    \,
    \star
     \dir{i_r}{\mathrm{s}}
     \cdots
     \dir{i_1}{\mathrm{s}}
     \!
     f
     \ar[
       rrr,
       <-|
     ]
     \ar[
       dd,
       |->,
       "{ \mathrm{d} }"
     ]
    &[-180pt]
    &[-50pt]
    &
     \dir{i_r}{\mathrm{s}}
     \cdots
     \dir{i_1}{\mathrm{s}}
     \!
     f
     \ar[
       ddd,
       |->,
       "{ \mathrm{d} }"
     ]
    \\
    \\
   (-1)^{n+r+1} 
  \sigma(r)
    \cdot
    (-1)^r
    \,
    \big(\!
    \star
     \dir{i_r}{\mathrm{s}}
     \cdots
     \dir{i_1}{\mathrm{s}}
     \,
     \big)
     (h_3 \, f +
     \sum_{r'}
       \dir{r'}{\omega}_2
       \,
       \dir{r'}{\mathrm{s}}\!f
     )
     \ar[
       d, 
       equals,
       "{
         \scalebox{.7}{
           \color{gray}
           \eqref{HodgeDualShift}
         }
       }"{pos=.5}
     ]
    \\
    \def\arraystretch{1.3}
    \begin{array}{l}
       (-1)^{n+r+1} \sigma(r)
      \,
      h_3 \, \star\,
      \dir{i_r}{\mathrm{s}}
      \cdots
      \dir{i_1}{\mathrm{s}} f
      \\
      +  (-1)^{nr+r+1} 
      \sigma(r)
      \,
      \sum_{r'}
      \dir{r'}{\mathrm{s}}h_3
      \,
      \star
      \,
      \dir{r'}{\mathrm{s}} \,
      \star \,
      \star \, 
      \dir{i_r}{\mathrm{s}}
      \cdots
      \dir{i_1}{\mathrm{s}}\!f
      \\
      +  (-1)^{r+1}
      \sigma(r)
      \,
      \sum_{r'}
      \dir{r'}{\omega}_2 
      \,
      \dir{r'}{\mathrm{s}}
      \,\star\,
      \dir{i_r}{\mathrm{s}}
      \cdots
      \dir{i_1}{\mathrm{s}}\!f
    \end{array}
     \ar[
       dr,
       equals,
       "{
         \scalebox{.7}{
           \color{gray}
           \eqref{IdemFiberwiseHodge}
         }
       }"{pos=.6, xshift=-2pt}
     ]
    &
    &&
    (-1)^r
    \;
     \dir{i_r}{\mathrm{s}}
     \cdots
     \dir{i_1}{\mathrm{s}}
     \big(
       h_3 \, f
       +
       \sum_{r'}
       \dir{r'}{\omega}_2
       \,
       \dir{r'}{\mathrm{s}}\!
       f
     \big)
     \ar[
       d, 
       equals,
       shorten >=-10pt,
       "{
         \scalebox{.7}{
           \color{gray}
           \eqref{HodgeDualShift}
         }
       }"{pos=.7}
     ]
     \\
     &
     \def\arraystretch{1.3}
     \begin{array}{l}
      (-1)^{n+r+1}  \sigma(r)
     \,
     h_3\, 
     \star
     \dir{i_r}{\mathrm{s}}
     \cdots
     \dir{i_1}{\mathrm{s}}\!
     f    
     \\
     \,-\,  (-1)^{r+1+rn} 
     \sigma(r-1)\,
     \sum_{r'}
     \,
     \dir{r'}{\omega}_2
     \,
     (
     \star
     \star 
     \, \dir{r'}{\mathrm{s}}\star
     \dir{i_r}{\mathrm{s}}
     \cdots
     \dir{i_1}{\mathrm{s}}\!
     f
     )
     \\
     \,-\,  (-1)^{r+1}
     \sigma(r+1)
     \sum_{r'}
     \;
     \dir{r'}{\mathrm{s}}h_3
     \,
     \star
     \dir{r'}{\mathrm{s}}\,
     \dir{i_r}{\mathrm{s}}
     \cdots
     \dir{i_1}{\mathrm{s}}\!
     f
     \end{array}
     \ar[
       rr,
       <-|
     ]
     &&
     \def\arraystretch{1.4}
     \begin{array}{l}
     h_3\, 
     \dir{i_r}{\mathrm{s}}
     \cdots
     \dir{i_1}{\mathrm{s}}\!
     f    
     \\
     \,+\,(-1)^{rn}
     \sum_{r'}
     \dir{r'}{\mathrm{s}}h_3
     \,
     (
     \star 
     \, \dir{r'}{\mathrm{s}}\star
     \dir{i_r}{\mathrm{s}}
     \cdots
     \dir{i_1}{\mathrm{s}}\!
     f         
     )
     \\
     \,+\,
     \sum_{r'} 
     \dir{r'}{\omega}_2
     \,
     \dir{r'}{\mathrm{s}}
     \,
     \dir{i_r}{\mathrm{s}}
     \cdots
     \dir{i_1}{\mathrm{s}}\!
     f
     \end{array}
  \end{tikzcd}
$$

\vspace{-4mm} 
\end{proof}
 
Of course, for $n=1$ this recovers the isomorphism of Ex. \ref{CyclificationOfTwistedKTTheorySpectra}. Similarly, generalizing the $n=1$ case of Lem. \ref{FurtherIsomorphismsOfCyclifiedTwistedKspectra}, we may characterize the set of all such swapping isomorphisms.

\begin{lemma}
[\bf All isomorphisms of n-toroidified twisted K-spectra]\label{FurtherIsomorphismsOfToroidifiedTwistedKspectra}
$\,$

\noindent {\bf (i)} There are in total 4 isomorphisms
\begin{equation*}
  \begin{tikzcd}[
    row sep=-2pt,
    column sep=15pt
  ]
     \mathrm{tor}^{n'}
    \,
    \mathfrak{l}
    \big(
      \Sigma^m
      \mathrm{KU}
      \!\sslash\!
      B\mathrm{U}(1)
    \big)
    \ar[
      rr,
      <->,
      "{ \sim }"
    ]
    &&
     \mathrm{tor}^{n'}
    \,
    \mathfrak{l}
    \big(
      \Sigma^{m+(n \, \mathrm{mod}\,  2)}
      \mathrm{KU}
      \!\sslash\!
      B\mathrm{U}(1)
    \big)\\
    &&
  \end{tikzcd}
\end{equation*}

\noindent with the property of swapping $\dir{r}{\mathrm{s}}h_3 \leftrightsquigarrow \dir{r}{\omega}_2$ and $ \dir{i_r}{\mathrm{s}}
    \cdots
    \dir{i_1}{\mathrm{s}}\!
    f\leftrightsquigarrow \star  \dir{i_r}{\mathrm{s}}
    \cdots
    \dir{i_1}{\mathrm{s}}\!
    f$, while mapping $h_3$ to $h_3$, up to relative sign prefactors. 

\noindent {\bf (ii)} Explicitly, in addition to 
\eqref{TnAutomorphism}, one has 
\begin{equation}\label{FirstExtraIsoOfToroidifiedTwistedKSpectra}
 \begin{tikzcd}[row sep=-2pt]
    \mathrm{tor}^{n'}
    \,
    \mathfrak{l}
    \big(
      \Sigma^m
      \mathrm{KU}
      \!\sslash\!
      B\mathrm{U}(1)
    \big)
    \ar[
      rr,
      "{ T^{n} }",
      "{ \sim }"{swap}
    ]
    &&
    \mathrm{tor}^{n'}
    \,
    \mathfrak{l}
    \big(
      \Sigma^{m+(n \, \mathrm{mod}\,  2)}
      \mathrm{KU}
      \!\sslash\!
      B\mathrm{U}(1)
    \big)
    \\
    {(-1)}^{n}\, 
    \dir{r}{\mathrm{s}}h_3
    &\longmapsfrom&
    \dir{r}{\omega}_2
    \\
    (-1)^{n}\,
    \dir{r}{\omega}_2    
    &\longmapsfrom&
    \dir{r}{\mathrm{s}}h_3
    \\
    (-1)^{(n+r+  r(r+1)/2)}
    \,
    \star
    \dir{i_r}{\mathrm{s}}
    \cdots
    \dir{i_1}{\mathrm{s}}\!
    f
    &\longmapsfrom&
    \dir{i_r}{\mathrm{s}}
    \cdots
    \dir{i_1}{\mathrm{s}}\!
    f
    \mathrlap{\,,}
  \end{tikzcd}
\end{equation}

\medskip 
\begin{equation}\label{SecondExtraIsoOfToroidifiedTwistedKSpectra}
  \begin{tikzcd}[row sep=-3pt]
    \mathrm{tor}^{n'}
    \,
    \mathfrak{l}
    \big(
      \Sigma^m
      \mathrm{KU}
      \!\sslash\!
      B\mathrm{U}(1)
    \big)
    \ar[
      rr,
      "{ T^{n} }",
      "{ \sim }"{swap}
    ]
    &&
    \mathrm{tor}^{n'}
    \,
    \mathfrak{l}
    \big(
      \Sigma^{m+(n \, \mathrm{mod}\,  2)}
      \mathrm{KU}
      \!\sslash\!
      B\mathrm{U}(1)
    \big)
    \\
    {(-1)}^{n+1}\, 
    \dir{r}{\mathrm{s}}h_3
    &\longmapsfrom&
    \dir{r}{\omega}_2
    \\
    (-1)^{n+1}\,
    \dir{r}{\omega}_2    
    &\longmapsfrom&
    \dir{r}{\mathrm{s}}h_3
    \\
    (-1)^{(n+  r(r+1)/2)}
    \,
    \star
    \dir{i_r}{\mathrm{s}}
    \cdots
    \dir{i_1}{\mathrm{s}}\!
    f
    &\longmapsfrom&
    \dir{i_r}{\mathrm{s}}
    \cdots
    \dir{i_1}{\mathrm{s}}\!
    f
    \mathrlap{\,,}
  \end{tikzcd}
\end{equation}
and
\begin{equation}\label{ThirdExtraIsoOfToroidifiedTwistedKSpectra}
  \begin{tikzcd}[row sep=-2pt]
    \mathrm{tor}^{n'}
    \,
    \mathfrak{l}
    \big(
      \Sigma^m
      \mathrm{KU}
      \!\sslash\!
      B\mathrm{U}(1)
    \big)
    \ar[
      rr,
      "{ T^{n} }",
      "{ \sim }"{swap}
    ]
    &&
    \mathrm{tor}^{n'}
    \,
    \mathfrak{l}
    \big(
      \Sigma^{m+(n \, \mathrm{mod}\,  2)}
      \mathrm{KU}
      \!\sslash\!
      B\mathrm{U}(1)
    \big)
    \\
    {(-1)}^{n+1}\, 
    \dir{r}{\mathrm{s}}h_3
    &\longmapsfrom&
    \dir{r}{\omega}_2
    \\
    (-1)^{n+1}\,
    \dir{r}{\omega}_2    
    &\longmapsfrom&
    \dir{r}{\mathrm{s}}h_3
    \\
    (-1)^{(n+ 1+  r(r+1)/2)}
    \,
    \star
    \dir{i_r}{\mathrm{s}}
    \cdots
    \dir{i_1}{\mathrm{s}}\!
    f
    &\longmapsfrom&
    \dir{i_r}{\mathrm{s}}
    \cdots
    \dir{i_1}{\mathrm{s}}\!
    f
    \mathrlap{\,,}
  \end{tikzcd}
\end{equation}

\noindent {\bf (iii)}
Evidently, the original isomorphism \eqref{TnAutomorphism} and the first above are the two possible extensions of the automorphism \eqref{AutomorphismOfRednTorOfB2U1} of $b\mathcal{T}^n$, while the latter two isomorphisms are the two possible extensions of the ``opposite'' automorphism of $b\mathcal{T}^n$ 
\smallskip 
\begin{equation*}
  \begin{tikzcd}[row sep=-2pt,
    column sep=0pt
  ]
   b\mathcal{T}^n
    \ar[
      rr,
      <->,
      "{ \sim }"
    ]
    &&
    b\mathcal{T}^n
    \\
    (-1)^{n+1} \, \dir{r}{\mathrm{s}} h_3
    &\longmapsfrom&
    \dir{r}{\omega}_2
    \\
    (-1)^{n+1} \, \dir{r}{\omega}_2
    &\longmapsfrom&\dir{r}{\mathrm{s}}h_3
    \mathrlap{\,.}
  \end{tikzcd}
  \end{equation*}
\end{lemma}
\begin{proof}This follows by direct inspection, similar to the case of $n=1$ from Lem. \ref{FurtherIsomorphismsOfCyclifiedTwistedKspectra}. Explicitly, starting (for instance) from the isomorphism \eqref{TnAutomorphism} we ask which possible extra set of signs one can insert in the image of the generators, such that it still commutes with the differentials. The relation $\dd h_3 = \sum_{r}\dir{r}{\omega}_2 \dir{r}{\mathrm{s}}h_3$ restricts the underlying map of graded commutative algebras to be of the form
\begin{equation*}
  \begin{tikzcd}[row sep=-2pt]
    \mathrm{tor}^{n'}
    \,
    \mathfrak{l}
    \big(
      \Sigma^m
      \mathrm{KU}
      \!\sslash\!
      B\mathrm{U}(1)
    \big)
    \ar[
      rr,
      "{ T^{n} }",
      "{ \sim }"{swap}
    ]
    &&
    \mathrm{tor}^{n'}
    \,
    \mathfrak{l}
    \big(
      \Sigma^{m+(n \, \mathrm{mod}\,  2)}
      \mathrm{KU}
      \!\sslash\!
      B\mathrm{U}(1)
    \big)
    \\
    {\color{purple}(-1)^{q}}{(-1)}^{n}\, 
    \dir{r}{\mathrm{s}}h_3
    &\longmapsfrom&
    \dir{r}{\omega}_2
    \\
   {\color{purple}(-1)^{q}}(-1)^{n}\,
    \dir{r}{\omega}_2    
    &\longmapsfrom&
    \dir{r}{\mathrm{s}}h_3
    \\
   {\color{purple}(-1)^{x(r)}} (-1)^{(n+r+1+  r(r+1)/2)}
    \,
    \star
    \dir{i_r}{\mathrm{s}}
    \cdots
    \dir{i_1}{\mathrm{s}}\!
    f
    &\longmapsfrom&
    \dir{i_r}{\mathrm{s}}
    \cdots
    \dir{i_1}{\mathrm{s}}\!
    f
    \mathrlap{\,,}
  \end{tikzcd}
\end{equation*}
for some $q \in \NN$ and some function $x(r): \NN\rightarrow \NN$. Demanding that it further commutes with the corresponding differentials on all the (shifted) flux generators $ \dir{i_r}{\mathrm{s}}
    \cdots
    \dir{i_1}{\mathrm{s}}\!
    f$ (cf. proof of Prop. \ref{TIsomorphismOf2kToroidifiedTwistedKTheory}) yields the condition
\smallskip 
$$
(-1)^q \cdot (-1)^{x(r-1)} \, =\, (-1)^{x(r)} \, = \, (-1)^q \cdot (-1)^{x(r+1)} ,
$$
whose set of solutions corresponds to the 3 extra isomorphisms above.
\end{proof}

\begin{theorem}[\bf Twisted K--theory cocycles under toroidal reduction--isomorphism--reoxidation]\label{TwistedKTheoryToroidalRedIsoReOxi}
\label{LInfinityAlgebraicToroidalTDuality}
$\,$

\noindent The composite operation of 

\begin{itemize}[
 leftmargin=.8cm,
 topsep=0pt,
 itemsep=5pt
]

 \item[\bf (a)] $n$-toroidally reducing \eqref{nToroidificationHomIsomorphism}
  twisted $\mathrm{KU}_{\mathcolor{purple}m}$ cocycles on a toroidally extended  super $L_\infty$-algebra $\dir{1\cdots n}{\widehat{g}}_A$
\begin{equation*}
  \begin{tikzcd}[
    row sep=-2pt, 
    column sep=10pt]
    \dir{1\cdots n}{\widehat{g}}_A
    \ar[
      rr
    ]
    &&
    \mathfrak{l}
    \big(
      \,
      \Sigma^{\color{purple}m}
      \mathrm{KU}
      \!\sslash\!\!B\mathrm{U}(1) \, \big)
    \\
    H_{A}^3
    &\longmapsfrom&
    h_3
    \\
    (F_{2k+m})_{k\in \mathbb{Z}}
    &\longmapsfrom&
    (f_{2k+m})_{k\in \mathbb{Z}} \, ,
  \end{tikzcd}
\end{equation*}
with geometric twists (Rem. \ref{GeometricTDuality}) of the form 
$$
H^3_A \, = \, H^3_\frg + \sum_{k=1}^{n} \dir{k}{e}_A \cdot {\dir{k}{p}_A}_* (H^3_A)\, ,
$$
along its fibration 
$$\dir{1\cdots n}{\widehat{g}}_A \xlongrightarrow{\quad \dir{1\cdots n}{p}_A\quad } \frg \xlongrightarrow{\quad \dir{1\cdots n}{c}_A \quad} b\mathbb{R}^{n} \, ,$$

\item[\bf (b)] applying the isomorphism \eqref{TnAutomorphism} on the target (geometric) toroidification of twisted $\mathrm{KU}_{\mathcolor{purple}m}$, hence viewing them instead as valued in the toroidification of  twisted $\mathrm{KU}_{{{\mathcolor{purple}m+(n\, \mathrm{mod}\, 2)}}}$, while noticing that this swaps the ``product Chern class'' 
 $\,\,  \dir{1\cdots n}{c}_A$ from that classifying the $\dir{1\cdots n}{\widehat{\frg}}_A$-extension to that classifying a different extension
$$\dir{1\cdots n}{\widehat{\mathfrak{g}}}_B \longrightarrow \frg \, ,$$
i.e., via
$$\dir{1\cdots n}{c}_B \, := \,(-1)^{n+1}\,  \big({\dir{1}{p}_A}_* (H_A^3),\cdots, {\dir{n}{p}_A}_* (H_A^3)\big) \quad : \quad  \frg \longrightarrow b \mathbb{R}^n \, ,
$$
  \item[\bf (c)] re-oxidizing \eqref{nToroidificationHomIsomorphism} the result, but now along the new $n$-toroidal fibration
  $$\dir{1\cdots n}{\widehat{\mathfrak{g}}}_B \xlongrightarrow{\quad \dir{1\cdots n}{p}_B\quad } \frg \xlongrightarrow{\quad \dir{1\cdots n}{c}_B \quad} b\mathbb{R}^n \, ,$$
\end{itemize}
  results in the twisted $\mathrm{KU}_{\mathcolor{purple}m +(n \, \mathrm{mod}\, 2)}$ cocycles given precisely by
\begin{equation}
  \label{Toroidal-Red-Iso-ReoxiAction}
  \hspace{-5mm} 
  \begin{tikzcd}[
    row sep=-1pt, 
    column sep=-2pt]
    \dir{1\cdots n}{\widehat{\mathfrak{g}}}_B
    \ar[
      rr
    ]
    &&
    \mathfrak{l}
    \big(
      \,
      \Sigma^{\color{purple}n \, \mathrm{mod} \, 2}
      \mathrm{KU}
      \!\sslash\!\!B\mathrm{U}(1) \, \big)
    \\
    H^3_{\frg} + (-1)^{n+1} \sum_{k=1}^{n} \dir{k}{e}_B \cdot \dir{k}{c}_A
    &\longmapsfrom&
    h_3
    \\
 \bigg((-1)^{1+n(n-1)/2} F_{2(k+n)\, n\cdots 1}   - \dir{n}{e}_B\cdots \dir{1}{e}_B \cdot F_{2k} 
    &\longmapsfrom&
    (f_{2k+(n \, \mathrm{mod} \, 2)})_{k\in \mathbb{Z}} \, \, ,
\\[-2pt]
+ \!\!\! \sum\limits_{\substack{1\leq r < n \\ 1 \leq i_1< \cdots < i_r \leq n}} \!\!\! (-1)^{n+r+1} (-1)^{(n-r)(n-r+1)/2}\,  \dir{i_r}{e}_B\cdots \dir{i_1}{e}_B  \star F_{2(k+n)\, i_r \cdots i_1}  \bigg)_{k\in \mathbb{Z}} \hspace{1cm} && \,
  \end{tikzcd}
\end{equation}
where we have abbreviated $\star F_{2k\, i_r \cdots i_1} := \tfrac{1}{(n-r)!}
  \,
  \epsilon^{ 
      i_{n} \cdots i_{r+1}}{}_{ 
      i_r \cdots i_2 i_1 }
 F_{2k\, i_{n} \cdots i_{r+1}
    }$ (cf. \eqref{TheFiberwiseHodgeDualityOperation}).
\end{theorem}
\begin{proof}
For notational brevity, we write down the $m=0$ case. Under the reduction \eqref{nToroidificationHomBijection} from Prop. \ref{UniversalPropertyOfnToroidification} the first step yields the map of super $L_\infty$-algebras
\vspace{-2mm} 
\begin{equation*}
  \begin{tikzcd}[
    row sep=-1pt, 
    column sep=10pt]
    \mathfrak{g}
    \ar[
      rr
    ]
    &&
    \mathrm{tor}^{n'}
    \,
    \mathfrak{l}
    \big(
      \,
      \Sigma^{\color{purple}0}
      \mathrm{KU}
      \!\sslash\!\!B\mathrm{U}(1)
      \,
    \big)
    \\
    H_{\frg }^3
    &\longmapsfrom&
    h_3
    \\
    -
     {\dir{r}{p}_A}_* (H^3_A)
    &\longmapsfrom&
    \dir{r}{\mathrm{s}}h_3
 \\
   (-1)^{r(r+1)/2} F_{2k\, i_r \cdots i_1}
    &\longmapsfrom&
    \dir{i_r}{\mathrm{s}}\cdots \dir{i_1}{\mathrm{s}}f_{2k} 
    \\
    \dir{r}{c}_A
    &\longmapsfrom&
    \dir{r}{\omega}_2
    \mathrlap{\,.}
  \end{tikzcd}
\end{equation*}
In the second step, postcomposition of the above morphism with the isomorphism in 
\smallskip 
\eqref{TnAutomorphism} $$\mathrm{tor}^{n'}
    \,
    \mathfrak{l}
    \big(
      \Sigma^0
      \mathrm{KU}
      \!\sslash\!
      B\mathrm{U}(1)
    \big) \xlongrightarrow{\quad \sim \quad} \mathrm{tor}^{n'}
    \,
    \mathfrak{l}
    \big(
      \Sigma^{n\, \mathrm{mod} \, 2}
      \mathrm{KU}
      \!\sslash\!
      B\mathrm{U}(1)
    \big)
$$ 
yields
\vspace{-3mm} 
    \begin{equation*}
  \begin{tikzcd}[
    row sep=-2pt, 
    column sep=10pt]
    \mathfrak{g}
    \ar[
      rr
    ]
    &&
    \mathrm{tor}^{n'}
    \,
    \mathfrak{l}
    \big(
      \,
      \Sigma^{\color{purple}n\, \mathrm{mod}\, 2}
      \mathrm{KU}
      \!\sslash\!\!B\mathrm{U}(1)
      \,
    \big)
    \\
    H_{\frg }^3
    &\longmapsfrom&
    h_3
    \\
    (-)^n
     \dir{r}{c}_A
    &\longmapsfrom&
    \dir{r}{\mathrm{s}}h_3
 \\
   (-1)^{n+r+1}  (-1)^{r(r+1)/2} (-1)^{(n-r)(n-r+1)/2} \star F_{2(k+n)\, i_r \cdots i_1}    &\longmapsfrom&
    \dir{i_r}{\mathrm{s}}\cdots \dir{i_1}{\mathrm{s}}f_{2k+(n\, \mathrm{mod}\, 2)} 
    \\
    \dir{r}{c}_B\, := \ (-1)^{n+1} {\dir{r}{p}_A}_* (H^3_A) 
    &\longmapsfrom&
    \dir{r}{\omega}_2
    \mathrlap{\,.}
  \end{tikzcd}
\end{equation*}
Lastly, in the third step $n$-toroidally  oxidizing \eqref{nToroidificationHomBijection} via the new $n$-tuple 2-cocycle
$$\dir{1\cdots n}{c}_B \, := \,  (-1)^{n+1}\,  \big({\dir{1}{p}_A}_* (H_A^3),\cdots, {\dir{n}{p}_A}_* (H_A^3)\big) \quad : \quad  \frg \longrightarrow b \mathbb{R}^n 
$$
 yields exactly the aforementioned morphism of super $L_\infty$-algebras out of the corresponding toroidal extension
 \vspace{-1mm} 
\begin{equation*}
  \begin{tikzcd}[
       row sep=-3pt,
    column sep=15pt]
    \dir{1\cdots n}{\widehat{\mathfrak{g}}}_B
    \ar[
      rr
    ]
    &&
    \mathfrak{l}
    \big(
      \,
      \Sigma^{\color{purple}n \, \mathrm{mod} \, 2}
      \mathrm{KU}
      \!\sslash\!\!B\mathrm{U}(1) \, \big) 
      , ,
  \end{tikzcd}
\end{equation*}
upon simplification of the sign prefactors.
\end{proof}

\begin{remark}
[\bf Effect of extra isomorphisms of toroidified twisted K-Spectra]\label{RedundancyOfExtraIsomorphismsOfToroidifiedTwistedKSpectra}
In complete analogy to Rem. \ref{RedundancyOfExtraIsomorphismsOfCyclifiedTwistedKSpectra}, applying instead any of the isomorphisms from Lem. \ref{FurtherIsomorphismsOfToroidifiedTwistedKspectra} in step {\bf (b)} of Thm. \ref{TwistedKTheoryToroidalRedIsoReOxi} yields three further distinct, but closely related, bijections between \textit{rational} twisted $K$-theory cocycles of toroidal extensions over $\frg$. The non-rational lifts of these four toroidal equally admissible T-duality operations, and their relation, retain an analogous interpretation in terms of the corresponding T-dual torus bundles and T-dual fluxes (with the relevant signs now depending on the dimension of the toroidal fibers). For brevity, we do not expand on the details here.
\end{remark}

\begin{remark}[\bf Reductions along products of higher spheres] At this point, it is straightforward to develop an analogous story along the lines of \S\ref{TorusReduction} using the notion of a ``higher toroidification''. That is, there is an immediate generalization of reducing and oxidizing along products of rational odd spheres (cf. Prop. \ref{UniversalPropertyOfnToroidification}) and a corresponding automorphism of higher toroidified twisted cocycle classifying spaces (cf. Prop. \ref{TIsomorphismOf2kToroidifiedTwistedKTheory}). We shall not need this further generalization here.
\end{remark}

\newpage

\medskip

\section{Super-Flux T-Dualities}
\label{SuperspaceTDuality}

We discuss how the abstract $L_\infty$-algebraic T-duality of \S\ref{ExtCycAdj} (Lem. \ref{FurtherIsomorphismsOfCyclifiedTwistedKspectra}) is realized (\S\ref{SuperspaceTDualityI}) on the super-invariant super-flux densities intrinsically carried by the 10D type IIA/B super-spacetimes fibered over 9D super-spacetime (\S\ref{SuperMinkowskiSpacetimes}).
Then in \S\ref{LiftingToMTheory} we consider the analogous discussion of $n$-toroidal $L_\infty$-algebraic T-duality (from \S\ref{TorusReduction})
realizing it as the full $1+9$D reduction of type IIA super-space to the (super-)point, and we observe that the resulting Poincar{\'e} super 2-form on 10D-doubled super-space lifts to a 3-form on the ``M-algebra''.

\smallskip

\subsection{Super-space-times}
\label{SuperMinkowskiSpacetimes}

We discuss the translational  supersymmetry algebras (Ex. \ref{SupersymmetryAlgebras}) for $D=10$ and of ``type II'', in terms of the algebraic data provided by the $D=11$ supersymmetry algebra. This is immediate for type IIA, but for type IIB it requires a little bit of finesse. With that in hand, though, the superspace T-duality in \S\ref{SuperspaceTDualityI} flows very naturally.

\medskip

\noindent
{\bf Dimensional reduction of 11D super Minkowski spacetime.}
Consider the projection operators
\begin{equation}
  \label{ProjectorsOn16}
  \left.
  \def\arraystretch{1.2}
  \begin{array}{ccl}
  P 
  &:=&
  \tfrac{1}{2}
  \big(
    1 + \Gamma_{\!\ten}
  \big)
  \\
  \overline{P}
  &:=&
  \tfrac{1}{2}
  \big(
    1 - \Gamma_{\!\ten}
  \big)
  \end{array}
  \!\!\! \right\}
  \;\in\;
  \mathrm{End}_{{}_{\mathbb{R}}}(
    \mathbf{32}
  )
\end{equation}
satisfying the following immediate but consequential relations:
\begin{equation}
  \label{PropertiesOfProjectorOn16}
  \def\arraystretch{1.2}
  \def\arraycolsep{3pt}
  \begin{array}{ccl}
    P \, P 
      &=& 
        P,
    \\
    \overline{P} 
      \,
    \overline{P} 
      &=& 
    \overline{P},
    \\
    P \, \overline{P} 
      &=& 
    0,
    \\
    \overline{P} 
      \,
    P
      &=& 
    0,
  \end{array}
  \hspace{1.2cm} 
  \def\arraystretch{1.2}
  \def\arraycolsep{3pt}
  \begin{array}{ccccc}
    \Gamma_{\!\leq 9}
    \,
    P
    &=&
    \overline{P}
    \,
    \Gamma_{\!\leq 9},
    \\
    \Gamma_{\!\leq 9}
    \,
    \overline{P}
    &=&
    P
    \,
    \Gamma_{\!\leq 9},
    \\
   \Gamma_{\ten}
    \,
    P
    &=&
    P
    \,
    \Gamma_{\!\ten}
    &=&
    + P,
    \\
    \Gamma_{\!\ten}
    \,
    \overline{P}
    &=&
    \overline{P}
    \,
    \Gamma_{\!\ten}
    &=&
    - \overline{P},
   \end{array}
  \qquad 
  P + \overline{P}
  \,=\,
  \mathrm{id},
  \qquad 
  \def\arraystretch{1.3}
  \def\arraycolsep{3pt}
  \begin{array}{ccc}
    \overline{P \psi}
    &=&
    \overline{\psi}
    \,
    \overline{P},
    \\
    \overline{\overline{P} \psi}
    &=&
    \overline{\psi}
    \,
    P
    \mathrlap{\,.}
  \end{array}
\end{equation}
Canonically identifying actions of spin subgroups on $\mathbf{32}$
$$
  \begin{tikzcd}
    \mathrm{Spin}(1,8)    
    \ar[r, hook]
    &
    \mathrm{Spin}(1,9)    
    \ar[r, hook]
    &
    \mathrm{Spin}(1,10)
  \end{tikzcd}
$$
by restriction of the Clifford algebra to products of its first $1 + d$ generators $\Gamma_0, \, \Gamma_1, \cdots \Gamma_d$, the projectors \eqref{ProjectorsOn16}
carve out two $\mathrm{Spin}(1,9)$-representations from the $\mathrm{Spin}(1,10)$-rep $\mathbf{32}$:
\begin{equation}
  \label{16Rep}
  \begin{tikzcd}[row sep=-2pt, column sep=0pt]
  \mathbf{16}
  &:=&
  P(\mathbf{32})
  &\in&
  \mathrm{Rep}_{{}_{\mathbb{R}}}\big(
    \mathrm{Spin}(1,9)
  \big)
  \\
  \overline{\mathbf{16}}
  &:=&
  \overline{P}(\mathbf{32})
  &\in&
  \mathrm{Rep}_{{}_{\mathbb{R}}}\big(
    \mathrm{Spin}(1,9)
  \big)
  \,,
  \end{tikzcd}
\end{equation}
and hence we have a branching of representations of this form:
\begin{equation}
  \label{32BranchingTo16Bar16}
  \begin{tikzcd}[row sep=-3pt,
    column sep=0pt
  ]
    \mathrm{Spin}(1,10)
    &&
    \mathrm{Spin}(1,9)
    \ar[
      ll,
      hook'
    ]
    \\
    \mathbf{32}
    \ar[rr]
    &&
    \mathbf{16}
    \oplus 
    \overline{\mathbf{16}}
    \\
    \psi
    &\longmapsto&
    \underbrace{
    P \psi
    }_{\color{gray} \psi_1 }
    +
    \underbrace{
      \overline{P}
      \psi
    }_{\color{gray} \psi_2 }.
  \end{tikzcd}
\end{equation}

Under this branching and decomposition ($\psi = \psi_1 + \psi_2 = P(\psi_1) + \overline{P}(\psi_2)$), the spinor pairing $\big(\overline{(-)}(-)\big)$ on $\mathbf{32}$ 
translates into pairings on $\mathbf{16} = P(\mathbf{32})$ and on $\overline{\mathbf{16}} = \overline{P}(\mathbf{32})$, by evaluating on pairs of  spinors belonging in each of the projected subspaces respectively. In terms of the two projected representations the vector-valued spinor pairing $\big(\overline{(-)}\Gamma(-)\big)$ decomposes as follows:
\begin{lemma}[\bf Decomposed vectorial spinor pairing]
\begin{equation}
  \label{DecomposedVectorialSpinorPairing}
  \def\arraystretch{1.4}
  \begin{array}{ccc}
  \big(\hspace{1pt}
    \overline{\psi}
    \,\Gamma_a\,
    \phi
  \big)
  &=&
  \big(\hspace{1pt}
    \overline{\psi}_1
    \,\Gamma_a\,
    \phi_1
  \big)  
  +
  \big(\hspace{1pt}
    \overline{\psi}_2
    \,\Gamma_a\,
    \phi_2
  \big)  
  \mathrlap{
  \,,
  \hspace{1cm}
  \mbox{
    \rm
    for 
    $a \neq \ten$,
  }
  }
  \\
    \big(\hspace{1pt}
      \overline{\psi}
      \,\Gamma_{\!\ten}\,
      \phi
    \big)
    &=&
    \big(\hspace{1pt}
      \overline{\psi}_2
      \,
      \phi_1
    \big)
    \,-\,
    \big(\hspace{1pt}
      \overline{\psi}_1
      \,
      \phi_2
    \big).
  \end{array}
\end{equation}
\end{lemma}
\begin{proof}
The first line in \eqref{DecomposedVectorialSpinorPairing} follows
since the mixed terms vanish due to the decomposition \eqref{32BranchingTo16Bar16}
and using the relations
\eqref{PropertiesOfProjectorOn16}:
$$
  \def\arraystretch{1.3}
  \begin{array}{ccl}
    \big(\hspace{1pt}
      \overline{\psi_1}
      \,\Gamma_a\,
      \phi_2
    \big)
    &=&
    \big(\hspace{1pt}
      \overline{P\psi_1}
      \,\Gamma_a\,
      \overline{P}\phi_2
    \big)
    \\
    &=&
    \big(\hspace{1pt}
      \overline{\psi_1}
      \,
      \overline{P}
      \,\Gamma_a\,
      \overline{P}\phi_2
    \big)
    \\
    &=&
    \big(\hspace{1pt}
      \overline{\psi_1}
      \,\Gamma_a\,
      P
      \overline{P}\phi_2
    \big)
    \\
    &=&
    0\, ,
  \end{array}
  \hspace{1.5cm}
  \def\arraystretch{1.3}
  \begin{array}{ccl}
    \big(\hspace{1pt}
      \overline{\psi_2}
      \,\Gamma_a\,
      \phi_1
    \big)
    &=&
    \big(\hspace{1pt}
      \overline{\overline{P}\psi_2}
      \,\Gamma_a\,
      P\phi_1
    \big)
    \\
    &=&
    \big(\hspace{1pt}
      \overline{\psi_1}
      \,
      P
      \,\Gamma_a\,
      P\phi_2
    \big)
    \\
    &=&
    \big(\hspace{1pt}
      \overline{\psi_1}
      \,\Gamma_a\,
      \overline{P}
      P\phi_2
    \big)
    \\
    &=&
    0\,,
  \end{array}
  \;\;\;\;\;\;\;\;\;\;\;\;\;\;
  \mbox{for $a \neq \ten$.}
$$
Similarly but complementarily, for the second line these relations give:
$$
  \def\arraystretch{1.4}
  \begin{array}{lll}
    \big(\hspace{1pt}
      \overline{\psi}
      \,\Gamma_{\!\ten}\,
      \phi
    \big)
    &=&
    \big(\hspace{1pt}
      \overline{\psi}_1
      \,\Gamma_{\!\ten}\,
      \phi_2
    \big)
    +
    \big(\hspace{1pt}
      \overline{\psi}_2
      \,\Gamma_{\!\ten}\,
      \phi_1
    \big)
    +
    \big(\hspace{1pt}
      \overline{\psi}_1
      \,\Gamma_{\!\ten}\,
      \phi_1
    \big)
    +
    \big(\hspace{1pt}
      \overline{\psi}_2
      \,\Gamma_{\!\ten}\,
      \phi_2
    \big)
    \\
    &=&
    \big(\hspace{1pt}
      \overline{P \psi_1}
      \,\Gamma_{\!\ten}\,
      \overline{P}\phi_2
    \big)
    +
    \big(\hspace{1pt}
      \overline{\overline{P}\psi_2}
      \,\Gamma_{\!\ten}\,
      P\phi_1
    \big)
    +
    \big(\hspace{1pt}
      \overline{P \psi_1}
      \,\Gamma_{\!\ten}\,
      P\phi_1
    \big)
    +
    \big(\hspace{1pt}
      \overline{\overline{P}\psi_2}
      \,\Gamma_{\!\ten}\,
      \overline{P}\phi_2
    \big)
    \\
    &=&
    \big(\hspace{1pt}
      \overline{\psi_1}
      \,\overline{P}\Gamma_{\!\ten}\,
      \overline{P}\phi_2
    \big)
    +
    \big(\hspace{1pt}
      \overline{\psi_2}
      \,P\Gamma_{\!\ten}\,
      P\phi_1
    \big)
    +
    \big(\hspace{1pt}
      \overline{\psi_1}
      \,\overline{P}\Gamma_{\!\ten}\,
      P\phi_1
    \big)
    +
    \big(\hspace{1pt}
      \overline{\psi_2}
      \,P\Gamma_{\!\ten}\,
      \overline{P}\phi_2
    \big)
    \\
    &=&
    -
    \big(\hspace{1pt}
      \overline{\psi_1}
      \,\overline{P}
      \, \overline{P}\phi_2
    \big)
    +
    \big(\hspace{1pt}
      \overline{\psi_2}
      \,P P\phi_1
    \big)
    +
    \big(\hspace{1pt}
      \overline{\psi_1}
      \,\overline{P}
      P\phi_1
    \big)
    +
    \big(\hspace{1pt}
      \overline{\psi_2}
      \,P\overline{P}\phi_2
    \big)
    \\
    &=&
    -
    \big(\hspace{1pt}
      \overline{\psi_1}
      \, 
      \phi_2
    \big)
    +
    \big(\hspace{1pt}
      \overline{\psi_2}
      \,
      \phi_1
    \big).
  \end{array}
$$

\vspace{-5mm} 
\end{proof}

This gives:
\begin{example}[\bf The M/IIA super-space extension]
  \label{11dAsExtensionFromIIA}
  We have 
  a $\mathrm{Spin}(1,9)$-equivariant 
  isomorphism
  \begin{equation}
    \label{CEOfIIASpacetime}
    \mathrm{CE}\big(
      \mathbb{R}^{
        1,10\,\vert\,
        \mathbf{32}
      }
    \big)
    \;\;
    \simeq
    \;\;
    \FDGCA\left[\!\!\!
      \def\arraystretch{1.5}
      \begin{array}{c}
      \big(
        \psi_1^\alpha
      \big)_{\alpha=1}^{16}
      ,\,
      \\
      \big(
        \psi_2^\alpha
      \big)_{\alpha=1}^{16}
      \,,
      \\
      \big(
        e^a
      \big)_{a = 0}^{9}
      ,\,
      \\
      e^{\ten}
      \end{array}
   \!\! \right]
    \Big/
    \left(\!\!
    \def\arraystretch{1.2}
    \def\arraycolsep{3pt}
    \begin{array}{clcl}
      \mathrm{d}\, \psi_1
      &=& 0
      \\
      \mathrm{d}\, \psi_2
      &=& 0
      \\
      \mathrm{d}\, e^a
      &=&
      \big(\hspace{1pt}
        \overline{\psi}{}_1
        \,\Gamma_a\,
        \psi_1
      \big)
      +
      \big(\hspace{1pt}
        \overline{\psi}{}_2
        \,\Gamma_a\,
        \psi_2
      \big)
      &
      \mbox{
        \rm
        for $a \neq \ten$
      }
      \\
      \mathrm{d}\,
      e^{\ten}
      &=&
      \smash{
      \grayunderbrace{
      \big(\hspace{1pt}
        \overline{\psi}{}_2
        \,\psi_1
      \big)
      -
      \big(\hspace{1pt}
        \overline{\psi}{}_1
        \,\psi_2
      \big)
      }{
        \big(\hspace{1pt}
        \overline{\psi}
        \,\Gamma_{\!\ten}\,
        \psi
        \big)
      }
      }
    \end{array}
    \!\!\!\right)
    \,,
  \end{equation}
  \vspace{.5cm}

  \noindent
  which exhibits $\mathbb{R}^{1,10\,\vert\,\mathbf{32}}$ 
  as a central extension 
  {\rm (Def. \ref{CentralExtension})}
  of $\mathbb{R}^{1,9\,\vert\,\mathbf{16} \oplus \overline{\mathbf{16}}}$:

\begin{equation}
  \label{11dSuperspaceAsExtensionOfIIASuperspace}
  \begin{tikzcd}[
    row sep=-4pt, column sep=large
  ]
    \mathbb{R}^{
      1,10
        \vert
      \mathbf{32}
    }
    \quad 
    \ar[
      rr,
      ->>,
      "{ \mathrm{hofib} }",
      "{
        \scalebox{.7}{
          \color{darkgreen}
          \bf
          extension of
        }
      }"{swap, yshift=-2pt}
    ]
    &&
    \;\;
    \mathbb{R}^{
      1,9\,\vert\,
      \mathbf{16}
      \oplus
      \overline{\mathbf{16}}
    }
    \quad 
    \ar[
      rrr,
      "{
        c_1^M
        \;:=\;
        (\hspace{1pt}
          \overline{\psi}
          \,\Gamma_{10}\,
        \psi) 
      }",
      "{
        \scalebox{.7}{
          \color{darkgreen}
          \bf
          classified by first Chern class
        }
      }"{swap,yshift=-2pt}
    ]
    &\phantom{---}&&
    b \mathbb{R}
    \,.
    \\
    \mathclap{
      \scalebox{.7}{
        \color{darkblue}
        \bf
        \def\arraystretch{.9}
        \begin{tabular}{c}
          $\mathbf{D=11}$
          \\
          super-Minkowski
          \\
          spacetime
        \end{tabular}
      }
    }
    \quad 
    &&
    \mathclap{
      \scalebox{.7}{
        \color{darkblue}
        \bf
        \def\arraystretch{.9}
        \begin{tabular}{c}
          $\mathbf{D=10}$ type IIA
          \\
          super-Minkowski
          \\
          spacetime
        \end{tabular}
      }
    }
  \end{tikzcd}
\end{equation}
\end{example}

Reducing one dimension further,
the two $\mathrm{Spin}(1,9)$-reps \eqref{16Rep} in turn become isomorphic when restricted to $\mathrm{Spin}(1,8)$-representations, the isomorphism given by acting with $\Gamma_9$:
\begin{equation}
  \label{IsoBetween16AndBar16}
  \begin{tikzcd}[row sep=small, 
    column sep=2pt
  ]
    \mathbf{16}
    &\defneq&
    P(\mathbf{32})
    \ar[
      out=180-65, 
      in=65, 
      looseness=3, 
      "\scalebox{.7}{
        $\mathrm{Spin}(1,8)$
      }"
    ]    
    \ar[
      rr,
      "{
        \Gamma_9
      }",
      "{ \sim }"{swap}
    ]
    \ar[
      dd,
      "{
        \Gamma_{ab}
      }"{swap}
    ]
    &\phantom{--------}&
    \overline{P}(\mathbf{32})
    \ar[
      out=180-65, 
      in=65, 
      looseness=3, 
      "\scalebox{.7}{
        $\mathrm{Spin}(1,8)$
      }"
    ]    
    \ar[
      dd,
      "{
        \Gamma_{ab}
      }"{swap}
    ]
    &\defneq&
    \overline{\mathbf{16}}
    \\
    \\
    &&
    P(\mathbf{32})
    \ar[
      rr,
      "{ 
         \Gamma_9 
      }",
      "{ \sim }"{swap}
    ]
    &&
    \overline{P}(\mathbf{32})
  \end{tikzcd}
  \;\;\;
  \;\;\; a,b \leq 8.
\end{equation}

This gives:

\begin{example}[\bf The IIA/9D super-space extension]
\label{IIAasExtensionOf9d}
We have a $\mathrm{Spin}(1,8)$-equivariant isomorphism
$$
  \mathrm{CE}\big(
    \mathbb{R}^{
      1,9\,\vert\,
      \mathbf{16}
      \oplus
      \overline{\mathbf{16}}
    }
  \big)
  \;\simeq\;
 \FDGCA
  \left[
  \def\arraystretch{1.4}
  \def\arraycolsep{0pt}
  \begin{array}{c}
    (\psi_1^\alpha)_{\alpha=1}^{16},
    \\
    (\psi_2^\alpha)_{\alpha=1}^{16},
    \\
    (e^a)_{a=0}^{8},
    \\
    e^9
  \end{array}
  \right]
  \Big/
  \left(
  \def\arraystretch{1.4}
  \def\arraycolsep{2pt}
  \begin{array}{lccl}
    \mathrm{d}\,\psi_1
    &=&
    0
    \\
    \mathrm{d}\,\psi_2
    &=&
    0
    \\
    \mathrm{d}\, e^a
    &=&
    \big(\hspace{1pt}
      \overline{\psi}_1
      \,\Gamma^a\,
      \psi_1
    \big)
    +
    \big(\hspace{1pt}
      \overline{\psi}_2
      \,\Gamma^a\,
      \psi_2
    \big)
    &
    \;\;
    \mbox{for $a \neq 9$}
    \\
    \mathrm{d}
    \,
    e^{9}
    &=&
    \grayunderbrace{
    (\hspace{1pt}
      \overline{\psi}_1
      \,\Gamma^9\,
      \psi_1
    )
    +
    \big(\hspace{1pt}
      \overline{\psi}_2
      \,\Gamma^9\,
      \psi_2
    \big)    
    }{
      \big(\hspace{1pt}
        \overline{\psi}
        \,\Gamma^9\,
        \psi
      \big)
    }
  \end{array}
  \!\right)
$$
which exhibits $\mathbb{R}^{1,9\,\vert\,\mathbf{16}\oplus \overline{\mathbf{16}}}$ as a central extension {\rm (Def. \ref{CentralExtension})} of $\mathbb{R}^{1,8\,\vert\,\mathbf{16}\oplus \mathbf{16}}$ by a 2-cocycle to be denoted $c_1^A$ (cf. footnote \ref{OnFirstChernClasses}):
\vspace{-2mm} 
\begin{equation}
  \label{IIAExtension}
  \begin{tikzcd}[
    row sep=-2pt, column sep=large
  ]
  \mathbb{R}^{
    1,9\,\vert\,
    \mathbf{16}
    \oplus
    \overline{\mathbf{16}}
  }
  \quad 
  \ar[
    rr,
    ->>,
    "{
      \mathrm{hofib}
    }",
    "{
      \scalebox{.7}{
        \color{darkgreen}
        \bf
        extension of
      }
    }"{swap, yshift=-2pt}
  ]
  &&
  \;\;
  \mathbb{R}^{
    1,8\,\vert\,
    \mathbf{16}
      \oplus
    \mathbf{16}
  }
  \;\;
  \ar[
    rrr,
    "{
    c_1^A
    \,:=\,
    (\hspace{1pt}
      \overline{\psi}
      \,\Gamma_9\,
      \psi
    )
    }",
    "{
      \scalebox{.7}{
       \color{darkgreen}
       \bf
       classified by first Chern class
      }
    }"{swap, yshift=-2pt}
  ]
  &{\phantom{----}}&&
  b\mathbb{R}
  \mathrlap{\;.}
  \\
  \mathclap{
    \scalebox{.7}{
      \color{darkblue}
      \bf
      \def\arraystretch{.9}
      \begin{tabular}{c}
        $\mathbf{D=10}$ type IIA 
        \\
        super-Minkowski
        \\
        spacetime
      \end{tabular}
    }
  }
  \quad 
  &&
  \mathclap{
    \scalebox{.7}{
      \color{darkblue}
      \bf
      \def\arraystretch{.9}
      \begin{tabular}{c}
        $\mathbf{D=9}$ type II 
        \\
        super-Minkowski
        \\
        spacetime
      \end{tabular}
    }
  }
  \end{tikzcd}
\end{equation}
\end{example}

\medskip

\noindent
{\bf The type IIB super-spacetime.}
We need a presentation of the type IIB super-spacetime analogous to the IIA case above, namely expressed in terms of the 11D spinors in $\mathbf{32}$ and their spinor pairing $\big(\hspace{1pt}\overline{(-)}(-)\big) : \mathbf{32} \otimes \mathbf{32} \xrightarrow{\;} \mathbb{R}$.
However, since along $\mathfrak{so}_{1,9} \xhookrightarrow{\;} \mathfrak{so}_{1,10}$ this representation branches as $\mathbf{32} \mapsto \mathbf{16} \oplus \overline{\mathbf{16}}$ 
\eqref{32BranchingTo16Bar16}
we have to
re-express the given $\overline{\mathbf{16}}$ as a $\mathbf{16}$.

\smallskip 
Observe that this may be achieved by ``compensating with a group automorphism'': The same diagram 
\eqref{IsoBetween16AndBar16} which shows how $P(\mathbf{32})$ and $\overline{P}(\mathbf{32})$ are isomorphic as $\mathrm{Spin}(1,8)$-representations also shows which transformation on the Lie algebra generators occurs when comparing them as $\mathrm{Spin}(1,9)$-representations, namely the Lorentz-generators with an index=9 pick up a minus sign:
\begin{equation}
  \label{BrokenIsoBetween16AndBar16}
  \begin{tikzcd}[row sep=small, 
    column sep=2pt
  ]
    \mathbf{16}
    &\defneq&
    P(\mathbf{32})
    \ar[
      out=180-65, 
      in=65, 
      looseness=3, 
      "\scalebox{.7}{
        $\mathrm{Spin}(1,8)$
      }"
    ]    
    \ar[
      rr,
      "{
        \Gamma_9
      }",
      "{ \sim }"{swap}
    ]
    \ar[
      dd,
      shift right=5pt,
      "{
        \Gamma_{ab}
      }"{swap}
    ]
    \ar[
      dd,
      shift left=5pt,
      "{
        \Gamma_{a9}
      }"
    ]
    &\phantom{--------}&
    \overline{P}(\mathbf{32})
    \ar[
      out=180-65, 
      in=65, 
      looseness=3, 
      "\scalebox{.7}{
        $\mathrm{Spin}(1,8)$
      }"
    ]    
    \ar[
      dd,
      shift right=5pt,
      "{
        \Gamma_{ab}
      }"{swap}
    ]
    \ar[
      dd,
      shift left=5pt,
      "{
        {\color{purple}-}
        \Gamma_{a9}
      }"
    ]
    &\defneq&
    \overline{\mathbf{16}}
    \\
    \\
    &&
    P(\mathbf{32})
    \ar[
      rr,
      "{ 
         \Gamma_9 
      }",
      "{ \sim }"{swap}
    ]
    &&
    \overline{P}(\mathbf{32})
    \mathrlap{\,.}
  \end{tikzcd}
  \;\;\;
  \;\;\; a,b \leq 8
\end{equation}
This means that the $\mathbf{16}$ of $\mathrm{Spin}(1,9)$ is the pullback of $\overline{\mathbf{16}}$ along the group homomorphism which on Lie algebras is given by this sign change
\eqref{IsoBetween16AndBar16}:

\begin{equation}
  \label{CompensatingIso}
  \begin{tikzcd}[
    row sep=-2pt
  ]
    \mathfrak{so}_{1,9}
    \ar[
      rr,
      "{ \sim }"
    ]
    \ar[
      rrrr,
      rounded corners,
      to path={ 
            ([yshift=+00pt]\tikztostart.north)
        --  ([yshift=+12pt]\tikztostart.north)
        --  node[yshift=5pt]{
             \scalebox{.7}{$
               \mathbf{16}
             $}
        }
            ([yshift=+10pt]\tikztotarget.north)
        --  ([yshift=+00pt]\tikztotarget.north)
      }
    ]
    &&
    \mathfrak{so}_{1,9}
    \ar[
      rr,
      "{
        \overline{\mathbf{16}}
      }"
    ]
    &&
    \mathfrak{gl}_{16}
    \\
    J_{a b < 9}
    &\longmapsto&
    J_{a b}
    &\longmapsto&
    \tfrac{1}{2}\Gamma_{a b}\vert_{\overline{P}(\mathbf{32})}
    \\
    J_{a 9}
    &\longmapsto&
    -J_{a 9}
    &\longmapsto&
    -
    \tfrac{1}{2}\Gamma_{a 9}\vert_{\overline{P}(\mathbf{32})}
    \\
    J_{a b}
    \ar[
      rrrr,
      |->,
      shorten=20pt
    ]
    &&&&
    \tfrac{1}{2}\Gamma^B_{a b}\vert_{
      \overline{P}(\mathbf{32})
    }
    \mathrlap{\,.}
  \end{tikzcd}
  \hspace{.9cm}
  (a,b < 9)
\end{equation}
Here in the last line we have summarized this situation by introducing the following notation -- recalling that our undecorated ``$\Gamma_a$'' are always those of the 11D Clifford algebra \eqref{The11dMajoranaRepresentation}:
\begin{equation}
  \label{TheTypeIIBQuasiCliffordGenerators}
  \def\arraystretch{1.2}
  \begin{array}{ll}
  \Gamma^B_{a}
  \;:=\;
  \Gamma_a
  &
  \mbox{for $a < 9$}
  \\
  \Gamma_9^B
  \;:=\;
  \Gamma_9\Gamma_{\!\ten}
  \\
  \Gamma_{a b}^B
  \;:=\;
  \Gamma_{a}^B \Gamma_b^B
  &
  \mbox{for
    $a < b \leq 9$
  }
  \\
  \Gamma_{b a}^B
  \;:=\;
  -
  \Gamma_{a}^B \Gamma_b^B
  &
  \mbox{for
    $a < b \leq 9$
    \rlap{,}
  }
  \end{array}
\end{equation}
which works since $\Gamma_{\!\ten}\vert_{\overline{P}(\mathbf{32})} = -\mathrm{id}$, by \eqref{PropertiesOfProjectorOn16}. But by the same relation also $\Gamma_{\!\ten}\vert_{P(\mathbf{32})} = +\mathrm{id}$, so that the $\Gamma^B$ operators reduce to the original Clifford generators $\Gamma$ \eqref{The11dMajoranaRepresentation} when restricted on $\mathbf{16}\equiv P(\mathbf{32})$ and hence encoding also a copy of the original representation
\begin{equation*}
   \begin{tikzcd}[
     row sep=0pt,
     column sep=20pt
   ]
     \mathfrak{so}_{1,9}
    \ar[
      rr,
      "{ \mathbf{16} }"
    ]
    &&
    \mathfrak{gl}_{16}
    \\
    J_{ab}&\longmapsto& \frac{1}{2}\Gamma^{B}_{ab}|_{P(\mathbf{32})} \equiv \frac{1}{2}\Gamma_{ab}|_{P(\mathbf{32})}
   \end{tikzcd}
\end{equation*}

In total we produced the type IIB spinor representation $\mathbf{16} \oplus \mathbf{16}$ of $\mathrm{Spin}(1,9)$ as a pullback of the type IIA spinor representation $\mathbf{16} \oplus \overline{\mathbf{16}}$, in terms of the 11D Clifford algebra expression \eqref{TheTypeIIBQuasiCliffordGenerators} as:
\begin{equation}
  \label{IIBSpinorsAsPullbackOfIIASpinors}
  \begin{tikzcd}[
    row sep=0pt,
    column sep=30pt
  ]
    \mathfrak{so}_{1,9}
    \ar[
      rr,
      "{ \sim }"
    ]
    \ar[
      rrrr,
      rounded corners,
      to path={ 
            ([yshift=+00pt]\tikztostart.north)
         -- ([yshift=+11.5pt]\tikztostart.north)
         -- node[yshift=5pt] {
            \scalebox{.7}{$
             \mathbf{16}
             \oplus
             \mathbf{16}
             $}
         }
           ([yshift=+10pt]\tikztotarget.north)
         -- ([yshift=+00pt]\tikztotarget.north)
      }
    ]
    &&
    \mathfrak{so}_{1,9}
    \ar[
      rr,
      "{
        \mathbf{16}
        \oplus
        \overline{\mathbf{16}}
      }"
    ]
    &&
    \mathfrak{gl}_{32}
    \\
    J_{ab}
    \ar[
      rrrr,
      |->,
      shorten=10pt
    ]
    &&
    &&
    \tfrac{1}{2}
    \Gamma_{a b}^B
  \end{tikzcd}
  \hspace{1cm}
  (a < b \leq 9)
\end{equation}
It is this transformation which turns out to make manifest superspace T-duality below. 

\begin{remark}[\bf Subtleties with type IIB Clifford elements.]
\label{SubtletiesWithIIBCliffordElements}
$\,$

\noindent {\bf (i)} 
Beware that the operators $\Gamma^B$ in \eqref{TheTypeIIBQuasiCliffordGenerators} do not generate a Clifford algebra (and not a $\mathrm{Pin}(1,9)$-group), but they do generate the correct $\mathrm{Spin}(1,9)$-group and -representation.

\noindent {\bf (ii)} As we will see in a moment, this defect is in a sense compensated by another defect, namely that $\Gamma^B_9$ is not skew-self-adjoint as the other Clifford generators \eqref{SkewSelfAdjointnessOfCliffordGenerators}:
\begin{equation}
  \label{AdjointnessForTypeIIB}
  \overline{\Gamma^B_{ab}}
  \;=\;
  \left\{\!\!\!
  \def\arraystretch{1.3}
  \begin{array}{ll}
    -\Gamma^B_{ab}
    &
    \mbox{for $a,b < 9$}
    \\
    + \Gamma^B_{ab}
    &
    \mbox{if $b = 9$}
    \mathrlap{\,,}
  \end{array}
  \right.
\end{equation}
where the second line comes about as 
$$
  \overline{
    \Gamma^B_{a 9}
  }
  \;\defneq\;
  \overline{
    \Gamma_a
    \Gamma_9
    \Gamma_{\!\ten}
  }
  \;=\;
    \overline{\Gamma_{\!\ten}}
    \,
    \overline{\Gamma_9}
    \,
    \overline{\Gamma_a}
  \;=\;
  (-1)^3
    \Gamma_{\!\ten}
    \Gamma_9
    \Gamma_a
  \;=\;
  (-1)^{3+1}
    \Gamma_a
    \Gamma_9
    \Gamma_{\!\ten}
  \;=\;
  \Gamma^B_{a9}
  \hspace{.4cm}
  \mbox{for $a < 9$}
  \,.
$$
\end{remark}

With this notation, we may set:

\begin{definition}[\bf Type IIB super-Minkowski super-Lie algebra.]
\label{TypeIIBSuperMinkowski}
The 10d type IIB super-Minkowski Lie algebra is given by
\begin{equation}
  \label{TheTypeIIBSuperMinkowskiAlgebra}
  \mathrm{CE}\big(
    \mathbb{R}
      ^{1,9 \,\vert\, \mathbf{16}\oplus\mathbf{16} }
  \big)
  \;\;
  =
  \;\;
 \FDGCA
  \left[\!\!\!
  \def\arraystretch{1.3}
  \begin{array}{l}
    (\psi_1^\alpha)_{\alpha=1}^{16},
    \\
    (\psi_2^\alpha)_{\alpha=1}^{16},
    \\
    (e^a)_{a=0}^9
  \end{array}
  \!\!\right]
  \Big/
  \left(\!\!
  \begin{array}{ccl}
    \mathrm{d}
    \, \psi_1
    &=&
    0
    \\
    \mathrm{d}
    \, \psi_2
    &=&
    0
    \\
    \mathrm{d}
    \, e^a
    &=&
    \big(\hspace{1pt}
      \overline{\psi}
      \,\Gamma_{\!\!B}^a\,
      \psi
    \big)
    \,,
    \end{array}
  \!\!\!\right)
\end{equation}
where the pairing is that of spinors in the $\mathbf{32}$ of 11D (!) under the identification \eqref{32BranchingTo16Bar16} and where, to compensate this, $\Gamma^a_B$ is from \eqref{TheTypeIIBQuasiCliffordGenerators}.
\end{definition}
Since it is evident that the differential in \eqref{TheTypeIIBSuperMinkowskiAlgebra} is at least $\mathrm{Spin}(1,8)$-equivariant,  we have the following analog of
Ex. \ref{IIAasExtensionOf9d}:
\begin{example}[\bf The IIB/9D super-space extension]
\label{IIBasExtensionOf9d}
We have a $\mathrm{Spin}(1,8)$-equivariant isomorphism
$$
  \mathrm{CE}\big(
    \mathbb{R}^{
      1,9\,\vert\,
      \mathbf{16}
      \oplus
      \mathbf{16}
    }
  \big)
  \;\simeq\;
  \FDGCA
  \left[
  \def\arraystretch{1.4}
  \def\arraycolsep{0pt}
  \begin{array}{c}
    (\psi_1^\alpha)_{\alpha=1}^{16},
    \\
    (\psi_2^\alpha)_{\alpha=1}^{16},
    \\
    (e^a)_{a=0}^{8},
    \\
    e^9
  \end{array}
  \right]
  \Big/
  \left(
  \def\arraystretch{1.2}
  \def\arraycolsep{2pt}
  \begin{array}{lccl}
    \mathrm{d}\,\psi_1
    &=&
    0
    \\
    \mathrm{d}\,\psi_2
    &=&
    0
    \\
    \mathrm{d}\, e^a
    &=&
    \big(\hspace{1pt}
      \overline{\psi}_1
      \,\Gamma^a\,
      \psi_1
    \big)
    +
    \big(\hspace{1pt}
      \overline{\psi}_2
      \,\Gamma^a\,
      \psi_2
    \big)
    &
    \;\;
    \mbox{for $a < 9$}
    \\
    \mathrm{d}
    \,
    e^{9}
    &=&
    \grayunderbrace{
    \big(\hspace{1pt}
      \overline{\psi}_1
      \,\Gamma^9\,
      \psi_1
    \big)
    -
    \big(\hspace{1pt}
      \overline{\psi}_2
      \,\Gamma^9\,
      \psi_2
    \big)
    }{
      (\hspace{1pt}
        \overline{\psi}
          \,\Gamma^9_B\,
        \psi
      )
      \,=\,
      (\hspace{1pt}
        \overline{\psi}
          \,\Gamma^9\Gamma_{\!\ten}\,
        \psi
      )
    }
  \end{array}
  \!\right)
$$
which exhibits $\mathbb{R}^{1,9\,\vert\,\mathbf{16}\oplus \mathbf{16}}$ as a central extension {\rm (Def. \ref{CentralExtension})} of $\mathbb{R}^{1,8\,\vert\,\mathbf{16}\oplus \mathbf{16}}$ classified by a 2-cocycle to be denoted $c_1^B$ (cf. footnote \ref{OnFirstChernClasses}):
\begin{equation}
  \label{IIBExtension}
  \begin{tikzcd}[
    row sep=-4pt,
    column sep=large
  ]
  \mathbb{R}^{
    1,9\,\vert\,
    \mathbf{16}
    \oplus
    \mathbf{16}
  }
  \quad 
  \ar[
    rr,
    ->>,
    "{
      \mathrm{hofib}
    }",
    "{
      \scalebox{.7}{
        \color{darkgreen}
        \bf
        extension of
      }
    }"{swap, yshift=-2pt}
  ]
  &&
  \;\;
  \mathbb{R}^{
    1,8\,\vert\,
    \mathbf{16}
      \oplus
    \mathbf{16}
  }
  \;\;
  \ar[
    rr,
    "{
      c_1^B
      \,:=\,
      (\hspace{1pt}
        \overline{\psi}
        \,\Gamma^9_B\,
        \psi
      )
      \,=\,
      (\hspace{1pt}
        \overline{\psi}
        \,\Gamma^9\Gamma_{\!\ten}\,
        \psi
      )
    }",
    "{
      \scalebox{.7}{
        \color{darkgreen}
        \bf
        classified by first Chern class
      }
    }"{swap, yshift=-2pt}
  ]
  &{\phantom{------}}&
  b\mathbb{R}
  \mathrlap{\,.}
  \\
  \mathclap{
    \scalebox{.7}{
      \color{darkblue}
      \bf
      \def\arraystretch{.9}
      \begin{tabular}{c}
        $\mathbf{D}=10$
        type IIB
        \\
        super-Minkowski
        \\
        spacetime
      \end{tabular}
    }
  }
  \quad 
  &&
  \mathclap{
    \scalebox{.7}{
      \color{darkblue}
      \bf
      \def\arraystretch{.9}
      \begin{tabular}{c}
        $\mathbf{D}=9$
        type II
        \\
        super-Minkowski
        \\
        spacetime
      \end{tabular}
    }
  }
  \end{tikzcd}
\end{equation}
\end{example}

What is less evident is that \eqref{TheTypeIIBSuperMinkowskiAlgebra} is also $\mathrm{Spin}(1,9)$-equivariant under the action \eqref{IIBSpinorsAsPullbackOfIIASpinors},
since the $\big\{\Gamma^a_B\big\}_{a=0}^9 \subset \mathrm{Cl}(1,10)$ by themselves do not generate a Clifford sub-algebra, by Rem. \ref{SubtletiesWithIIBCliffordElements}.
However, the failure of $\Gamma^9_B$ to be skew-self-adjoint \eqref{AdjointnessForTypeIIB} compensates this defect, as follows:
\begin{lemma}[\bf Lorentz-equivariance of the type IIB spacetime]
\label{LorentzEquivarianceOfIIBSpacetime}
  The differential in \eqref{TheTypeIIBSuperMinkowskiAlgebra} is indeed equivariant 
  \eqref{VectorialSpinorPairing}
  under the $\mathrm{Spin}(1,9)$-action \eqref{IIBSpinorsAsPullbackOfIIASpinors}.
\end{lemma}
\noindent
Checking this is straightforward, but we spell out the proof because this statement was omitted in \cite{FSS18-TDualityA}:
\begin{proof}
For a Lie action of 
$\mathfrak{so}_{1,10}$ on $\psi$ by 
$$
  J^{ab} \psi 
  \,=\,
  -
  \tfrac{1}{2}\Gamma_{\!\!B}^{ab}\psi
$$
we need to check that the term $\big(\hspace{1pt}\overline{\psi}\,\Gamma_{\!\!B}\,\psi\big)$ transforms in the vector representation, namely that
$$
  J^{ab} 
  \big(\hspace{1pt}
    \overline{\psi}
    \,\Gamma^c_{\!\!B}\,
    \psi
  \big)
  \,\defneq\,
  \big(\hspace{1pt}
    \overline{J^{a b}\psi}
    \,\Gamma^c_{\!\!B}\,
    \psi
  \big)
  +
  \big(\hspace{1pt}
    \overline{\psi}
    \,\Gamma^c_{\!\!B}\,
    J^{ab}\psi
  \big)
  \;=\;
  \eta^{bc} 
  \big(\hspace{1pt}
    \overline{\psi}
    \,\Gamma^a_{\!\!B}\,
    \psi
  \big)
  -
  \eta^{a c}
  \big(\hspace{1pt}
    \overline{\psi}
    \,\Gamma^b_{\!\!B}\,
    \psi
  \big)
  \,.
$$
In the case where $a,b,c < 9$ this is, by \eqref{TheTypeIIBQuasiCliffordGenerators}, just the ordinary case which, just to recall, works out as usual:
$$
    J^{ab}
    \big(\hspace{1pt}
      \overline{\psi}
      \,\Gamma^c_B\,
      \psi
    \big)
    \;=\;
    \big(\hspace{1pt}
      \overline{\psi}
      \,
      [
        \tfrac{1}{2}
        \Gamma^{ab}
        ,
        \Gamma^c
      ]
       \,
      \psi
    \big)
    \;=\;
    \big(\hspace{1pt}
      \overline{\psi}
      (
        \eta^{bc}
        \Gamma^a
        -
        \eta^{ac}
        \Gamma^b
      )
      \psi
    \big)
    \,,
    \hspace{.7cm}
    \mbox{for $a,b,c < 9$}
    \,,
$$
where in the first step we use  -- via the first line in \eqref{AdjointnessForTypeIIB} --
that $\overline{\tfrac{1}{2}\Gamma_{ab}} \,=\, - \tfrac{1}{2}\Gamma_{ab}$,
which in the second step gives rise to the commutator $[\mbox{-},\mbox{-}]$ (in the 11D Clifford algebra).

Next, the case where $a,b < 9$ but $c = 9$ does involve the modified $\Gamma^9_{\!\!B} = \Gamma^9\Gamma^{\ten}$, but 
gives the correct answer trivially since the same kind of commutator appears and evidently vanishes:
$$
    J^{ab}
    \big(\hspace{1pt}
      \overline{\psi}
      \,\Gamma^9_B\,
      \psi
    \big)
    \;=\;
    \big(\hspace{1pt}
      \overline{\psi}
      \,
      [
        \tfrac{1}{2}
        \Gamma^{ab}
        ,
        \Gamma^9\Gamma^{\ten}
      ]
       \,
      \psi
    \big)    
    \;=\;
    0
    \,,
    \hspace{1cm}
    \mbox{for $a, b < 9$.}
$$
The interesting effect is in the next case, where one of the first two indices take the value $9$. Here the modified adjointness relation in the second line of \eqref{AdjointnessForTypeIIB} makes instead an anti-commutator $\{\mbox{-},\mbox{-}\}$ appear, which however becomes a commutator after pulling out the factor of $\Gamma^{\ten}$ that comes with $\Gamma^9_{\!\!B}$, and that again gives the correct result:
$$
    J^{a9}
    \big(\hspace{1pt}
      \overline{\psi}
      \,\Gamma^c_B\,
      \psi
    \big)
    \;=\;
    -
    \big(\hspace{1pt}
      \overline{\psi}
      \,
      \{
        \tfrac{1}{2}
        \Gamma^{a9}\Gamma^{\ten}
        ,
        \Gamma^c
      \}
       \,
      \psi
    \big)    
    \;=\;
    \big(\hspace{1pt}
      \overline{\psi}
      \,
      [
        \tfrac{1}{2}
        \Gamma^{a9}
        ,
        \Gamma^c
      ]
      \Gamma^{\ten}
       \,
      \psi
    \big)    
    \;=\;
    -
    \eta^{ac}
    \big(\hspace{1pt}
      \overline{\psi}
      \,\Gamma^9_{\!\!B}\,
      \psi
    \big)
  \hspace{.6cm}
  \mbox{for $a,c < 9$}
  \,.
$$
Finally, a similar computation passing through an anti-commutator also confirms the last remaining case:
$$
    J^{a9}
    \big(\hspace{1pt}
      \overline{\psi}
      \,\Gamma^9_B\,
      \psi
    \big)
    \;=\;
    -
    \big(\hspace{1pt}
      \overline{\psi}
      \{
        \tfrac{1}{2}
        \Gamma^{a9}
        \Gamma^{\ten}
        ,
        \Gamma^9\Gamma^{\ten}
      \}
      \psi
    \big)
    \;=\;
    \big(\hspace{1pt}
      \overline{\psi}
      \,\Gamma^a_{\!\!B}\,
      \psi
    \big)
    \hspace{1cm}
  \mbox{for $a < 9$.}
  \qedhere
$$
\end{proof}

Observing that the Clifford elements
\begin{equation}
  \label{SpinorialSDualityGenerators}
  \begin{array}{ccl}
    \sigma_1
    &:=&
    \Gamma_9
    \\
    \sigma_2
    &:=&
    \Gamma_{\!\ten}
  \end{array}
  \;\;\;\mbox{
    (whose product we denote $\sigma_3 := \sigma_1\sigma_2 \,=\,\Gamma_9\Gamma_{\!\ten}$)
  }
\end{equation}
anti-commute with {\it all} the $(\Gamma_a^B)_{a=0}^9$ \eqref{TheTypeIIBQuasiCliffordGenerators} we also have:
\begin{proposition}[{\bf R-Symmetry of type IIB} {\cite[Rem. 2.11]{FSS18-TDualityA}}]
\label{RSymmetryOfTypeII}
  The elements \eqref{SpinorialSDualityGenerators} generate a $\mathrm{Pin}(2)$-action on $\mathbf{16} \oplus \mathbf{16}$, which commutes with the $\mathrm{Spin}(1,9)$-action \eqref{IIBSpinorsAsPullbackOfIIASpinors}, making a direct product action of
  $
    \mathrm{Spin}(1,9)
    \times
    \mathrm{Pin}(2)
  $.
\end{proposition}
(This effectively 12-dimensional spin-action is seen to be related to the ``F-theory''-perspective on type IIB, in Def. \ref{FTheorySuperSpacetime} and Prop. \ref{SuperspaceSDuality}.)

\medskip

Putting these pieces together, we have more generally:
\begin{proposition}[\bf Lorentz-invariants on type IIB super-spacetime]
\label{LorentzInvariantsOnIIBSuperSpacetimes}
Setting
\begin{equation}
    \label{CliffordBasisIIB}
    \Gamma^B_{a_1 \cdots a_p}
    \;:=\;
    \left\{\!\!
    \def\arraystretch{1.2}
    \begin{array}{lcl}
      (-1)^{\mathrm{sgn}(\sigma)}
      \,
      \Gamma^B_{\sigma(a_1)}
      \cdots
      \Gamma^B_{\sigma(a_p)}
      &\vert&
      \mbox{\rm
        $\exists$ 
        $\sigma \in \mathrm{Sym}(n)$
        s.t. 
        $a_{\sigma(1)}
        < 
        \cdots
        <
        a_{\sigma(n)}$
      }
      \\
      0 &\vert& 
      \mbox{\rm
        otherwise
      }
    \end{array}
    \right.
  \end{equation}
the following expressions are invariants for the $\mathrm{Spin}(1,9)$-action \eqref{IIBSpinorsAsPullbackOfIIASpinors} type IIB super-spacetime \eqref{TheTypeIIBSuperMinkowskiAlgebra}:
\begin{equation}
  \left.
  \def\arraystretch{1.4}
  \def\arraycolsep{0pt}
  \begin{array}{lr}
    \big(\hspace{1pt}
      \overline{\psi}
      \,
      \Gamma^B_{a_1 \cdots a_p}
      \,
      \psi
    \big)
    &
    e^{a_1} \cdots e^{a_p}
    \\
    \big(\hspace{1pt}
      \overline{\psi}
      \,
      \Gamma^B_{a_1 \cdots a_p}
      \,
      \sigma_i
      \,
      \psi
    \big)
    &
    e^{a_1} \cdots e^{a_p}
    \\
  \end{array}
  \right\}
  \;\in\;
  \mathrm{CE}\big(
    \mathbb{R}^{
      1,9\,\vert\,
      \mathbf{16}
      \oplus
      \mathbf{16}
    }
  \big)
\end{equation}
for $\sigma_i$ from \eqref{SpinorialSDualityGenerators}.
\end{proposition}
(Of course, many of these expressions vanish by \eqref{SkewSpinorPairings}.)
\begin{proof}
  The statement for the second line follows immediately from that for the first line by Prop. \ref{SpinorialSDualityGenerators}. We proceed to prove the statement for the first line.
  
  For $p = 0$ the statement holds trivially, since the given expression vanishes, by \eqref{SkewSpinorPairings}.
  
  From $p=1$ on we shall prove the stronger statement that
  $
    \big(\hspace{1pt}
      \overline{\phi}
      \,
      \Gamma_{a_1 \cdots a_9}
      \,
      \psi
    \big)
    e^{a_1}\cdots e^{a_p}
  $
  is invariant
  for $\phi$ possibly any other element
  transforming as $J^{ab}\phi \,=\, -\tfrac{1}{2}\Gamma^{ab}_B\phi$:

  \smallskip

  The case $p = 1$ follows verbatim as in Lem. \ref{LorentzEquivarianceOfIIBSpacetime}, with the first factor of ``$\psi$`` there generalized to ``$\phi$''.

  For the remaining cases
  we may, by 
  \eqref{CliffordBasisIIB},  assume without restriction of generality that $a_1 < \cdots < a_{p+1}$, hence in particular that $a_n < 9$ whenever $n \leq p$, and we need to show that the following is invariant:
  $$
    \def\arraystretch{1.6}
    \begin{array}{ccll}
      \big(\hspace{1pt}
        \overline{\phi}
        \,\Gamma^B_{a_1 \cdots a_9}\,
        \psi
      \big)
      e^{a_1} \cdots e^{a_p}
      &=&
      \big(\hspace{1pt}
        \overline{\phi}
        \,\Gamma^B_{a_1 \cdots a_p}
        \Gamma^B_{a_{p+1}}\,
        \psi
      \big)
      e^{a_1}\cdots e^{a_p}
      e^{a_{p+1}}
      &
      \proofstep{
        by assmptn
      }
      \\
      &=&
      \pm
      \big(\hspace{1pt}
        \overline{
          \Gamma^B_{a_1 \cdots a_p}
          \phi
        }
        \,
        \Gamma^B_{a_{p+1}}
        \,
        \psi
      \big)
      e^{a_1}\cdots e^{a_p}
      \;
      e^{a_{p+1}}      
      &
      \proofstep{
        by \eqref{AdjointnessOfCliffordBasisElements}\,.
      }
    \end{array}
$$

Now observe that the expression
$$
  \phi'
  \;:=\;
  \Gamma^B_{a_1 \cdots a_p}
  \phi\, e^{a_1} \cdots e^{a_p}
$$
transforms as a spinor under any $J^{ab}$: For $a,b < 9$ this is the standard situation, and then for $J^{a9}$ it follows since \eqref{CliffordBasisIIB} makes any $\Gamma_{\!\ten}$-factor ``stay on the right''. But with this we are reduced to seeing that
$
  \big(\hspace{1pt}
    \overline{\phi'}
    \,\Gamma_{a_{p+1}}\,
    \psi
  \big)
  e^{a_{p+1}}
$
is invariant, which is the case $p=1$ already proven.  
\end{proof}

\subsection{Type IIA/B-Duality}
\label{SuperspaceTDualityI}

We give a streamlined review of the core part of the formulation (and in fact a derivation) 
from \cite{FSS18-TDualityA}
of T-duality along a single isometry between type IIA and type IIB super-flux densities on super-Minkowski spacetime.\footnote{
Note that the focus on Minkowski super-spacetime is only superficially a specialization: The torsion constraints that govern supergravity theories say that -- in the manner of Cartan geometry -- the actual super-flux densities on on-shell super-spacetimes are tangent space-wise constrained to have fermionic components given by these Minkowskian super-invariants -- and in fact in 11D SuGra that condition is equivalent to the supergravity equations of motion \cite[Thm. 3.1]{GSS24-SuGra}. In this, the super-invariants on super-Minkowski spacetimes are the archetypes that govern full on-shell supergravity theories. }

\smallskip

The key observations driving this are the following:

\begin{itemize}[
 leftmargin=.8cm,
 topsep=2pt,
 itemsep=2pt
]

\item[\bf (i)] The type IIA super-flux densities on super-Minkowki spacetime satisfying their Bianchi identities are equivalently (Prop. \ref{TheTypeIIACocyclesSummarized})
super-$L_\infty$ cocycles with coefficients in the real Whitehead $L_\infty$-algebra of the twisted K-theory spectrum $\mathrm{KU}$ (Ex. \ref{WhiteheadLInfinityOfTwistedKTheorySpectrum}).

\item[\bf (ii)] Double dimensional reduction of super-flux densities on super-Minkowski spacetime is equivalently given by {\it cyclifying} (as in {\it cyclic cohomology}) their coefficient $L_\infty$-algebra (Prop. \ref{TheExtCycAdjunction}). 

\item[\bf (iii)] The cyclification of twisted $\mathrm{KU}$ is equivalent to that of twisted $\Sigma\mathrm{KU}$ by swapping the Chern class with the wrapping mode of the 3-form (Ex. \ref{CyclificationOfTwistedKTTheorySpectra}).

\item[\bf (iv)] The type IIA flux densities are equivalent to the IIB flux densities after reduction via cyclification to 9d, whereby the type IIA spacetime is swapped for the type IIB spacetime (Prop. \ref{TheTypeIIBCocyclesSummarized}).

\end{itemize}

\medskip

\noindent
{\bf M/IIA duality.}
Recall from Ex. \ref{4SphereValuedSuperFlux} that the C-field super-flux densities on 11D super-Minkowski spacetime are encoded by a super-$L_\infty$ homomorphism of the form:
\vspace{-1mm} 
\begin{equation}
  \label{MBraneCocycles}
  \begin{tikzcd}[
    row sep=-3pt
    ,column sep=30pt
    ,/tikz/column 1/.append style={anchor=base east}
  ]
    \mathbb{R}^{1,10\,\vert\,\mathbf{32}}
    \ar[
      rr,
      "{
        (G_4,\, G_7)
      }"
    ]
    &&
    \mathfrak{l}S^4
    \\
    \tfrac{1}{2}
    \big(\hspace{1pt}
      \overline{\psi}
      \,\Gamma_{a_1 a_2}\,
      \psi
    \big)
    e^{a_1} e^{a_2}
    \;=:\;
    G_4
    &\longmapsfrom&
    g_4
    \\
    \tfrac{1}{5!}
    \big(\hspace{1pt}
      \overline{\psi}
      \,\Gamma_{a_1 \cdots a_5}\,
      \psi
    \big)
    e^{a_1} \cdots e^{a_5}
    \;=: \;
    G_7
    &\longmapsfrom&
    g_7
  \end{tikzcd}  
\end{equation}

\begin{example}[{\bf IIA-Reduction of C-field super-flux} {\cite[pp 11]{FSS17}, cf. \cite[Ex. 2.13]{SS24-Flux}}]
\label{IIAReductionOfCFieldSuperFlux}
Under reduction via the Ext/Cyc-adjunction (Prop. \ref{TheExtCycAdjunction}) with respect to the M/IIA extension (Ex.  \ref{11dAsExtensionFromIIA}) 
$$
  \begin{tikzcd}[
  ]
    \mathbb{R}^{1,10\,\vert\,\mathbf{32}}
    \ar[
      rr,
      "{
        (G_4,\, G_7)
      }"
    ]
    &&
    \mathfrak{l}S^4
  \end{tikzcd}
  \hspace{1cm}
    \longleftrightarrow
  \hspace{1cm}
  \begin{tikzcd}[
      row sep=-1pt, column sep=huge
  ]
    \mathbb{R}^{
      1,9\,\vert\, 
      \mathbf{16} 
        \oplus 
      \overline{\mathbf{16}}
    }
    \ar[
      rr,
      "{
        \mathrm{rdc}_{
          \scalebox{.7}{$c_1^M$}
        }
          (G_4,G_7)
      }"
    ]
    \ar[
      dr,
      "{
        (\hspace{1pt}
          \overline{\psi}
          \,\Gamma_{\!\ten}\,
          \psi) 
      }"{sloped, swap}
    ]
    &&
    \mathrm{cyc}(\mathfrak{l}S^4)
    \ar[
      dl,
      "{ \omega_2 }"
    ]
    \\
    &
    b\mathbb{R}
  \end{tikzcd}
$$
the flux densities from 11D \eqref{MBraneCocycles} become:
\vspace{-2mm} 
\begin{equation}
  \label{DDReducedMBraneCocycles}
  \hspace{-6mm}
  \begin{tikzcd}[sep=0pt]
    \mathbb{R}^{1,10\,\vert\,\mathbf{32}}
    \ar[
      rr,
      "{
        (G_4,\, G_7)
      }"
    ]
    &&
    \mathfrak{l}S^4
    \\
    \tfrac{1}{2}
    \big(\hspace{1pt}
      \overline{\psi}
      \,\Gamma_{a_1 a_2}\,
      \psi
    \big)
    e^{a_1} e^{a_2}
    &\longmapsfrom&
    g_4
    \\
    \tfrac{1}{5!}
    \big(\hspace{1pt}
      \overline{\psi}
      \,\Gamma_{a_1 \cdots a_5}\,
      \psi
    \big)
    e^{a_1} \cdots e^{a_5}
    &\longmapsfrom&
    g_7
  \end{tikzcd}  
  \hspace{.35cm}
  \leftrightsquigarrow
  \hspace{-4mm}
  \begin{tikzcd}[
    row sep=-2pt
    ,column sep=25pt
    ,/tikz/column 1/.append style={anchor=base east}
  ]
    \mathbb{R}^{
      1,9\,\vert\,
      \mathbf{16} 
        \oplus 
      \overline{\mathbf{16}}
    }
    \ar[
      rr,
      "{
          \begin{array}{c}
            \mathrm{rdc}_{
              \scalebox{.65}{$c_1^M$}
            }
              (G_4, G_7) 
            \;=
            \\
            (F_2, F_4, F_6, H_3^A, H_7^A)
          \end{array}
      }"
    ]
    &&
    \mathrm{cyc}\big(
      \mathfrak{l}S^4
    \big)
    \\
    \big(\hspace{1pt}
      \overline{\psi}
      \,\Gamma_{\!\ten}
      \psi
    \big)
    \;=:\;
    \mathrm{F}_2
    &\longmapsfrom&
    \omega_2
    \\
    \tfrac{1}{2}\big(\hspace{1pt}
      \overline{\psi}
      \,\Gamma_{a_1 a_2}\,
      \psi
    \big)
    e^{a_1}\, e^{a_2}
    \;=:\;
    \mathrm{F}_4
    &\longmapsfrom&
    g_4
    \\
    \big(\hspace{1pt}
      \overline{\psi}
      \,
      \Gamma_a
      \Gamma_{\!\ten}
      \,
      \psi
    \big)
    e^a
    \;=:\;
    H_3^{A}
    &\longmapsfrom&
    \mathrm{s}g_4
    \\
    \tfrac{1}{5!}
    \big(\hspace{1pt}
      \overline{\psi}
      \,\Gamma_{a_1 \cdots a_5}\,
      \psi
    \big)
    e^{a_1} \cdots e^{a_5}
    \;=:\;
    H^A_7
    &\longmapsfrom&
    g_7
    \\
    -
    \tfrac{1}{4!}
    \big(\hspace{1pt}
      \overline{\psi}
      \,
      \Gamma_{a_1 \cdots a_4}
      \Gamma_{\!\ten}
      \,
      \psi
    \big)
    e^{a_1} \cdots e^{a_4}
    \;=:\;
    -F_6
    &\longmapsfrom&
    \mathrm{s}g_7
  \end{tikzcd}
\end{equation}
satisfying: 
\vspace{2mm} 
\begin{equation}
  \label{DimReducedBianchiIdentities}
  \left.
  \def\arraystretch{1.3}
  \def\arraycolsep{4pt}
  \begin{array}{lcl}
    \mathrm{d}\, G_4 &=& 0
    \\
    \mathrm{d}\, G_7 &=&
    \tfrac{1}{2}\,
    G_4 \, G_4
  \end{array}
  \!\!\!\right\}
  \hspace{.5cm}
  \leftrightsquigarrow
  \hspace{.5cm}
  \left\{\!\!
  \def\arraystretch{1.1}
  \def\arraycolsep{4pt}
  \begin{array}{lcl}
    \mathrm{d}\, F_2
    &=&
    0
    \\
    \mathrm{d}\, F_4
    &=&
    H_3^A \, F_2
    \\
    \mathrm{d}\, F_6
    &=&
    H_3^A \, F_4
    \\
    \mathrm{d}\, H_3^A
    &=&
    0
    \\
    \mathrm{d}\, H_7^A
    &=&
    \tfrac{1}{2}
    \,
    F_4 \, F_4
    \,-\,
    F_2 \, F_6\;.
  \end{array}
  \right.
\end{equation}
\end{example}

However, on the type IIA super-spacetime there appear further/higher super-invariants satisfying analogous differential equations --- this observation is essentially due to  \cite[\S 6.1]{CdAIPB00}, except that we also consider $F_{12}$ 
\footnote{
  \label{OnD10Branes}
  A 12-form term like $F_{12}$ in \eqref{HigherTypeASuperFluxes} -- nominally the WZW term for ``D10-branes'' -- is rarely considered in the literature (an exception is \cite[p 30]{CallisterSmith09}) since it is evidently invisible on ordinary (bosonic) spacetimes. But on super-space it is non-vanishing and must be considered \cite[Rem. 4.3]{BMSS19}
  to complete the flux densities $F_{2\bullet}$ to an $\mathfrak{l}\mathrm{KU}$-valued cocycle. On the other hand, the yet higher degree fluxes of this form $F_{2k+2} = \tfrac{1}{(2k)!}
  \big(\hspace{1pt}
    \overline{\psi}
    \,\Gamma_{a_1 \cdots a_{2k}}\,
    \psi
  \big)
  e^{a_1} \cdots e^{a_{2k}}
  $
  do vanish on $\mathbb{R}^{1,9\,\vert\,\mathbf{16}
      \oplus
      \overline{\mathbf{16}}}$ and hence need not be further considered.
}.

\begin{definition}[\bf Higher type IIA super-flux densities.]
\label{HigherIIAFluxDensities}
Consider the following super-invariants, beyond those appearing via reduction from 11D in Ex. \ref{IIAReductionOfCFieldSuperFlux}
\footnote{
  In the string theory lore the higher flux densities related to the higher super-invariants  \eqref{HigherTypeASuperFluxes} and corresponding to (1.) the D6-brane, (2.) the D8-brane and (3.) the ``D10-brane'',  are meant to have 11D ancestors given, more or less informally, by (1.) the 11D Kaluza-Klein monopole, (2.) a Scherk-Schwarz compactification to massive type IIA theory, respectively, while the M-theory lift of (3.) the ``D10-brane'' seems not to have been discussed (cf. footnote \ref{OnD10Branes}). 
  
  More in the spirit of the rigorous  derivations here, we have shown in \cite{BMSS19} that the relevant higher generators appear when the 4-sphere coefficient (Ex. \ref{4SphereValuedSuperFlux}) for the fluxes in 11D are subjected to fiberwise ``stabilization'' over the 3-sphere (in the sense of stable homotopy theory).
  
}
\begin{equation}
  \label{HigherTypeASuperFluxes}
  \left.
  \def\arraystretch{1.5}
  \begin{array}{ccl}
    F_8
    &:=&
    +
    \tfrac{1}{6!}
    \big(\hspace{1pt}
      \overline{\psi}
      \,\Gamma_{a_1 \cdots a_6}\,
      \psi
    \big)
    e^{a_1} \cdots e^{a_6}
    \\
    F_{10}
    &:=&
    +
    \tfrac{1}{8!}
    \big(\hspace{1pt}
      \overline{\psi}
      \,
      \Gamma_{a_1 \cdots a_8}
      \Gamma_{\!\ten}
      \,
      \psi
    \big)
    e^{a_1} \cdots e^{a_8}
    \\
    F_{12}
    &:=&
    +
    \tfrac{1}{10!}
    \big(\hspace{1pt}
      \overline{\psi}
      \,
      \Gamma_{a_1 \cdots a_{\ten}}\,
      \psi
    \big)
    e^{a_1} \cdots e^{a_{\ten}}
  \end{array}
  \right\}
  \;
  \in
  \;
  \mathrm{CE}\big(
    \mathbb{R}^{
      1,9\,\vert\,
      \mathbf{16}
      \oplus
      \overline{\mathbf{16}}
    }
  \big).
\end{equation}
\end{definition}

\begin{lemma}[\bf Bianchi identities for higher IIA super-fluxes]
\label{BianchiIdentityForHigherIIASuperFluxes}
The higher super-flux densities \eqref{HigherTypeASuperFluxes} satisfy
\begin{equation}
  \label{HigherBianchiIdentities}
  \def\arraystretch{1.2}
  \begin{array}{ccl}
    \mathrm{d}\,F_8
    &=&
    H^A_3 \, F_6
    \\
    \mathrm{d}\, F_{10}
    &=&
    H_3^A \, F_8
    \\
    \mathrm{d}\, F_{12}
    &=&
    H_3^A \, F_{10}
    \\
    0 &=&
    H_3^A \, F_{12}
  \end{array}
\end{equation}
with $H_3^A$ from \eqref{DDReducedMBraneCocycles}.
\end{lemma}
\begin{proof}
The proof for the first equation in \eqref{HigherBianchiIdentities} is also given in \cite[\S B]{CdAIPB00}, which we follow.
First, note that the closure of $H_3^A$ from \eqref{DimReducedBianchiIdentities}
$$
  \def\arraystretch{1.5}
  \begin{array}{ccl}
    0
    &=&
    \mathrm{d}\, H_3^A
    \\
    &=&
    \mathrm{d}
    \Big(
    \big(\hspace{1pt}
      \overline{\psi}
      \,
      \Gamma_a
      \Gamma_{\!\ten}
      \,
      \psi
    \big)
    e^a
    \Big)
    \\
    &=&
    \big(\hspace{1pt}
      \overline{\psi}
      \,
      \Gamma_a
      \Gamma_{\!\ten}
      \,
      \psi
    \big)
    \big(\hspace{1pt}
      \overline{\psi}
      \,\Gamma^a\,
      \psi
    \big)
  \end{array}
  \hspace{1cm}
  a < \ten
$$
means in components that
\begin{equation}
  \label{ConsequenceOfClosureOfH3A}
  \def\arraystretch{1.6}
  \begin{array}{ll}
  0 
 & \;=\;
  \big(
    \Gamma_a
    \Gamma_{\!\ten}
  \big)_{(\alpha \beta}
  \Gamma^a_{\gamma \delta)}
  \\
 & \;=\;
  \tfrac{2 \cdot 3!}{4!}
  \Big(
  \big(
    \Gamma_a
    \Gamma_{\!\ten}
  \big)_{(\alpha \beta}
  \Gamma^a_{\delta) \gamma}
  \,+\,
  \big(
    \Gamma_a
    \Gamma_{\!\ten}
  \big)_{\gamma (\alpha }
  \Gamma^a_{\beta \delta)}
  \Big)
  \,,
  \end{array}
  \hspace{1cm}
  a < \ten 
\end{equation}
where in the second line we used \eqref{SymmetricSpinorPairings} that $(\Gamma_a\Gamma_{\!\ten})_{\alpha\beta}$ and $(\Gamma_a)_{\gamma\delta}$ both already are symmetric in their spinor indices.

With this in hand, we compute as follows:
$$
  \def\arraystretch{1.5}
  \begin{array}{ccll}
    \mathrm{d}\, F_8
    &=&
    \mathrm{d}
    \,
    \tfrac{1}{6!}
    \big(\hspace{1pt}
      \overline{\psi}
      \,\Gamma_{a_1 \cdots a_6}\,
      \psi
    \big)
    e^{a_1} \cdots e^{a_6}
    &
    \proofstep{
      by
      \eqref{HigherTypeASuperFluxes}
    }
    \\
    &=&
    -
    \tfrac{1}{5!}
    \big(\hspace{1pt}
      \overline{\psi}
      \,\Gamma_{a_1 \cdots a_5\, b}\,
      \psi
    \big)
    \big(\hspace{1pt}
      \overline{\psi}
      \,\Gamma^b\,
      \psi
    \big)
    e^{a_1} \cdots e^{a_5}
    &
    \proofstep{
      by \eqref{SuperMinkowskiCE}
    }
    \\
    &=&
    -
    \tfrac{1}{5!}
    \big(\hspace{1pt}
      \overline{\psi}
      \,\Gamma_{a_1 \cdots a_5}
      \Gamma_b\,
      \psi
    \big)
    \big(\hspace{1pt}
      \overline{\psi}
      \,\Gamma^b\,
      \psi
    \big)
    e^{a_1} \cdots e^{a_5}
    &
    \proofstep{
      by
      \eqref{GeneralCliffordProduct}
      \& 
      \eqref{SkewSpinorPairings}
    }
    \\
    &=&
    +
    \tfrac{1}{5!}
    \big(\hspace{1pt}
      \overline{\psi}
      \,\Gamma_{a_1 \cdots a_5}
      \Gamma_{\!\ten}
      \,
      \Gamma_b
      \Gamma_{\!\ten}
      \,
      \psi
    \big)
    \big(\hspace{1pt}
      \overline{\psi}
      \,\Gamma^b\,
      \psi
    \big)
    e^{a_1} \cdots e^{a_5}
    &
    \proofstep{
      by
      \eqref{CliffordDefiningRelation}
    }
    \\
    &=&
    +
    \tfrac{1}{5!}
    \big(
      \Gamma_{a_1 \cdots a_5}
      \Gamma_{\!\ten}
    \big)_{(\alpha \vert{\color{darkblue}\kappa}\vert}
    \tensor{
      \big(
        \Gamma_{b}
        \Gamma_{\!\ten}
      \big)
      }{^{\color{darkblue}\kappa}_\beta}
    \big(
      \Gamma^b
    \big)_{\gamma\delta)}
    \;
    \psi^{\alpha}
    \psi^{\beta}
    \psi^{\gamma}
    \psi^{\delta}
    \,
    e^{a_1}
    \cdots
    e^{a_5}
    &
    \proofstep{
      matrix multip.
    }
    \\
    &=&
    -
    \tfrac{1}{5!}
    \big(
      \Gamma_{a_1 \cdots a_5}
      \Gamma_{\!\ten}
    \big)_{(\alpha \vert{\color{darkblue}\kappa}\vert}
    \big(
      \Gamma_{b}
      \Gamma_{\!\ten}
    \big)_{\gamma\beta}
    \big(
      \Gamma^{b}
    \big)^{\color{darkblue}\kappa}{}_{\delta)}
    \;
    \psi^{\alpha}
    \psi^{\beta}
    \psi^{\gamma}
    \psi^{\delta}
    \,
    e^{a_1}
    \cdots
    e^{a_5}
    &
    \proofstep{
      by \eqref{ConsequenceOfClosureOfH3A}
    }
    \\
    &=&
    +
    \tfrac{1}{5!}
    \big(
      \Gamma_{a_1 \cdots a_5}
      \Gamma_{b}
      \Gamma_{\!\ten}
    \big)_{(\alpha \delta}
    \big(
      \Gamma^{b}
      \Gamma_{\!\ten}
    \big)_{\gamma\beta)}
    \;
    \psi^{\alpha}
    \psi^{\beta}
    \psi^{\gamma}
    \psi^{\delta}
    \,
    e^{a_1}
    \cdots
    e^{a_5}
    &
    \proofstep{
      matrix multip.
    }
    \\
    &=&
    +
    \tfrac{1}{5!}
    \big(\hspace{1pt}
      \overline{\psi}
      \,
      \Gamma_{a_1 \cdots a_5}
      \Gamma_{b}
      \Gamma_{\!\ten}
      \,
      \psi
    \big)
    \big(
      \overline{\psi}
      \,
      \Gamma^{b}
      \Gamma_{\!\ten}
      \,
      \psi
    \big)
    \,
    e^{a_1}
    \cdots
    e^{a_5}
    \\
    &=&
    +
    \tfrac{1}{4!}
    \big(\hspace{1pt}
      \overline{\psi}
      \,
      \Gamma_{
        [a_1 \cdots a_4}
      \Gamma_{\!\ten}
      \,
      \psi
    \big)
    \big(
      \overline{\psi}
      \,
      \Gamma_{a_5]}
      \Gamma_{\!\ten}
      \,
      \psi
    \big)
    \,
    e^{a_1}
    \cdots
    e^{a_5}
    &
    \proofstep{
      by
      \eqref{GeneralCliffordProduct}
      \& 
      \eqref{SkewSpinorPairings}
    }
    \\
    &=&
    \;+\;
    F_6\, H_3^A
    &
    \proofstep{
      by
      \eqref{DDReducedMBraneCocycles}
    }.
  \end{array}
$$
The remaining two cases (and in fact all cases) work analogously:
$$
  \def\arraystretch{1.5}
  \begin{array}{ccll}
    \mathrm{d}\, F_{10}
    &=&
    \mathrm{d}
    \,
    \tfrac{1}{8!}
    \big(\hspace{1pt}
      \overline{\psi}
      \,
      \Gamma_{a_1 \cdots a_{8}}
      \Gamma_{\!\ten}
      \,
      \psi
    \big)
    e^{a_1} \cdots e^{a_{8}}
    &
    \proofstep{
      by
      \eqref{HigherTypeASuperFluxes}
    }
    \\
    &=&
    -
    \tfrac{1}{7!}
    \big(\hspace{1pt}
      \overline{\psi}
      \,
      \Gamma_{a_1 \cdots a_7\, b}
      \Gamma_{\!\ten}
      \,
      \psi
    \big)
    \big(\hspace{1pt}
      \overline{\psi}
      \,\Gamma^b\,
      \psi
    \big)
    e^{a_1} \cdots e^{a_7}
    &
    \proofstep{
      by \eqref{SuperMinkowskiCE}
    }
    \\
    &=& 
    -
    \tfrac{1}{7!}
    \big(\hspace{1pt}
      \overline{\psi}
      \,
      \Gamma_{a_1 \cdots a_7}
      \Gamma_b
      \Gamma_{\!\ten}
      \,
      \psi
    \big)
    \big(\hspace{1pt}
      \overline{\psi}
      \,\Gamma^b\,
      \psi
    \big)
    e^{a_1} \cdots e^{a_7}
    &
    \proofstep{
      by
      \eqref{GeneralCliffordProduct}
      \& 
      \eqref{SkewSpinorPairings}
    }
    \\
    &=&
    -
    \tfrac{1}{7!}
    \big(
      \Gamma_{a_1 \cdots a_7}
    \big)_{(\alpha \vert{\color{darkblue}\kappa}\vert}
    \tensor{
      \big(
        \Gamma_{b}
        \Gamma_{\!\ten}
      \big)
      }{^{\color{darkblue}\kappa}_\beta}
    \big(
      \Gamma^b
    \big)_{\gamma\delta)}
    \;
    \psi^{\alpha}
    \psi^{\beta}
    \psi^{\gamma}
    \psi^{\delta}
    \,
    e^{a_1}
    \cdots
    e^{a_7}
    &
    \proofstep{
      matrix multip.
    }
    \\
    &=&
    +
    \tfrac{1}{7!}
    \big(
      \Gamma_{a_1 \cdots a_7}
    \big)_{(\alpha \vert{\color{darkblue}\kappa}\vert}
    \big(
      \Gamma_{b}
      \Gamma_{\!\ten}
    \big)_{\gamma\beta}
    \big(
      \Gamma^{b}
    \big)^{\color{darkblue}\kappa}{}_{\delta)}
    \;
    \psi^{\alpha}
    \psi^{\beta}
    \psi^{\gamma}
    \psi^{\delta}
    \,
    e^{a_1}
    \cdots
    e^{a_7}
    &
    \proofstep{
      by \eqref{ConsequenceOfClosureOfH3A}
    }
    \\
    &=&
    +
    \tfrac{1}{7!}
    \big(
      \Gamma_{a_1 \cdots a_7}
      \Gamma_{b}
    \big)_{(\alpha \delta}
    \big(
      \Gamma^{b}
      \Gamma_{\!\ten}
    \big)_{\gamma\beta)}
    \;
    \psi^{\alpha}
    \psi^{\beta}
    \psi^{\gamma}
    \psi^{\delta}
    \,
    e^{a_1}
    \cdots
    e^{a_7}
    &
    \proofstep{
      matrix multip.
    }
    \\
    &=&
    +
    \tfrac{1}{7!}
    \big(\hspace{1pt}
      \overline{\psi}
      \,
      \Gamma_{a_1 \cdots a_7}
      \Gamma_{b}
      \,
      \psi
    \big)
    \big(
      \overline{\psi}
      \,
      \Gamma^{b}
      \Gamma_{\!\ten}
      \,
      \psi
    \big)
    \,
    e^{a_1}
    \cdots
    e^{a_7}
    \\
    &=&
    +
    \tfrac{1}{6!}
    \big(\hspace{1pt}
      \overline{\psi}
      \,
      \Gamma_{
        [a_1 \cdots a_6}
      \,
      \psi
    \big)
    \big(\hspace{1pt}
      \overline{\psi}
      \,
      \Gamma_{a_7]}
      \Gamma_{\!\ten}
      \,
      \psi
    \big)
    \,
    e^{a_1}
    \cdots
    e^{a_7}
    &
    \proofstep{
      by
      \eqref{GeneralCliffordProduct}
      \& 
      \eqref{SkewSpinorPairings}
    }
    \\
    &=&
    +
    F_{8}\, H_3^A
    &
    \proofstep{
      by
      \eqref{DDReducedMBraneCocycles}
    }
  \end{array}
$$

and:

$$
  \def\arraystretch{1.5}
  \begin{array}{ccll}
    \mathrm{d}\, F_{12}
    &=&
    \mathrm{d}
    \,
    \tfrac{1}{10!}
    \big(\hspace{1pt}
      \overline{\psi}
      \,
      \Gamma_{a_1 \cdots a_{\ten}}\,
      \psi
    \big)
    e^{a_1} \cdots e^{a_{\ten}}
    &
    \proofstep{
      by
      \eqref{HigherTypeASuperFluxes}
    }
    \\
    &=&
    -
    \tfrac{1}{9!}
    \big(\hspace{1pt}
      \overline{\psi}
      \,\Gamma_{a_1 \cdots a_9\, b}\,
      \psi
    \big)
    \big(\hspace{1pt}
      \overline{\psi}
      \,\Gamma_b\,
      \psi
    \big)
    e^{a_1} \cdots e^{a_9}
    &
    \proofstep{
      by \eqref{SuperMinkowskiCE}
    }
    \\
    &=& 
    -
    \tfrac{1}{9!}
    \big(\hspace{1pt}
      \overline{\psi}
      \,\Gamma_{a_1 \cdots a_9}
      \Gamma_b\,
      \psi
    \big)
    \big(\hspace{1pt}
      \overline{\psi}
      \,\Gamma^b\,
      \psi
    \big)
    e^{a_1} \cdots e^{a_9}
    &
    \proofstep{
      by
      \eqref{GeneralCliffordProduct}
      \& 
      \eqref{SkewSpinorPairings}
    }
    \\
    &=&
    +
    \tfrac{1}{9!}
    \big(\hspace{1pt}
      \overline{\psi}
      \,\Gamma_{a_1 \cdots a_9}
      \Gamma_{\!\ten}
      \,
      \Gamma_b
      \Gamma_{\!\ten}
      \,
      \psi
    \big)
    \big(\hspace{1pt}
      \overline{\psi}
      \,\Gamma^b\,
      \psi
    \big)
    e^{a_1} \cdots e^{a_9}
    &
    \proofstep{
      by
      \eqref{CliffordDefiningRelation}
    }
    \\
    &=&
    +
    \tfrac{1}{9!}
    \big(
      \Gamma_{a_1 \cdots a_9}
      \Gamma_{\!\ten}
    \big)_{(\alpha \vert{\color{darkblue}\kappa}\vert}
    \tensor{
      \big(
        \Gamma_{b}
        \Gamma_{\!\ten}
      \big)
      }{^{\color{darkblue}\kappa}_\beta}
    \big(
      \Gamma^b
    \big)_{\gamma\delta)}
    \;
    \psi^{\alpha}
    \psi^{\beta}
    \psi^{\gamma}
    \psi^{\delta}
    \,
    e^{a_1}
    \cdots
    e^{a_9}
    &
    \proofstep{
      matrix multip.
    }
    \\
    &=&
    -
    \tfrac{1}{9!}
    \big(
      \Gamma_{a_1 \cdots a_9}
      \Gamma_{\!\ten}
    \big)_{(\alpha \vert{\color{darkblue}\kappa}\vert}
    \big(
      \Gamma_{b}
      \Gamma_{\!\ten}
    \big)_{\gamma\beta}
    \big(
      \Gamma^{b}
    \big)^{\color{darkblue}\kappa}{}_{\delta)}
    \;
    \psi^{\alpha}
    \psi^{\beta}
    \psi^{\gamma}
    \psi^{\delta}
    \,
    e^{a_1}
    \cdots
    e^{a_9}
    &
    \proofstep{
      by \eqref{ConsequenceOfClosureOfH3A}
    }
    \\
    &=&
    +
    \tfrac{1}{9!}
    \big(
      \Gamma_{a_1 \cdots a_9}
      \Gamma_{b}
      \Gamma_{\!\ten}
    \big)_{(\alpha \delta}
    \big(
      \Gamma^{b}
      \Gamma_{\!\ten}
    \big)_{\gamma\beta)}
    \;
    \psi^{\alpha}
    \psi^{\beta}
    \psi^{\gamma}
    \psi^{\delta}
    \,
    e^{a_1}
    \cdots
    e^{a_9}
    &
    \proofstep{
      matrix multip. 
    }
    \\
    &=&
    +
    \tfrac{1}{9!}
    \big(\hspace{1pt}
      \overline{\psi}
      \,
      \Gamma_{a_1 \cdots a_9}
      \Gamma_{b}
      \Gamma_{\!\ten}
      \,
      \psi
    \big)
    \big(
      \overline{\psi}
      \,
      \Gamma^{b}
      \Gamma_{\!\ten}
      \,
      \psi
    \big)
    \,
    e^{a_1}
    \cdots
    e^{a_9}
    \\
    &=&
    +
    \tfrac{1}{8!}
    \big(\hspace{1pt}
      \overline{\psi}
      \,
      \Gamma_{[a_1 \cdots a_8}
      \Gamma_{\!\ten}
      \,
      \psi
    \big)
    \big(
      \overline{\psi}
      \,
      \Gamma_{a_9]}
      \Gamma_{\!\ten}
      \,
      \psi
    \big)
    \,
    e^{a_1}
    \cdots
    e^{a_9}
    &
    \proofstep{
      by
      \eqref{GeneralCliffordProduct}
      \& 
      \eqref{SkewSpinorPairings}
    }
    \\
    &=&
    +
    F_{10}\, H_3^A
    &
    \proofstep{
      by
      \eqref{DDReducedMBraneCocycles}\,.
    }
  \end{array}
$$

\vspace{-4mm} 
\end{proof}

It is worth summarizing this state of affairs in super-$L_\infty$ algebraic language:
\begin{proposition}[{\bf The type IIA super-cocycles} {\cite[Prop. 4.8]{FSS18-TDualityA}}]
\label{TheTypeIIACocyclesSummarized}
  On the type IIA super-Minkowski spacetime $\mathbb{R}^{1,9\,\vert\, \mathbf{16}\oplus \overline{\mathbf{16}}}$ we have the following super-invariants
  \begin{equation}
    \label{IIACocycles}
    \left.
    \def\arraystretch{1.3}
    \begin{array}{ccl}
      H_3^A
      &:=&
      \big(\hspace{1pt}
        \overline{\psi}\,
        \Gamma_a
        \Gamma_{\!\ten}
        \,\psi
      \big)
      e^a
      \\
      F_{-2k} &:=& \;\;0
      \,,\;\;\;
      k \in \mathbb{N}
      \\
      F_2 &:=& 
      \big(\hspace{1pt}
        \overline{\psi}
        \,\Gamma_{\!\ten}\,
        \psi
      \big)
      \\
      F_4 
      &:=&
      \tfrac{1}{2}
      \big(\hspace{1pt}
        \overline{\psi}
        \,\Gamma_{a_1 a_2}\,
        \psi
      \big)
      e^{a_1} e^{a_2}
      \\
      F_6 
      &:=&
      \tfrac{1}{4!}
      \big(\hspace{1pt}
        \overline{\psi}
        \,
        \Gamma_{\!\ten}
        \Gamma_{a_1 \cdots a_4}\,
        \psi
      \big)
      e^{a_1} \cdots e^{a_4}
      \\
      F_8 
      &:=&
      \tfrac{1}{6!}
      \big(\hspace{1pt}
        \overline{\psi}
        \,\Gamma_{a_1 \cdots a_6}\,
        \psi
      \big)
      e^{a_1} \cdots e^{a_6}
      \\
      F_{10} 
      &:=&
      \tfrac{1}{8!}
      \big(\hspace{1pt}
        \overline{\psi}
        \,
        \Gamma_{\!\ten}
        \Gamma_{a_1 \cdots a_8}\,
        \psi
      \big)
      e^{a_1} \cdots e^{a_8}
      \\
      F_{12} 
      &:=&
      \tfrac{1}{10!}
      \big(\hspace{1pt}
        \overline{\psi}
        \,
        \Gamma_{a_1 \cdots a_{\ten}}\,
        \psi
      \big)
      e^{a_1} \cdots e^{a_{\ten}}
      \\
      F_{14+2k} &=& 0
      \,,\;\;\;\;
      k \in \mathbb{N}
    \end{array}
    \right\}
    \;\in\;
    \mathrm{CE}\big(
      \mathbb{R}^{
        1,9\,\vert\,
        \mathbf{16}\oplus\overline{\mathbf{16}}
    }
    \big)
    \hspace{.5cm}
    \mbox{\rm s.t.}
    \hspace{.5cm}
    \left\{\!\!\!
    \def\arraystretch{1.3}
    \begin{array}{l}
       \mathrm{d}\, H_3^A 
       \;=\;
       0
       \\
       \mathrm{d}\, F_{2\bullet+2}
       \;=\;
       H_3^A \, F_{2\bullet}
       \,,\;\;\;
    \end{array}
    \right.
  \end{equation}
  hence equivalently constituting a super-$L_\infty$ homomorphism
  \eqref{LieHomomorphism} 
  to the real Whitehead $L_\infty$-algebra of twisted $\mathrm{KU}_0$ \eqref{WhiteheadAlgOfTwistedKU}:
  \begin{equation}
    \label{HigherIIACocyclesAsSingleCocycle}
    \begin{tikzcd}[row sep=-2pt, 
      column sep=50pt
    ]
      \mathbb{R}^{
        1,9\,\vert\,\mathbf{16}\oplus\overline{\mathbf{16}}
      }
      \ar[
        rr,
        "{
          \scalebox{1}{$
            \scalebox{1.3}{$($}
              H_3^{\color{purple}A}, (F_{2k})_{k \in \mathbb{Z}}
            \scalebox{1.3}{$)$}
          $}
        }"
      ]
      \ar[
        dr,
        "{
          H_3^{\mathcolor{purple}A}
        }"{swap}
      ]
      &&
      \mathfrak{l}
      \big(
        \Sigma^0
        \mathrm{KU}
        \!\sslash\!\!
        B\mathrm{U}(1)
      \big)  \mathrlap{\,.}
      \ar[
        dl,
        "{
          h_3
        }"
      ]
      \\
      &
      \mathfrak{l}
      B^2\mathrm{U}(1)
    \end{tikzcd}
  \end{equation}
\end{proposition}

\noindent
\begin{example}[\bf Reduction of IIA super-cocycles to 9d]
The reduction of
of the type IIA cocycles (Prop. \ref{TheTypeIIACocyclesSummarized})
via the Ext/Cyc-adjunction (Prop. \ref{TheExtCycAdjunction})  along the IIA/9D extension (Ex. \ref{IIAasExtensionOf9d})
$$
  \hspace{-1mm}
  \begin{tikzcd}[
    column sep=30pt
  ]
    \mathbb{R}^{
      1,9\,\vert\,
      \mathbf{16}\oplus\overline{\mathbf{16}}
    }
    \ar[
      rr,
      "{
          \scalebox{1}{$
            \scalebox{1.3}{$($}
              H_3^A, (F_{2k})_{k \in \mathbb{Z}}
            \scalebox{1.3}{$)$}
          $}
      }"
    ]
    &&
    \mathfrak{l}\big(
      \Sigma^0
      \mathrm{KU}
      \!\sslash\!\!
      B\mathrm{U}(1)
    \big)
  \end{tikzcd}
  \;
  \leftrightsquigarrow
  \;\,
  \begin{tikzcd}[
    column sep=55pt,
    row sep=0pt
  ]
    \mathbb{R}^{
      1,8\,\vert\,
      \mathbf{16}\oplus
      \mathbf{16}
    }
    \ar[
     rr,
     "{
        \mathrm{rdc}_{c_1^A}
          \scalebox{1}{$
            \scalebox{1.3}{$($}
              H_3^A, (F_{2k})_{k \in \mathbb{Z}}
            \scalebox{1.3}{$)$}
          $}     
     }"
    ]
    \ar[
      dr,
      "{
        c_1^A
        \,=\,
        (\hspace{1pt}
          \overline{\psi}
          \,\Gamma^9\,
          \psi
        )
      }"{swap, sloped, pos=.4}
    ]
    &&
    \mathrm{cyc}
    \,
    \mathfrak{l}\big(
      \Sigma^0\mathrm{KU}
      \!\sslash\!\!
      B\mathrm{U}(1)
    \big)    
    \ar[
      dl,
      "{ \omega_2 }"
    ]
    \\
    &
    b\mathbb{R}
  \end{tikzcd}
$$
gives the following system of super-invariants in 9D, where on the right we show their equivalent incarnation as having coefficients either in the cyclification of twisted $\mathrm{KU}_0$ or of twisted $\mathrm{KU}_1$, via the first isomorphism in \eqref{IsoOfDoublyCyclifiedTwistedKSpectra}:

\begin{equation}
  \label{CocyclesIn9D}
  \begin{tikzcd}[row sep=-3pt,
   column sep=10pt
  ]
    \mathbb{R}^{
      1,8\,\vert\,
      \mathbf{16}
      \oplus
      \mathbf{16}
    }
    \ar[
      rr,
      "{
          \scalebox{1}{$
            \mathrm{rdc}_{
              c_1^A
            }
            \scalebox{1.3}{$($}
              H_3^A, (F_{2k})_{k \in \mathbb{Z}}
            \scalebox{1.3}{$)$}
          $}
      }",
      "{
        \scalebox{.7}{
          \color{gray}
          \eqref{HigherIIACocyclesAsSingleCocycle}
        }
      }"{swap}
    ]
    &&
    \mathrm{cyc}
    \,
    \mathfrak{l}\big(
      \Sigma^{\color{purple}0}
      \mathrm{KU}
      \!\sslash\!\!
      B\mathrm{U}(1)
    \big)    
    \ar[
      rr,
      "{ \sim }",
      "{
        \scalebox{.7}{
          \color{gray}
          \eqref{IsoOfDoublyCyclifiedTwistedKSpectra}
        }
      }"{swap}
    ]
    &&
    \mathrm{cyc}
    \,
    \mathfrak{l}\big(
      \Sigma^{\color{purple}1}
      \mathrm{KU}
      \!\sslash\!\!
      B\mathrm{U}(1)
    \big)    
    \\
    \mathllap{
      +c_1^A
      \;=\;\;\;\;\;\;\;
    }
    \big(\hspace{1pt}
      \overline{\psi}
      \,\Gamma_9\,
      \psi
    \big)
    &\longmapsfrom&
    \omega_2
    &\longmapsfrom&
    -
    \mathrm{s}h_3
    \\
    \big(\hspace{1pt}
      \overline{\psi}
      \,
      \Gamma_a
      \Gamma_{\!\ten}
      \,
      \psi
    \big)
    e^a
    &\longmapsfrom&
    h_3
    &\longmapsfrom&
    h_3
    \\
    \mathllap{
      -c_1^B
      \;=\;\;\;
    }
    -
    \big(\hspace{1pt}
      \overline{\psi}
      \,
      \Gamma_9
      \Gamma_{\!\ten}
      \,
      \psi
    \big)
    &\longmapsfrom& 
    \mathrm{s}h_3
    &\longmapsfrom&
    -\omega_2
    \\
    0
    &\longmapsfrom&
    f_{\leq 0}
    &\longmapsfrom&
    \mathrm{s}\!f_{\leq 1}
    \\
    0
    &\longmapsfrom&
    \mathrm{s}\!f_{\leq 0}
    &\longmapsfrom&
    f_{\leq 1}
    \\
    \big(\hspace{1pt}
      \overline{\psi}
      \,\Gamma_{\!\ten}\,
      \psi
    \big)
    &\longmapsfrom&
    f_2
    &\longmapsfrom&
    \mathrm{s}\!f_{3}
    \\
    0
    &\longmapsfrom&
    \mathrm{s}\!f_2
    &\longmapsfrom&
    f_1
    \\
    \tfrac{1}{2}
    \big(\hspace{1pt}
      \overline{\psi}
      \,\Gamma_{a_1 a_2}\,
      \psi
    \big)
    e^{a_1}e^{a_2}
    &\longmapsfrom&
    f_4
    &\longmapsfrom&
    \mathrm{s}\!f_5
    \\
    \big(\hspace{1pt}
      \overline{\psi}
      \,
      \Gamma_a
      \Gamma_9
      \,
      \psi
    \big)
    e^a
    &\longmapsfrom&
    \mathrm{s}\!f_4
    &\longmapsfrom&
    f_3
    \\
    \tfrac{1}{4!}
    \big(\hspace{1pt}
      \overline{\psi}
      \,
      \Gamma_{a_1 \cdots a_4}
      \Gamma_{\!\ten}
      \,
      \psi
    \big)
    e^{a_1}\cdots e^{a_4}
    &\longmapsfrom&
    f_6
    &\longmapsfrom&
    \mathrm{s}\!f_7
    \\
    \tfrac{1}{3!}
    \big(\hspace{1pt}
      \overline{\psi}
      \Gamma_{a_1 a_2 a_3}
      \Gamma_9
      \Gamma_{\!\ten}
      \,
      \psi
    \big)
    e^{a_1} e^{a_2} e^{a_3}
    &\longmapsfrom&
    \mathrm{s}\!f_6
    &\longmapsfrom&
    f_5
    \\
    \tfrac{1}{6!}
    \big(\hspace{1pt}
      \overline{\psi}
      \,\Gamma_{a_1 \cdots a_6}\,
      \psi
    \big)
    e^{a_1} \cdots e^{a_6}
    &\longmapsfrom&
    f_8
    &\longmapsfrom&
    \mathrm{s}\!f_{9}
    \\
    \tfrac{1}{5!}
    \big(\hspace{1pt}
      \overline{\psi}
      \,
      \Gamma_{a_1 \cdots a_5}
      \Gamma_9
      \,
      \psi
    \big)
    e^{a_1} \cdots e^{a_5}
    &\longmapsfrom&
    \mathrm{s}\!f_8
    &\longmapsfrom&
    f_7
    \\
    \tfrac{1}{8!}
    \big(\hspace{1pt}
      \overline{\psi}
      \,
      \Gamma_{a_1 \cdots a_8}
      \Gamma_{\!\ten}
      \,
      \psi
    \big)
    e^{a_1}\cdots e^{a_8}
    &\longmapsfrom&
    f_{10}
    &\longmapsfrom&
    \mathrm{s}\!f_{11}
    \\
    \tfrac{1}{7!}
    \big(\hspace{1pt}
      \overline{\psi}
      \,
      \Gamma_{a_1 \cdots a_7}
      \Gamma_9
      \Gamma_{\!\ten}
      \,
      \psi
    \big)
    e^{a_1}\cdots e^{a_7}
    &\longmapsfrom&
    \mathrm{s}\!f_{10}
    &\longmapsfrom&
    f_{9}
    \\
    \tfrac{1}{10!}
    \big(\hspace{1pt}
      \overline{\psi}
      \,
      \Gamma_{a_1 \cdots a_{\ten}}\,
      \psi
    \big)
    e^{a_1}\cdots e^{a_{\ten}}
    &\longmapsfrom&
    f_{12}
    &\longmapsfrom&
    \mathrm{s}\!f_{13}
    \\
    \tfrac{1}{9!}
    \big(\hspace{1pt}
      \overline{\psi}
      \,
      \Gamma_{a_1 \cdots a_{9}}
      \Gamma_9
      \,
      \psi
    \big)
    e^{a_1}\cdots e^{a_{9}}
    &\longmapsfrom&
    \mathrm{s}\!f_{12}
    &\longmapsfrom&
    f_{11}
    \\
    0
    &\longmapsfrom&
    f_{\geq 14}
    &\longmapsfrom&
    \mathrm{s}\!f_{\geq 15}
    \\
    0
    &\longmapsfrom&
    \mathrm{s}\!f_{\geq 14}
    &\longmapsfrom&
    f_{\geq 13}
  \end{tikzcd}
\end{equation}
\end{example}

\smallskip

\noindent
{\bf T-dualization.}
Using the super-$L_\infty$ machinery, we now obtain with mechanical ease the T-dual version of Prop. \ref{TheTypeIIACocyclesSummarized} (the latter was claimed also in \cite[\S 2]{Sakaguchi00}):

We can now give a streamlined proof of the following statement from {\cite[Prop. 4.10]{FSS18-TDualityA}}:
\begin{proposition}[{\bf The type IIB super-cocycles}]
\label{TheTypeIIBCocyclesSummarized}
  On the type IIB super-Minkowski spacetime $\mathbb{R}^{1,9\,\vert\, \mathbf{16}\oplus \mathbf{16}}$ {\rm (Def. \ref{TypeIIBSuperMinkowski})} we have the following super-invariants
  \begin{equation}
    \label{IIBCocycles}
    \left.
    \def\arraystretch{1.5}
    \begin{array}{ccl}
      H_3^B
      &=&
      \big(\hspace{1pt}
        \overline{\psi}
        \,
        \Gamma^B_a 
        \,\Gamma_{\!\ten}
        \,
        \psi
      \big)
      e^a
      \\
      F_{\leq 1}
      &=& 0
      \\
      F_3 &=&
      \;
      \big(\hspace{1pt}
        \overline{\psi}
        \,
        \Gamma^B_a
        \Gamma_9
        \,
        \psi
      \big)
      e^a
      \\
      F_5
      &=&
      \tfrac{1}{3!}
      \big(\hspace{1pt} 
        \overline{\psi}
        \,
        \Gamma^B_{a_1 a_2 a_3}
        \Gamma_9\Gamma_{\!\ten}
        \,
        \psi
      \big)
      e^{a_1} e^{a_2} e^{a_3}
      \\
      F_7
      &=&
      \tfrac{1}{5!}
      \big(\hspace{1pt} 
        \overline{\psi}
        \,
        \Gamma^B_{a_1 \cdots a_5}
        \Gamma_9
        \,
        \psi
      \big)
      e^{a_1} \cdots e^{a_5}
      \\
      F_9
      &=&
      \tfrac{1}{7!}
      \big(\hspace{1pt} 
        \overline{\psi}
        \,
        \Gamma^B_{a_1 \cdots a_7}
        \Gamma_9\Gamma_{\!\ten}
        \,
        \psi
      \big)
      e^{a_1} \cdots e^{a_7}
      \\
      F_{11}
      &=&
      \tfrac{1}{7!}
      \big(\hspace{1pt} 
        \overline{\psi}
        \,
        \Gamma^B_{a_1 \cdots a_9}
        \Gamma_9
        \,
        \psi
      \big)
      e^{a_1} \cdots e^{a_9}
            \\
      F_{13+2k} &=& 0
      \,,\;\;\;\;
      k \in \mathbb{N}
    \end{array}
    \right\}
    \;\in\;
    \mathrm{CE}\big(
      \mathbb{R}^{
        1,9 \,\vert\,
        \mathbf{16}\oplus \mathbf{16}
      }
    \big)
    \;\;\;\;
    \mbox{s.t.}
    \;\;\;\;
    \left\{\!\!\!
    \def\arraystretch{1.2}
    \begin{array}{ccl}
      \mathrm{d}\,
      H_3^B &=& 0
      \\
      \mathrm{d}\,
      F_{2\bullet+1}
      &=&
      H_3^B \, F_{2\bullet-1}
      \mathrlap{\,,}
    \end{array}
    \right.
  \end{equation}
  where the $\Gamma^B_{a_1 \cdots a_p}$ are from \eqref{CliffordBasisIIB}, 
  hence equivalently constituting a super-$L_\infty$ homomorphism
  \eqref{LieHomomorphism} 
  to the real Whitehead $L_\infty$-algebra of twisted $\Sigma^{\color{purple}1}\mathrm{KU}$ \eqref{WhiteheadAlgOfTwistedKU}:
  \begin{equation}
    \label{HigherIIBCocyclesAsSingleCocycle}
    \begin{tikzcd}[row sep=-4pt, 
      column sep=55pt
    ]
      \mathbb{R}^{
        1,9\,\vert\,\mathbf{16}\oplus\mathbf{16}
      }
      \ar[
        rr,
        "{
          \scalebox{1}{$
            \scalebox{1.3}{$($}
              H_3^{\mathcolor{purple}B}, (F_{2k+1})_{k \in \mathbb{Z}}
            \scalebox{1.3}{$)$}
          $}
        }"
      ]
      \ar[
        dr,
        "{
          H_3^{\mathcolor{purple}B}
        }"{swap}
      ]
      &&
      \mathfrak{l}
      \big(
        \Sigma^{\color{purple}1}
        \mathrm{KU}
        \!\sslash\!\!
        B\mathrm{U}(1)
      \big)
      \mathrlap{\,.}
      \ar[
        dl,
        "{
          h_3
        }"
      ]
      \\
      &
      \mathfrak{l}\big(
        B^2 \mathrm{U}(1)
      \big)
    \end{tikzcd}
  \end{equation}
\end{proposition}
\begin{proof}
  This is the {\it T-dual} statement of Prop. \ref{TheTypeIIACocyclesSummarized} in that the super-invariants in \eqref{IIBCocycles} are the result of an application of the composite operation from Thm. \ref{TwistedKTheoryRedIsoReOxi}, namely:

\vspace{-.3cm}
\hspace{-.8cm}
\def\tabcolsep{2pt}
  \begin{itemize}[
    leftmargin=.7cm,
    itemsep=2pt,
    topsep=2pt
  ]
  \item[\bf (i)] reducing \eqref{CyclificationHomIsomorphism}
  the IIA super-invariants
  \eqref{HigherIIACocyclesAsSingleCocycle}
  along the type IIA fibration \eqref{IIAExtension} to 9d,

  \item[\bf (ii)] equivalently re-regarding their coefficients in the cyclification of twisted $\mathrm{KU}_{\mathcolor{purple}1}$ instead of twisted $\mathrm{KU}_{\mathcolor{purple}0}$, via \eqref{IsoOfDoublyCyclifiedTwistedKSpectra}, noticing that this swaps the Chern class from that classifying the type IIA extension to that classifying the IIB extension \eqref{IIBExtension},

  \item[\bf (iii)] re-oxidizing \eqref{CyclificationHomIsomorphism} the result, but now along the IIB fibration \eqref{IIBExtension}:
\end{itemize}

\vspace{-1mm} 
\begin{equation}
  \label{TDualityDiagram}
  \begin{tikzcd}[
    row sep=15pt, 
    column sep=40pt
  ]
    &[-30pt]
    &[-75pt]
    &[+5pt]
    &
    &[-60pt]
    \mathfrak{l}
    \big(
      \Sigma^{\color{purple}1}
      \mathrm{KU}
      \!\sslash\!\! B\mathrm{U}(1)
    \big)
   \\[-20pt]
    &[-30pt]
    &[-75pt]
    \mathbb{R}^{
      1,9\,\vert\,\mathbf{16}\oplus \mathbf{16}
    }
    \ar[
      urrr,
      "{
        \mathrm{oxd}_{c_1^{\color{purple}B}}
        \scalebox{.9}{$
          {\mathcolor{blue}T}
          \scalebox{1.5}{$($}
          \mathrm{rdc}_{c_1^{\color{purple}A}}
          \scalebox{1.3}{$($}
            H_3^{\mathcolor{purple}A},\, (F_{2k})_{k \in \mathbb{Z}}
          \scalebox{1.3}{$)$}
          \scalebox{1.5}{$)$}
        $}
      }"{sloped},
      "{
        \scalebox{.9}{$
          =\,
          \scalebox{1.3}{$($}
            H_3^{\mathcolor{purple}B},\, (F_{2k+1})_{k \in \mathbb{Z}}
          \scalebox{1.3}{$)$}
        $}
      }"{sloped, swap}
    ]
    \ar[
      ddddl,
      ->>
    ]
    &[+5pt]
    &
    &[-60pt]
    \\[-12pt]
    \mathbb{R}^{
      1,9\,\vert\,
      \mathbf{16}\oplus\overline{\mathbf{16}}
    }
    \ar[
      rrrrd,
      crossing over,
      "{
        \scalebox{.9}{$
          \scalebox{1.3}{$($}
            H_3^{\mathcolor{purple}A},\, (F_{2k})_{k \in \mathbb{Z}}
          \scalebox{1.3}{$)$}
        $}
      }"{sloped, pos=.63}
    ]
    \ar[
      dddr,
      ->>
    ]
    &
    &&&
    &[-30pt]
    \\[-32pt]
    &&&&
    \mathfrak{l}
    \big(
      \Sigma^{\color{purple}0}
      \mathrm{KU}
      \!\sslash\!\! B\mathrm{U}(1)
    \big)
    \\
    &
    &&&
    &[-30pt]
    \mathrm{cyc}\,
    \mathfrak{l}
    \big(
    \Sigma^{\color{purple}1}
    \mathrm{KU}
    \!\sslash\!\! B\mathrm{U}(1)
    \big)
    \\[-10pt]
    &
    \mathbb{R}^{
      1,8\,\vert\,
      \mathbf{16}\oplus\mathbf{16}
    }
    \ar[
      drrr,
      shorten <=-7pt,
      "{
        \scalebox{.9}{$
          \mathrm{rdc}_{c_1^{\color{purple}A}}
          \scalebox{1.3}{$($}
            H_3^{\mathcolor{purple}A},\, (F_{2k})_{k \in \mathbb{Z}}
          \scalebox{1.3}{$)$}
        $}
      }"{sloped,swap, pos=.35}
    ]
    \ar[
      urrrr,
      "{
        \scalebox{.9}{$
          {\mathcolor{blue}T}
          \scalebox{1.5}{$($}
          \mathrm{rdc}_{c_1^{\color{purple}A}}
          \scalebox{1.3}{$($}
            H_3^{\mathcolor{purple}A},\, (F_{2k})_{k \in \mathbb{Z}}
          \scalebox{1.3}{$)$}
          \scalebox{1.5}{$)$}
        $}
      }"{sloped}
    ]
    &&
    \\
    &
    &&&
    \mathrm{cyc}\,
    \mathfrak{l}
    \big(
    \Sigma^{\color{purple}0}
    \mathrm{KU}
    \!\sslash\!\! B\mathrm{U}(1)
    \big)
    \ar[
      uur,
      "{
         \mathcolor{blue}
         T 
       }",
      "{ 
        \sim 
      }"{swap, sloped}
    ]
  \end{tikzcd}
\end{equation}

\vspace{-1mm}

By the Ext/Cyc-adjunction (Prop. \ref{TheExtCycAdjunction}),
the result of this process is guaranteed to be a super-$L_\infty$ homomorphism of the form shown in the top right of the above diagram,
which implies the claimed Bianchi identities \eqref{HigherIIBCocyclesAsSingleCocycle}. 

Hence, all that remains to be shown is that the super-invariants produced by this process are indeed those shown in \eqref{IIBCocycles}. This is a straightforward matter of plugging the IIA super-invariants \eqref{HigherIIACocyclesAsSingleCocycle} into the ``winding / non-winding swapping'' formula \eqref{Red-Iso-ReoxiAction}, or equivalently  plugging the reduced 9d super-invariants \eqref{CocyclesIn9D} into the oxidation formula \eqref{CyclificationHomBijection}: 
\vspace{-2.5mm}
\begin{equation}
  \label{ObtainingTheIIBfluxes}
  \def\arraystretch{2.2}
  \begin{array}{ccl}
    F_1
    &=&
    \grayoverbrace{
      0
    }{
      f_1
    }
    \;-\;
    e^9
    \grayoverbrace{
      0
    }
    { \mathrm{s}\!f_1 }
    \\[-13pt]
    &=& 0\;,
    \\[-5pt]
    F_3
    &=&
    \grayoverbrace{
    \textstyle{\underset{\scalebox{.58}{$a < 9$}}{\sum}}
    \big(\hspace{1pt}
      \overline{\psi}
      \,\Gamma_a\Gamma_9\,
      \psi
    \big)
    e^a
    }{
      f_3
    }
    \,-\,  
    e^9
    \grayunderbrace{
    \grayoverbrace{
    \big(\hspace{1pt}
      \overline{\psi}
      \,
      \Gamma_{\!\ten}
      \,
      \psi
    \big)
    }{
      \mathrm{s}\!f_3
    }
    }{
      \mathclap{
        \scalebox{.8}{$
        -
        \big(\hspace{1pt}
          \overline{\psi}
          \,
          \underbrace{
            \Gamma_9 \Gamma_{\!\ten}
          }_{
            \Gamma_9^B
          }
          \Gamma_9
          \,
          \psi
        \big)
        $}
      }
    }
    \\[-20pt]
    &=&
    \big(\hspace{1pt}
      \overline{\psi}
      \,\Gamma^B_a\Gamma_9\,
      \psi
    \big)
    e^a,
    \\
    F_5
    &=&
    \grayoverbrace{
    \textstyle{\underset{\scalebox{.58}{$a < 9$}}{\sum}}
    \tfrac{1}{3!}
    \big(\hspace{1pt}
      \overline{\psi}
      \,\Gamma_{a_1 a_2 a_2}
      \Gamma_9\Gamma_{\!\ten}
      \,
      \psi
    \big)
    e^{a_1} e^{a_2} e^{a_3}
    }{
      f_5
    }
    \;-\;
    e^9
    \grayoverbrace{
    \textstyle{\underset{\scalebox{.58}{$a < 9$}}{\sum}}
    \tfrac{1}{2}
    \grayunderbrace{
    \big(\hspace{1pt}
      \overline{\psi}
      \,\Gamma_{a_1 a_2}\,
      \psi
    \big)
    }{
      \mathclap{
      \scalebox{.8}{$
        -
        \big(\hspace{1pt}
          \overline{\psi}
          \,
          \Gamma_{a_1 a_2}
          \underbrace{
          \Gamma_{9}\Gamma_{\!\ten}
          }_{\Gamma_9^B}
          \Gamma_9\Gamma_{\ten}
          \,
          \psi
        \big)
      $}
      }
    }
    e^{a_1}e^{a_2}
    }{
      \mathrm{s}\!f_5
    }
    \\[-20pt]
    &=&
    \tfrac{1}{3!}
    \big(\hspace{1pt}
      \overline{\psi}
        \Gamma^B_{a_1 a_2 a_3}
        \Gamma_{9}\Gamma_{\!\ten}
      \psi
    \big)
    e^{a_1} e^{a_2} e^{a_3},
    \\[+2pt]
    F_7
    &=&
    \grayoverbrace{
    \textstyle{\underset{\scalebox{.58}{$a < 9$}}{\sum}}
    \tfrac{1}{5!}
    \big(\hspace{1pt}
      \overline{\psi}
      \,
      \Gamma_{a_1 \cdots a_5}
      \Gamma_9
      \,
      \psi
    \big)
    e^{a_1} \cdots e^{a_5}
    }{
      f_7
    }
    \;-\;
    e^9
    \grayoverbrace{
    \tfrac{1}{4!}
    \grayunderbrace{
    \textstyle{\underset{\scalebox{.58}{$a < 9$}}{\sum}}
    \big(\hspace{1pt}
      \overline{\psi}
      \,\Gamma_{a_1 \cdots a_4}
      \Gamma_{\!\ten}\,
      \psi
    \big)
    }{
      \mathclap{
      \scalebox{.8}{$
      -
      \big(\hspace{1pt}
        \overline{\psi}
        \,\Gamma_{a_1 \cdots a_4}
        \underbrace{
          \Gamma_9 \Gamma_{\!\ten}
        }_{
          \Gamma_9^B
        }
        \Gamma_9\,
        \psi
      \big)
      $}
      }
    }
    e^{a_1} \cdots e^{a_4}
    }{
      \mathrm{s}\!f_7
    }
    \\[-20pt]
    &=&
    \tfrac{1}{5!}
    \big(\hspace{1pt}
      \overline{\psi}
      \,
      \Gamma^B_{a_1 \cdots a_5}
      \Gamma_9
      \,
      \psi
    \big)
    e^{a_1} \cdots e^{a_5},
    \\[+1pt]
    F_9
    &=&
    \grayoverbrace{
    \textstyle{\underset{\scalebox{.58}{$a_i < 9$}}{\sum}}
    \tfrac{1}{7!}
    \big(\hspace{1pt}
      \overline{\psi}
      \,
      \Gamma_{a_1 \cdots a_7}
      \Gamma_9
      \Gamma_{\!\ten}
      \,
      \psi
    \big)
    e^{a_1}\cdots e^{a_7}
    }{
      f_9
    }
    \;-\;
    e^9
    \grayoverbrace{
    \textstyle{\underset{\scalebox{.58}{$a_i < 9$}}{\sum}}
    \tfrac{1}{6!}
    \grayunderbrace{
    \big(\hspace{1pt}
      \overline{\psi}
      \,\Gamma_{a_1 \cdots a_6}\,
      \psi
    \big)
    }{
      \mathclap{
      \scalebox{.8}{$
        -
        \big(\hspace{1pt}
          \overline{\psi}
          \,\Gamma_{a_1 \cdots a_6}
          \underbrace{
          \Gamma_9\Gamma_{\!\ten}
          }_{\Gamma_9^B}
          \Gamma_9\Gamma_{\!\ten}\,
          \psi
        \big)
      $}
      }
    }
    e^{a_1 \cdots a_6}
    }{
      \mathrm{s}\!f_9
    }
    \\[-15pt]
    &=&
    \tfrac{1}{7!}
    \big(\hspace{1pt}
      \overline{\psi}
      \,\Gamma^B_{a_1 \cdots a_7}
      \Gamma_9\Gamma_{\!\ten}
      \,
      \psi
    \big)
    e^{a_1}\cdots e^{a_7},
    \\[+2pt]
    F_{11}
    &=&
    \grayoverbrace{
    \textstyle{\underset{\scalebox{.58}{$a_i < 9$}}{\sum}}
    \tfrac{1}{9!}
    \big(\hspace{1pt}
      \overline{\psi}
      \,\Gamma_{a_1 \cdots a_9}
      \Gamma_9\,
      \psi
    \big)
    e^{a_1} \cdots e^{a_9}
    }{
      f_{11}
    }
    \;-\;
    e^9
    \grayoverbrace{
    \textstyle{\underset{\scalebox{.58}{$a_i < 9$}}{\sum}}
    \tfrac{1}{8!}
    \grayunderbrace{
    \big(\hspace{1pt}
      \Gamma_{a_1 \cdots a_8}
      \Gamma_{\!\ten}
    \big)
    }{
      \mathclap{
        \scalebox{.8}{$
          -
          \big(\hspace{1pt}
            \overline{\psi}
            \,\Gamma_{a_1 \cdots a_8}
            \underbrace{
            \Gamma_9 \Gamma_{\!\ten}}_{
              \Gamma_9^B
            }
            \Gamma_9\,
            \psi
          \big)
        $}
      }
    }
    e^{a_1} \cdots e^{a_8}
    }{
      \mathrm{s}\!f_{11}
    }
    \\[-10pt]
    &=&
    \tfrac{1}{9!}\big(\hspace{1pt}
      \overline{\psi}
      \,\Gamma^B_{a_1 \cdots a_9}
      \Gamma_9\,
      \psi
    \big)
    e^{a_1}\cdots e^{a_9}.
  \end{array}
\end{equation}

\vspace{-6mm}
\end{proof}

\begin{remark}[\bf Lorentz invariance of IIB super-fluxes]
  While the proof of Prop. \ref{TheTypeIIBCocyclesSummarized} does not make manifest that the resulting super-translation invariants \eqref{IIBCocycles} are also $\mathrm{Spin}(1,9)$-invariant, 
  this is immediate by Prop. \ref{LorentzInvariantsOnIIBSuperSpacetimes}.
\end{remark}

\begin{remark}[{\bf T-Dual NS-Flux and topological T-duality} {\cite[Rem. 5.4]{FSS18-TDualityA}}]
\label{TDualNSFlux}
$\,$

\noindent {\bf (i)}  The action of the T-duality operation from Prop. \ref{TheTypeIIBCocyclesSummarized} on the NS-flux densities $H_3^{A/B}$ is particularly interesting: Note that both these fluxes come out as the sum of (1.) the basic $H_3$-flux in 9d (pulled back to 10d along the corresponding fibration) with (2.) the fiber form $e^9$ times the Chern class classifying the {\it other} extension:
\vspace{-2mm} 
\begin{equation}
  \label{TDualityOnNSFlux}
  \hspace{-4.5mm} 
  \begin{tikzcd}[row sep=-3pt,
    column sep=43pt,
    ampersand replacement=\&
  ]
  \def\arraystretch{1.4}
  \def\arraycolsep{3pt}
  \begin{array}{ccccl}
    H_3^{\color{purple}A}
    \!&\!=\!&\!\
    \grayoverbrace{
      H_3 
    }{
      h_3
    }
    &+&
    e^9_{A}
    \grayoverbrace{
    \big(\hspace{1pt}
      \overline{\psi}
      \,
      \Gamma_9
      \Gamma_{\!\ten}
      \,
      \psi
    \big)
    }{
     - \mathrm{s}h_3
    }
    \\
    \!&=\!&
    H_3 
    &+&
    e^9
    \;\;\;\;
    c_1^{\color{purple}B}
  \end{array}
  \ar[
    dr,
    |->,
    "{
      \mathrm{rdc}_{c_1^A}
    }"{sloped},
    "{
      \scalebox{.7}{
        \color{darkgreen}
        \bf
        \def\arraystretch{.9}
        \begin{tabular}{c}
          reduction along
          \\
          IIA extension
        \end{tabular}
      }
    }"{swap, sloped, yshift=-1pt, pos=.35}
  ]
 \!\! \ar[
    rr,
    |->,
    shorten=5pt,
    "{
      \scalebox{.7}{
        \color{darkgreen}
        \bf
        \begin{tabular}{c}
          superspace
          T-duality
        \end{tabular}
      }
    }"
  ]
  \hspace{-2mm}
  \&\&
  \hspace{-3mm}
  \def\arraystretch{1.4}
  \def\arraycolsep{3pt}
  \begin{array}{ccccl}
    H_3^{\color{purple}B}
    &=&
    H_3 
    &+&
    e^9
    \big(\hspace{1pt}
      \overline{\psi}
      \,\Gamma_9\,
      \psi
    \big)
    \\
    &=&
    H_3 
    &+&
    e^9_B
    \;\;\;\;
    c_1^{\color{purple}A}
  \end{array}
  \ar[
    dl,
    <-|,
    "{
      \mathrm{oxd}_{c_1^B}
    }"{sloped},
    "{
      \scalebox{.7}{
        \color{darkgreen}
        \bf
        \def\arraystretch{.9}
        \begin{tabular}{c}
          oxidation along
          \\
          IIB extension
        \end{tabular}
      }
    }"{swap, sloped, yshift=-2pt, pos=.35}
  ]
  \\
  \&
  \def\arraystretch{1.2}
  \def\arraycolsep{3pt}
  \begin{array}{cccc}
    H_3 &=&
    \underset{
    \scalebox{.6}{$a < 9$}
    }{\sum}
    e^a
    \big(\hspace{1pt}
      \overline{\psi}
      \,
      \Gamma_a
      \Gamma_{\!\ten}
      \,
      \psi
    \big)
    \\
    c_1^A
    &=&
    \big(\hspace{1pt}
      \overline{\psi}
      \,
      \Gamma_9
      \,
      \psi
    \big)
    \\
    - c_1^B
    &=&
    -\big(\hspace{1pt}
      \overline{\psi}
      \,
      \Gamma^B_9
      \,
      \psi
    \big)\, ,
  \end{array}
  \end{tikzcd}
\end{equation}

\vspace{0mm} 
\noindent where, by the closure of $H_3^A$ (or that of $H_3^B$) from \eqref{DDReducedMBraneCocycles}, the basic $H_3$-flux satisfies
\begin{align}\label{DifferentialOfBasicH3Flux}
\dd \, H_3 \, =\, - c_1^A \cdot c_1^B \, .  
\end{align}
\noindent {\bf (ii)}  In particular, this says that the fiber integration of $H_3^{A}$ from the type IIA spacetime down to 9d is the Chern class classifying the type IIB extension, and vice versa:
\begin{align}\label{FiberIntegrationOfH3Flux}
  p^A_\ast
  \,H_3^A\,
  \;=\;
  c_1^B
  \;\;\;\;\;\;
  \mbox{and}
  \;\;\;\;\;\;
  p^B_\ast
  \,H_3^B\,
  \;=\;
  c_1^A
  \, . 
\end{align}
\noindent {\bf (iii)}  The analogous phenomenon in ordinary T-duality (i.e., not on super-flux densities over super-spacetime as considered here) was originally proposed in \cite[(1.8)]{BEM04} and gave rise to the mathematical notion of ``topological T-duality''. 

\noindent {\bf (iv)}  While the formalism of topological T-duality has worked wonders, its actual relation to string/M-theory rests on educated guesses (though much progress was recently made when \cite{Waldorf24} related it to the Buscher rules). 
Here it is interesting that we find \eqref{TDualityOnNSFlux} this relation being hard-coded in the DNA of supergravity.

\noindent {\bf (v)} Also note that the form of the 3-flux in \eqref{TDualityOnNSFlux} is analogous, up to degrees, to the form of the 11D 7-flux in its closed Page-charge form,
$
  \widetilde{G}_7
  =
  G_7
  -
  \tfrac{1}{2}
  H_3\,G_4
  \,.
$
This analogy is made precise by the notion of higher T-duality discussed in \S\ref{HigherTduality}.
\end{remark}

\medskip

Next, we obtain a maybe more vivid perspective on this super-space T-duality by turning attention from the dimensionally reduced $L_\infty$-cocycles to the higher spacetime extensions that these classify, which reveals the appearance of {\it doubled} superspace:

\medskip

\noindent
{\bf Doubled super-spacetime and the Poincar{\'e} form.}
The following Prop. \ref{DoubledSuperSpaceAsHmotopyFiber} speaks of the {\it homotopy fiber} of non-abelian $L_\infty$-algebra cocycles. This concept -- standard in (rational)  homotopy theory --  is explained in some detail in \cite{FSS23-Char}, but the reader not to be bothered by such notions may take the following \eqref{CEOfStringExtendedDoubledSpace} as a definition. In any case, either abstractly or by inspection, one sees that the following serves to express an equivalent point of view on the above T-duality isomorphism.
\begin{lemma}[{\bf Extended doubled super-space as homotopy fiber of reduced 3-flux} {\cite[Prop. 7.5]{FSS18-TDualityA}}]
\label{DoubledSuperSpaceAsHmotopyFiber}
The homotopy fibers of the A/B-reduced \eqref{CyclificationHomIsomorphism} $H^{A/B}_3$-flux cocycle \eqref{TDualityOnNSFlux}, to be denoted
\begin{equation}
  \label{Reduced3FluxExtension}
  \begin{tikzcd}
    \widehat{
    \mathbb{R}^{
      1,8+(1+1)\,\vert\, 32
    }_{\mathcolor{purple}{A}/\mathcolor{purple}{B}}
    }
    \ar[
      rr,
      "{
        \mathrm{hofib}
      }"
    ]
    &&
    \mathbb{R}^{
      1,8\,\vert\, 
      \mathbf{16} 
      \oplus
      \mathbf{16}
    }
    \ar[
      rr,
      "{
        \mathrm{red}_{
          c^{\mathcolor{purple}{A}/\mathcolor{purple}{B}}_1
        }
        \big(
        H_3^{\mathcolor{purple}{A}/\mathcolor{purple}{B}}
        \big)
      }"
    ]
    &\phantom{--}&
    \mathrm{cyc}
    \mathfrak{l}
    B^2 \mathrm{U}(1)
  \end{tikzcd}
\end{equation}
are given by 
\begin{equation}
  \label{CEOfStringExtendedDoubledSpace}
  \mathrm{CE}\Big(
    \widehat{
    \mathbb{R}
      ^{1,8+(1+1)\,\vert\, 32}
      _{\mathcolor{purple}{A}/\mathcolor{purple}{B}}
      }
  \Big)
  \;\simeq\;
  \FDGCA
  \left[
  \def\arraystretch{1.3}
  \def\arraycolsep{1pt}
  \begin{array}{c}
    (\psi^\alpha)_{\alpha=1}^{32}
    \\
    (e^a)_{a = 0}^8
    \\
    e^9_A 
    \\
    e^9_B
    \\
    b_2
  \end{array}
  \right]
  \Big/
  \left(
  \def\arraystretch{1.3}
  \def\arraycolsep{1pt}
  \begin{array}{ccl}
    \mathrm{d}\,
    \psi &=& 0
    \\
    \mathrm{d}\,
    e^a &=& 
    \big(\hspace{1pt}
      \overline{\psi}
        \,\Gamma^a\,
      \psi
    \big)
    \\
    \mathrm{d}\,
    e^9_A
    &=&
    c_1^A
    \\
    \mathrm{d}\,
    e^9_B
    &=&
    c_1^B
    \\
    \mathrm{d}\, b_2
    &=&
    H_3^{\mathcolor{purple}{A}/\mathcolor{purple}{B}}
  \end{array}
  \!\!\right)
\end{equation}
and, as such, are equivalently 
further extensions of the type A/B string-extended super-spacetime (Ex. \ref{StringExtendedSuperSpace})
by a further (``doubled'') copy of the fiber coframe $e^9$.
\end{lemma}

\newpage 
It is natural to see that we have the following (first stated in {\cite[Prop. 6.2]{FSS18-TDualityA}}):
\begin{lemma}[{\bf T-duality of $\mathfrak{string}$-extended doubled super-spacetimes}]
\label{TDualityOfStringExtendedSuperSpacetimes}
  The equivalence 
  that is induced by the $\mathcolor{blue}{T}$-duality isomorphism \eqref{10ToroidalTDualityDiagram}
  between the doubled \& extended homotopy-fiber spaces \eqref{Reduced3FluxExtension} sends all generators to the generators of the
  same name, except for $b_2$ (the avatar of the string's ``B-field''), which is instead shifted by the ``Poincar{\'e} form'' {\rm (cf. Rem. \ref{PoincareFormInLiterature} below)}
  $
    P_2
    \;:=\;
    e^9_B e^9_A
  $:
\begin{equation}
  \begin{tikzcd}[
   column sep=-10pt,
   row sep=20pt
  ]
    b_2
    +
    \grayoverbrace{
    e^9_B
    \,
    e^9_A
    }{
      P_2
    }
    \ar[
      rrrr,
      <-|,
      shorten=15pt
    ]
    && &&
    b_2
    \\[-20pt]
    \widehat{
    \mathbb{R}^{
      1,8+(1+1)
      \,\vert\,
      32
    }_{
      \mathcolor{purple}{B}
    }
    }
    \ar[
      ddrr,
      "{
        \mathrm{hofib}
      }"{sloped}
    ]
    \ar[
      rrrr,
      <->,
      "{ \sim }"
    ]
    && &&
    \widehat{
    \mathbb{R}^{
      1,8+(1+1)
      \,\vert\,
      32
    }_{
      \mathcolor{purple}{A}
    }
    }
    \ar[
      ddll,
      "{
        \mathrm{hofib}
      }"{sloped}
    ]
    \\
    {}
    \\
    &&
    \mathbb{R}^{
      1,8
      \,\vert\,
      \mathbf{16}
      \oplus
      \mathbf{16}
    }
    \ar[
      dl,
      "{
        \mathrm{red}_{c_1^{\mathcolor{purple}{A}}}
        (
          H_3^{\mathcolor{purple}{A}},\,
          F_{2\bullet}
        )
      }"{swap, xshift=3pt,yshift=-3pt}
    ]
    \ar[
      dr,
      "{
        \mathrm{red}_{c_1^{\mathcolor{purple}{B}}}
        (
          H_3^{\mathcolor{purple}{B}},\,
          F_{2\bullet\mathcolor{purple}{+1}}
        )
      }"{xshift=-4pt, yshift=-3pt}
    ]
    \\
    &
    \mathrm{cyc}
    \mathfrak{l}
    \big(      \Sigma^{\mathcolor{purple}{0}}
      \mathrm{KU}
      \!\sslash\!
      B\mathrm{U}(1)
    \big)
    \ar[
      rr,
      <->,
      shorten=-1pt,
      "{
        \sim
      }"
    ]
    \ar[
      dl,
      ->>
    ]
    &
    {}
    \ar[
      d,
      phantom,
      "{
        \scalebox{.7}{
          \color{gray}
          \eqref{TDualityOnCyclifiedKFibration}
        }
      }"
    ]
    &
    \mathrm{cyc}
    \mathfrak{l}
    \big(
      \Sigma^{\mathcolor{purple}{1}}
      \mathrm{KU}
      \!\sslash\!
      B\mathrm{U}(1)
    \big)
    \ar[
      dr,
      ->>
    ]
    \\
    \mathrm{cyc}
    \mathfrak{l}
    B^2\mathrm{U}(1)
    \ar[
      rrrr,
      "{ \sim }",
    ]
    && 
    {}
    &&
    \mathrm{cyc}
    \mathfrak{l}
    B^2\mathrm{U}(1)
    \mathrlap{\,.}
  \end{tikzcd}
\end{equation}
\end{lemma}

This makes us to turn attention to the ``doubled'' or ``correspondence space'' of the IIA/B superspacetimes:
\begin{definition}[{\bf $\mathbb{R}^1$-Doubled super-space}, cf.  {\cite[Def. 6.1]{FSS18-TDualityA}}]
\label{1DoubledSuperSpace}
  Write $\mathbb{R}^{1,8=(1+1)\,\vert\, 32}$ for the super-Lie algebra given by extending the 9d super type II spacetime by {\it both} the Chern class classifying the IIA extension as well as that classifying the IIB extension:
  \begin{equation}
    \label{CEOf1DoubledSuperSpace}
    \mathrm{CE}\big(
      \mathbb{R}^{
        1,8+(1+1) \,\vert\, 32
      }
    \big)
    \;\;
    \simeq
    \;\;    
   \FDGCA
    \left[\!\!
    \def\arraystretch{1.3}
    \begin{array}{c}
      (\psi^\alpha)_{\alpha=0}^{32}
      \\
      (e^a)_{a=0}^8
      \\
      e_A^9,
      e_B^9
    \end{array}
   \! \right]
    \Big/
    \left(\!\!
    \def\arraystretch{1.3}
    \begin{array}{rl}
      \mathrm{d}\,
      \psi
      =&
      0
      \\
      \mathrm{d}\,
      e^a
      =&
      \big(\hspace{1pt}
        \overline{\psi}
        \,\Gamma^a\,
        \psi
      \big)
      \\
      \mathrm{d}\,
      e^9_{A/B}
      =&
      \big(\hspace{1pt}
        \overline{\psi}
        \,\Gamma^9_{A/B}\,
        \psi
      \big)
    \end{array}
   \!\!\! \right).
  \end{equation}

  \vspace{1mm} 
\noindent   This may be understood as the fiber product \eqref{CommutingPairOfExtensions} of the IIA- with the IIB extension, making a Cartesian square of super-Lie algebras, as follows:
  \vspace{-1mm} 
  $$
    \begin{tikzcd}[row sep=0pt, column sep=large]
      &
      \mathbb{R}^{
        1,8+(1+1)   
        \,\vert\,
         32
      }
      \ar[
        dr, 
        ->>,
        "{ \pi_B }"
      ]
      \ar[
        dl, 
        ->>,
        "{ \pi_A }"{swap}
      ]
      \ar[
        dd,
        phantom,
        "{
          \scalebox{.7}{
            \color{gray}
            (pb)
          }
        }"{pos=.3}
      ]
      \\
      \mathbb{R}^{
        1,9\,\vert\,
        \mathbf{16}
        \oplus
        \overline{\mathbf{16}}
      }
      \ar[
        dr,
        ->>
      ]
      &&
      \mathbb{R}^{
        1,9\,\vert\,
        \mathbf{16}
        \oplus
        \mathbf{16}
      }
      \ar[
        dl,
        ->>
      ]
      \\
      &
      \mathbb{R}^{
        1,8
        \,\vert\,
        \mathbf{16}
        \oplus
        \mathbf{16}
      }
    \end{tikzcd}
  $$
  We say that 
\begin{equation}
    \label{Poincare2FormForSingleFiber}
    P_2
    \;:=\;
    e^9_B \, e^9_A
    \;\in\;
    \mathrm{CE}\big(
      \mathbb{R}^{1,8+(1+1)\,\vert\,  32}
    \big)
  \end{equation}
  is the {\it twisted Poincar{\' e} super 2-form} (or just {\it Poincar{\'e} form}, for short) on the doubled super spacetime.
\end{definition}
With this we may concisely re-state Rem. \ref{TDualNSFlux} as follows:
\begin{proposition}[{\bf Poincar{\'e} form is coboundary for difference of T-dual NS super-fluxes}, cf. {\cite[Prop. 6.2]{FSS18-TDualityA}}]
\label{PoincarFormTrivializesDifferenceOfPullbacksOfNSFluxes}
  The pullbacks of the NS super-flux densities to the doubled super-spacetime \eqref{CEOf1DoubledSuperSpace} differ by the differential of the twisted Poincar{\'e} 2-form \eqref{Poincare2FormForSingleFiber}
  \begin{equation}
    \label{BianchiOfPoincare2Form}
    \mathrm{d} P_2
    \;=\;
    \pi_A^\ast H_3^A
    -
    \pi_B^\ast H_3^B
    \,.
  \end{equation}
\end{proposition}
\begin{proof}
Via Prop. \ref{TheTypeIIBCocyclesSummarized} and Rem. \ref{TDualNSFlux} this follows straightforwardly:
$$
  \def\arraystretch{1.4}
  \begin{array}{ccll}
    \pi_A^\ast H_3^A
    -
    \pi_B^\ast H_3^B
    &=&
    \; \Big(
    \;
    \underset{
      \scalebox{.6}{$a\!<\!9$}
    }{\sum}
    e^a
    \big(\hspace{1pt}
      \overline{\psi}
      \,\Gamma_a\,
      \psi
    \big)
    +
    e_A^9\, 
    \big(\hspace{1pt}
      \overline{\psi}
      \,\Gamma_9\Gamma_{\!\ten}\,
      \psi
    \big)
    \Big) \; -  \; 
     \Big(
    \; \underset{
      \scalebox{.6}{$a\!<\!9$}
    }{\sum}
    e^a
    \big(\hspace{1pt}
      \overline{\psi}
      \,\Gamma_a\,
      \psi
    \big)
    +
    e_B^9\, 
    \big(\hspace{1pt}
      \overline{\psi}
      \,\Gamma_9\,
      \psi
    \big)
    \Big) 
    \\
    &=&  e_A^9\, 
    \big(\hspace{1pt}
      \overline{\psi}
      \,\Gamma_9\Gamma_{\!\ten}\,
      \psi
    \big)
    \;-\; 
     e_B^9\, 
    \big(\hspace{1pt}
      \overline{\psi}
      \,\Gamma_9\,
      \psi
    \big) 
    \\
    &=&
    \mathrm{d}
    \big(
      e^9_B \, e^9_A
    \big)
    \;\defneq\;
    \mathrm{d}\, P_2
    \,.
  \end{array}
$$

\vspace{-6mm} 
\end{proof}

\newpage 
\begin{remark}[{\bf Poincar{\'e} 2-form and Buscher rules in T-duality literature} {\cite[Rem. 6.3]{FSS18-TDualityA}}]
\label{PoincareFormInLiterature}
$\,$

\noindent {\bf (i)} That the analog of the relation \eqref{BianchiOfPoincare2Form} should hold for ordinary T-duality (i.e., disregarding super-flux densities on super-spacetimes as considered here) was originally proposed by \cite[(1.12)]{BEM04}. As previously in Rem. \ref{TDualNSFlux}, here it is interesting to find these phenomena hard-coded in the DNA of supergravity.

\noindent {\bf (ii)} In fact, understanding the super-Lie algebraic content of Prop. \ref{PoincarFormTrivializesDifferenceOfPullbacksOfNSFluxes} through the lens of (super-)rational homotopy theory (essentially via Ex. \ref{WhiteheadLInfinityAlgebra}), it reproduces the image under rationalization of topological T-duality in the form proposed in \cite[Def. 2.8]{BRS06} in the sense of Def. \ref{TdualityCorrespondence} and Prop. \ref{TDuality/FourierMukaiIsomorphism} (in their purely even form).

\noindent {\bf (iii)}
A transparent understanding of how the (twisted) Poincar{\'e} 2-form and its Bianchi identity \eqref{BianchiOfPoincare2Form}
controls the {\it Buscher rules} of T-duality was more recently obtained in \cite[Lem. 3.3.1(c)]{Waldorf24}.
\end{remark}

\begin{remark}[\bf Classifier for the Poincar{\'e} 2-form]
  \label{ClassifierForPoincare2Form}
  The $L_\infty$-algebra which classifies the Bianchi identity \eqref{BianchiOfPoincare2Form} of the Poincar{\'e} 2-form is the homotopy fiber of the universal map that forms the difference of a pair of degree=3 classes:
  \begin{equation}
    \label{ClassifierAlgebraForPoinara2Form}
    \begin{tikzcd}[row sep=-3pt,
      column sep=0pt
    ] 
      \mathfrak{poin}_2
      \ar[
        r,
        "{
          \mathrm{hofib}
        }"
      ]
      &[30pt]
      \mathfrak{l}B^2\mathrm{U}(1)
      \times
      \mathfrak{l}B^2\mathrm{U}(1)
      \ar[
        rr
      ]
      &&
      \mathfrak{l}B^2\mathrm{U}(1)
      \\
      &
      \pi_L^\ast
      \omega_3
      -
      \pi_R^\ast
      \omega_3
      &\longmapsfrom&
      \omega_3
    \end{tikzcd}
  \end{equation}
  (where $\pi_{L/R}$ are the two projections out of the direct product). This is 
  given by
  \begin{equation}
    \label{CEOfClassifierOfPoincare2Form}
    \mathrm{CE}\big(
      \mathfrak{poin}_2
    \big)
    \;\;
    \simeq
    \;\;
    \FDGCA
    \left[\!\!
    \def\arraystretch{1.2}
    \begin{array}{c}
      \omega^A_3
      \\
      \omega^B_3
      \\
      p_2
    \end{array}
    \!\!\right]
    \Big/
    \left(
    \def\arraystretch{1.2}
    \def\arraycolsep{2pt}
    \begin{array}{ccl}
      \mathrm{d}\,
      \omega^A_3
      &=&
      0
      \\
      \mathrm{d}\,
      \omega^B_3
      &=&
      0
      \\
      \mathrm{d}\, 
      p_2
      &=&
      \omega^A_3
      -
      \omega^B_3
    \end{array}
    \! \right)
  \end{equation}
  and in that the Bianchi identity \eqref{BianchiOfPoincare2Form}
  on $P_2$ characterizes dashed maps, making the following diagram commute:
  \begin{equation}
    \label{ModulatingPoincare2Form}
    \begin{tikzcd}[
      row sep=20pt,
      column sep=35pt
    ]
      \mathbb{R}^{
        1,9\,\vert\,
        \mathbf{16}
        \oplus
        \overline{\mathbf{16}}
      }
      \quad 
      \underset{
        \mathclap{
        \mathbb{R}^{
          1,8\,\vert\,
          \mathbf{16}
          \oplus
          \mathbf{16}
        }
        }
      }{\times}
      \quad
      \mathbb{R}^{
        1,9\,\vert\,
        \mathbf{16}
        \oplus
        \overline{\mathbf{16}}
      }
      \ar[
        rr,
        dashed,
        "{
          P_2
        }"
      ]
      \ar[
        d,
        shorten=-2pt,
        "{
          (\pi_A, \pi_B)
        }"
        {swap}
      ]
      &&
      \mathfrak{poin}_2
      \ar[
        d
      ]
      \\
      \mathbb{R}^{
        1,9\,\vert\,
        \mathbf{16}
        \oplus
        \overline{\mathbf{16}}
      }
      \times
      \mathbb{R}^{
        1,9\,\vert\,
        \mathbf{16}
        \oplus
        \overline{\mathbf{16}}
      }      
      \ar[
        rr,
        "{
          (H_3^A,\, H_3^B)
        }"
      ]
      &&
      \mathfrak{l}
      B^2 \mathrm{U}(1)
      \times
      \mathfrak{l}
      B^2 \mathrm{U}(1)
      \mathrlap{\,.}
    \end{tikzcd}
  \end{equation}
\end{remark}

\noindent
{\bf T-duality as a Fourier-Mukai transform.} 
\label{FourierMukaiTransform}
The above doubled super-space picture
\eqref{Poincare2FormForSingleFiber} coupled with the observations from \eqref{FiberIntegrationOfH3Flux} and 
\eqref{BianchiOfPoincare2Form} lead to an alternative but equivalent formulation of the $T$-duality phenomenon within (super-)rational homotopy theory in terms of correspondences and an induced Fourier-Mukai integral-transform (\cite[(1.1)]{Hori99}\cite[(1.9)]{BEM04}\cite[\S 4.1]{GevorgyanSarkissian14}), here on 3-twisted Chevalley--Eilenberg cochain complexes (Def. \ref{3TwistedPeriodicChevalleyEilenberg}).

\begin{definition}[\bf T-duality Correspondence]\label{TdualityCorrespondence}
Pairs $(\, \widehat{\frg}_A, H_A)$ and $(\, \widehat{\frg}_B, H_B)$ of centrally extended super $L_\infty$-algebras (Def. \ref{CentralExtension}) over $\frg$ via even $2$-cocycles $c_A,c_B \in \CE(\frg)$, supplied with $3$-cocycle twists, respectively, are said to be in \textit{T-duality correspondence} if:
\begin{itemize}[
  leftmargin=.8cm,
  topsep=2pt,
  itemsep=6pt
]
\item[\bf (i)] The respective fiber integration of the twists $H_{A/B}$ yields the opposite extension cocycles $c_{B/A}$ (cf. Eq. \eqref{FiberIntegrationOfH3Flux} of Rem. \ref{TDualNSFlux})
$$ 
(p_{A/B})_* H_{A/B} \, \, = \, \, c_{B/A} \qquad \in \qquad  \CE(\frg) \, . 
$$
\item[\bf (ii)] On the doubly extended space $\widehat{\frg}_{A}\times_\frg \widehat{\frg}_B$ (cf. Def. \ref{1DoubledSuperSpace})
 $$
    \begin{tikzcd}[row sep=2pt, column sep=huge]
      &
     \widehat{\frg}_A \times_\frg \widehat{\frg}_B
      \ar[
        dr, 
        ->>,
        "{ \pi_B }"
      ]
      \ar[
        dl, 
        ->>,
        "{ \pi_A }"{swap}
      ]
      \ar[
        dd,
        phantom,
        "{
          \scalebox{.7}{
            \color{gray}
            (pb)
          }
        }"{pos=.3}
      ]
      \\
   \widehat{\frg}_A 
      \ar[
        dr,
        ->>,
        "{ p_A }"{swap}
      ]
      &&
      \widehat{\frg}_B \, ,
      \ar[
        dl,
        ->>,
        "{ p_B }" 
      ]
      \\
      &
    \frg 
    \end{tikzcd}
  $$
defined dually via
$$
\CE(\, \widehat{\frg}_{A}\times_\frg \widehat{\frg}_B) \quad = \quad  \CE(\frg)\big[e_A,e_B\big]\big/ \big(\dd e_{A/B}= c_{A/B}\big) \, ,
$$
the {\it  Poincar{\' e} form} (cf. \eqref{Poincare2FormForSingleFiber}) 
$$
P\, \, :=\, \, e_B \cdot e_A \quad \quad \in \quad \quad \CE(\, \widehat{\frg}_{A}\times_\frg \widehat{\frg}_B) 
$$
is a coboundary of the difference between the pullbacks of  the two twisting cocycles (cf. Prop. \ref{PoincarFormTrivializesDifferenceOfPullbacksOfNSFluxes})
$$\dd P\, \, = \, \, \pi_A^* H_A - \pi_B^* H_B \, .$$
\end{itemize}
\end{definition}

The observation for the decomposed structure of the twisting NS-fluxes from Rem. \ref{TDualNSFlux} yields in fact an equivalent characterization of such a T-duality correspondence.

\newpage 
\begin{lemma}[\bf T-duality conditions on the base]\label{TDualityConditionsOnTheBase}
Two pairs $(\, \widehat{\frg}_A, H_A)$ and $(\, \widehat{\frg}_B, H_B)$ are in T-duality correspondence (Def. \ref{HigherTdualityCorrespondence}) if and only if the twists are of the form
$$
H_{A/B}\; =\;  H_\frg + e_{A/B} \cdot c_{B/A} 
$$
for a common basic $3$-cochain $H_\frg \in \CE(\frg)$ whose differential trivializes the product of the corresponding extending $2$-cocycles
$$
\dd H_\frg\; =\;  - c_A \cdot c_B \, .
$$
\begin{proof}
The first condition from Def. \ref{HigherTdualityCorrespondence} yields immediately that
$$
H_{A/B} = H^{A/B}_\frg + e_{A/B} \cdot c_{B/A}
$$
for potentially different basic $H^A_\frg, H^B_\frg \in \CE(\frg)$. The second condition then expands as
\begin{align*}
e_A \cdot c_B -e_B\cdot c_A   = \dd P&=  \pi_A^* H_A - \pi_B^* H_B \\
&=H_\frg^A + e_A \cdot c_B - H_\frg^B - e_B \cdot c_A  
\end{align*}
which holds if and only if 
$$
H_\frg^A= H_\frg^B \, .
$$

\smallskip 
\noindent Strictly speaking, this is an equation on $\CE(\, \widehat{\frg}_A\times_\frg \widehat{\frg}_B)$ of (doubly) basic forms via $\pi_A^* \circ p_A^*$ and $\pi_B^* \circ p_B^*$, but these two 
\textit{injective} morphisms actually coincide as maps $\CE(\frg) \hookrightarrow \CE(\, \widehat{\frg}_A\times_\frg \widehat{\frg}_B)$, by construction, as can be seen immediately by their action on generators. Hence the equation holds equivalently on $\CE(\frg)$.

Finally since the twists $H_{A/B}$ are by assumption closed on the respective extensions $\widehat{\frg}_{A/B}$, it follows in particular that 
$$
0=\dd H_A = \dd H_\frg + c_A \cdot c_B 
$$
as an equation on $\frg_A$, which also holds as an equation on $\frg$ since both $H_\frg$ and $c_A\cdot c_B$ are implicitly pullbacks of basic forms via the dgca morphism $p_A^*: \CE(\frg) \hookrightarrow \CE(\,\widehat{\frg}_{A})$, which is injective. The closure of $h_B$ yields the same condition.
 
The reverse implication follows by the same computations.
\end{proof}
\end{lemma}
\begin{corollary}[\bf T-duality correspondence classifying space] It follows that the T-duality $L_\infty$-algebra $b\mathcal{T}\cong \mathrm{cyc}\, b^2 \mathbb{R}$ from Ex. \ref{TDualityClassifyingAlgebra} classifies the set of T-duality correspondences over any super-$L_\infty$ algebra $\frg$, in that morphisms of super $L_\infty$-algebras 
\begin{equation*}
  \begin{tikzcd}[row sep=-3pt,
    column sep=0pt
  ]
    \frg
    \ar[
      rr
    ]
    &&
   b\mathcal{T}
    \\
   c_A
    &\longmapsfrom&
    \omega_2
    \\
c_B
    &\longmapsfrom&
    - \widetilde \omega_2
    \\
H_\frg
    &\longmapsfrom&
     h_3 
    \mathrlap{\, ,}
  \end{tikzcd}
\end{equation*}
are in canonical bijection with the set of T-duality correspondences over $\frg$ in terms of Lem. \ref{TDualityConditionsOnTheBase} (hence equivalently in terms of Def. \ref{TdualityCorrespondence}). 
\end{corollary}
 
For any such $T$-duality correspondence, the following natural ``pull-push'' homomorphism of 3-twisted cochain complexes
\begin{align}\label{TDualityPullPush}
T_{\mathrm{FM}} \quad := \quad -  {\pi_B}_*\circ e^{-P} \circ \pi^*_A \equiv - {\pi_B}_{*} \circ (1+e_A \cdot e_B) \circ \pi_A^*  
\end{align}
is an isomorphism of degree $( -1 \,\mathrm{mod} \, 2, \, \mathrm{evn})$. In particular, it descends to an isomorphism on cohomology {\cite[Prop. 6.4]{FSS18-TDualityA}}{\cite[Thm. 3.17]{FSS20-HigherT}}.

\smallskip 
\begin{proposition}[\bf T-duality/Fourier-Mukai isomorphism]\label{TDuality/FourierMukaiIsomorphism}
Let $(\widehat{\,\frg}_A, H_A)$ and $(\,\widehat{\frg}_B, H_B)$ be in T-duality correspondence. 

\vspace{1mm} 
\noindent {\bf (i)} Then the pull-push homomorphism  \eqref{TDualityPullPush} is in fact an isomorphism of 3-twisted cochain complexes (Def. \ref{3TwistedPeriodicChevalleyEilenberg}) and acts explicitly as
\begin{align}\label{TdualityIso}
T_{\mathrm{FM}}\; \; :\; \; \CE^{\bullet + H_A}(\,\widehat{\frg}_A) &\xlongrightarrow{\quad \sim \quad} \CE^{(\bullet -1) +H_B}(\,\widehat{\frg}_B)\\
F = F_{\mathrm{bas}} + e_A \cdot {p_A}_{*} F &\quad \longmapsto \quad    - {p_A}_{*} F - e_B \cdot F_{\mathrm{bas}} \nonumber \, , 
\end{align}
thereby swapping the ``winding'' and ``non-winding'' modes.

\vspace{1mm} 
\noindent {\bf (ii)}  In particular, it descends to a ``Fourier-Mukai'' isomorphism of twisted cocycles and furthermore twisted cohomologies
\begin{align*}
T_{\mathrm{FM}}\; \; :\; \; H_\CE^{\bullet + H_A}(\,\widehat{\frg}_A) &\xlongrightarrow{\quad \sim \quad}
H_\CE^{(\bullet -1) +H_B}(\,\widehat{\frg}_B) \, .
\end{align*}
\end{proposition}

\newpage 
\begin{proof}
The explicit form of the mapping follows immediately as
\begin{align*}
{\pi_B}_{*}\big((1+ e_A \cdot e_B)\cdot (F_\bas + e_A \cdot {p_A}_*F  )\big) &= {\pi_B}_*\big(F_\bas + e_A\cdot e_B\cdot F_\bas + e_A \cdot {p_A}_*F +0 \big) \\
&= {p_A}_* F + e_B \cdot F_\bas 
\end{align*}
where we absorbed the explicit mention of the injective morphism $\pi_A^*$, and used the fact that the fiber integration along $\pi_B : \widehat{\frg}_A \times_\frg \widehat{\frg}_B \rightarrow \frg_B$ is the derivation that takes the value $1$ on the $e_A$ generator, and 0 on the rest. 

The fact that $T_{\mathrm{FM}}$ is a linear isomorphism of cochains follows by the existence of the explicit inverse
\begin{align*}
T_{\mathrm{FM}}^{-1} \quad := \quad - {\pi_A}_*\circ e^{P} \circ \pi^*_B \equiv - {\pi_A}_{*} \circ (1+ e_B \cdot e_A) \circ \pi_B^*  
\end{align*}
acting by ``swapping back'' the winding and non-winding modes,
\begin{align*}
T_{\mathrm{FM}}^{-1}\; \; :\; \; \CE^{\bullet + H_B}(\,\widehat{\frg}_B) &\xlongrightarrow{\quad \sim \quad} \CE^{(\bullet -1 +H_A)}(\,\widehat{\frg}_A)\\
\tilde{F} = + \tilde{F}_{\mathrm{bas}} + e_B \cdot {p_B}_{*} \tilde{F} &\quad \longmapsto \quad    -{p_B}_{*} \tilde{F} - e_B \cdot \tilde{F}_{\mathrm{bas}} \nonumber \, , 
\end{align*}
precisely by the same calculation under the exchange of indices $A\leftrightarrow B$.

Lastly, to see that the linear isomorphism $T_{\mathrm{FM}}$ interwines with the twisted differentials, up to a sign (due to the odd degree of $T_{\mathrm{FM}}$ and that of the differentials), we compute
\begin{align*}
(\dd_{\widehat{\frg}_B} - H_B) \circ T_{\mathrm{FM}} (F) &=  (\dd_{\widehat{\frg}_B} - H_\frg - e_B \cdot c_A) \big( -{p_A}_* F - e_B \cdot F_\bas \big)\\
&= - ( \dd\,  {p_{A}}_* F + c_B \cdot F_{\bas} - H_\frg \cdot {p_{A}}_* F) 
\\
&\quad - e_B \cdot ( - \dd F_\bas + H_\frg \cdot F_\bas - c_A \cdot {p_A}_* F)  
\end{align*}
where we used the explicit swapping action of $T_{\mathrm{FM}}$ and the explicit form of the twist $h_B$ from Lem. \ref{HigherTDualityConditionsOnTheBase}.

Computing similarly,
\begin{align*}
T_{\mathrm{FM}}\circ (\dd_{\widehat{\frg}_A} - H_A)  (F) &=  T_{\mathrm{FM}} \big(  \dd F_\bas + c_A \cdot {p_A}_* F - H_\frg \cdot F_\bas + e_A \cdot( -\dd \, {p_A}_* F + H_\frg \cdot {p_A}_*F - c_B \cdot F_\bas )\,\big) \\
&= + ( \dd\,  {p_{A}}_* F + c_B \cdot F_{\bas} - H_\frg \cdot {p_{A}}_* F) 
\\
&\quad + e_B \cdot ( - \dd F_\bas + H_\frg \cdot F_\bas - c_A \cdot {p_A}_* F) \, ,
\end{align*}
and hence
$$
(\dd_{\widehat{\frg}_B} - H_B) \circ T_{\mathrm{FM}} \; \;=\; \;-\, T_{\mathrm{FM}}\circ (\dd_{\widehat{\frg}_A} - H_A) \, .
$$

\vspace{-4mm}
\end{proof}

\begin{corollary}[\bf Pull-push via automorphism of cyclified twisted K-theory spectra]
\label{PullPushViaAutomorphismOfCyclifiedCoefficients}
Under the identification of twisted K-theory cocycles with  $3$-twisted cocycles from Eq. \eqref{TwistedKTheorySpectraClassifyTwistedCocycles}, the action of the T-duality isomorphism  \eqref{TdualityIso}
by Fourier-Mukai transform coincides with that of the composite operation of (a) reduction, (b) automorphism of cyclified twisted K-theory spectra, and  (c) reoxidation from Thm. \ref{TwistedKTheoryRedIsoReOxi}
$$
  \begin{tikzcd}[
    column sep={40pt, between origins},
    row sep=-3pt
  ]
    \mathfrak{g}_A
    \ar[
      rr
    ]
    \ar[
      dr,
      "{
        H_{\mathcolor{purple}{A}}
      }"{swap, xshift=2pt, yshift=2pt, pos=.6}
    ]
    \ar[
      d,
      phantom,
      shift right=10pt,
      "{
        \Bigg\{
      }"{pos=.2}
    ]
    &&
    \mathfrak{l}
    \big(
      \,
      \Sigma^{
        \mathcolor{purple}0
      }
      \mathrm{KU}
      \!\sslash\!
      B\mathrm{U}(1)
      \,
    \big)
    \ar[
      dl,
      shorten <=-5pt
    ]
    \ar[
      d,
      phantom,
      shift left=43pt,
      "{
        \Bigg\}
      }"{pos=-.2}
    ]
    \ar[
      rr,
      <->,
      shift right=8pt,
      shorten=10pt,
      "{
        \scalebox{.7}{
          \rm
          Thm. \ref{TwistedKTheoryRedIsoReOxi}
        }
      }"
    ]
    &[20pt]
    &[20pt]
    \mathfrak{g}_B
    \ar[
      rr
    ]
    \ar[
      dr,
      "{
        H_{\mathcolor{purple}{B}}
      }"{swap, xshift=2pt, yshift=2pt, pos=.6}
    ]
    \ar[
      d,
      phantom,
      shift right=10pt,
      "{
        \Bigg\{
      }"{pos=.2}
    ]
    &&
    \mathfrak{l}
    \big(
      \,
      \Sigma^{
        \mathcolor{purple}1
      }
      \mathrm{KU}
      \!\sslash\!
      B\mathrm{U}(1)
      \,
    \big)
    \ar[
      dl,
      shorten <=-5pt
    ]
    \ar[
      d,
      phantom,
      shift left=43pt,
      "{
        \Bigg\}.
      }"{pos=-.2}
    ]
    \\
    {}
    &
    b^2\mathbb{R}
    \ar[
      d,
      hook,
      shorten=7pt
    ]
    &
    {}
    &&
    {}
    &
    b^2\mathbb{R}
    \ar[
      d,
      hook,
      shorten=7pt
    ]
    &
    {}
    \\[30pt]
    &
    \mathrm{CE}^{0 + H_{\mathcolor{purple}{A}}}(\frg_A)
    \ar[
      rrrr,
      <->,
      "{
         T_{\mathrm{FM}}
      }"
    ]
    & && &
    \mathrm{CE}^{-1 + H_{\mathcolor{purple}{B}}}(\frg_B)
  \end{tikzcd}
$$
The analogous statement follows for the Fourier--Mukai transform given by $-T_\mathrm{FM}$, where its action on cocycles coincides with the composite operation from {\bf (i)} Rem. \ref{RedundancyOfExtraIsomorphismsOfCyclifiedTwistedKSpectra}. Swapping further the sign in (the exponential of) $P_2$ corresponds to the further two composite operations from Rem. \ref{RedundancyOfExtraIsomorphismsOfCyclifiedTwistedKSpectra}.
\end{corollary}

\begin{remark}[\bf Strict isomorphism vs quasi-isomorphism]
The original article \cite{FSS18-TDualityA} focuses on the induced isomorphism on twisted $L_\infty$-cohomology. Nevertheless, we stress that the map actually defines a strict isomorphism even at the level of twisted cochain complexes, and hence most importantly at the level of cocycles (and further on cohomology). From a physical perspective, this means that the isomorphism is realized at the level of flux densities (prior to flux-quantization) and not only at the level of the corresponding gauge equivalency classes.  
\end{remark}

\subsection{M-brane T-duality}
\label{HigherTduality}

The  $L_\infty$-algebraic formulation of T-duality 
makes immediate a much larger generality of $L_\infty$-algebraic (hence: rational-topological) ``higher T-duality'' \cite{FSS20-HigherT} in the sense of ``higher structures'' and ``categorified symmetries''. Here, the NS-field twist $H_3$ in ordinary T-duality (typically thought of as the curvature 3-form of a ``bundle gerbe'') is allowed to have higher degrees (as befits higher bundle gerbes and yet richer higher fiber bundles) -- see \cite{Sa-Higher} for the appearance of degree 7 (bosonic) twists in string theory. At the rational level, this story amounts to applying the constructions related to higher (odd) central extensions along the lines of \S\ref{HigherExtensions}.

\medskip

This ``higher T-duality'' was identified in {\cite[Thm. 3.17]{FSS20-HigherT}}, where it was encoded in terms of an isomorphism between the $(4t-1)$-twisted periodic Chevalley--Eilenberg cohomologies (Def. \ref{TwistedPeriodicChevalleyEilenberg})
\begin{align*}
T\; \; :\; \; H_\CE^{\bullet + h_A}(\,\widehat{\frg}_A) &\xlongrightarrow{\quad \sim \quad}
H_\CE^{(\bullet -2t+1) +h_B}(\,\widehat{\frg}_B) 
\end{align*}
of certain ``higher T-dual pairs'' of centrally extended super $L_\infty$-algebras $\widehat{\frg}_A$, $\widehat{\frg}_B$ via $2t$-cocycles $\omega^{A},\omega^{B} \in \mathrm{CE}(\frg)$, for any $t\in \NN$, over a common base $\frg$, supplied with suitably related $(4t-1)$-cocycle twists $h_{A/B} \in \CE(\,\widehat{\frg}_{A/B})$. 

\medskip 
Here we provide an equivalent description via an isomorphism of the higher cyclifications of the corresponding twisted periodic cocycle classifying $L_\infty$-algebras from Ex. \ref{CyclificationOfHigherTwistedCocyclesClassifyingSpaces} and Thm. \ref{HigherTwistedCocyclesRedIsoReOxi}. This may be viewed as a justification for the existence of the indicated isomorphism of twisted cohomologies. We motivate this new description by first further elaborating the description of the higher-self duality on the $\mathfrak{m}2\mathfrak{brane}$ from \cite[\S 4.3]{FSS20-HigherT}.

\medskip
\noindent
{\bf The $\mathfrak{m}2\mathfrak{brane}$ higher self T-duality.} Recall the higher extension of $\mathbb{R}^{1,10\vert \mathbf{32}}$ via the 4-cocycle $G_4$ from Ex. \ref{TheM2braneExtensionOf11D}
 \begin{equation*}
   \begin{tikzcd}[
   row sep=3pt
  ]
   \mathfrak{m}2\mathfrak{brane}
    \ar[
      rr,
      "{
       p \, := \, \mathrm{hofib}
      }"
    ]
    &&
    \mathbb{R}^{
      1,10\,\vert\,
      \mathbf{32}}   
    \ar[
      rr,
      "{
        G_4
      }"
    ]
    &&
    b^3 \mathbb{R} \, ,
    \end{tikzcd}
  \end{equation*}
and consider the ``Page charge'' 7-cocycle
\begin{align}
\label{PageCharge7Cocycle}
\widetilde{G}_7 \, \, := \, \,2\,  G_7 - c_3 \cdot G_4  \qquad \in \qquad \CE(\mathfrak{m}2\mathfrak{brane}) \, .
\end{align}
Note that this is indeed closed, i.e., constitutes a homomorphism
$$
  \widetilde{G}_7 
  \,:\, 
  \mathfrak{m}2\mathfrak{brane} 
  \longrightarrow 
  b^6 \mathbb{R} 
\,,
$$
since 
\begin{align*}
\dd \, \widetilde{G}_7 &= 2 
 \,  \dd G_7 - G_4 G_4  \\
&= G_4 G_4 -  G_4 G_4 \\
&=0 \, ,  
\end{align*}
by the fact that $(G_4,G_7)$ forms a $\mathfrak{l}S^4$-cocycle (Ex. \ref{4SphereValuedSuperFlux}).
As such, it may be thought of as a $7$-twisting cocycle analogous to the $3$-twisting IIA/IIB NS-fluxes
$
  H_3^{A/B} 
  \;=\;
  H_3 + e_{A}^9 \cdot c^{B/A}_1
$.
Apart from its higher degree, the crucial property of the $7$-cocycle $\widetilde{G}_7$  is that its fiber integration down to 11D super-spacetime, via $p: \mathfrak{m}2\mathfrak{brane} \rightarrow \mathbb{R}^{1,10\vert \mathbf{32}}$, yields the original extending $4$-cocycle $G_4$ up to a sign prefactor
$$
p_* \widetilde{G}_7\, \, = \, \, - \,  G_4  \, . 
$$
This suggests that the corresponding higher $T$-duality should be a \textit{self-duality} acting on 7-twisted periodic cocycles on $\mathfrak{m}2\mathfrak{brane}$, in an appropriate sense.

More precisely, we may consider the higher central extension of $\mathbb{R}^{1,10\vert \mathbf{32}}$ by the opposite 4-cocycle $p_* \widetilde{G}_7 = -\, G_4$ instead, yielding an isomorphic copy of the $\mathfrak{m}2\mathfrak{brane}$-algebra 
 \begin{equation*}
   \begin{tikzcd}[
   column sep=35pt
  ]
   \mathfrak{m}2\mathfrak{brane}^{\color{purple}-}
    \ar[
      rr,
      "{
       p_{-} \, := \, \mathrm{hofib}
      }"
    ]
    &&
    \mathbb{R}^{
      1,10\,\vert\,
      \mathbf{32}}   
    \ar[
      r,
      "{
        - G_4
      }"
    ]
    &
    b^3 \mathbb{R} \, ,
    \end{tikzcd}
  \end{equation*}
given by
$$
    \mathrm{CE}\big(
      \mathfrak{m}2\mathfrak{brane}^{\color{purple}-}
    \big)
    \;\simeq\;
    \FDGCA
    \left[
    \def\arraystretch{1.4}
    \def\arraycolsep{1pt}
    \begin{array}{c}
      (\psi^\alpha)_{\alpha=1}^{32}
      \\
      (e^a)_{a=1}^{\ten}
      \\
      c_3^-
    \end{array}
    \right]
    \Big/
    \left(
    \def\arraystretch{1.4}
    \def\arraycolsep{1pt}
    \begin{array}{ccl}
      \mathrm{d}\,
      \psi &=&
      0
      \\
      \mathrm{d}\,
      e^a
      &=&
      \big(\hspace{1pt}
        \overline{\psi}
        \,\Gamma^a\,
        \psi
      \big)
      \\
      \mathrm{d}\,c_3^-
      &=&
      \smash{
      \grayunderbrace{
      {\color{purple}-}\tfrac{1}{2}
      \big(\hspace{1pt}
        \overline{\psi}
        \Gamma_{a b }
        \psi
      \big)
      e^a e^b
      }{
        - G_4
      }
      }
    \end{array}
    \right)
    . 
  $$

\vspace{4mm}
\noindent
\begin{lemma}[\bf Tautological Reflection isomorphism]
    The (tautological) isomorphism between the two versions of the $4$-cocycle extended $\FR^{1,10\vert \mathbf{32}}$ given by `reflecting' the extending generators
\begin{equation}\label{M2BraneReflectionIsomorphism}
 \begin{tikzcd}[row sep=-3pt,
    column sep=0pt
  ]
    \mathfrak{m}2\mathfrak{brane}
    \ar[
      rr,
      <->,
      "{ \sim }"
    ]
    &&
    \mathfrak{m}2\mathfrak{brane}^-
    \\
   -c_3
    &\longmapsfrom&
   c_3^-
    \mathrlap{\,, }
  \end{tikzcd}
\end{equation}
maps the original twisting 7-cocycle  to the ``dual'' twisting cocycle
\begin{align}\label{DualPageCharge7Cocycle}
\widetilde{G}_7^- \, \, := \, \,2\,  G_7 + c_3^- \cdot G_4  \qquad \in \qquad \CE(\mathfrak{m}2\mathfrak{brane}^-) \, .
\end{align}
\end{lemma}
Evidently, the fiber integration of the dual 7-cocycle under $p_{-} : \mathfrak{m}2\mathfrak{brane}^- \rightarrow \mathbb{R}^{1,10\vert \mathbf{32}}$ recovers the 4-cocycle classifying the original $\mathfrak{m}2\mathfrak{brane}$
$$
{p_{-}}_* \widetilde{G}_7^-\, \, = \, \,  G_4  \, . 
$$
\begin{remark}[\bf Higher $\mathfrak{m}2\mathfrak{brane}$ twisting cocycles under higher T-duality]\label{HigherM2BraneTwistingCocyclesUnderHigherTDuality}
$\,$

\noindent{\bf (i)}
The two $b^6 \mathbb{R}$-cocycles $\widetilde{G}_7$ and $ \widetilde{G}_7^-$ are related in a manner analogous to that of the IIA and IIB NS-fluxes from Rem. \ref{TDualNSFlux}. That is, it follows immediately that $ \widetilde{G}_7^-$ is related to $\widetilde{G}_7$  via the higher T-duality operation on higher twists (Thm. \ref{HigherTwistedCocyclesRedIsoReOxi}), namely : the composite of (1.) reduction \eqref{HigherCyclificationHomIsomorphism} along $\omega_4=G_4$ followed by (2.) automorphism \eqref{AutomorphismOfCycOfB2tU1} of $b \mathcal{T}_2 \cong  \mathrm{cyc}_3(b^6\mathbb{R})$ and then (3.) oxidation \eqref{HigherCyclificationHomIsomorphism} along $\widetilde{\omega}_4=-G_4$ 

\vspace{-3mm} 
\begin{equation}\label{TDualityOnPageChargeTwists}
  \hspace{-4mm} 
  \begin{tikzcd}[row sep=-3pt,
    column sep=60pt,
    ampersand replacement=\&
  ]
  \def\arraystretch{1.4}
  \def\arraycolsep{1.1pt}
  \begin{array}{ccccl}
    \widetilde{G}_7
    &=&\
    \grayoverbrace{
      2G_7 
    }{
      \omega_7
    }
    &-&
    c_3
    \grayoverbrace{G4}{
      \mathrm{s}_3 \omega_7
    }
  \end{array}
  \ar[
    dr,
    |->,
    "{
      \mathrm{rdc}_{G_4}
    }"{sloped},
    "{
      \scalebox{.7}{
        \color{darkgreen}
        \bf
        \def\arraystretch{.9}
        \begin{tabular}{c}
          reduction along
          \\
        extension
        \end{tabular}
      }
    }"{swap, sloped, yshift=-2pt, pos=.5}
  ]
  \ar[
    rr,
    |->,
    shorten=5pt,
    "{
      \scalebox{.7}{
        \color{darkgreen}
        \bf
        \begin{tabular}{c}
          superspace higher
          T-duality
        \end{tabular}
      }
    }"
  ]
  \&\&
  \def\arraystretch{1.4}
  \def\arraycolsep{1.1pt}
  \begin{array}{ccccl}
    \widetilde{G}_7^-
    &=&
    2 G_7 
    &+&
    c_3^-
    G_4\,.
  \end{array}
  \ar[
    dl,
    <-|,
    "{
      \mathrm{oxd}_{-G_4}
    }"{sloped},
    "{
      \scalebox{.7}{
        \color{darkgreen}
        \bf
        \def\arraystretch{.9}
        \begin{tabular}{c}
          oxidation along
          \\
         reflected extension
        \end{tabular}
      }
    }"{swap, sloped, yshift=-2pt, pos=.5}
  ]
  \\
  \&
  \def\arraystretch{1.2}
  \def\arraycolsep{1.5pt}
  \begin{array}{cccc}
    2 G_7 &=&
   \tfrac{2}{5!}
      \big(\,
        \overline{\psi}
        \,\Gamma_{a_1 \cdots a_5}\,
        \psi
      \big)
      e^{a_1} \cdots e^{a_5}
    \\
    G_4
    &=&
    \hspace{1pt}
      \tfrac{1}{2}
      \big(\,
        \overline{\psi}
        \,\Gamma_{a_1 a_2}\,
        \psi
      \big)
      e^{a_1} e^{a_2}
    
    \\
    - (-G_4)
    &=&
    \hspace{1pt}
      \tfrac{1}{2}
      \big(\,
        \overline{\psi}
        \,\Gamma_{a_1 a_2}\,
        \psi
      \big)
      e^{a_1} e^{a_2}
  \end{array}
  \end{tikzcd}
\end{equation}

\vspace{-3mm} 
\noindent {\bf (ii)} The cobounding condition on the base $\mathbb{R}^{1,10\vert \mathbf{32}}$, analogous to \eqref{DifferentialOfBasicH3Flux}, is now
\begin{align}\label{DifferentialOfBasicG7Flux}
\dd (2G_7)\, \,  = \, \,  - \, G_4 \cdot (-G_4) \, , 
\end{align}
being satisfied automatically by the $\mathfrak{l}S^4$-cocycle condition (or equivalently, the closure of either $\widetilde{G}_7$ or $\widetilde{G}_7^-$).

\vspace{0mm} 
\noindent {\bf(iii)} On the doubly higher extended space
\begin{align}\label{M2BraneCorrespondenceSpace}
\mathfrak{m}2\mathfrak{brane} \times_{\mathbb{R}^{1,10\vert \mathbf{32} }} \mathfrak{m}2\mathfrak{brane}^{-}
\end{align}
the analogous \textit{higher Poincare form} (cf. \eqref{Poincare2FormForSingleFiber})
\begin{align}\label{M2BranePoincareForm}
P_{6}\, \, := \, \, c_3^{-} \cdot c_3
\end{align}
is a coboundary for the difference of the (pullbacks of) the twisting cocycles
$$
\dd P_6\, \,  = \, \, \pi^*\widetilde{G}_7 - \pi_-^*\widetilde{G}_7^{-}\, ,
$$
as can be seen immediately since $\pi^*\widetilde{G}_7 - \pi_-^*\widetilde{G}_7^{-}= - G_4 c_3 - c_3^- G_4 = \dd( c_3^- \cdot c_3) \, .$
\end{remark}

\begin{proposition}[\bf Higher $\mathfrak{m}2\mathfrak{brane}$ twisted cocycles under higher T-duality]
\label{HigherM2BraneTwistedCocyclesUnderHigherTDuality}
The higher T-duality operation  \eqref{TDualityOnPageChargeTwists} between $\widetilde{G}_7$ and $\widetilde{G}_7^{-}$ extends to a bijection of the corresponding 7-twisted cocycles. In particular, it maps any $\widetilde{G}_7$-twisted cocycle of degree $(m\, \, \mathrm{mod}\, \, 6, \, \mathrm{evn})$ on $\mathfrak{m}2\mathfrak{brane}$
\begin{equation*}
  \begin{tikzcd}[
    row sep=-3pt, 
    column sep=0pt]
    \mathfrak{m}2\mathfrak{brane}
    \ar[
      rr
    ]
    &&
    \mathfrak{l}
    \big(
      \,
      \Sigma^{\color{purple}m}
      \mathrm{K}^{3}{U}
      \!\sslash\!\!B^{5}\mathrm{U}(1) \, \big)
    \\
    \widetilde{G}_7
    &\longmapsfrom&
    h_{7}
    \\
    (F_{6k +m})_{k\in \mathbb{Z}}
    &\longmapsfrom&
    (f_{6k +m})_{k\in \mathbb{Z}}
  \end{tikzcd}
\end{equation*}
to the $\widetilde{G}_7^{-}$-twisted cocycle of degree $(m-3\, \, \mathrm{mod}\, \, 6, \, \mathrm{evn})$ on $\mathfrak{m}2\mathfrak{brane}^{-}$
\begin{equation*}
  \begin{tikzcd}[
    row sep=-3pt, 
    column sep=0pt]
    \mathfrak{m}2\mathfrak{brane}^{-}
    \ar[
      rr
    ]
    &&  
     \mathfrak{l}
    \big(
      \,
      \Sigma^{\color{purple}m-3}
      \mathrm{K}^{3}{U}
      \!\sslash\!\!B^{5}\mathrm{U}(1) \, \big)
    \\
    \widetilde{G}_7^{-}
    &\longmapsfrom&
    h_{7}
    \\
 \big(- {p}_*F_{6k+m} - c^{-}_3 \cdot (F_{6(k-1)+m})_{\bas} \, \big)_{k\in \mathbb{Z}}
    &\longmapsfrom&
    (f_{6k-3 +m })_{k\in \mathbb{Z}}
  \end{tikzcd}
\end{equation*}
and vice-versa.
\end{proposition}
\begin{proof}
This follows as an application of Thm. \ref{HigherTwistedCocyclesRedIsoReOxi}.
\end{proof}

\begin{corollary}[\bf As a self-duality on $\mathfrak{m}2\mathfrak{brane}$]\label{AsASelfDualityOnm2brane}
Applying a further pullback operation via the reflection isomorphism \eqref{M2BraneReflectionIsomorphism} on the result of Prop. \ref{HigherM2BraneTwistedCocyclesUnderHigherTDuality}, this yields an actual self-duality on $\mathfrak{m}2\mathfrak{brane}$  as a \textit{non-trivial} isomorphism of $\widetilde{G}_7$-twisted cocycles. Explicitly, this self-duality maps any $\widetilde{G}_7$-twisted cocycle of degree $(m\, \, \mathrm{mod}\, \, 6, \, \mathrm{evn})$ on $\mathfrak{m}2\mathfrak{brane}$
\begin{equation*}
  \begin{tikzcd}[
    row sep=-3pt, 
    column sep=0pt]
    \mathfrak{m}2\mathfrak{brane}
    \ar[
      rr
    ]
    &&
    \mathfrak{l}
    \big(
      \,
      \Sigma^{\color{purple}m}
      \mathrm{K}^{3}{U}
      \!\sslash\!\!B^{5}\mathrm{U}(1) \, \big)
    \\
    \widetilde{G}_7
    &\longmapsfrom&
    h_{7}
    \\
    (F_{6k +m})_{k\in \mathbb{Z}}
    &\longmapsfrom&
    (f_{6k +m})_{k\in \mathbb{Z}}
  \end{tikzcd}
\end{equation*}
to the $\widetilde{G}_7$-twisted cocycle of degree $(m-3\, \, \mathrm{mod}\, \, 6, \, \mathrm{evn})$ on $\mathfrak{m}2\mathfrak{brane}$
\begin{equation*}
  \begin{tikzcd}[
    row sep=-3pt, 
    column sep=0pt]
    \mathfrak{m}2\mathfrak{brane}
    \ar[
      rr
    ]
    &&  
     \mathfrak{l}
    \big(
      \,
      \Sigma^{\color{purple}m-3}
      \mathrm{K}^{3}{U}
      \!\sslash\!\!B^{5}\mathrm{U}(1) \, \big)
    \\
    \widetilde{G}_7
    &\longmapsfrom&
    h_{7}
    \\
 \big(- {p}_*F_{6k+m} + c_3 \cdot (F_{6(k-1)+m})_{\bas} \, \big)_{k\in \mathbb{Z}}
    &\longmapsfrom&
    (f_{6k-3 +m })_{k\in \mathbb{Z}} \, .
  \end{tikzcd}
\end{equation*}

\vspace{-2mm}
\noindent
The analogous self-duality statement for $\widetilde{G}_7^{-}$-twisted cocycles on $\mathfrak{m}2\mathfrak{brane}^{-}$ follows verbatim.

\newpage 
\end{corollary}
\begin{remark}[\bf Relation to the $\Gamma^{10}$-parity isomorphism]
\label{RelationToParity}
The self-T-duality construction discussed so far actually applies verbatim in more generality, i.e., for \textit{any} base super-$L_\infty$ algebra $\frg$ supporting a (non-trivial) $\mathfrak{l}S^4$-cocycle (Rem. \ref{CanonicalHigherTDualityCorrespondencesofABasicSphereCocycle}). However, in the current example of the $(G_4,G_7)$-cocycle on $\mathbb{R}^{1,10\vert \mathbf{32}}$, the self-duality described in Cor. \ref{AsASelfDualityOnm2brane} is, curiously, directly related to its $\Gamma^{10}$-parity isomorphism (Ex. \ref{ParityIsomorphism}). More precisely, recalling that $(\Gamma^{10})^*(G_4,G_7)= (-G_4,G_7)$, it follows immediately that:

\begin{itemize}[
  leftmargin=.85cm,
  topsep=2pt,
  itemsep=2pt
]
\item[\bf (i)] The \textit{$\Gamma^{10}$-reflection augmented} composite operation of (a) reduction \eqref{HigherCyclificationHomIsomorphism} along $\omega_4=G_4$ followed by (b) postcomposition with the automorphism \eqref{AutomorphismOfCycOfB2tU1} of $b \mathcal{T}_2 \cong  \mathrm{cyc}_3(b^6\mathbb{R})$ \textit{along with} precomposition with the \textit{parity isomorphism} \eqref{ComponentsOfParityIso} $\Gamma^{10}:\mathbb{R}^{1,10\vert\mathbf{32}}\rightarrow \mathbb{R}^{1,10\vert\mathbf{32}}$, and then (c) oxidation \eqref{HigherCyclificationHomIsomorphism} along the now resulting $\widetilde{\omega}_4={\color{purple}+}G_4$ acts trivially on the twist $\widetilde{G}_7$, i.e., precisely as the higher self T-duality of Cor. \ref{AsASelfDualityOnm2brane}.
\item[\bf (ii)] However, it may immediately be seen that the $\Gamma^{10}$-reflection augmented higher self-T duality of {\bf (i)} extends (via Thm. \ref{HigherTwistedCocyclesRedIsoReOxi}) instead to a \textit{different} higher self T-duality of $\widetilde{G}_7$-twisted cocycles. Explicitly, a twisted cocycle \begin{equation*}
  \begin{tikzcd}[
    row sep=-2pt, 
    column sep=0pt]
    \mathfrak{m}2\mathfrak{brane}
    \ar[
      rr
    ]
    &&
    \mathfrak{l}
    \big(
      \,
      \Sigma^{\color{purple}m}
      \mathrm{K}^{3}{U}
      \!\sslash\!\!B^{5}\mathrm{U}(1) \, \big)
    \\
    \widetilde{G}_7
    &\longmapsfrom&
    h_{7}
    \\
    (F_{6k +m})_{k\in \mathbb{Z}}
    &\longmapsfrom&
    (f_{6k +m})_{k\in \mathbb{Z}}
  \end{tikzcd}
\end{equation*}
is instead mapped to the twisted cocycle
\begin{equation*}
  \begin{tikzcd}[
    row sep=-2pt, 
    column sep=0pt]
    \mathfrak{m}2\mathfrak{brane}
    \ar[
      rr
    ]
    &&  
     \mathfrak{l}
    \big(
      \,
      \Sigma^{\color{purple}m-3}
      \mathrm{K}^{3}{U}
      \!\sslash\!\!B^{5}\mathrm{U}(1) \, \big)
    \\
    \widetilde{G}_7
    &\longmapsfrom&
    h_{7}
    \\
 \big(- (\Gamma^{10})^*{p}_*F_{6k+m} {\color{purple}-} c_3 \cdot (\Gamma^{10})^*(F_{6(k-1)+m})_{\bas} \, \big)_{k\in \mathbb{Z}}
    &\longmapsfrom&
    (f_{6k-3 +m })_{k\in \mathbb{Z}} \, .
  \end{tikzcd}
\end{equation*}
\item[\bf (iii)] The $(\Gamma^{10})^*$-action on the components of the fluxes and the sign in front of $c_3$ above may be compensated by further augmenting the operation from {\bf (i)} with an initial precomposition with the lifted parity automorphism \eqref{ParityIsomorphism} $\mathrm{par}: \mathfrak{m}2\mathfrak{brane} \rightarrow \mathfrak{m}2\mathfrak{brane}$ which conveniently also preserves the twist
\begin{align*}
\mathrm{par}^* \widetilde{G}_7 \,& = \, \mathrm{par}^* 2G_7 - \mathrm{par}^*c_3 
\cdot \mathrm{par}^*G_4 \\
&= \, 2G_7 - (-)^{2} c_3 \cdot G_4 
\\
&= \, \widetilde{G}_7\, .
\end{align*}
That is, since $(\Gamma^{10}\Gamma^{10})^*=\mathrm{id}^*$, the final product of a twisted cocycle $\big(\widetilde{G}_7, \, (F_{6k+m})_{k\in \mathbb{Z}} \big)$ on $\mathfrak{m}2\mathfrak{brane}$ under this (doubly) parity augmented higher T-duality operation is precisely the twisted cocycle 
\begin{equation*}
  \begin{tikzcd}[
    row sep=-2pt, 
    column sep=0pt]
    \mathfrak{m}2\mathfrak{brane}
    \ar[
      rr
    ]
    &&  
     \mathfrak{l}
    \big(
      \,
      \Sigma^{\color{purple}m-3}
      \mathrm{K}^{3}{U}
      \!\sslash\!\!B^{5}\mathrm{U}(1) \, \big)
    \\
    \widetilde{G}_7
    &\longmapsfrom&
    h_{7}
    \\
 \big(-{p}_*F_{6k+m} {\color{purple}-} c_3 \cdot (F_{6(k-1)+m})_{\bas} \, \big)_{k\in \mathbb{Z}}
    &\longmapsfrom&
    (f_{6k-3 +m })_{k\in \mathbb{Z}} \, .
  \end{tikzcd}
\end{equation*}
\end{itemize}
Remarkably, the (doubly) parity augmented operation of {\bf (iii)} coincides precisely with the higher self-duality from Cor. \ref{AsASelfDualityOnm2brane}, which instead made use of the tautological reflection isomorphism \eqref{M2BraneReflectionIsomorphism}. This  seems to indicate a further peculiar relation between the specific spacetime properties of $\mathbb{R}^{1,10\vert \mathbf{32}}$ and the nature of fluxes (via the target classifying space $\mathfrak{l}S^4$) that live on it.
\end{remark}

\noindent {\bf Higher T-duality as a Fourier--Mukai transform} In complete analogy to the case of standard superspace T-duality, Rem. \ref{HigherM2BraneTwistingCocyclesUnderHigherTDuality} and Prop. \ref{HigherM2BraneTwistedCocyclesUnderHigherTDuality} suggest that there should be an equivalent description in terms of the doubly extended correspondence space (cf. Def. \ref{TdualityCorrespondence}) and a pull-push isomorphism (cf. Prop. \ref{TDuality/FourierMukaiIsomorphism}). This is indeed the case, and in fact directly generalizes the definitions and results of the $3$-twisted case with degree $1$ central extensions not only to the $7$-twisted case with degree $3$ extensions, but to the cases of all odd $(4t-1)$-twisted cases with odd rational extensions of degree $2t-1$. We now spell out how this works in its full generality.

\begin{definition}[\bf Higher T-duality correspondence]\label{HigherTdualityCorrespondence}
Pairs $(\,\widehat{\frg}_A, H_A)$ and $(\,\widehat{\frg}_B, H_B)$ of higher centrally extended super $L_\infty$-algebras over $\frg$ via even $2t$-cocycles $\omega_{2t}^A, \omega_{2t}^B \in \CE(\frg)$, supplied with $(4t-1)$-cocycle twists, respectively, are said to be in \textit{higher T-duality correspondence} if:
\begin{itemize}[
  leftmargin=.85cm,
  topsep=2pt,
  itemsep=2pt
]
\item[\bf (i)] The respective fiber integration of the twists $H_{A/B}$ yields the opposite extension cocycles $c_{B/A}$ (cf. Eq. \eqref{TDualityOnPageChargeTwists} of Rem. \ref{HigherM2BraneTwistingCocyclesUnderHigherTDuality})
$$ 
(p_{A/B})_* H_{A/B} \, \, = \, \, \omega_{2t}^{B/A} \qquad \in \qquad  \CE(\frg) \, . 
$$
\item[\bf (ii)] On the doubly extended space $\widehat{\frg}_{A}\times_\frg \widehat{\frg}_B$ (cf. Eq. \eqref{M2BraneCorrespondenceSpace})
 $$
    \begin{tikzcd}[row sep=2pt, column sep=huge]
      &
     \widehat{\frg}_A \times_\frg \widehat{\frg}_B
      \ar[
        dr, 
        ->>,
        "{ \pi_B }"
      ]
      \ar[
        dl, 
        ->>,
        "{ \pi_A }"{swap}
      ]
      \ar[
        dd,
        phantom,
        "{
          \scalebox{.7}{
            \color{gray}
            (pb)
          }
        }"{pos=.3}
      ]
      \\
   \widehat{\frg}_A 
      \ar[
        dr,
        ->>,
        "{ p_A }"{swap}
      ]
      &&
      \widehat{\frg}_B \, ,
      \ar[
        dl,
        ->>,
        "{ p_B }" 
      ]
      \\
      &
    \frg 
    \end{tikzcd}
  $$
defined dually via
$$
\CE(\, \widehat{\frg}_{A}\times_\frg \widehat{\frg}_B) \quad = \quad  \CE(\frg)\big[b_A,b_B\big]\big/ \big(\dd b_{A/B}= \omega_{2t}^{A/B}\big) \, ,
$$
the {\it higher Poincar{\' e} form} (cf. Eq. \eqref{M2BranePoincareForm}) 
$$
P\, \, :=\, \, b_B \cdot b_A \quad \quad \in \quad \quad \CE(\, \widehat{\frg}_{A}\times_\frg \widehat{\frg}_B) 
$$
is a coboundary of the difference between the pullbacks of  the two twisting cocycles (cf. Prop. \ref{PoincarFormTrivializesDifferenceOfPullbacksOfNSFluxes})
$$\dd P\, \, = \, \, \pi_A^* H_A - \pi_B^* H_B \, .$$
\end{itemize}
\end{definition}

\noindent
The conditions of a higher T-duality correspondence (Def. \ref{HigherTdualityCorrespondence}) may be equivalently -- and concisely -- expressed via data over the original base  super-$L_\infty$ algebra $\frg$ {\cite[Prop. 3.13]{FSS20-HigherT}}.
\begin{lemma}[\bf Higher T-duality conditions on the base]\label{HigherTDualityConditionsOnTheBase}
Two pairs $(\widehat{\, \frg}_A, H_A)$ and $(\widehat{\, \frg}_B, H_B)$ are in higher T-duality correspondence (Def. \ref{HigherTdualityCorrespondence}) if and only if the twists are of the form
$$
H_{A/B}\; =\;  H_\frg + b_{A/B} \cdot \omega^{B/A}_{2t} 
$$
for a common basic $(4t-1)$-cochain $H_\frg \in \CE(\frg)$ whose differential trivializes the product of the corresponding extending $2t$-cocycles
$$
\dd H_\frg\; =\;  - \, \omega_{2t}^A \cdot \omega_{2t}^B \, .
$$
\begin{proof}
Follows verbatim as that of Lem. \ref{TDualityConditionsOnTheBase}, by modifying the degrees appropriately.
\end{proof}
\end{lemma}

\begin{corollary}[\bf Higher T-duality correspondence classifying space] It follows that the T-duality $L_\infty$-algebra $b\mathcal{T}_{t}\cong \mathrm{cyc}_{2t-1}\, b^{2t} \mathbb{R}$ from Ex. \ref{HigherTDualityClassifyingAlgebra} classifies the set of higher T-duality correspondences over any super-$L_\infty$ algebra $\frg$, in that morphisms of super $L_\infty$-algebras 
\begin{equation*}
  \begin{tikzcd}[row sep=-3pt,
    column sep=0pt
  ]
    \frg
    \ar[
      rr
    ]
    &&
   b\mathcal{T}_{t}
    \\
   \omega^A_{2t}
    &\longmapsfrom&
    \omega_{2t}
    \\
\omega^{B}_{2t}
    &\longmapsfrom&
    - \widetilde \omega_{2t}
    \\
H_\frg
    &\longmapsfrom&
     h_{4t-1} 
    \mathrlap{\, ,}
  \end{tikzcd}
\end{equation*}
are in canonical bijection with the set of higher T-duality correspondences over $\frg$ in terms of Lem. \ref{HigherTDualityConditionsOnTheBase} (hence equivalently in terms of Def. \ref{HigherTdualityCorrespondence}). 
\end{corollary}

\begin{example}[\bf $\mathfrak{m}2\mathfrak{brane}$ higher T-duality correspondence via the $(G_4,G_7)$-cocycle]
The cobounding condition \eqref{DifferentialOfBasicG7Flux} translates to a higher T-duality correspondence via the map
\begin{equation*}
  \begin{tikzcd}[row sep=-3pt,
    column sep=0pt
  ]
    \mathbb{R}^{1,10\vert \mathbf{32}}
    \ar[
      rr
    ]
    &&
   b\mathcal{T}_{2}
    \\
   G_4
    &\longmapsfrom&
    \omega_{4}
    \\
- G_4
    &\longmapsfrom&
    - \widetilde \omega_{4}
    \\
2G_7
    &\longmapsfrom&
     h_{7} 
    \mathrlap{\, ,}
  \end{tikzcd}
\end{equation*}
which evidently factors through the fixed $(G_4, G_7)$-cocycle as
\ \begin{equation}\label{FactoringThroughG4G7}
   \begin{tikzcd}[
   row sep=3pt
  ]
   \mathbb{R}^{1,10\vert \mathbf{32}}
    \ar[
      rr,
      "{
       (G_4, G_7)
      }"
    ]
    &&
    \mathfrak{l}S^4   
    \ar[
      rr,
      "{
        \iota_2
      }"
    ]
    &&
     b\mathcal{T}_{2} \, ,
    \end{tikzcd}
  \end{equation}
where the latter `embedding' morphism of the rational 4-sphere into the higher T-duality algebra  is given by 
\begin{equation*}
  \begin{tikzcd}[row sep=-3pt,
    column sep=0pt
  ]
 \mathfrak{l}S^4
    \ar[
      rr, "{
        \iota_2
      }"
    ]
    &&
   b\mathcal{T}_{2}
    \\
   g_4
    &\longmapsfrom&
    \omega_{4}
    \\
g_4
    &\longmapsfrom&
    \widetilde \omega_{4}
    \\
2g_7
    &\longmapsfrom&
     h_{7} 
    \mathrlap{\, .}
  \end{tikzcd}
\end{equation*}

\end{example}
\begin{remark}[\bf Canonical higher T-duality correspondences of a basic $\mathfrak{l} S^{2t}$-cocycle]\label{CanonicalHigherTDualityCorrespondencesofABasicSphereCocycle}  
$\,$

\noindent{\bf (i)} The factorization \eqref{FactoringThroughG4G7} implies immediately that  a fixed $\mathfrak{l}S^4$-cocycle on base super-$L_\infty$ algebra $\frg$, yields a class of different higher T-duality correspondences via different choices of embeddings of $\mathfrak{l}S^4$ into the higher T-duality algebra $b\mathcal{T}_t$. For example, it is immediate to see (Lem. \ref{HigherTDualityConditionsOnTheBase}) that post-composition of $(G_4,G_7): \mathbb{R}^{1,10\vert \mathbf{32}} \longrightarrow \mathfrak{l} S^4$ with the embedding 
\begin{equation*}
  \begin{tikzcd}[row sep=-3pt,
    column sep=0pt
  ]
   \mathfrak{l}S^4
    \ar[
      rr, "{
        \widehat{\iota}_2
      }"
    ]
    &&
   b\mathcal{T}_{2}
    \\
   g_4
    &\longmapsfrom&
    \omega_{4}
    \\
-g_4
    &\longmapsfrom&
    \widetilde \omega_{4}
    \\
-2g_7
    &\longmapsfrom&
     h_{7} 
    \mathrlap{\, , }
  \end{tikzcd}
\end{equation*}
yields (directly) a self-correspondence on
$$
\mathfrak{m}2\mathfrak{brane}
$$
but instead with the opposite twist of \eqref{PageCharge7Cocycle}
$$
\widehat{G}_7 \, \, := \, \, -2G_7 + c_3 G_4 \, \equiv \, - \widetilde{G}_7 \, .
$$

\noindent{\bf (ii)} Evidently, this observation generalizes to higher even sphere-valued cocycles on any super-$L_\infty$ algebra. That is, for any fixed $\mathfrak{l}S^{2t}$-cocycle
\vspace{0mm} 
\begin{equation*}
   \begin{tikzcd}[
     row sep=-2pt,
     column sep=20pt
   ]
     \frg
    \ar[
      rr,
      "{ (G_{2t},\,G_{4t-1}) }"
    ]
    &&
    \mathfrak{l}S^{2t}
    \\
    G_{2t} &\longmapsfrom& g_{2t}
    \\
    G_{4t-1} &\longmapsfrom& g_{4t-1}
    \, ,
   \end{tikzcd}
 \end{equation*}
postcomposition with (any of) the embeddings 
$$\mathfrak{l}S^{2t} \xlongrightarrow{\iota_t} b \mathcal{T}_t $$
immediately yields a higher T-duality correspondence (Lem. \ref{HigherTDualityConditionsOnTheBase}).
\end{remark}

For any higher $T$-duality correspondence, the  natural pull-push homomorphism of $(4t-1)$-twisted cochain complexes
\begin{align}\label{HigherTDualityPullPush}
T_{\mathrm{FM}} \quad := \quad - {\pi_B}_*\circ e^{-P} \circ \pi^*_A \equiv - {\pi_B}_{*} \circ (1+b_A \cdot b_B) \circ \pi_A^*  
\end{align}
is an isomorphism of degree $( -n_t \,\, \mathrm{mod} \, \, 2n_t, \, \mathrm{evn})$, for $$ n_t \, \, := \, \, 2t-1 \, .$$ In particular, it descends to an isomorphism on cohomology {\cite[Thm. 3.17]{FSS20-HigherT}}.

\begin{proposition}[\bf Higher T-duality/Fourier-Mukai isomorphism]\label{HigherTDualityFourierMukaiIsomorphism}
Let $(\,\widehat{\frg}_A, H_A)$ and $(\,\widehat{\frg}_B, H_B)$ be in higher T-duality correspondence. 

\noindent {\bf (i)} Then the pull-push morphism  \eqref{HigherTDualityPullPush} is an isomorphism of $(4t-1)$-twisted cochain complexes (Def. \ref{TwistedPeriodicChevalleyEilenberg}) and acts explicitly as
\begin{align}\label{HigherTdualityIso}
T_{\mathrm{FM}}\; \; :\; \; \CE^{\bullet + H_A}(\,\widehat{\frg}_A) &\xlongrightarrow{\quad \sim \quad} \CE^{(\bullet -n_t) +H_B}(\,\widehat{\frg}_B)\\
F = F_{\mathrm{bas}} + b_A \cdot {p_A}_{*} F &\quad \longmapsto \quad    - {p_A}_{*} F - b_B \cdot F_{\mathrm{bas}} \nonumber \, , 
\end{align}
thereby swapping the ``winding'' and ``non-winding'' modes.

\noindent {\bf (ii)}
In particular, it descends to a ``Fourier-Mukai'' isomorphism of higher twisted cocycles and furthermore higher twisted cohomologies
\begin{align*}
T_{\mathrm{FM}}\; \; :\; \; H_\CE^{\bullet + H_A}(\,\widehat{\frg}_A) &\xlongrightarrow{\quad \sim \quad}
H_\CE^{(\bullet -n_t) +H_B}(\,\widehat{\frg}_B) \, .
\end{align*}
\end{proposition}
\begin{proof}
Follows verbatim as that of Prop. \ref{TDuality/FourierMukaiIsomorphism}, by modifying the degrees appropriately.
\end{proof}

\begin{corollary}[\bf Pull-push via automorphism of higher cyclified twisted cocycle classifying algebras]
\label{PullPushViaAutomorphismOfCyclifiedHigherCoefficients}
Under the identification of $(2n_t+1)=(4t-1)$-twisted Chevalley-Eilenberg cocycles with maps into the corresponding classifying super-$L_\infty$-algebras from Eq. \eqref{HigherTwistedKTheorySpectraClassifyTwistedCocycles}, the action of the higher T-duality isomorphism  \eqref{HigherTdualityIso}
by Fourier-Mukai transform coincides with that of the composite operation of (a) reduction, (b) automorphism of higher cyclified twisted cocycle classifying algebras, and (c) reoxidation from Lem. \ref{HigherLInfinityAlgebraicTDuality}
$$
  \begin{tikzcd}[
    column sep={62pt, between origins},
    row sep=-3pt
  ]
    \mathfrak{g}_A
    \ar[
      rr
    ]
    \ar[
      dr,
      "{
        H_{\mathcolor{purple}{A}}
      }"{swap, xshift=2pt, yshift=2pt, pos=.6}
    ]
    \ar[
      d,
      phantom,
      shift right=9pt,
      "{
        \Bigg\{
      }"
      {pos=.2}
    ]
    &&
    \mathfrak{l}
    \big(
      \,
      \Sigma^{
        \mathcolor{purple}m
      }
      \mathrm{K}^{n_t}\mathrm{U}
      \!\sslash\!
      B^{2n_t-1}\mathrm{U}(1)
      \,
    \big)
    \ar[
      dl,
      shorten <=-5pt
    ]
    \ar[
      d,
      phantom,
      shift left=60pt,
      "{
        \Bigg\}
      }"{pos=-.2}
    ]
    \ar[
      rr,
      <->,
      shift right=8pt,
      shorten=10pt,
      "{
        \scalebox{.7}{
          \rm
          Thm. \ref{TwistedKTheoryRedIsoReOxi}
        }
      }"
    ]
    &[20pt]
    &[20pt]
    \mathfrak{g}_B
    \ar[
      rr
    ]
    \ar[
      dr,
      "{
        H_{\mathcolor{purple}{B}}
      }"{swap, xshift=2pt, yshift=2pt, pos=.6}
    ]
    \ar[
      d,
      phantom,
      shift right=10pt,
      "{
        \Bigg\{
      }"{pos=.2}
    ]
    &&
    \mathfrak{l}
    \big(
      \,
      \Sigma^{
        \mathcolor{purple}m-n_t
      }
      \mathrm{K}^{n_t}\mathrm{U}
      \!\sslash\!
      B^{2n_t-1}\mathrm{U}(1)
      \,
    \big)
    \ar[
      dl,
      shorten <=-5pt
    ]
    \ar[
      d,
      phantom,
      shift left=68pt,
      "{
        \Bigg\}.
      }"{pos=-.2}
    ]
    \\
    {}
    &
    b^{2n_t}\mathbb{R}
    \ar[
      d,
      hook,
      shorten=7pt
    ]
    &
    {}
    &&
    {}
    &
    b^{2n_t}\mathbb{R}
    \ar[
      d,
      hook,
      shorten=7pt
    ]
    &
    {}
    \\[30pt]
    &
    \mathrm{CE}^{m + H_{\mathcolor{purple}{A}}}(\frg_A)
    \ar[
      rrrr,
      <->,
      "{
         T_{\mathrm{FM}}
      }"
    ]
    & && &
    \mathrm{CE}^{m-n_t + H_{\mathcolor{purple}{B}}}(\frg_B)
  \end{tikzcd}
$$
\end{corollary}

\medskip

\subsection{Toroidal II$\mathrm{A}/\widetilde{\mathrm{A}}$-duality}

Next we turn to standard T-duality not just along a 1-dimensional fiber, but along {\it all} spacetime directions.

\medskip

\noindent
{\bf 10-Toroidal T-duality on super-fluxes.}
By combining the discussion of super-space T-duality in \S\ref{SuperspaceTDualityI} with that of higher-dimensional torus reductions in \S\ref{TorusReduction} it is now immediate to T-dualize super-fluxes on super $n$-tori for higher $n$.
In particular, since the type IIA super-spacetime is a $(1+9\vert 32)$-dimensional torus extension of the super-point $\mathbb{R}^{0\,\vert\, \mathbf{32}}$ (Ex. \ref{SuperMinkowskiAsToroidalExtensionOfSuperPoint}) we may consider the composite operation of Thm. \ref{TwistedKTheoryRedIsoReOxi}, in immediate higher dimensional analogy to \eqref{TDualityDiagram}:

\begin{itemize}[
  leftmargin=1cm,
  topsep=2pt,
  itemsep=2pt
]
\item[\bf (i)] toroidally reducing the type IIA super-fluxes to the super-point, via Prop. \ref{UniversalPropertyOfToroidification},

\item[\bf (ii)] transforming the result along the $T^{10}$-automorphism of Prop. \ref{TIsomorphismOf2kToroidifiedTwistedKTheory},

\item[\bf (iii)] toroidally re-oxidizing the result, but now with respect to the 10 resulting T-dual Chern classes $-\dir{a}p_\ast H_3^A$
\end{itemize}

\begin{equation}
  \label{10ToroidalTDualityDiagram}
  \begin{tikzcd}[
    row sep=15pt, 
    column sep=50pt
  ]
    &[-30pt]
    &[-75pt]
    &[+5pt]
    &
    &[-60pt]
    \mathfrak{l}
    \big(
      \mathrm{KU}
      \!\sslash\!\! B\mathrm{U}(1)
    \big)
   \\[-20pt]
    &[-30pt]
    &[-75pt]
    \widetilde{\mathbb{R}}^{
      1,9\,\vert\,
      \mathbf{16} \oplus 
      \overline{\mathbf{16}}
    }
    \ar[
      urrr,
      "{
        \mathrm{oxd}_{
          (\hspace{1pt}\overline{\psi}\Gamma^\bullet\psi)
        }
        \;
        \scalebox{.9}{$
          {\mathcolor{blue}T}^{10}
          \scalebox{1.5}{$($}
          \mathrm{rdc}_{
            (\hspace{1pt}\overline{\psi}\Gamma^\bullet\psi)
          }
          \scalebox{1.3}{$($}
            H_3^{\mathcolor{purple}A},\, (F_{2k})_{k \in \mathbb{Z}}
          \scalebox{1.3}{$)$}
          \scalebox{1.5}{$)$}
        $}
      }"{sloped},
      "{
        =:\;
        \big(
        H_3^{\mathcolor{purple}{\widetilde A}}
        ,\,
        (\widetilde F_{2k})_{k \in \mathbb{Z}}
        \big)
      }"{sloped, swap}
    ]
    \ar[
      ddddl,
      ->>,
      "{ \tilde p }"
    ]
    &[+5pt]
    &
    &[-60pt]
    \\[-12pt]
    \mathbb{R}^{
      1,9\,\vert\,
      \mathbf{16}\oplus\overline{\mathbf{16}}
    }
    \ar[
      rrrrd,
      crossing over,
      "{
        \scalebox{.9}{$
          \scalebox{1.3}{$($}
            H_3^{\mathcolor{purple}A},\, (F_{2k})_{k \in \mathbb{Z}}
          \scalebox{1.3}{$)$}
        $}
      }"{sloped, pos=.63}
    ]
    \ar[
      dddr,
      ->>,
      "{ p }"{swap}
    ]
    &
    &&&
    &[-30pt]
    \\[-32pt]
    &&&&
    \mathfrak{l}
    \big(
      \mathrm{KU}
      \!\sslash\!\! B\mathrm{U}(1)
    \big)
    \\
    &
    &&&
    &[-30pt]
    \mathrm{tor}^{10}\,
    \mathfrak{l}
    \big(
    \mathrm{KU}
    \!\sslash\!\! B\mathrm{U}(1)
    \big)
    \\[-10pt]
    &
    \mathbb{R}^{
      0 \,\vert\,\mathbf{32}
    }
    \ar[
      drrr,
      shorten <=-7pt,
      "{
        \scalebox{.9}{$
          \mathrm{rdc}_{
           (\hspace{1pt}\overline{\psi}\Gamma^\bullet\psi)
          }
          \scalebox{1.3}{$($}
            H_3^{\mathcolor{purple}A},\, (F_{2k})_{k \in \mathbb{Z}}
          \scalebox{1.3}{$)$}
        $}
      }"{sloped,swap, pos=.35}
    ]
    \ar[
      urrrr,
      "{
        \scalebox{.9}{$
          {\mathcolor{blue}T}^{10}
          \scalebox{1.5}{$($}
          \mathrm{rdc}_{
            (\hspace{1pt}\overline{\psi}\Gamma^\bullet\psi)
          }
          \scalebox{1.3}{$($}
            H_3^{\mathcolor{purple}A},\, (F_{2k})_{k \in \mathbb{Z}}
          \scalebox{1.3}{$)$}
          \scalebox{1.5}{$)$}
        $}
      }"{sloped}
    ]
    &&
    \\
    &
    &&&
    \mathrm{tor}^{10}\,
    \mathfrak{l}
    \big(
    \mathrm{KU}
    \!\sslash\!\! B\mathrm{U}(1)
    \big)
    \ar[
      uur,
      "{
         \mathcolor{blue}T^{10} 
       }",
      "{ 
        \sim 
      }"{swap, sloped}
    ]
  \end{tikzcd}
\end{equation}
Here the duality operation $T^{10}$
\eqref{TnAutomorphism} is seen to act on the NS- and D0 fluxes as 
\begin{equation}
  \begin{tikzcd}[row sep=-4pt, column sep=small]
    \mathbb{R}^{0\,\vert\, \mathbf{32}}
    \ar[
      rr,
      "{
        \scalebox{.9}{$
        \mathrm{rdc}_{
          (
            \overline{\psi}
            \,\Gamma^\bullet
            \psi
          )
        }
        \left(
          H_3^A
          ,\,
          (F_{2\bullet})
        \right)
        $}
      }"
    ]
    &[16pt]&[16pt]
    \mathrm{tor}^{10}
    \,
    \mathfrak{l}\big(
      \mathrm{KU}
      \!\sslash\!
      B\mathrm{U}(1)
    \big)
    \ar[
      rr,
      "{
        T^{10}
      }",
      "{ \sim }"{swap}
    ]
    &&
    \mathrm{tor}^{10}
    \,
    \mathfrak{l}\big(
      \mathrm{KU}
      \!\sslash\!
      B\mathrm{U}(1)
    \big)
    \\
    \mathllap{
      \dir{a}{c}_1
      =\;\;\;\;\;\,
    }
    \big(\hspace{1pt}
      \overline{\psi}
      \,\Gamma^a\,
      \psi
    \big)
    &\longmapsfrom&
    \dir{a}{\omega}_2
    &\longmapsfrom&
    \dir{a}{\mathrm{s}}h_3
    \\
    0
    &\longmapsfrom& 
    h_3
    &\longmapsfrom&
    h_3
    \\
    \mathllap{
      \dir{a}{\tilde c}_1
      =\;- \,
    }
    \big(\hspace{1pt}
      \overline{\psi}
      \,\Gamma_a \Gamma_{\!\ten}\,
      \psi
    \big)
    &\longmapsfrom&
    \dir{a}{\mathrm{s}}h_3
    &\longmapsfrom&
    \dir{a}{\omega}_2
  \end{tikzcd}
\end{equation}

\medskip

This means that the re-oxidized dual structure is given by:
\begin{definition}[\bf The fully T-dual super-spacetime]
\label{TheFullyTDualSuperSpacetime}
The fully T-dual super-spacetime in \eqref{10ToroidalTDualityDiagram}
is $\widetilde{\mathbb{R}}{}^{1,9\,\vert\,\mathbf{16} \oplus \overline{\mathbf{16}}}$ given by
\begin{equation}
  \label{DualSpacetime}
  \mathrm{CE}\big(
    \widetilde{\mathbb{R}}^{
      1,9\,\vert\,
      \mathbf{16} \oplus \overline{\mathbf{16}}
    }
  \big)
  \;\;
  \simeq
  \;\;
  \FDGCA
  \left[\!
  \def\arraystretch{1.4}
  \def\arraycolsep{2pt}
  \begin{array}{c}
    (\psi^\alpha)_{\alpha = 1}^{32}
    \\
    (\tilde{e}_a)_{a=0}^9
  \end{array}
 \right]
  \Big/
  \left(\!
  \def\arraystretch{1.1}
  \def\arraycolsep{2pt}
  \begin{array}{ccl}
    \mathrm{d}\, \psi &=& 0
    \\
    \mathrm{d}\, \tilde{e}_a
    &=& -
    \big(\hspace{1pt}
      \overline{\psi}
      \,\Gamma_a\Gamma_{\!\ten}\,
      \psi
    \big)
  \end{array}
  \!\right)
\end{equation}
carrying the fully dual super-flux density 
\begin{equation}
  \label{FullyDual3FluxDensity}
  H_3^{
    \mathcolor{purple}{
      \widetilde A
    }
  }
  \;\;
  =
  \;- \;
  \sum_{a=0}^9
  \tilde{e}_a \, \dir{a}{c}_1
  \;=\;- 
  \tilde{e}_a   
  \big(\hspace{1pt}
    \overline{\psi}
      \,\Gamma^a\,
    \psi
  \big).
\end{equation}
\end{definition}

\begin{remark}[\bf Nature of the fully T-dual IIA super-spacetime.]
\label{NatureOfFullyTDualSuperSpacetime}
The fully T-dual super-Lie algebra $\widetilde{\mathbb{R}}{}^{1,9\,\vert\, \mathbf{16} \oplus \overline{\mathbf{16}}}$ \eqref{DualSpacetime} is not actually isomorphic to the ordinary IIA super-spacetime $\mathbb{R}^{1,9\,\vert\, \mathbf{16} \oplus \overline{\mathbf{16}}}$ as real super-Lie algebras. It behaves like the IIA super-spacetime but with the signature convention of the metric, or equivalently of its Clifford anti-commutator, swapped:

While the IIA super-spacetime is controlled by the original Clifford generators 
$$
  (\Gamma_a)_{a=0}^9
  \,,
  \hspace{1cm}
  \Gamma_a \, \Gamma_b \,
  +
  \Gamma_b \Gamma_a
  \;=\;
  2\eta_{a b}
  \,,
$$
its fully T-dualized version is controlled by the 10D variant of \eqref{TheTypeIIBQuasiCliffordGenerators}:
$$
  \big(\,
    \widetilde \Gamma_a
    \,:=\,
    \Gamma_a \Gamma_{\!\ten}
  )_{a=0}^9
  \,,
  \hspace{1cm}
  \widetilde \Gamma_a \, 
  \widetilde \Gamma_b
  +
  \widetilde \Gamma_b \,
  \widetilde \Gamma_a
  \;=\;
  -2\eta_{a b}
  \,.
$$
Following the terminology of \cite{Hull98-TimelikeTDualityI}\cite{Hull98-TimelikeTDualityII} this ``fully T-dualized IIA spacetime'' would be called a ``$\mathrm{IIA}^{\!\ast}$-spacetime'', due to the T-duality along a timelike direction that it involves.
\end{remark}

\smallskip 
In variation of Def. \ref{1DoubledSuperSpace}, we now have:

\begin{definition}[\bf The fully doubled super-spacetime]
\label{FullyDoubleSuperSpacetime}
The fully doubled super-spacetime is the homotopy-fiber product \eqref{CommutingPairOfExtensions} of the type IIA spacetime (Ex. \ref{11dAsExtensionFromIIA}) with its full T-dual (Def. \ref{TheFullyTDualSuperSpacetime}) over the super-point \eqref{CEOfSuperpoint}:
\begin{equation}
  \label{FiberProductForFullyDoubledSpacetime}
  \begin{tikzcd}[row sep=-3pt, column sep=large]
    &
    \mathfrak{Dbl}
    \ar[
      dr,
      ->>,
      "{ 
        \pi_{\widetilde A} 
      }"{pos=.4}
    ]
    \ar[
      dl,
      ->>,
      "{ 
        \pi_{A} 
      }"{swap, pos=.4}
    ]
    \\
    \mathbb{R}^{
      1,9\,\vert\, 
      \mathbf{16} 
      \oplus
      \overline{\mathbf{16}}
    }
    \ar[
      dr,
      "{ p_A }"{swap}
    ]
    &&
    \widetilde{\mathbb{R}}{}^{
      1,9\,\vert\, 
      \mathbf{16} 
      \oplus
      \overline{\mathbf{16}}
    }
    \ar[
      dl,
      "{
        \widetilde{p}_A
      }"
    ]
    \\
    &
    \mathbb{R}^{
      0\,\vert\,\mathbf{32}
    }
  \end{tikzcd}
\end{equation}
given by
\begin{equation}
  \label{CEOfFullyDoubledSuperSpacetime}
  \mathrm{CE}\big(
    \mathfrak{Dbl}
  \big)
  \;\;
  \simeq
  \;\;
  \FDGCA
  \left[
  \def\arraystretch{1.2}
  \def\arraycolsep{1pt}
  \begin{array}{c}
    (\psi^\alpha)_{\alpha=1}^{32}
    \\
    (e^a)_{a=0}^9
    \\
    (\tilde e_a)_{a=0}^9
  \end{array}
  \right]
  \Big/
  \left(
  \def\arraystretch{1.2}
  \def\arraycolsep{1pt}
  \begin{array}{ccl}
    \mathrm{d}\, \psi &=& 0
    \\
    \mathrm{d}\,e^a &=& 
    \big(\hspace{1pt}
      \overline{\psi}
      \,\Gamma^a\,
      \psi
    \big)
    \\
    \mathrm{d}\,\tilde e_a &=& -
    \big(\hspace{1pt}
      \overline{\psi}
      \,\Gamma_a\Gamma_{\!\ten}\,
      \psi
    \big)
  \end{array}
  \right).
\end{equation}
\end{definition}

\medskip

A simple but important consequence for us is now the following variant of Prop. \ref{PoincarFormTrivializesDifferenceOfPullbacksOfNSFluxes}:

\begin{proposition}[\bf Full Poincar{\'e} super 2-form]
\label{Poinre2FormForFullDualization}
On the fully doubled super-spacetime
\eqref{CEOfFullyDoubledSuperSpacetime} the 2-form
\begin{equation}
  \label{FullPoincare2Form}
  P^{10}_2
  \;:=\; - 
  \tilde e_a \, e^a
  \;\in\;
  \mathrm{CE}\big(
    \mathfrak{Dbl}
  \big)
\end{equation}
is a coboundary of the difference between the NS 3-flux $H_3^A$ 
\eqref{IIACocycles}
in IIA and its fully T-dual $H_3^{\widetilde{A}}$ \eqref{FullyDual3FluxDensity}
\begin{equation}
  \def\arraystretch{1.4}
  \begin{array}{ccl}
  \mathrm{d}
  \,
  P^{10}_2
  &=&
   \big(\hspace{1pt}
    \overline{\psi}
    \,
    \Gamma_a\Gamma_{\!\ten}
    \,
    \psi
  \big)
  e^a
  +
  \tilde e_a
  \big(\hspace{1pt}
    \overline{\psi}
    \,
    \Gamma^a
    \,
    \psi
  \big)
  \\
  &=&
  \pi_A^\ast
  \,
  H_3^A
  -
  \pi_{\widetilde A}^\ast
  \,
  H_3^{\tilde A}
  \mathrlap{\,.}
  \end{array}
\end{equation}
\end{proposition}

\noindent
{\bf Toroidal T-duality as a Fourier-Mukai transform.} 
\label{ToroidalFourierMukaiTransform}
In analogy to the $n=1$ T-duality case, the above (fully) doubled super-space picture
leads to an equivalent formulation of toroidal T-duality in terms of correspondences and an induced Fourier-Mukai integral-transform on 3-twisted Chevalley--Eilenberg cochain complexes (Def. \ref{3TwistedPeriodicChevalleyEilenberg}).

\begin{definition}[\bf Toroidal T-duality Correspondence]\label{ToroidalTdualityCorrespondence}
Pairs $(\, \, \,  \dir{1\cdots n}{\widehat{\frg}}_A, H_A)$ and $(\, \, \, \dir{1\cdots n}{\widehat{\frg}}_B, H_B)$ of toroidally extended super $L_\infty$-algebras (Def. \ref{CentralTorusExtension}) over $\frg$ via n-tuples of even $2$-cocycles $\, \, \dir{1\cdots n}{c}_A,\, \dir{1\cdots n}{c
}_B \longrightarrow b \mathbb{R}^n$, supplied respectively with \textit{geometric} $3$-cocycle twists of the form 
$$ 
H_{A/B} \, = \, H_\frg + \sum_{k=1}^{n} \dir{k}{e}_{A/B} \cdot {\dir{k}{p}_{A/B}}_* (H_{A/B})
\, , 
$$ 
are said to be in toroidal \textit{T-duality correspondence} if:
\begin{itemize}[
  leftmargin=.8cm,
  topsep=2pt,
  itemsep=2pt
]
\item[\bf (i)] The $n$-tuple of (partial) fiber integrations of the twists $H_{A/B}$ yields the opposite extension cocycles $\, \, \dir{1\cdots n}{c}_{B/A}$, up to a dimensionally dependent sign  (cf. Eq. \eqref{DDReducedMBraneCocycles} and Eq. \eqref{FullyDual3FluxDensity}, where $H_\frg^3=0$)
$$ 
\,\, \, \dir{1\cdots n}{c}_{B/A}\, \, = \, \, (-1)^{n+1} \,\,  \,   \big(\, {\dir{1}{p}_{A/B}}_* (H_{A/B}),\cdots, {\dir{n}{p}_{A/B}}_* (H_{A/B})\big)\qquad : \qquad \frg \longrightarrow b \mathbb{R}^n \, . 
$$
\item[\bf (ii)] On the (fully) doubled extended space $\dir{1\cdots n}{\widehat{\frg}}_{A}\times_\frg \;\dir{1\cdots n}{\widehat{\frg}}_B$ (cf. Def. \ref{FullyDoubleSuperSpacetime})
 $$
    \begin{tikzcd}[row sep=2pt, column sep=huge]
      &
     \dir{1\cdots n}{\widehat{\frg}}_A \times_\frg \; \dir{1\cdots n}{\widehat{\frg}}_B
      \ar[
        dr, 
        ->>,
        "{ \dir{1\cdots n}{\pi}_B }"
      ]
      \ar[
        dl, 
        ->>,
        "{ \dir{1\cdots n}{\pi}_A }"{swap}
      ]
      \ar[
        dd,
        phantom,
        "{
          \scalebox{.7}{
            \color{gray}
            (pb)
          }
        }"{pos=.3}
      ]
      \\
   \dir{1\cdots n}{\widehat{\frg}}_A 
      \ar[
        dr,
        ->>,
        "{ \dir{1\cdots n}{p}_A }"{swap}
      ]
      &&
      \dir{1\cdots n}{\widehat{\frg}}_B \, ,
      \ar[
        dl,
        ->>,
        "{ \dir{1\cdots n}{p}_B }" 
      ]
      \\
      &
    \frg 
    \end{tikzcd}
  $$
defined dually via
$$
\CE(\,\, \,  \dir{1\cdots n}{\widehat{\frg}}_A \times_\frg \; \dir{1\cdots n}{\widehat{\frg}}_B) \quad = \quad  \CE(\frg)\big[(\dir{r}{e}_A)_{r=1}^{n}\,\, , (\dir{r}{e}_B)_{r=1}^{n}
    \, \, \big]
    \big/
    \big(
      \mathrm{d}
      \,
      \dir{r}{e}_{A/B}
      \;=\;
      \dir{r}{c}_{A/B}
\big)_{r=1}^n
    \, ,
$$
the {\it  Poincar{\' e} form} (cf. \eqref{FullPoincare2Form}) 
$$
P\, \, :=\,  (-1)^{n+1} \, \sum_{r=1}^n \dir{r}{e}_B \cdot \dir{r}{e}_A \quad \quad \in \quad \quad \CE(\,\, \,  \dir{1\cdots n}{\widehat{\frg}}_A \times_\frg \; \dir{1\cdots n}{\widehat{\frg}}_B)
$$
is a coboundary of the difference between the pullbacks of  the two twisting cocycles (cf. Prop. \ref{PoincarFormTrivializesDifferenceOfPullbacksOfNSFluxes})
$$
  \dd P 
    \;=\; 
    (\;\,\dir{\;\,1\cdots n}{\pi}_A)^* H_A 
      - 
    (\;\,\dir{\;\,1\cdots n}{\pi}_B)^* H_B
  \,.
$$
\end{itemize}
\end{definition}

The observation for the ``geometric'' (cf. Rem. \ref{GeometricTDuality}) structure of the twisting NS-fluxes yields an equivalent characterization toroidal T-duality correspondences.

\begin{lemma}[\bf Toroidal T-duality conditions on the base]\label{ToroidalTDualityConditionsOnTheBase}
Two pairs $(\,\, \,  \dir{1\cdots n}{\widehat{\frg}}_A, H_A)$ and $(\, \, \, \dir{1\cdots n}{\widehat{\frg}}_B, H_B)$ are in T-duality correspondence (Def. \ref{ToroidalTdualityCorrespondence}) if and only if the twists are of the geometric form
$$ 
H_{A/B} \, = \, H_\frg + (-1)^{n+1} \sum_{k=1}^{n} \dir{k}{e}_{A/B} \cdot \dir{k}{c}_{B/A}
\, , 
$$
for a common basic $3$-cochain $H_\frg \in \CE(\frg)$ whose differential trivializes the sum of  the product of the corresponding extending $2$-cocycles
$$
\dd H_\frg\; =\;  (-1)^{n} \sum_{k=1}^{n} \dir{k}{c}_{A} \cdot \dir{k}{c}_{B}\, .
$$
\begin{proof}
Follows by analogous calculations as the proof 
the $n=1$ case in Lem. \ref{TDualityConditionsOnTheBase}, while carrying along the dimensionally dependent signs.
\end{proof}
\end{lemma}
\begin{corollary}[\bf Toroidal T-duality correspondence classifying space] It follows that the T-duality $L_\infty$-algebra $b\mathcal{T}^n\cong \mathrm{tor}^{n'}\, b^2 \mathbb{R}$ from Ex. \ref{nToroidificationOfBundleGerbeClassifyingSpace} classifies the set of $n$-toroidal T-duality correspondences over any super-$L_\infty$ algebra $\frg$, in that morphisms of super $L_\infty$-algebras
\begin{equation*}
  \begin{tikzcd}[row sep=-3pt,
    column sep=0pt
  ]
    \frg
    \ar[
      rr
    ]
    &&
   b\mathcal{T}^n
    \\
   \dir{r}{c}_A
    &\longmapsfrom&
    \dir{r}{\omega}_2
    \\
\dir{r}{c}_B
    &\longmapsfrom&
    (-1)^{n} \widetilde {\dir{r}{\omega}}_2
    \\
H_\frg
    &\longmapsfrom&
     h_3 
    \mathrlap{\, ,}
  \end{tikzcd}
\end{equation*}
are in canonical bijection with the set of toroidal T-duality correspondences over $\frg$ in terms of Lem. \ref{ToroidalTDualityConditionsOnTheBase} (hence equivalently in terms of Def. \ref{ToroidalTdualityCorrespondence}). 
\end{corollary}
 
For any $n$-toroidal T-duality correspondence, the natural ``pull-push'' homomorphism of 3-twisted cochain complexes 
\begin{align}\label{ToroidalTDualityPullPush}
T_{\mathrm{FM}}^n \quad :&= \quad (-1)^{1+n(n-1)/2} \quad  {\dir{1\cdots n}{\pi}_B}_*\circ e^{-P} \circ \dir{1\cdots n}{\pi}_A{}^*
\\
& \equiv\quad (-1)^{1+n(n-1)/2}  \quad   {\dir{1\cdots n}{\pi}_B}_{*} \circ \bigg(1+ \sum\limits_{\substack{1\leq k \leq n \\ 1 \leq j_1< \cdots < j_k \leq n}}  \, (-1)^{kn}\,  \dir{j_k}{e}_B \cdot \dir{j_k}{e}_A \cdots \dir{j_1}{e}_B \cdot \dir{j_1}{e}_A \bigg) \circ \dir{1\cdots n}{\pi}_A{}^*
\nonumber \\
& = \quad (-1)^{1+n(n-1)/2}  \quad   {\dir{1\cdots n}{\pi}_B}_{*} \circ \bigg(1+(-1)^{n} \, \sum_{j=1}^n  \dir{j}{e}_B \cdot \dir{j}{e}_A +  \sum_{1\leq j_1 < j_2 \leq n}^n  \dir{j_2}{e}_B \cdot \dir{j_2}{e}_A \cdot \dir{j_1}{e}_B \cdot \dir{j_1}{e}_A  +\cdots \bigg) \circ \dir{1\cdots n}{\pi}_A{}^* \nonumber
\end{align}
is an isomorphism of degree $( -n \,\mathrm{mod} \, 2, \, \mathrm{evn})$. 
In particular, it descends to an isomorphism on cohomology. Here we provide the full details of the proof for any $n\in \NN$, which is only sketched  for $n\geq 2$ in {\cite[Thm. 3.17]{FSS20-HigherT}}.

\begin{proposition}[\bf Toroidal T-duality/Fourier-Mukai isomorphism]
\label{ToroidalTDuality/FourierMukaiIsomorphism}
Let  $(\,\, \,  \dir{1\cdots n}{\widehat{\frg}}_A, H_A)$ and $(\, \, \, \dir{1\cdots n}{\widehat{\frg}}_B, H_B)$ be in toroidal T-duality correspondence. 

\noindent {\bf (i)} Then the pull-push homomorphism  \eqref{ToroidalTDualityPullPush} is in fact an isomorphism of 3-twisted cochain complexes (Def. \ref{3TwistedPeriodicChevalleyEilenberg}) and acts explicitly as
\begin{align}\label{ToroidalTdualityIso}
  T^n_{\mathrm{FM}}
    \;\;:\;\; 
    \CE^{\bullet + H_A}
    \big(\,\,\, 
      \dir{1\cdots n}{\widehat{\frg}}_A
    \big) 
    &\xlongrightarrow{\quad \sim \quad} \CE^{(\bullet -n\, \mathrm{mod} \, 2) +H_B}
    \big(
      \,\,\, 
      \dir{1\cdots n}{\widehat{\frg}}_B
    \big)
    \\
    F 
      \;=\;  
    F_\frg
    +   
    \sum\limits_{
      \mathclap{
      \substack{
        1\leq r \leq n 
        \\ 
        1 \leq i_1 < \cdots < i_r \leq n
      }
      }
    } 
    \;
    \dir{i_r}{e}_A\cdots \dir{i_1}{e}_A \cdot F_{i_r \cdots i_1}  &
      \;\longmapsto\;
    \bigg((-1)^{1+n(n-1)/2} F_{ n\cdots 1}   + \dir{n}{e}_B\cdots \dir{1}{e}_B \cdot F_\frg  \nonumber 
\\[-20pt]
&\quad \phantom{\longmapsto\bigg(AA} 
+ 
  \sum\limits_{
    \mathclap{
    \substack{
      1\leq r < n 
      \\ 
      1 \leq i_1 < \cdots < i_r \leq n
    }
    }
  }
  \;
  (-1)^{1+n+r}  (-1)^{(n-r)(n-r+1)/2} \,  \dir{i_r}{e}_B\cdots \dir{i_1}{e}_B  \star F_{ i_r \cdots i_1}\bigg)
\nonumber \, , 
\end{align}
where we have abbreviated $\star F_{ i_r \cdots i_1} := \tfrac{1}{(n-r)!}
  \,
  \epsilon^{ 
      i_{n} \cdots i_{r+1}}{}_{ 
      i_r \cdots i_2 i_1 }
 F_{ i_{n} \cdots i_{r+1}
    }$, thereby swapping modes winding along certain torus directions to those winding along the Hodge-dual directions.
    
\noindent {\bf (ii)}  In particular, it descends to a ``Fourier-Mukai'' isomorphism of twisted cocycles and furthermore twisted cohomologies
\begin{align*}
T^n_{\mathrm{FM}}\; \; :\; \; H_\CE^{\bullet + H_A}(\,\, \, \dir{1\cdots n}{\widehat{\frg}}_A) &\xlongrightarrow{\quad \sim \quad}
H_\CE^{(\bullet -n\, \mathrm{mod} \, 2) +H_B}(\,\, \, \dir{1\cdots n}{\widehat{\frg}}_B) \, .
\end{align*}
\end{proposition}
\begin{proof} 
For brevity, we shall omit reference to the injective pullback morphisms  $(\;\dir{\;\,1\cdots n}{\pi}_{A/B})^*$. Recall that 
$$
{\dir{1\cdots n}{\pi}_B}_{*} \, : \, \CE(\, \, \,  \dir{1\cdots n}{\widehat{\frg}}_A \times_\frg \dir{1\cdots n}{\widehat{\frg}}_B) \longrightarrow  \CE(\, \, \, \dir{1\cdots n}{\widehat{\frg}}_B)
$$
is the total $n$-toroidal fiber integration (cf. \eqref{2ToroidalFiberIntegration}), mapping the highest-degree product of generators $\dir{n}{e}_A\cdots \dir{1}{e}_A$ of the first $n$-toroidal extension to $1$, and all lower-degree products to $0$. Then the explicit form of the mapping \eqref{ToroidalTDualityPullPush} follows by a tedious but straightforward calculation:
\begin{align*}
T^n_\mathrm{FM} \bigg( F_\frg
    &+
    \sum\limits_{
      \mathclap{
      \substack{
        1\leq r \leq n 
        \\ 
        1 \leq i_1 < \cdots < i_r \leq n
      }
      }
    } 
    \;
    \dir{i_r}{e}_A\cdots \dir{i_1}{e}_A \cdot F_{i_r \cdots i_1} \bigg) 
    \\
    &= (-1)^{1+n(n-1)/2}\quad {\dir{1\cdots n}{\pi}_B}_* \bigg(  F_\frg
    +
    \sum\limits_{
      \mathclap{
      \substack{
        1 \leq r \leq n 
        \\ 
        1 \leq i_1 < \cdots < i_r \leq n
      }
      }
    }
    \;
    \dir{i_r}{e}_A\cdots \dir{i_1}{e}_A \cdot F_{i_r \cdots i_1} 
    \\&
    \qquad + 
    \sum\limits_{
      \mathclap{
      \substack{
        1 \leq k \leq n 
        \\ 
        1 \leq j_1 < \cdots < j_k \leq n
      }
      }
    } 
    \; 
    (-1)^{kn}
    \;\;\;\;\;
    \sum\limits_{
      \mathclap{
      \substack{
        1\leq r < n 
        \\ 
        1 \leq i_1 < \cdots < i_r \leq n
      }
      }
    }
    \;
    \dir{j_k}{e}_B\cdot \dir{j_k}{e}_A \cdots \dir{j_1}e_B \cdot \dir{j_1} e_A \cdot   \dir{i_r}{e}_A\cdots \dir{i_1}{e}_A \cdot F_{i_r \cdots i_1}  \\
    &\qquad + 
    \sum\limits_{
      \mathclap{
      \substack{
        1\leq k \leq n 
        \\ 
        1 \leq j_1 < \cdots < j_k \leq n
      }
      }
    } 
    \; (-1)^{kn}\,  \dir{j_k}{e}_B\cdot \dir{j_k}{e}_A \cdots \dir{j_1}e_B \cdot \dir{j_1} e_A \cdot F_\frg \bigg)
        \\
    & = 
    (-1)^{1+n(n-1)/2} 
      \quad 
    {\dir{1\cdots n}{\pi}_B}_* 
    \bigg(
        \dir{n}{e}_A \cdots \dir{1}{e}_A \cdot F_{n \cdots 1}  
          \,+\, 
        (-1)^{ n(n-1)/2} \, 
        \dir{n}{e}_A \cdots \dir{1}{e}_A \cdot \dir{r}e_B \cdots \dir{1}e_B \cdot F_\frg 
    \\&
    \qquad 
    + 
    \sum\limits_{
      \mathclap{
      \substack{
        1 \leq r < n 
        \\ 
        1 \leq i_1 < \cdots < i_r \leq n 
        \\
        1 \leq j_1 < \cdots < j_{n-r}\leq n
      }
      }
    }  
    \;
    (-1)^{n+r}  (-1)^{(n-r)(n-r+1)/2} 
    \;\;  
    \dir{j_{n-r}}{e}_A \cdots\dir{j_1} e_A \cdot \dir{i_r}{e}_A\cdots \dir{i_1}{e}_A \cdot \dir{j_{n-r}}{e}_B\cdots \dir{j_1}e_B    \cdot F_{i_r \cdots i_1} \bigg) 
    \\
    &=
      (-1)^{1+n(n-1)/2} 
      \quad 
    {\dir{1\cdots n}{\pi}_B}_* 
    \bigg(
        \dir{n}{e}_A \cdots \dir{1}{e}_A \cdot F_{n \cdots 1}  
          \,+\, 
        (-1)^{ n(n-1)/2} \, 
        \dir{n}{e}_A \cdots \dir{1}{e}_A \cdot \dir{r}e_B \cdots \dir{1}e_B \cdot F_\frg 
    \\&
    \qquad 
    + 
    \sum\limits_{
      \mathclap{
      \substack{
        1 \leq r < n 
      }
      }
    }  
    \;
    (-1)^{n+r}  (-1)^{(n-r)(n-r+1)/2} 
    \;\;\frac{1}{r!(n-r)!} \epsilon_{i_n\cdots i_{r+1}}{}^{i_r\cdots i_1} 
    \dir{n}{e}_A \cdots\dir{1} e_A \cdot \dir{i_n}{e}_B\cdots \dir{i_{r+1}}e_B    \cdot F_{i_r \cdots i_1} \bigg) 
    \\
    &=  (-1)^{1+n(n-1)/2} \, F_{n \cdots 1}  -  \dir{n}e_B \cdots \dir{1}e_B \cdot F_\frg 
    \\
    & \qquad  + 
    \sum\limits_{
      \mathclap{
      \substack{
        1 \leq r < n 
      }
      }
    }  
    \;
  (-1)^{1+r(r-1)/2} 
    \;\;\frac{1}{r!(n-r)!} \epsilon^{i_r\cdots i_1} {}_{i_n\cdots i_{r+1}}
   \dir{i_n}{e}_B\cdots \dir{i_{r+1}}e_B    \cdot F_{i_r \cdots i_1} 
\end{align*}
where in the second equality we dropped the terms that map to zero under total fiber integration, so in particular setting $k=n-r$, and picked signs by moving the $\dir{r}e_{A}$-generators to the left through the $\dir{r}e_B$-generators. In the third equality we renamed $j_1 \mapsto i_{r+1}, \cdots, j_{n-r}\mapsto i_n$, while letting the corresponding ordered sums be unordered and compensating with the $\frac{1}{r!(n-r)!}$ prefactor. In the fourth equality we applied the total fiber integration and simplified the sign prefactors, while swapping the indices in the $\epsilon$-symbol. The result then follows by relabelling the dummy indices $r\mapsto n-r$ under the sum.

The fact that $T_{\mathrm{FM}}$ is a linear isomorphism of cochains follows by the existence of the explicit inverse
\begin{align*}(T_{\mathrm{FM}}^{n})^{-1} \quad := \quad (-1)^{1+n(n-1)/2} \quad  {\dir{1\cdots n}{\pi}_A}_*\circ e^{P} \circ \dir{1\cdots n}{\pi}_B{}^*
\end{align*}
acting by ``swapping back'' the winding and Hodge-dual winding modes,  by the same calculation under the exchange of indices $A\leftrightarrow B$.

Lastly, to see that the linear isomorphism $T_{\mathrm{FM}}^n$ interwines with the twisted differentials (up to a sign), recall that the total fiber integration ${\dir{1\cdots n}{\pi}_B}_{*}$ is equivalently given by the composition of the $n$-consecutive partial fiber integrations ${\dir{1\cdots n}{\pi}_B}_{*}= {\dir{1}{\pi}_B}_{*}\circ \cdots \circ {\dir{n}{\pi}_B}_{*}$ and hence commutes with the corresponding (untwisted) differentials up to $(-1)^{n}$ sign (cf. Rem. \ref{BasicAndFiberFormsOn2ToroidalExtension}),
as each of the partial fiber integrations do so up to a $(-1)$ sign (cf. Rem.  \ref{BasicAndFiberForms}). Noting also that multiplication by $H_B$ trivially commutes with the total fiber integration, it follows that
\begin{align*}
(\dd - H_B ) \circ T^n_\mathrm{FM} (F) &= 
(-1)^{1+n(n-1)/2}\, \,  (\dd - H_B ) \circ  {\dir{1\cdots n}{\pi}_B}_*\circ e^{-P} (F)  
\\
&=  (-1)^{1+n+n(n-1)/2} \quad  {\dir{1\cdots n}{\pi}_B}_\ast \circ  \big(d-H_B) (e^{-P}F)  
\\
&=  (-1)^{1+n+n(n-1)/2} \quad  {\dir{1\cdots n}{\pi}_B}_\ast \big( - e^{-P} \cdot \dd P \cdot F + e^{-P} \cdot \dd F - e^{-P} \cdot H_B \cdot F \big) 
\\
& =  (-1)^{1+n+n(n-1)/2} \quad  {\dir{1\cdots n}{\pi}_B}_\ast \circ e^{-P} \circ ( \dd - \dd P  -   H_B ) (F)   
\\
&= (-1)^n \, T^{n}_\mathrm{FM} \circ (\dd - H_A ) (F) 
\end{align*}
where the final line employs the {\bf (ii)} condition $\dd P = H_A - H_B$ of a toroidal correspondence from Def. \ref{ToroidalTdualityCorrespondence}.
\end{proof}

\begin{corollary}[\bf Pull-push via automorphism of toroidified twisted K-theory spectra]
\label{PullPushViaAutomorphismOfToroidifiedCoefficients}
Under the identification of twisted K-theory cocycles with  $3$-twisted cocycles from Eq. \eqref{TwistedKTheorySpectraClassifyTwistedCocycles}, the action of the T-duality isomorphism  \eqref{ToroidalTdualityIso}
by Fourier-Mukai transform coincides with that of the composite operation of (a) reduction, (b) automorphism of n-toroidified twisted K-theory spectra (c) reoxidation from Thm. \ref{TwistedKTheoryToroidalRedIsoReOxi}
$$
  \begin{tikzcd}[
    column sep={55pt, between origins},
    row sep=-3pt
  ]
    \dir{1\cdots n}{\widehat{\frg}}_A
    \ar[
      rr
    ]
    \ar[
      dr,
      "{
        H_{\mathcolor{purple}{A}}
      }"{swap, xshift=2pt, yshift=2pt, pos=.6}
    ]
    \ar[
      d,
      phantom,
      shift right=15pt,
      "{
        \Bigg\{
      }"{pos=-.6}
    ]
    &&
    \mathfrak{l}
    \big(
      \,
      \Sigma^{
        \mathcolor{purple}0
      }
      \mathrm{KU}
      \!\sslash\!
      B\mathrm{U}(1)
      \,
    \big)
    \ar[
      dl,
      shorten <=-5pt
    ]
    \ar[
      d,
      phantom,
      shift left=43pt,
      "{
        \Bigg\}
      }"{pos=-.6}
    ]
    \ar[
      rr,
      <->,
      shift right=8pt,
      shorten=10pt,
      "{
        \scalebox{.7}{
          \rm
          Thm. \ref{TwistedKTheoryToroidalRedIsoReOxi}
        }
      }"
    ]
    &[20pt]
    &[20pt]
    \;\;
    \dir{1\cdots n}{\widehat{\frg}}_B
    \ar[
      rr
    ]
    \ar[
      dr,
      "{
        H_{\mathcolor{purple}{B}}
      }"{swap, xshift=2pt, yshift=2pt, pos=.6}
    ]
    \ar[
      d,
      phantom,
      shift right=14pt,
      "{
        \Bigg\{
      }"{pos=-.6}
    ]
    &&
    \mathfrak{l}
    \big(
      \,
      \Sigma^{
        \mathcolor{purple}n \, \mathrm{mod} \, 2
      }
      \mathrm{KU}
      \!\sslash\!
      B\mathrm{U}(1)
      \,
    \big)
    \ar[
      dl,
      shorten <=-5pt
    ]
    \ar[
      d,
      phantom,
      shift left=56pt,
      "{
        \Bigg\}.
      }"{pos=-.6}
    ]
    \\
    {}
    &
    b^2\mathbb{R}
    \ar[
      d,
      hook,
      shorten=7pt
    ]
    &
    {}
    &&
    {}
    &
    b^2\mathbb{R}
    \ar[
      d,
      hook,
      shorten=7pt
    ]
    &
    {}
    \\[30pt]
    &
    \mathrm{CE}^{0 + H_{\mathcolor{purple}{A}}}(\, \, \, \dir{1\cdots n}{\widehat{\frg}}_A)
    \ar[
      rrrr,
      <->,
      "{
        T_{\mathrm{FM}}^n
      }"
    ]
    & && &
    \mathrm{CE}^{n \,\mathrm{mod}\,  2 + H_{\mathcolor{purple}{B}}}(\, \, \, \dir{1\cdots n}{\widehat{\frg}}_B)
  \end{tikzcd}
$$
\end{corollary}

\newpage 
\section{Lift to the M-algebra}
\label{LiftingToMTheory}

After recalling the super Lie algebra of the F-theory super-spacetime (Def. \ref{FTheorySuperSpacetime}),
we observe that its analogue $\mathfrak{F}$ \eqref{frakF} 
for T-duality reduction on a $(1,\!9\,\vert\, 32)$-dimensional super-torus all the way down to the point is
further extended \eqref{CommutingPairOfExtensions}
by the ``M-algebra'', with the ``double'' copy of the (full 10-dimensional) super-spacetime now constituted entirely by membrane wrapping modes, and 
with the Poincar{\'e} super 2-form  \eqref{Poincare2FormForSingleFiber} lifted to an M-theoretic Poincar{\'e} super 3-form (Rem. \ref{LiftOfPoincare3Form}).
In the companion article \cite{GSS24-MGroup} we use this to explain the M-algebra as the super-space version of the exceptional-geometric tangent space for 11D supergravity reduced all the way to the point.

\medskip

\noindent
{\bf F-Theory super-spacetime.}
Given that the derivation began in \eqref{MBraneCocycles} on 11D super-spacetime, going through its reduction to 10D IIA super-spacetime \eqref{11dAsExtensionFromIIA},  to arrive at its ``doubled'' version \eqref{CEOf1DoubledSuperSpace}, it is natural to ask for the doubling of the T-dualized fiber to take place already in 11D, hence for extending 9D super-spacetime by {\it all three} extra dimensions: 

\vspace{0mm} 
{\bf (i) }the IIA fiber, 

\vspace{0mm} 
{\bf (ii)} the IIB fiber, 

\vspace{0mm} 
{\bf (iii)} the M fiber.

\smallskip

At the level of super-Lie algebraic local model spaces, this request is immediate to satisfy:

\begin{definition}[{\bf F-theory super-spacetime} {\cite[Def. 8.1, Prop. 8.3]{FSS18-TDualityA}}]
\label{FTheorySuperSpacetime}
Write $\mathbb{R}^{1,9+(1+1)\,\vert\,32}$ for the super-Lie algebra given by
\vspace{-2mm}
\begin{equation}
  \label{TheFTheorySuperSpacetime}
  \mathrm{CE}\big(
    \mathbb{R}^{
      1,9+(1+1)\,\vert\,32
    }
  \big)
  \;\simeq\;
 \FDGCA
  \left[
  \def\arraystretch{1.4}
  \def\arraycolsep{2pt}
  \begin{array}{c}
    (\psi^\alpha)_{\alpha=1}^{32}
    \\
    (e^a_A)_{a=1}^{9}
    \\
    e^{\ten},
    e^{9}_A,
    e^9_B
  \end{array}
  \right]
  \Big/
  \left(
  \def\arraystretch{1.2}
  \def\arraycolsep{2pt}
  \begin{array}{cclcl}
    \mathrm{d}\,
    \psi^\alpha
    &=&
    0
    \\
    \mathrm{d}
    \,
    e^a 
    &=&
    \big(\hspace{1pt}
      \overline{\psi}
      \,\Gamma^a\,
      \psi
    \big)
    \\
    \mathrm{d}
    \,
    e^{9}_A 
    &=&
    \big(\hspace{1pt}
      \overline{\psi}
      \,\Gamma^9\,
      \psi
    \big)
    &=&
    \big(\hspace{1pt}
      \overline{\psi}
      \,\sigma_1\,
      \psi
    \big)
    \\
    \mathrm{d}
    \,
    e^{\ten} 
    &=&
    \big(\hspace{1pt}
      \overline{\psi}
      \,\Gamma^{\ten}\,
      \psi
    \big)
    &=&
    \big(\hspace{1pt}
      \overline{\psi}
      \,\sigma_2\,
      \psi
    \big)
    \\
    \mathrm{d}
    \,
    e^{9}_B 
    &=&
    \big(\hspace{1pt}
      \overline{\psi}
      \,\Gamma^9\Gamma^{\ten}\,
      \psi
    \big)
    &=&
    \big(\hspace{1pt}
      \overline{\psi}
      \,\sigma_3\,
      \psi
    \big)
  \end{array}
\!  \right)
\end{equation}
(using the notation \eqref{SpinorialSDualityGenerators} on the right)
which is equivalently the homotopy-fiber product \eqref{CommutingPairOfExtensions} of the 11D super-spacetime 
with the doubled super-spacetime \eqref{CEOf1DoubledSuperSpace}:
\begin{equation}
  \label{FTheorySpacetimesAsFiberProduct}
  \begin{tikzcd}[row sep=small, column sep=large]
    &
    \mathbb{R}^{
      1,9+(1+1)\,\vert\, 32
    }
    \ar[
      d,
      phantom,
      "{
        \scalebox{.7}{
          \color{darkblue}
          \bf
          \def\arraystretch{.9}
          \begin{tabular}{c}
            F super-
            \\
            spacetime
          \end{tabular}
        }
      }"{pos=.3}
    ]
    \ar[
      dr
    ]
    \ar[
      dl
    ]
    \\
    \mathbb{R}^{
      1,10\,\vert\, \mathbf{32}
    }
    \ar[
      d,
      phantom,
      "{
        \scalebox{.7}{
          \color{darkblue}
          \bf
          \def\arraystretch{.9}
          \begin{tabular}{c}
            11D super-
            \\
            spacetime
          \end{tabular}
        }
      }"{pos=.3}
    ]
    \ar[
      dr,
      ->>,
      "{ 
         p_M 
       }"{swap}
    ]
    &
    {}
    &
    \mathbb{R}^{
      1,8+(1+1)\,\vert\, 32
    }
    \ar[
      d,
      phantom,
      "{
        \scalebox{.7}{
          \color{darkblue}
          \bf
          \def\arraystretch{.85}
          \begin{tabular}{c}
            doubled 
            \\
            super-
            spacetime
          \end{tabular}
        }
      }"{pos=.3}
    ]
    \ar[
      dl,
      ->>,
      "{
        \pi_A
      }"{swap}
    ]
    \ar[
      dr,
      ->>,
      "{
        \pi_B
      }"
    ]
    \\
    {}
    &
    \mathbb{R}^{
      1,9\,\vert\,
      \mathbf{16}
      \oplus
      \overline{\mathbf{16}}
    }
    \ar[
      d,
      phantom,
      "{
        \scalebox{.7}{
          \color{darkblue}
          \bf
          \def\arraystretch{.9}
          \begin{tabular}{c}
            IIA super-
            \\
            spacetime
          \end{tabular}
        }
      }"{pos=.3}
    ]
    \ar[
      dr,
      ->>
    ]
    &{}&
    \mathbb{R}^{
      1,9\,\vert\,
      \mathbf{16}
        \oplus
      \mathbf{16}
    }
    \ar[
      d,
      phantom,
      "{
        \scalebox{.7}{
          \color{darkblue}
          \bf
          \def\arraystretch{.9}
          \begin{tabular}{c}
            IIB super-
            \\
            spacetime
          \end{tabular}
        }
      }"{pos=.3}
    ]
    \ar[
      dl,
      ->>
    ]
    \\
    {}
    &
    {}
    &
    \mathbb{R}^{
      1,8\vert 
      \mathbf{16}
      \oplus
      \mathbf{16}
    }
    \ar[
      d,
      phantom,
      "{
        \scalebox{.7}{
          \color{darkblue}
          \bf
          \def\arraystretch{.9}
          \begin{tabular}{c}
            9D super-
            \\
            spacetime
          \end{tabular}
        }
      }"{pos=.5}
    ]
    &
    {}
    \\
    & & {} & {}
  \end{tikzcd}
\end{equation}
\end{definition}

\vspace{-1mm} 
By inspection, one sees that (cf. also \cite{Sakaguchi00}):

\begin{proposition}[{\bf Superspace S-duality on F-theory super-spacetime} {\cite[Prop. 8.6]{FSS18-TDualityA}}]
\label{SuperspaceSDuality}
  The group $\mathrm{Pin}(2)$ of Prop. \ref{RSymmetryOfTypeII} acts by super-Lie automorphisms on the F-theory super-spacetime \eqref{TheFTheorySuperSpacetime}
  under which (the pullback of) flux densities $H_3^B$ and $F_3$
  {\rm (from Prop. \ref{TheTypeIIBCocyclesSummarized}: the F1- and the D1-string couplings)} span the 2-dimensional vector representation
  \vspace{-1mm} 
  $$
    \begin{tikzcd}[
     row sep=-2pt, column sep=large
    ]
      \mathbb{R}^{
        1,9+(1+1)\,\vert\,32
      }
      \ar[
        rr,
        "{
          \exp\big(
            \tfrac{t}{2}
            \sigma_3
          \big)
        }",
        "{
          \sim
        }"{swap}
      ]
      &&
      \mathbb{R}^{
        1,9+(1+1)\,\vert\,32
      }
      \\
      e^{\tfrac{t}{2}\Gamma_9\Gamma_{\!\ten}}
      \psi
      &\longmapsfrom&
      \psi
      \\
      \mathrm{cos}(t)
      \, e^9
      +
      \mathrm{sin}(t)
      \, e^{\ten}
      &\longmapsfrom&
      e^9
      \\
      \mathrm{cos}(t)
      \, e^{\ten}
      -
      \mathrm{sin}(t)
      \, e^{9}
      &\longmapsfrom&
      e^{\ten}
      \\
      \\
      \mathrm{cos}(t)
      \,
      F_3
      +
      \mathrm{sin}(t)\,
      H_3^A
      &\longmapsfrom&
      F_3
      \\
      \mathrm{cos}(t)
      \,
      H_3^A
      -
      \mathrm{sin}(t)\,
      F_3
      &\longmapsfrom&
      H_3^A
    \end{tikzcd}
  $$
\end{proposition}
Noting that $\mathrm{SO}(2) \subset \mathrm{SL}(2,\mathbb{R})$ is the maximal compact subgroup of the S-duality group of IIB supergravity, hence the respective {\it local} U-duality group, this justifies the ``F-theory'' terminology (as concerned with the lift of T-duality on a single fiber in 10d to M-theory on a torus fiber, cf. \cite{Johnson97}); leaving open, however, the question if this local model space \eqref{CEOf1DoubledSuperSpace} supports a global super-field theory the way that its projection to $\mathbb{R}^{1,10\vert\mathbf{32}}$
\eqref{11dSuperFluxInIntroduction} supports 11D SuGra. In this vein, we now ask for yet further extension to bring out more of the structure expected in M-theory.

\medskip
\noindent
{\bf The fully brane-extended IIA super-algebra.}
Now we observe that the {\it fully doubled} super-spacetime Def. \ref{FullyDoubleSuperSpacetime} is further extended by what is known as the {\it fully extended} IIA super-spacetime (in the sense of \cite{vanHoltenVanProeyen82}):
\begin{definition}[\bf The fully brane-extended type IIA algebra]
\label{FullyExtebdedIIASpacetime}
The translational type IIA fully extended supersymmetry algebra 
$\IIA$
is (e.g. \cite[(2.2)]{Hull98}\cite[(2.16)]{CdAIPB00} \footnote{
  In \cite[(2.2)]{Hull98}\cite[(2.16)]{CdAIPB00} also the D0-brane charge with differential $(\overline{\psi} \,\Gamma_{\!\ten}\, \psi)$ --  is included in the extended IIA-algebra \eqref{CEFullyExtendedIIA}.
  But condensing D0-brane charge of course means opening up the 11th dimension, and hence
  here we regard this term instead as providing the further extension to the M-algebra, see Ex. \ref{BasicMAlGebraAsExtensionOfFullyExtendedIIASpacetime}.
}) given by \footnote{
  The signs in \eqref{CEFullyExtendedIIA} are a convention that is natural in view of the further extension by the M-algebra \eqref{CEOfBasicMAlgebra}, 
  where these signs 
  align with the Fierz identity \eqref{FierzDecomposition}, and makes the exceptional brane rotating symmetry in Prop. \ref{ManifestGL32} come out naturally.
}
\vspace{1mm} 
\begin{equation}
  \label{CEFullyExtendedIIA}
  \mathrm{CE}\big(
    \IIA
  \big)
  \;\;
  \simeq
  \;\;
  \FDGCA
  \left[\!
  \def\arraystretch{1.3}
  \begin{array}{c}
    (\psi^\alpha)_{\alpha=1}^{32}
    \\
    (e^a)_{a = 1}^9
    \\
    (\tilde e_a)_{a = 1}^9
    \\
    (e_{a_1 a_2} = e_{[a_1 a_2]})
      _{ a_i = 0 }^9
    \\
    (e_{a_1 \cdots a_4} = e_{[a_1 \cdots a_4]})
      _{ a_i = 0 }^9
    \\
    (e_{a_1 \cdots a_5} = e_{[a_1 \cdots a_5]})
      _{ a_i = 0 }^9
  \end{array}
  \!\right]
  \Big/
  \left(
  \def\arraystretch{1.3}
  \def\arraycolsep{2pt}
  \begin{array}{ccl}
    \mathrm{d}\, \psi 
      &=& 
    0
    \\
    \mathrm{d}\, e^a
    &=&
    +
    \big(\hspace{1pt}
      \overline{\psi}
      \,\Gamma^a\,
      \psi
    \big)
    \\
    \mathrm{d}\, \tilde e_a
    &=&
    -
    \big(\hspace{1pt}
      \overline{\psi}
      \,\Gamma_a\Gamma_{\!\ten}\,
      \psi
    \big)
    \\
    \mathrm{d}\, e_{a_1 a_2}
    &=&
    -
    \big(\hspace{1pt}
      \overline{\psi}
      \,\Gamma_{a_1 a_2}\,
      \psi
    \big)
    \\
    \mathrm{d}\, e_{a_1 \cdots a_4}
    &=&
    +
    \big(\hspace{1pt}
      \overline{\psi}
      \,\Gamma_{a_1 \cdots a_4}
      \Gamma_{\!\ten}\,
      \psi
    \big)
    \\
    \mathrm{d}\, e_{a_1 \cdots a_5}
    &=&
    +
    \big(\hspace{1pt}
      \overline{\psi}
      \,\Gamma_{a_1 \cdots a_5}\,
      \psi
    \big)
  \end{array}
  \!\!\right)
  .
\end{equation}
\end{definition}
 
\begin{remark}[\bf Extended IIA-algebra and brane charges]
\label{ExtendedIIAAlgebraAndBraneCharges}
$\,$

  \noindent {\bf (i)}  The bosonic body of the fully extended type IIA algebra \eqref{CEFullyExtendedIIA} may be suggestively re-arranged as
  \begin{equation}
    \label{BraneChargesInFullyExtededIIAAlgebra}
    \def\arraystretch{1.6}
    \def\arraycolsep{1pt}
    \begin{array}{ccccccccccc}
    \big(
      \IIA
    \big)_{\mathrm{bos}}
    &\;\simeq_{{}_{\mathbb{R}}}\;&
    \mathbb{R}^{1,9}
    &\oplus&
    (\mathbb{R}^{1,9})^\ast
    &\oplus&
    \wedge^2 (\mathbb{R}^{1,9})^\ast
    &\oplus&
    \wedge^4 (\mathbb{R}^{1,9})^\ast
    &\oplus&
    \wedge^5 (\mathbb{R}^{1,9})^\ast
    \\
    &\simeq_{{}_{\mathbb{R}}}&
    \underset{
      \mathclap{
        \hspace{4pt}
        \adjustbox{
          rotate=-50
        }{
          \rlap{
          \hspace{-14pt}
          \scalebox{.7}{
            \color{gray}
            space-time 
          }
          }
        }
      }
    }{
      \mathbb{R}^{1,9}
    }
    &\oplus&
    \underset{
      \mathclap{
        \hspace{4pt}
        \adjustbox{
          rotate=-50
        }{
          \rlap{
          \hspace{-27pt}
          \scalebox{.7}{
            \color{gray}
            \def\arraystretch{.8}
            \def\tabcolsep{0pt}
            \begin{tabular}{c}
              T-dual
              \\
              space-time /
              \\
              string charges
            \end{tabular}
          }
          }
        }
      }    
    }{
      (\mathbb{R}^{1,9})^\ast
    }
    &\oplus&
    \underset{
      \mathclap{
        \hspace{4pt}
        \adjustbox{
          rotate=-50
        }{
          \rlap{
          \hspace{-18pt}
          \scalebox{.7}{
            \color{gray}
            \def\arraystretch{.8}
            \def\tabcolsep{0pt}
            \begin{tabular}{c}
              D2-brane
              \\
              charges
            \end{tabular}
          }
          }
        }
      }        
    }{
      \wedge^2(\mathbb{R}^{9})^\ast
    }
    \oplus
    \underset{
      \mathclap{
        \hspace{4pt}
        \adjustbox{
          rotate=-50
        }{
          \rlap{
          \hspace{-18pt}
          \scalebox{.7}{
            \color{gray}
            \def\arraystretch{.8}
            \def\tabcolsep{0pt}
            \begin{tabular}{c}
              D8-brane
              \\
              charges
            \end{tabular}
          }
          }
        }
      }            
    }{
      \wedge^8(\mathbb{R}^{9})
    }
    &\oplus&
    \underset{
      \mathclap{
        \hspace{4pt}
        \adjustbox{
          rotate=-50
        }{
          \rlap{
          \hspace{-18pt}
          \scalebox{.7}{
            \color{gray}
            \def\arraystretch{.8}
            \def\tabcolsep{0pt}
            \begin{tabular}{c}
              D4-brane
              \\
              charges
            \end{tabular}
          }
          }
        }
      }
    }{
      \wedge^4(\mathbb{R}^9)^\ast
    }
    \oplus
    \underset{
      \mathclap{
        \hspace{4pt}
        \adjustbox{
          rotate=-50
        }{
          \rlap{
          \hspace{-18pt}
          \scalebox{.7}{
            \color{gray}
            \def\arraystretch{.8}
            \def\tabcolsep{0pt}
            \begin{tabular}{c}
              D6-brane
              \\
              charges
            \end{tabular}
          }
          }
        }
      }    
    }{
      \wedge^6(\mathbb{R}^9)
    }
    &\oplus&
    \underset{
      \mathclap{
        \hspace{4pt}
        \adjustbox{
          rotate=-50
        }{
          \rlap{
          \hspace{-18pt}
          \scalebox{.7}{
            \color{gray}
            \def\arraystretch{.8}
            \def\tabcolsep{0pt}
            \begin{tabular}{c}
              NS5-brane
              \\
              charges
            \end{tabular}
          }
          }
        }
      }    
    }{
      \wedge^5(\mathbb{R}^{1,9})^\ast
    }    \end{array}
  \end{equation}
  \vspace{.6cm}

  \noindent
  where in the second line we Hodge-dualized all temporal components (following \cite[(2.12)]{Hull98}) by the rule
  $$
    \wedge^p(\mathbb{R}^{1,d})^\ast
    \;\simeq_{\mathbb{R}}\;
    \grayunderbrace{
    \wedge^p(\mathbb{R}^d)^\ast
    }{
      \mathclap{
        \scalebox{.7}{
          spatial
        }
      }
    }
    \,\oplus\,
    \grayunderbrace{
    \wedge^{1+d-p}( \mathbb{R}^d )
    }{
      \mathclap{
        \scalebox{.7}{
          \def\arraystretch{.9}
          \begin{tabular}{c}
            dualized
            \\
            temporal
          \end{tabular}
        }
      }
    }.
  $$

  \vspace{-2mm} 
\noindent {\bf (ii)}     Note how the string charges in \eqref{BraneChargesInFullyExtededIIAAlgebra} play a special role as compared to the (other) brane charges: They appear with their temporal component included and, as such, may equivalently be understood as, in fact, being the T-dual {\it spacetime dimensions}.
Thus, the fully extended IIA algebra
\eqref{TheFullyTDualSuperSpacetime}
is a toroidal extension (Def. \ref{CentralTorusExtension})
of the fully doubled super-spacetime 
\eqref{FiberProductForFullyDoubledSpacetime}
by the D-brane and NS5-brane charges:
\begin{equation}
  \label{FullYTDualSpacetimeAsExtensionOfFullyDoubledSpacetime}
  \begin{tikzcd}[row sep=-3pt,
   column sep=0pt
  ]
    \IIA
    \ar[
      rr,
      ->>
    ]
    &&
    \mathfrak{Dbl}
    \ar[
      rr,
      "{
        \scalebox{.7}{
          brane charges
        }
      }"
    ]
    &\phantom{-----}&
    b\mathbb{R}^{507}
    \\
    \psi
    &\longmapsfrom&
    \psi
    \\
    e^a
    &\longmapsfrom&
    e^a
    \\
    \widetilde e_a
    &\longmapsfrom&
    \widetilde e_a
    \mathrlap{\,.}
  \end{tikzcd}
\end{equation}
\end{remark}
This relation between extended IIA super-symmetry and T-duality correspondence super-space may not have previously been appreciated as such. To make it fully manifest:

\begin{remark}[\bf Fully extended IIA as fiber product of doubled super-space with brane charges]
\label{FullyExtendedIIAAsFiberProduct}
To make manifest the double role that the string charges $\tilde e_a$  play in \eqref{FullYTDualSpacetimeAsExtensionOfFullyDoubledSpacetime} --- on the one hand as a doubled copy of spacetime and on the other as part of the general brane charges --- consider the following super-Lie algebra of {\it pure brane charges} $\Brn$,
given by
\vspace{-1mm} 
\begin{equation}
  \label{PureBraneAlgebra}
  \mathrm{CE}\big(
  \Brn
  \big)
  \;\simeq\;
  \left(
  \def\arraystretch{1.4}
  \def\arraycolsep{2pt}
  \begin{array}{ccl}
    \mathrm{d}\, 
    \psi &=& 0
    \\
    \mathrm{d}\, 
    \tilde{e}_{a}
    &=&
    -
    \big(\hspace{1pt}
      \overline{\psi}
      \,\Gamma_{a}
      \Gamma_{\!\ten}\,
      \psi
    \big)
    \\
    \mathrm{d}\, 
    e_{a_1 a_2}
    &=&
    -
    \big(\hspace{1pt}
      \overline{\psi}
      \,\Gamma_{a_1 a_2}\,
      \psi
    \big)
    \\
    \mathrm{d}\, 
    e_{a_1 \cdots a_4}
    &=&
    +
    \big(\hspace{1pt}
      \overline{\psi}
      \,\Gamma_{a_1 \cdots a_4}
      \Gamma_{\!\ten}\,
      \psi
    \big)
    \\
    \mathrm{d}\, 
    e_{a_1 \cdots a_5}
    &=&
    +
    \big(\hspace{1pt}
      \overline{\psi}
      \,\Gamma_{a_1 \cdots a_5}\,
      \psi
    \big)
  \end{array}
  \!\!\right)
  ,
\end{equation}
hence the extension of the fully dual super-spacetime \eqref{DualSpacetime} (which may be regarded as consisting entirely of string charges)
by the remaining D/NS-brane charges \eqref{BraneChargesInFullyExtededIIAAlgebra}.

Then the fully extended IIA algebra $\IIA$
\eqref{CEFullyExtendedIIA}
is the fiber product 
\eqref{CommutingPairOfExtensions} 
over the fully T-dual super-spacetime \eqref{DualSpacetime}
of this pure brane charge algebra \eqref{PureBraneAlgebra} 
with the fully doubled super-spacetime \eqref{CEOfFullyDoubledSuperSpacetime}:
\begin{equation}
  \label{FiberProductOfDoubledSpacetimeWitPureBraneAlgebra}
  \begin{tikzcd}[
    row sep={18pt, between origins},
    column sep={75pt, between origins}
  ]
    &
    &
    \IIA
    \ar[
      dr,
      ->>
    ]
    \ar[
      dl,
      ->>
    ]
    \\
    &
    \mathfrak{Dbl}
    \ar[
      dr,
      ->>
    ]
    \ar[
      dl,
      ->>
    ]
    &&
    \Brn
    \ar[
      dl,
      ->>
    ]
    \\
    \mathbb{R}^{
      1,9\,\vert\,
      \mathbf{16}
      \oplus
      \overline{\mathbf{16}}
    }
    &
    &
    \widetilde{\mathbb{R}}{}^{
      1,9\,\vert\,
      \mathbf{16}
      \oplus
      \overline{\mathbf{16}}
    }
    \\[-9pt]
    \mathclap{
      \scalebox{.7}{
        \color{gray}
        IIA spacetime
      }
    }
    &&
    \mathclap{
      \scalebox{.7}{
        \color{gray}
        \def\arraystretch{.9}
        \begin{tabular}{c}
          IIA string charges
        \end{tabular}
      }
    }
    &
    \mathclap{
      \scalebox{.7}{
        \color{gray}
        \def\arraystretch{.9}
        \begin{tabular}{c}
          IIA brane charges
        \end{tabular}
      }
    }
  \end{tikzcd}
\end{equation}
\end{remark}

\smallskip

\noindent
{\bf The M-algebra.} Similar to the fully extended IIA super-symmetry algebra,
we have the full extension of the 11D supersymmetry algebra, which may be understood (\cite[(13)]{Townsend95}\cite[(1)]{Townsend98}) as incorporating charges $Z^{a_1 a_2}$ of M2-branes and $Z^{a_1 \cdots a_5}$ of M5-branes (the terminology {\it M-algebra} 
follows \cite{Sezgin97}\cite{BDPV05}\cite[(3.1)]{Bandos17}
\footnote{
  \cite{Sezgin97} uses the term ``M-algebra'' for a large further extension of \eqref{CEOfBasicMAlgebra} which includes the ``hidden algebra'' of \cite{DF82}\cite{AndrianopoliDAuria24}; whereas other authors like \cite{BDPV05} say ``M-algebra'' for just \eqref{CEOfBasicMAlgebra}. Here we disambiguate this situation by speaking of the ``basic'' M-algebra. More discussion of this point is in the companion article \cite{GSS24-MGroup}.
}):

\begin{definition}[\bf Basic M-algebra]
\label{BasicMAlgebra}
The {\it basic M-algebra} is the super-Lie algebra $\mathfrak{M}$ given by \footnote{
  The sign convention in \eqref{CEOfBasicMAlgebra} is natural in view of the Fierz identity \eqref{FierzDecomposition}, and makes the exceptional brane rotating symmetry in Prop. \ref{ManifestGL32} come out naturally.
}
\begin{equation}
  \label{CEOfBasicMAlgebra}
  \hspace{-2mm} 
  \mathrm{CE}\big(
    \mathfrak{M}
  \big)
  \simeq
  \;
  \FDGCA
  \left[\!\!
    \def\arraystretch{1.4}
    \begin{array}{c}
      (\psi^\alpha)_{\alpha=1}^{32}
      \\
      (e^a)_{a=0}^{10}
      \\
      \big(
        e_{a_1 a_2}
        =
        e_{[a_1 a_2]}
      \big)_{a_i=0}^{10}
      \\
      \big(
        e_{a_1 \cdots a_5}
        =
        e_{[a_1 \cdots a_5]}
      \big)_{a_i=0}^{10}
    \end{array}
  \!\!\right]
  \Big/
  \left(
  \def\arraystretch{1.3}
  \def\arraycolsep{1pt}
  \begin{array}{ccl}
    \mathrm{d}
    \,\psi &=& 0
    \\
    \mathrm{d}\, 
    e^a &=&
    +
    \big(\hspace{1pt}
      \overline{\psi}
      \,\Gamma^a\,
      \psi
    \big)
    \\
    \mathrm{d}\, 
    e_{a_1 a_2} 
      &=&
    -
    \big(\hspace{1pt}
      \overline{\psi}
      \,\Gamma_{a_1 a_2}\,
      \psi
    \big)
    \\
    \mathrm{d}\, 
    e_{a_1 \cdots a_5} 
    &=&
    +
    \big(\hspace{1pt}
      \overline{\psi}
      \,\Gamma_{a_1 \cdots a_5}\,
      \psi
    \big)
  \end{array}
  \right).
\end{equation}
\vspace{-.3cm}

\end{definition}

\begin{example}[\bf Basic M-Algebra as extension of fully extended type IIA super-spacetime]
\label{BasicMAlGebraAsExtensionOfFullyExtendedIIASpacetime}
The basic M-algebra 
\eqref{CEOfBasicMAlgebra}
is a central extension (Def. \ref{CentralExtension})
of the fully extended type IIA algebra \eqref{CEFullyExtendedIIA} by (the pullback of) the same 2-cocycle 
\eqref{11dSuperspaceAsExtensionOfIIASuperspace}
that classifies the M/IIA extension:
\vspace{-2mm} 
\begin{equation}
  \label{MAlgebraExtensionOfFullyExtendedIIA}
  \adjustbox{
    raise=-2.5cm
  }{
  \begin{tikzpicture}
  \node at (0,0) {
  \begin{tikzcd}[
    row sep=-1pt,
   column sep=5pt
  ]
    \mathfrak{M}
    \ar[
      rr,
      ->>
    ]
    &&
    \IIA
    \ar[
      rr,
      "{
        (\overline{\psi}
        \Gamma^{\ten} \psi)
      }"
    ]
    &\phantom{-----}&
    b\mathbb{R}
    \\
    \psi
    &\longmapsfrom&
    \psi
    \\
    e^a
    &\longmapsfrom&
    e^a
    \\
    e_{a \, \ten}
    &\longmapsfrom&
    \tilde e_a
    \\
    e_{a_1 a_2}
    &\longmapsfrom&
    e_{a_1 a_2}
    \\
    e_{\ten\, a_1 \cdots a_4}
    &\longmapsfrom&
    e_{a_1 \cdots a_4}
    \\
    e_{a_1 \cdots a_5}
    &\longmapsfrom&
    e_{a_1 \cdots a_5}
    \mathrlap{\,.}
  \end{tikzcd}
  };
  \draw[
    draw=olive,
    fill=olive,
    draw opacity=.3,
    fill opacity=.3,
  ]
    (-3.4,-.4) rectangle
    (.38,.13);
    \node at
     (1.45,-.16)  {
       \scalebox{.7}{
         \color{gray}
         \def\arraystretch{.9}
         \begin{tabular}{c}
           string charges /
           \\
           doubled spacetime 
         \end{tabular}
       }
     };
    \node at
     (-4.35,-.16)  {
       \scalebox{.7}{
         \color{gray}
         \def\arraystretch{.9}
         \begin{tabular}{c}
           wrapped M2- 
           \\
           brane charges
         \end{tabular}
       }
     };
  \end{tikzpicture}
  }
\end{equation}
which means that the M-algebra is equivalently the fiber product \eqref{CommutingPairOfExtensions}
of the fully extended IIA spacetime with the 11D super-space over the 10D IIA spacetime:
\vspace{-2mm}
$$
  \begin{tikzcd}[
    row sep=0pt
  ]
    &
    \mathfrak{M}
    \ar[
      dr,
      ->>
    ]
    \ar[
      dl,
      ->>
    ]
    \\
    \mathbb{R}^{1,10\,\vert\,\mathbf{32}}
    \ar[
      dr,
      ->>
    ]
    &&
    \IIA
    \ar[
      dl,
      ->>
    ]
    \\
    &
    \mathbb{R}^{
      1,9\,\vert\,
      \mathbf{16}
      \oplus
      \overline{\mathbf{16}}
    }
  \end{tikzcd}
$$
The assignment \eqref{MAlgebraExtensionOfFullyExtendedIIA} reflects the isomorphism
$$
  \def\arraystretch{1.5}
  \def\arraycolsep{1pt}
  \begin{array}{ccccccccc}
    \mathfrak{M}_{\mathrm{bos}}
    &\;\simeq_{{}_{\mathbb{R}}}\;&
    \mathbb{R}^{1,10}
    &\oplus&
    \wedge^2(\mathbb{R}^{1,10})^\ast
    &\oplus&
    \wedge^5(\mathbb{R}^{1,10})^\ast
    \\
    &\simeq_{{}_{\mathbb{R}}}&
    \mathbb{R}
    \oplus
    \mathbb{R}^{1,9}
    &\oplus&
    (\mathbb{R}^{1,9})^\ast
    \oplus
    \wedge^2(\mathbb{R}^{1,9})^\ast
    &\oplus&
    \wedge^4(\mathbb{R}^{1,9})^\ast
    \oplus
    \wedge^5(\mathbb{R}^{1,9})^\ast
    \\[0pt]
    &\;\simeq_{{}_{\mathbb{R}}}\;&
    \hspace{-34pt}
    \mathrlap{
    \mathbb{R}
    \oplus
    \big(
      \IIA
    \big)_{\mathrm{bos}}
    \mathrlap{\,,}
    }
    
  \end{array}
$$
where in the second line we have decomposed into components that are parallel resp. orthogonal to the $\ten$-coordinate axis, by the rule
$$
  \wedge^p\big(
    \mathbb{R}^{1,d}
  \big)^\ast
  \;
  \simeq_{{}_{\mathbb{R}}}
  \;
  \wedge^{p-1}\big(
    \mathbb{R}^{1,d-1}
  \big)^\ast
  \oplus
  \wedge^{p}\big(
    \mathbb{R}^{1,d-1}
  \big)^\ast
  \mathrlap{\,.}
$$
Of particular importance for the picture of \eqref{SummaryDiagram} is the boxed assignment in \eqref{MAlgebraExtensionOfFullyExtendedIIA}, which identifies the string charges (alternatively: doubled spacetime directions) in the fully extended IIA algebra with the charges of M2-branes wrapping the $\ten$-axis.
The assignments below the box in \eqref{MAlgebraExtensionOfFullyExtendedIIA} are the remaining D- and NS5-brane charges. 
\end{example}
In order to make this distinction manifest and in view of Rem. \ref{FullyExtendedIIAAsFiberProduct}, consider:

\begin{remark}[\bf M-algebra as extended doubled super-space]
With the $\mathfrak{M}$-algebra being a spacetime-extension 
\eqref{11dSuperspaceAsExtensionOfIIASuperspace}
of the $\IIA$-algebra, it must, by \eqref{FiberProductOfDoubledSpacetimeWitPureBraneAlgebra} also be an extension by brane charges of a spacetime extension of the fully doubled superspacetime \eqref{CEOfFullyDoubledSuperSpacetime}.
To make this manifest ---  recalling from \eqref{FTheorySpacetimesAsFiberProduct} that the analogous spacetime extension of the partially doubled super-spacetime is the F-theory super-spacetime --- we shall write $\mathfrak{F}$ for the super-Lie algebra given by
\begin{equation}
  \label{frakF}
  \mathrm{CE}\big(
    \mathfrak{F}
  \big)
  \;=\;
  \FDGCA
  \left[\!\!\!
  \def\arraystretch{1.5}
  \begin{array}{c}
    (\psi^\alpha)_{\alpha=1}^{32}
    \\
    \big(
      e^a
    \big)_{a = 0}^{
      \mathcolor{purple}{10}
    }
    \\
    \big(
      \tilde e_a
    \big)_{a = 0}^{
      \mathcolor{purple}9
    }
  \end{array}
  \!\!\right]
  \Big/
  \left(\!
  \def\arraystretch{1.3}
   \def\arraycolsep{1pt}
  \begin{array}{ccl}
    \mathrm{d}
    \,
    \psi &=& 0
    \\
    \mathrm{d}\,
    e^a &=&
    \big(\hspace{1pt}
      \overline{\psi}
      \,\Gamma^a\,
      \psi
    \big)
    \\
    \mathrm{d}\,
    \tilde e_a &=& -
    \big(\hspace{1pt}
      \overline{\psi}
      \,\Gamma_a\Gamma_{\!\ten}\,
      \psi
    \big)
  \end{array}
 \! \right)
\end{equation}
and thus extending the fully doubled super-spacetime in the same way that 11D super-spacetime extends 10D type IIA:
\begin{equation}
  \label{FasExtensionOfDbl}
  \begin{tikzcd}[row sep=-3pt, column sep=0pt]
    \mathfrak{F}
    \ar[
      rr,
      ->>
    ]
    &&
    \mathfrak{Dbl}
    \ar[
      rr,
      "{
        (\overline{\psi}\Gamma_{\!\ten}\psi)
      }"
    ]
    &\phantom{----}&
    b\mathbb{R}
    \\
    e^a
    &\longmapsfrom&
    e^a
    \\
    \tilde e_a
    &\longmapsfrom&
    \tilde e_a
    \mathrlap{\,.}
  \end{tikzcd}
\end{equation}
Then the M-algebra is the fiber product \eqref{CommutingPairOfExtensions} over 
the fully doubled super-spacetime of this full $\mathfrak{F}$-spacetime with the fully extended $\IIA$-spacetime:
\begin{equation}
  \label{MAlgebraAsFiberProductOverDoubledSuperspace}
  \begin{tikzcd}[
    row sep={15pt, between origins},
    column sep={70pt, between origins}
  ]
  &
  \mathfrak{M}
  \ar[
    dl,
    ->>,
    "{
      p^{\mathrm{Brn}}
    }"{swap, pos=.3}
  ]
  \ar[
    dr,
    ->>,
    "{
      p^M
    }"{pos=.4}
  ]
  \\
  \mathfrak{F}
  \ar[
    dr,
    ->>,
    "{
      p^M
    }"{swap}
  ]
  &&
  \IIA
  \ar[
    dl,
    ->>,
    "{
      p^{\mathrm{Brn}}
    }"
  ]
  \\
  &
  \mathfrak{Dbl}
  \end{tikzcd}
\end{equation}
 This perspective reveals a  relation of the M-algebra to T-duality, cf. Prop. \ref{ReductionOfPoincare3Form}.
\end{remark}

\noindent

\begin{remark}[\bf Lift of the Poincar{\'e} super 2-form to the M-algebra]
\label{LiftOfPoincare3Form}
$\,$

\noindent {\bf (i)} On the basic M-algebra \eqref{CEOfBasicMAlgebra}, we consider the element
\begin{equation}
  \label{Poincare3Form}
  P_3
  \;=\;
  \tfrac{1}{2}
  e^{a_1} e_{a_1 a_2}  e^{a_2}
  \;\;\;
  \in
  \;
  \mathrm{CE}(\mathfrak{M})
\end{equation}
which we shall call the {\it Poincar{\'e} super 3-form}, since it is an M-theoretic lift of the Poincar{\'e} super 2-form from \eqref{Poincare2FormForSingleFiber}.

\noindent {\bf (ii)}  It is immediate that the 
dimensional reduction 
  the Poincar{\'e} 3-form \eqref{Poincare3Form} 
  on the M-algebra,
  by fiber integration \eqref{FiberIntegration} along the M-theory direction,
  reproduces the full Poincar{\'e} 2-form 
  \eqref{Poincare3Form} on the doubled super-space:
  \begin{equation}
    \label{ReductionOfPoincare3Form}
    p^M_\ast
    \,
    P_3
    \;=\;
    P_2
    \,.
  \end{equation}

\noindent {\bf (iii)}  Given that the Poincar{\'e} 2-form $P_2$ in $\mathfrak{Dbl}$ entirely controls rational-topological T-duality in 10D, this exhibits $P_3$ on $\mathfrak{M}$ as analogously reflecting the T-duality phenomenon in M-theory.
\end{remark}

We summarize the resulting picture in \S\ref{ConclusionAndOutlook}, see \eqref{SummaryDiagram}.

\begin{remark}[\bf M-algebra as fiber product of spacetime with brane charges]
Another, equivalent, perspective is that the M-algebra is the fiber product of 11D superspace with the M-brane charges over the super-point. To make this manifest, observe that the pure brane algebra \eqref{PureBraneAlgebra}, which we introduced from the point of view of IIA branes,
  has the following  isomorphic CE-algebra
  $$
  \begin{tikzcd}[
    ampersand replacement=\&,
    column sep=2pt
  ]
  \&[-10pt]
  \FDGCA
  \left[\!
  \def\arraystretch{1.4}
  \def\arraycolsep{2pt}
  \begin{array}{ccl}
    (\psi^\alpha)_{\alpha=1}^{32}
    \\
    (\tilde{e}_a)_{a=0}^{\mathcolor{purple}{9}}
    \\
    (e_{a_1 a_2} = e_{[a_1 a_2]})_{a_i=0}^{\mathcolor{purple}{9}}
    \\
    (e_{a_1 \cdots a_4} = e_{[a_1\cdots a_4]})_{a_i=0}^{\mathcolor{purple}{9}}
    \\
    (e_{a_1 \cdots a_5} = e_{[a_1\cdots a_5]})_{a_i=0}^{\mathcolor{purple}{9}}
  \end{array}
  \right]
  \Big/
  \left(
  \def\arraystretch{1.4}
  \def\arraycolsep{2pt}
  \begin{array}{ccl}
    \mathrm{d}\, 
    \psi &=& 0
    \\
    \mathrm{d}\, 
    \tilde{e}_{a}
    &=&
    -
    \big(\hspace{1pt}
      \overline{\psi}
      \,\Gamma_{a}
      \Gamma_{\!\ten}\,
      \psi
    \big)
    \\
    \mathrm{d}\, 
    e_{a_1 a_2}
    &=&
    -
    \big(\hspace{1pt}
      \overline{\psi}
      \,\Gamma_{a_1 a_2}\,
      \psi
    \big)
    \\
    \mathrm{d}\, 
    e_{a_1 \cdots a_4}
    &=&
    +
    \big(\hspace{1pt}
      \overline{\psi}
      \,\Gamma_{a_1 \cdots a_4}
      \Gamma_{\!\ten}\,
      \psi
    \big)
    \\
    \mathrm{d}\, 
    e_{a_1 \cdots a_5}
    &=&
    \big(\hspace{1pt}
      \overline{\psi}
      \,\Gamma_{a_1 \cdots a_5}\,
      \psi
    \big)
  \end{array}
  \!\!\right)
  \ar[
    dd,
    "{ \sim }"{sloped}
  ]
  \&
  \psi
  \&
  \tilde{e}_a
  \&
  e_{a_1 a_2}
  \&
  e_{a_1 \cdots a_4}
  \&
  e_{a_1 \cdots a_5}
  \\
  \mathrm{CE}\big(
    \Brn
  \big)
  \ar[
    ur,
    "{ \sim }"{sloped}
  ]
  \ar[
    dr,
    shorten >=20,
    "{ \sim }"{sloped, swap, pos=.2}
  ]
  \\
  \&
  \FDGCA
  \left[
  \def\arraystretch{1.4}
  \def\arraycolsep{2pt}
  \begin{array}{ccl}
    (\psi^\alpha)_{\alpha=1}^{32}
    \\
    (e_{a_1 a_2} = e_{[a_1 a_2]})_{a_i=0}^{\mathcolor{purple}{10}}
    \\
    (e_{a_1 \cdots a_5} = e_{[a_1 \cdots a_5]})_{a_i=0}^{\mathcolor{purple}{10}}
  \end{array}
  \right]
  \Big/
  \left(
  \def\arraystretch{1.4}
  \def\arraycolsep{2pt}
  \begin{array}{ccl}
    \mathrm{d}\, \psi
    &=&
    0
    \\
    \mathrm{d}\,
    e_{a_1 a_2}
    &=&
    -\big(\hspace{1pt}
      \overline{\psi}
      \,\Gamma_{a_1 a_2}\,
      \psi
    \big)
    \\
    \mathrm{d}\,
    e_{a_1 \cdots a_5}
    &=&
    +\big(\hspace{1pt}
      \overline{\psi}
      \,\Gamma_{a_1 \cdots a_5}\,
      \psi
    \big)
  \end{array}
  \right)
  \&
  \psi
  \ar[
    uu,
    |->,
    shorten=15pt
  ]
  \&
  e_{a \, \mathcolor{purple}{\ten}}
  \ar[
    uu,
    |->,
    shorten=15pt
  ]
  \&
  \underset{
    \mathclap{
      \adjustbox{
        raise=-10pt,
        scale=.7
      }{$
        \;\;\;\;\;\;\;\;\;\;\;\;
        (a_i \leq 9)
      $}
    }
  }{
  e_{a_1 a_2}
  }
  \ar[
    uu,
    |->,
    shorten=15pt
  ]
  \&
  e_{\mathcolor{purple}{\ten}\, a_1 \cdots a_4}
  \ar[
    uu,
    |->,
    shorten=15pt
  ]
  \&
  e_{a_1 \cdots a_5}
  \ar[
    uu,
    |->,
    shorten=15pt
  ]
  \end{tikzcd}
  $$
  witnessing the duality between IIA-branes and M-branes. But this makes it manifest that the M-algebra is also the fiber product \eqref{CommutingPairOfExtensions} over the super-point
  of 11D super-spacetime with the brane charges:
  \begin{equation}
  \label{MAlgebraAsFiberProductOf11DSpacetimeWithBraneCharges}
  \begin{tikzcd}[
    row sep={15pt, between origins},
    column sep={70pt, between origins}
  ]
    & \mathfrak{M}
    \ar[dl, ->>]
    \ar[dr, ->>]
    \\
    \mathbb{R}^{
     1,10\,\vert\,
     \mathbf{32}
    }
    \ar[
      dr,->>
    ]
    &&
    \Brn
    \ar[dl,->>]
    \\
    &
    \mathbb{R}^{0 \,\vert\, \mathbf{32}}
  \end{tikzcd}
  \end{equation}
\end{remark}

\noindent
{\bf U-duality realized on the M-algebra.} 
With the M-algebra thus emerging as an M-theoretic extension of the fully doubled super-spacetime on which T-duality becomes manifest, we just note that it carries a canonical action of the expected {\it local} hidden U-duality symmetry of M-theory, namely of the ``maximal compact subalgebra'' of $\mathfrak{e}_{11}$:

\begin{proposition}[{\bf Manifestly $\mathrm{GL}(32;\mathbb{R})$-equivariant incarnation of basic M-algebra} {\cite[\S 4]{West03}\cite[\S 5]{BaerwaldWest00}}]
  \label{ManifestGL32}
  Unifying all the even generators of the M-algebra \eqref{MAlgebraExtensionOfFullyExtendedIIA}
  into a symmetric bispinorial form like this
  \begin{equation}
    \label{TheBispinorCEElement}
    e^{\alpha \beta}
    \;\;
    :=
    \;\;
    \tfrac{1}{32}
    \big(
      e^a
      \,
      \Gamma
        _a 
        ^{\alpha \beta}
      +
      \tfrac{1}{2}
      e^{a_1 a_2}
      \,
      \Gamma
        _{a_1 a_2}
        ^{\alpha \beta}
      +
      \tfrac{1}{5!}
      e^{a_1 \cdots a_5}
      \,
      \Gamma
        _{a_1 \cdots a_5}
        ^{\alpha \beta}
    \big)
  \end{equation}
  the CE-differential 
  acquires equivalently the compact form
  \begin{equation}
    \label{ManifestEquivariantDifferentialOnBasicMAlgebra}
    \def\arraystretch{1.2}
    \begin{array}{lcl}
      \mathrm{d}
      \,
      \psi^\alpha
      &=&
      0
      \\
      \mathrm{d}
      \,
      e
        ^{\alpha \beta}
      &=&
      \psi^\alpha
      \,
      \psi^\beta,
    \end{array}
  \end{equation}
  which makes manifest that any $g\in \mathrm{GL}(32, \mathbb{R})$ acts via super-Lie algebra automorphisms of the M-algebra 
  \begin{equation}
    \label{ManifestGL32Equivariance}
    \begin{tikzcd}[row sep=-3pt,
      column sep=0pt
    ]
      \mathllap{
        g
        \;:\;
      }
      \mathrm{CE}\big(
        \mathfrak{M}
      \big)
      \ar[rr]
      &&
      \mathrm{CE}\big(
        \mathfrak{M}
      \big)
      \\
      \psi^\alpha
      &\longmapsto&
      g^{\alpha}_{\alpha'}
      \,
      \psi^{\alpha'}
      \\
      e^{\alpha \beta}
      &\longmapsto&
      g^{\alpha}_{\alpha'}
      \,
      g^{\beta}_{\beta'}
      \,
      e^{\alpha'\beta'}.
    \end{tikzcd}
  \end{equation}
\end{proposition}
\begin{proof}
 First, to see that the transformation \eqref{TheBispinorCEElement} is invertible,
  the trace-property
  \eqref{VanishingTraceOfCliffordElements}
  allows to recover:
  \begin{equation}
    \label{OriginalBosonicGeneratorsFromManifestlyGL32EquivariantBasis}
    \def\arraystretch{1.4}
    \begin{array}{lcl}
      e^a
      &=&
      \phantom{+}
      \Gamma
        ^a
        _{\alpha \beta}
      \,
      e^{\alpha \beta}
      \\
      e^{a_1 a_2}
      &=&
      -
      \Gamma
        ^{a_1 a_2}
        _{\alpha \beta}
      \,
      e^{\alpha \beta}
      \\
      e^{a_1 \cdots a_5}
      &=&
      \phantom{+}
      \Gamma
        ^{a_1 \cdots a_5}
        _{\alpha \beta}
      \,
      e^{\alpha \beta}.
    \end{array}
  \end{equation}
  Finally, the differential is as claimed due to the Fierz expansion formula \eqref{FierzDecomposition}:
  $$
    \def\arraystretch{1.5}
    \begin{array}{ccll}
      \mathrm{d}
      \;
      e^{\alpha \beta}
     &=&
    \tfrac{1}{32}
    \Big(
      \Gamma
        _a 
        ^{\alpha\beta}
      \,
      \big(\hspace{1pt}
        \overline{\psi}
        \,\Gamma^a\,
        \psi
      \big)
      -
      \tfrac{1}{2}
      \Gamma
        _{a_1 a_2}
        ^{\alpha\beta}
      \big(\hspace{1pt}
        \overline{\psi}
        \,\Gamma^{a_1 a_2}
        \,
        \psi
      \big)
      +
      \tfrac{1}{5!}
      \Gamma
        _{a_1 \cdots a_5}
        ^{\alpha\beta}
      \big(\hspace{1pt}
        \overline{\psi}
        \,\Gamma^{a_1 \cdots a_5}\,
        \psi
      \big)
    \Big)
    &
    \proofstep{
     by
     \eqref{TheBispinorCEElement}
     \&
     \eqref{CEOfBasicMAlgebra}
    }
    \\
    &=&
    \psi^\alpha
    \,
    \psi^{\beta}
    &
    \proofstep{
      by 
      \eqref{FierzDecomposition}.
    }
    \end{array}
  $$

  \vspace{-4mm} 
\end{proof}

\begin{example}[\bf Brane rotating symmetry]
  \label{BraneRotatingSymmetry}
 On the original bosonic generators \eqref{CEOfBasicMAlgebra} -- the spacetime momentum $e^a$, the M2-brane charges $e_{a_1 a_2}$ and the M5-brane charges $e_{a_1 \cdots a_5}$ --- the 
 $\mathrm{GL}(32;\mathbb{R})$ symmetry of \eqref{ManifestEquivariantDifferentialOnBasicMAlgebra} 
 acts by mixing them all among each other, e.g.
 $$
   \def\arraystretch{1.4}
   \begin{array}{lcll}
     e^a
     &=&
     \Gamma^a_{\alpha \beta}
     e^{\alpha \beta}
     &
     \proofstep{
       by \eqref{OriginalBosonicGeneratorsFromManifestlyGL32EquivariantBasis}
     }
     \\
     &\overset{g}{\mapsto}&
     \Gamma^a_{\alpha \beta}
     g^\alpha_{\alpha'}
     g^{\beta}_{\beta'}
     e^{\alpha' \beta'}
     &
     \proofstep{
       by
       \eqref{ManifestGL32Equivariance}
     }
     \\
     &=&
     \big(
     \tfrac{1}{32}
     \Gamma^a_{\alpha \beta}
     g^\alpha_{\alpha'}
     g^{\beta}_{\beta'}
     \Gamma_b^{\alpha'\beta'}
     \big)
     e^b
     \,+\,
     \big(
     \tfrac{1}{64}
     \Gamma^a_{\alpha \beta}
     g^\alpha_{\alpha'}
     g^{\beta}_{\beta'}
     \Gamma_{b_1 b_2}^{\alpha'\beta'}
     \big)
     e^{b_1 b_2}
     \,+\,
     \big(
     \tfrac{1}{5!\cdot 32}
     \Gamma^a_{\alpha \beta}
     g^\alpha_{\alpha'}
     g^{\beta}_{\beta'}
     \Gamma_{b_1 \cdots b_5}
       ^{\alpha'\beta'}
     \big)
     e^{b_1 \cdots b_5}
     &
     \proofstep{
       by
       \eqref{TheBispinorCEElement},
     }
   \end{array}
 $$

 \vspace{1mm} 
\noindent  as befits a U-duality symmetry. For this reason, the authors \cite{BaerwaldWest00} speak of a ``brane rotating symmetry''.
\end{example}

\begin{remark}[\bf Relation to $\mathfrak{e}_{11}$-duality]
  This enhanced equivariance 
  \eqref{ManifestGL32Equivariance}
  of the M-algebra, which makes the basic super Lie bracket a morphism of $\mathfrak{sl}_{32}$-representations
  $\mathbf{32} \otimes_{\mathrm{sym}} \mathbf{32} \;\simeq\; \mathbf{526}$, will have to be understood as the effective part of the action of the ``maximal compact'' subalgebra of $\mathfrak{e}_{11}$, according to \cite[p. 42]{BossardKleinschmidtSezgin19}. 
\end{remark}

\newpage 
\section{Conclusion \& Outlook}
\label{ConclusionAndOutlook}

In view of the fact that on-shell 11D supergravity is entirely determined by the structure of its avatar super-flux densities on the typical super-tangent space \cite[Thm. 3.1]{GSS24-SuGra}, which there are naturally understood as an $\mathfrak{l}S^4$-valued super-$L_\infty$ algebra cocycle \cite{FSS17}, we have given a super-$L_\infty$ algebraic re-analysis of T-duality on the analogous avatar super-flux densities of 10D type II supergravity, first for type A/B dualization along one spacetime direction (streamlining the previous result of \cite{FSS18-TDualityA}) and then for type A/$\widetilde{\mathrm{A}}$-duality along {\it all} spacetime directions, via reduction/oxidation all the way to/from the super-point $\mathbb{R}^{0\vert \mathbf{32}}$.

\smallskip

\noindent
The result of this analysis is that:

\smallskip
\begin{itemize}
\item[\bf (i)] 
this super-space T-duality is entirely controlled (Prop. \ref{ToroidalTDuality/FourierMukaiIsomorphism}) by a super Fourier-Mukai transform with integral kernel a Poincar{\'e} super 2-form $P_2$ \eqref{FullPoincare2Form}
on fully doubled 10D superspace $\mathfrak{Dbl}$ 
\eqref{CEOfFullyDoubledSuperSpacetime}
serving as a correspondence between type IIA super-spacetime and its full $\mathrm{II}\widetilde{\mathrm{A}}$ dual,

\item[\bf (ii)] extending this super-correspondence by 
  (i) the 11th spacetime dimension and  (ii) the remaining brane charges,
yields a whole network \eqref{SummaryDiagram}
of further correspondences that bring in increasingly more M-theoretic structure into type II T-duality,

\item[\bf (iii)] the tip of this network is formed by the basic M-algebra $\mathfrak{M}$ \eqref{CEOfBasicMAlgebra}
carrying a super-invariant 3-form $P_3$
\eqref{Poincare3Form}
whose dimensional reduction along the 11th spacetime dimension coincides 
\eqref{ReductionOfPoincare3Form}
with the above Poincar{\'e} 2-form $P_2$.
\end{itemize}

\smallskip

This suggests that the M-algebra carrying $P_3$ is in some sense the M-theoretic lift of the fully doubled super-space carrying the Poincar{\'e} form $P_2$ that controls 10D T-duality
---
which makes sense in view of the fact \cite{GSS24-E11}
that the M-algebra may also be understood as the ``super-exceptional tangent space'' of 11D SuGra, on which at least some of the hidden M-theoretic U-dualities are meant to become manifest in a geometric way.

\smallskip

Indeed, by way of outlook we observe now that the ``decomposed'' 3-form known to exist on the further ``hidden'' extension of the M-algebra \cite{DF82}\cite{BDIPV04}\cite{GSS24-MGroup} is just this $P_3$ --  up to M-brane corrections. 

(For definiteness we will now be referring to the respective equation numbers in \cite{GSS24-MGroup}.)

\smallskip

\medskip

\noindent
{\bf The hidden M-algebra and its ``decomposed'' 3-form.}
Namely, a further ``hidden'' extension $\widehat{\mathfrak{M}}$ of the M-algebra by a further odd spinor-valued generator --- and hence leaving its bosonic body un-affected! --- is
\vspace{2mm} 
\begin{equation}
  \label{TheHiddenMAlgebra}
  \mathrm{CE}\big(\,
    \widehat{\mathfrak{M}}
  \, \big)
  \;\simeq\;
  \mathrm{CE}(\mathfrak{M})
  \Big[
    \grayunderbrace{
    (\DFSpinor^\alpha)
     _{\alpha=1}^{32}
    }{
      \mathrm{deg}
      =
      (1,\mathrm{odd})
    }
  \Big]
  \big/
  \Big(
    \tfrac{1}{2}
    \mathrm{d}\, \DFSpinor
    \,=\,
      (1+s)
      \Gamma_a\psi
      \, e^a
      \,+\,
      \Gamma^{a_1 a_2}\psi
      \, e_{a_1 a_2}
      \,+\,
      \tfrac{6+s}{6!}
      \,
      \Gamma^{a_1 \cdots a_5}
      \, e_{a_1 \cdots a_5}
   \Big)
\end{equation}
for any $s \in \mathbb{R}$, with the remarkable property that (for $s \neq 0$) it carries a super-invariant 3-form $\widehat P_3$
\cite[(34)]{GSS24-MGroup}
which is a coboundary for (the pullback of) the avatar C-field super-flux  $G_4$ \eqref{11dSuperFluxInIntroduction}:
\begin{equation}
   \label{TheDecomposed3Form}
    \widehat{P}_3
    \;\in\;
    \mathrm{CE}(\widehat{\mathfrak{M}})^{\mathrm{Spin}(1,10)}
    \,,
    \;\;\;\;
    \mathrm{d}\, 
    \widehat{P}_3
    \;=\;
    G_4 \,\defneq\,
    \tfrac{1}{2}\big(\,
      \overline{\psi}
      \,\Gamma_{a_1 a_2}\,
      \psi
    \big)
    e^{a_1}
    e^{a_2}
\end{equation}

\medskip

\noindent

Incidentally, this means 
\cite[Lem. 3.7]{FSS20Exc}
that the hidden M-algebra serves as an {\it atlas} (in the sense of stacks) for the M2-brane-extended super-spacetime (Ex. \ref{TheM2braneExtensionOf11D}), in that we have a homomorphism \eqref{LieHomomorphism} of super-$L_\infty$-algebras
\begin{equation}
  \label{HiddenMAlgebraAsAtlas}
  \begin{tikzcd}[sep=0pt]
    \widehat{\mathfrak{M}}
    \ar[
      rr,
      ->>
    ]
    &&
    \mathfrak{m}2\mathfrak{brane}
    \\
    \psi
    &\longmapsfrom&
    \psi
    \\
    e^a
    &\longmapsfrom&
    e^a
    \\
    \widehat{P}_3
    &\longmapsfrom&
    c_3
  \end{tikzcd}
\end{equation}
which is surjective in degree=0 and whose domain is, by construction, an ordinary super-Lie algebra (instead of a higher super-$L_\infty$ algebra).
Under this atlas, the rational higher T-duality on $\mathfrak{m}2\mathfrak{brane}$ from \S\ref{HigherTduality} transfers to a corresponding higher duality on $\widehat{\mathfrak{M}}$ (see \cite[Prop. 4.17]{FSS20-HigherT}).

\medskip

\noindent
{\bf The hidden Poincar{\'e} 3-form.}
Previously, the meaning or further preferred specification of the parameter $s$ in \eqref{TheHiddenMAlgebra} had remained mysterious, 
but now we may observe \cite[\S 2.2.23]{GSS24-MGroup} that:

\smallskip

\noindent
{\bf (i)}
There are exactly two values of $s$ for which $\widehat{P}_3$ \eqref{TheDecomposed3Form} equals (the pullback of) the ``M-theoretic Poincar{\'e} form'' $P_3$ \eqref{Poincare3Form} up to M-brane corrections 
$$
  {p^{\mathrm{Brn}}_{\mathrm{bas}}}_\ast
  \, 
  \widehat{P}_3
  \;=\;
  P_3
  \hspace{1cm}
  \Leftrightarrow
  \hspace{1cm}
  s = - 1
  \;\;\;\;\;
  \mbox{or}
  \;\;\;\;\;
  s = 3/2
  \,.
$$

\newpage 
\noindent
{\bf (ii)} The case $s = -1$ is moreover special because it is exactly here that the hidden extension \eqref{TheHiddenMAlgebra} becomes independent of the spacetime coframe and hence induces a corresponding fermionic extension $\widehat{\IIA}$ of the IIA-algebra (as well as $\widehat{\mathfrak{Brn}}$ of the pure brane-charge algebra), such that the pullback square for the M-algebra 
\eqref{SummaryDiagram}
has a further ``hidden'' lift as follows:
\vspace{0mm} 
\begin{equation}
  \begin{tikzcd}[
    row sep=20pt, 
    column sep=37pt
  ]
    &&
    \mathfrak{M}
    \ar[
      dddd,
      phantom,
      "{
        \adjustbox{
          raise=3pt,
          scale=.7
        }{
          \color{darkblue}
          \bf
          \def\arraystretch{.8}
          \begin{tabular}{c}
            {
              \color{gray}
              \eqref{CEOfBasicMAlgebra}
            }
            \\
            Basic
            \\
            M-algebra
        \end{tabular}
        }
      }"{pos=-.33}
    ]
    \ar[
      dddd,
      phantom,
      "{
        \begin{array}{c}
          \hspace{-9pt}
          \frac{1}{2}
          e^{a}  
          e_{a b}   
          e^b
          \mathrlap{
            \scalebox{.8}{$
              \, =: \!P_3
            $}
          }
          \\
          \rotatebox[origin=c]{-90}{
            $\longmapsto$
          }
          \\
          e_a \, \tilde e^a
          \mathrlap{
            \scalebox{.8}{$
            \, =: \! P_2
            $}
          }
        \end{array}
      }",
      "{
        {p^{\mathrm{Brn}}_{\mathrm{bas}}}\, 
        p^{\mathrm{M}}_\ast
      }"{xshift=18pt, scale=.65},
      "{
        \scalebox{.67}{
          \color{olive}
          \bf
          \def\arraystretch{.85}
          \begin{tabular}{c}
            Poincar{\'e}
            \\
            super-forms
            \\
            \color{gray}
            \eqref{ReductionOfPoincare3Form}
          \end{tabular}
        }
      }"{xshift=-27pt, pos=.58}
    ]
    \ar[
      ddll,
      ->>,
      "{
        p^{\mathrm{Brn}}
      }"{swap}
    ]
    \ar[
      ddrr,
      ->>,
      "{
        p^{\mathrm{M}}
      }"{yshift=4pt, pos=.6},
      "{
        \scalebox{.7}{
          \color{gray}
          \eqref{MAlgebraExtensionOfFullyExtendedIIA}
        }
      }"{description, pos=.6}
    ]
    \\
    \\
    \mathfrak{F}
    \ar[
      ddrr,
      ->>
    ]
    &&
    &&
    \IIA
    \ar[
      ddll,
      ->>,
      "{
        \scalebox{.7}{
          \color{gray}
          \eqref{FullYTDualSpacetimeAsExtensionOfFullyDoubledSpacetime}
        }
      }"{description}
    ]
    \\
    \\
    &&
  \mathfrak{Dbl}
  \end{tikzcd}
  \hspace{.3cm}
  \longmapsto
  \hspace{.3cm}
  \begin{tikzcd}[
    row sep=20pt, 
    column sep=37pt
  ]
    &&
    \widehat{\mathfrak{M}}
    \ar[
      dddd,
      phantom,
      "{
        \adjustbox{
          raise=3pt,
          scale=.7
        }{
          \color{darkblue}
          \bf
          \def\arraystretch{.8}
          \begin{tabular}{c}
            {
              \color{gray}
              \eqref{TheHiddenMAlgebra}
            }
            \\
            Hidden
            \\
            M-algebra
        \end{tabular}
        }
      }"{pos=-.33}
    ]
    \ar[
      dddd,
      phantom,
      "{
        \begin{array}{c}
          \widehat{P}_3
          \\
          \rotatebox[origin=c]{-90}{
            $\longmapsto$
          }
          \\
          P_2
        \end{array}
      }",
      "{
        {p^{\mathrm{Brn}}_{\mathrm{bas}}}\, 
        p^{\mathrm{M}}_\ast
      }"{xshift=18pt, scale=.65},
      "{
        \scalebox{.67}{
          \color{olive}
          \bf
          \def\arraystretch{.85}
          \begin{tabular}{c}
            hidden M-
            \\
            Poincar{\'e}
            \\
            form
            \\
            \eqref{TheDecomposed3Form}
          \end{tabular}
        }
      }"{xshift=-27pt, pos=.4}
    ]
    \ar[
      ddll,
      ->>,
      "{
        p^{\mathrm{Brn}}
      }"{swap}
    ]
    \ar[
      ddrr,
      ->>,
      "{
        p^{\mathrm{M}}
      }"{yshift=4pt, pos=.6}
    ]
    \\
    \\
    \mathfrak{F}
    \ar[
      ddrr,
      ->>
    ]
    &&
    &&
    \widehat{\IIA}
    \ar[
      ddll,
      ->>
    ]
    \\
    \\
    &&
  \mathfrak{Dbl}
  \end{tikzcd}
\end{equation}

This plausibly exhibits the hidden M-algebra $\widehat{\mathfrak{M}}$ (at $s=-1$) as the Kleinian local model space for a topological T-duality- (and possibly U-duality-)covariant completion of 11D superspace supergravity. We hope to further discuss this elsewhere.

\smallskip

A first step in this direction would be to uncover the mathematical formalization of the action of this M-theoretic lift of T-duality on the $(G_4,G_7)$-flux, by somehow employing the hidden Poincar\'{e} 3-form. Certainly, the Fourier-Mukai transform picture cannot apply here verbatim, since the $(G_4,G_7)$-flux is a twisted \textit{non-abelian} cocycle (as in Ex. \ref{S4cocyclesAsTwistedS7}) and hence cannot be identified with an element in (abelian) twisted Chevalley-Eilenberg cohomology. Nevertheless, some form of the duality-by-adjunction (Theorems \ref{TwistedKTheoryRedIsoReOxi}, \ref{TwistedKTheoryToroidalRedIsoReOxi} and \ref{HigherTwistedCocyclesRedIsoReOxi}) may persist in this non-abelian case. 


\smallskip 
In this regard, it may be noteworthy that the Bianchi identity 
$\mathrm{d}\, \widehat{P}_3 \,=\, G_4$ \eqref{TheDecomposed3Form} on $\widehat{\mathfrak{M}}$ simply trivializes the class $[G_4]$ of the twist (instead of identifying it with a dual twist as in $\mathrm{d} \, P_2 \,=\, H_3^A - H_3^{\widetilde{A}}$ ), which actually matches the observation, shown in the following diagram, that
reducing and then re-oxidizing this M-theoretic twist through the grand M-theoretic correspondence diagram \eqref{SummaryDiagram} does make it disappear, in that all the components of $G_4$ (and of $G_7$) get absorbed as coordinates on $\mathfrak{Brn}$:
\begin{equation}
  \label{MTheoreticCorrespondence}
  \adjustbox{
    raise=-4cm
  }{
  \begin{tikzpicture} 
  \node at (0,0) {
  \begin{tikzcd}[column sep=huge]
    &[+20pt]
    &[-10pt]    
    \widehat{\mathfrak{M}}
    \ar[d, ->>]
    \\
    &[+20pt]
    &[-10pt]
    \mathfrak{M}
    \ar[
      dr, 
      ->>,
      "{
        \tilde p
      }"
    ]
    \ar[
      dl, 
      ->>,
      "{ p }"{swap}
    ]
    &
    &[+20pt]
    \\
    \mathfrak{l}S^4
    \ar[
      r,
      <-,
      "{
        (G_4, G_7)
      }"
    ]
    &
    \mathbb{R}^{
      1,10\,\vert\,\mathbf{32}
    }
    \ar[dr, ->>]
    &&
    \Brn
    \ar[dl, ->>]
    \ar[
      r,
      "{
        (
          \overline{\psi}
          \,\Gamma\,
          \psi
        )
      }"
    ]
    &
    b\mathbb{R}^{11}\, .
    \\
    &
    &
    \mathbb{R}^{0\,\vert\,\mathbf{32}}
    \ar[
      d,
      "{
        \left(
        \adjustbox{scale=.8}{$
        \def\arraycolsep{-1pt}
        \begin{array}{c}
          (\overline{\psi}\Gamma\psi),
          \\
          (\overline{\psi}\Gamma\Gamma\psi),
          \\
          (\overline{\psi}\Gamma\Gamma\Gamma\Gamma\Gamma\psi)
        \end{array}
        $}
        \right)
      }"{description}
    ]
    \\[30pt]
    &&
    b\mathbb{R}^{528}
  \end{tikzcd}
  };
  \begin{scope}[
    shift={(0,.2)}
  ]
  \node at (0,.9) {
    \adjustbox{scale=.7}{
      \color{darkblue}
      \bf
      \def\arraystretch{.9}
      \begin{tabular}{c}
        Basic
        \\
        M-algebra
      \end{tabular}
    }
  };
  \node at (0,3.3) {
    \adjustbox{scale=.7}{
      \color{darkblue}
      \bf
      \def\arraystretch{.9}
      \begin{tabular}{c}
        Hidden
        \\
        M-algebra
      \end{tabular}
    }
  };
  \node at (-3.6,-.4) {
    \adjustbox{scale=.7}{
      \color{darkblue}
      \bf
      \def\arraystretch{.9}
      \begin{tabular}{c}
        11D super-spacetime 
        \\ 
        carrying 
        \\
        M-brane charges
      \end{tabular}
    }
  };
  \node at (+3.5,-.4) {
    \adjustbox{scale=.7}{
      \color{darkblue}
      \bf
      \def\arraystretch{.9}
      \begin{tabular}{c}
        M-brane charges 
        \\ 
        carrying
        \\
        spacetime extension
      \end{tabular}
    }
  };
  \node at (0,-3.9) {
    \adjustbox{scale=.7}{
      \color{darkblue} \bf
      \def\arraystretch{.9}
      \begin{tabular}{c}
        Superpoint
        \\
        carrying 528
        \\
        0-brane charges
      \end{tabular}
    }
  };
  \draw[
    |->,
    bend right=10,
    gray
  ]
    (-3,-1.2) -- 
    node[
      sloped,
      scale=.7,
      yshift=-6pt,
      color=darkgreen
    ]{
      \bf
      full reduction
    }
    (-1.1,-3.2);
  \draw[
    |->,
    bend right=10,
    gray
  ]
    (+3,-1.2) -- 
    node[
      sloped,
      scale=.7,
      yshift=-6pt,
      color=darkgreen
    ]{
      \bf
      full reduction
    }
    (+1.1,-3.2);
  \end{scope}
  \end{tikzpicture}
  }
\end{equation}

\newpage 
\appendix

\section*{Background}

For ease of reference we briefly record some conventions, definitions and facts used in the main text.

\subsection*{Tensor conventions}
\label{TensorConventionsAnd11dSpinors}

Our tensor conventions are standard, but since superspace computations crucially depend on the corresponding prefactors, here to briefly make them explicit:\begin{itemize}[leftmargin=.4cm]
\item
  The Einstein summation convention applies throughout: Given a product of terms indexed by some $i \in I$, with the index of one factor in superscript and the other in subscript, then a sum over $I$ is implied:
  $
    x_i \, y^i
    :=
    \sum_{i \in I} 
    x_i \, y^i
  $.

\item
Our Minkowski metric is ``mostly plus''
\begin{equation}
  \label{MinkowskiMetric}
  \big(\eta_{ab}\big)
    _{a,b = 0}
    ^{ d }
  \;\;
    =
  \;\;
  \big(\eta^{ab}\big)
    _{a,b = 0}
    ^{ d }
  \;\;
    :=
  \;\;
  \big(
    \mathrm{diag}
      (-1, +1, +1, \cdots, +1)
  \big)_{a,b = 0}^{d} \;.
\end{equation}
\item
  Shifting position of frame indices always refers to contraction with the  Minkowski metric \eqref{MinkowskiMetric}:
  $$
    V^a 
      \;:=\;
    V_b \, \eta^{a b}
    \,,
    \;\;\;\;
    V_a \;=\; V^b \eta_{a b}
    \,.
  $$
\item Skew-symmetrization of indices is denoted by square brackets ($(-1)^{\vert\sigma\vert}$ is sign of the permutation $\sigma$):
$$
  V_{[a_1 \cdots a_p]}
  \;:=\;
  \tfrac{1}{p!}
  \sum_{
    \sigma \in \mathrm{Sym}(n)
  }
  (-1)^{\vert \sigma \vert}
  V_{ a_{\sigma(1)} \cdots a_{\sigma(p)} }\,.
$$
\item
We normalize the Levi-Civita symbol to \begin{equation}
  \label{transversalizationOfLeviCivitaSymbol}
  \epsilon_{0 1 2 \cdots} 
    \,:=\, 
  +1
  \;\;\;\;\mbox{hence}\;\;\;\;
  \epsilon^{0 1 2 \cdots} 
    \,:=\, 
  -1
  \,.
\end{equation}
\end{itemize}

\subsection*{Super-algebra}
\label{SuperAlgebraConventions}

In {\it homological} super-algebra, where a homological degree $n \in \mathbb{Z}$ (such as of flux densities) interacts with a super-degree $\sigma \in \ZTwo$ there are -- beware -- two different {\it sign rules} in use (cf. \cite[p. 62]{DeligneMorgan99}), whose relation is a little subtle. The traditional sign rule in supergravity 
(e.g. \cite[(II.2.106-9)]{CDF91})
that we follow here comes from $\mathbb{Z} \times \ZTwo$-{\it bi-grading}. (The alternative sign rule which collapses this bi-degree to a single ``parity'' degree in $\ZTwo$ is popular with authors who say the word ``Q-manifold''). 

\medskip

\noindent
{\bf Sign rule.}
For homological super-algebra we consider bigrading in the direct product ring $\mathbb{Z} \times \ZTwo$ ---  where the first factor $\mathbb{Z}$ is the homological degree and the second $\ZTwo \simeq \{\mathrm{evn}, \mathrm{odd}\}$ the super-degree -- with sign rule
\begin{equation}
  \label{Signs}
  \mathrm{deg}_1 = 
  (n_1, \sigma_1),
  \;
  \mathrm{deg}_2
  =
  (n_2, \sigma_2)
    \,\in\,
  \mathbb{Z}\times \ZTwo
  \hspace{1cm}
    \yields
  \hspace{1cm}
  \mathrm{sgn}
  \big(
    \mathrm{deg}_1,
    \,
    \mathrm{deg}_2
  \big)
  \;:=\;
  (-1)^{n_1 \cdot n_2 + \sigma_1 \cdot \sigma_2}
  \,.
\end{equation}

For $(v_i)_{i \in I}$ a set of generators with bi-degrees $(\mathrm{deg}_i)_{i \in I}$ we write:
\begin{itemize}[
  leftmargin=1.2cm,
  topsep=2pt,
  itemsep=4pt
]
\item[\bf (i)]
$
  \mathbb{R}\big\langle
    (v_i)_{i \in I}
  \big\rangle
$
for the graded super-vector space spanned by these elements,
\item[\bf (ii)]
$
  \mathbb{R}\big[
    (v_i)_{i \in I}
  \big]
$
for the graded-commutative polymonial algebra generated by these elements, 

hence the tensor algebra on $\vert I\vert$ generators modulo the relation
\begin{equation}
  \label{TheSignRule}
  v_1 \cdot v_2
  \;=\;
  (-1)^{\mathrm{sgn}(
    \mathrm{deg}_1
    ,
    \mathrm{deg}_2
  )}
  \;
  v_2 \cdot v_1
  \,,
\end{equation}

hence the (graded, super) {\it symmetric algebra} on the above super-vector space:
$$
  \mathbb{R}\big[ (v_i)_{i \in I} \big]
  \;\;
  :=
  \;\;
  \mathrm{Sym}
  \big(
    \mathbb{R}\big\langle 
      (v_i)_{i \in I} 
    \big\rangle
  \big).
$$
\item[\bf (iii)]
$\FDGCA\big[ (v_i)_{i \in I} \big]$ for the (free) differential graded-commutative algebra (dgca) generated by these elements and their \textit{differentials} 
$$(\dd v_i)_{i\in I}
$$ 
treated as primitive elements 
with $\mathrm{deg}(\dd e_i) = \deg(e_i) + (1,\mathrm{evn})$ and modulo the corresponding relation \eqref{TheSignRule}, with differential defined by 
\begin{align*} e_i \longmapsto \dd e_i  \hspace{1cm} ,  \hspace{1cm} 
\dd e_i \longmapsto 0
\end{align*}
and extended as a (graded) `derivation, hence the dgca 
\begin{equation}\label{FreeDGCA}
 \FDGCA \big[ (v_i)_{i \in I} \big]     \;
  :=
  \;
  \Big( \mathrm{Sym}
  \big(
    \mathbb{R}\big\langle 
      (v_i)_{i \in I}, \,  (\dd v_i)_{i \in I} 
    \big\rangle
  \big), \, \dd \Big).
\end{equation}

\end{itemize}

\subsection*{Spinors in 11D}
\label{SpinorsIn11d}

We briefly record the following standard facts about the Majorana spinor representation $\mathbf{32}$ of $\mathrm{Spin}(1,10)$ (proofs and references may be found in \cite[\S 2.5]{MiemiecSchnakenburg06}\cite[\S 2.2.1]{GSS24-SuGra}).

\smallskip

(We may and do take this to be the only spinor representation that we construct from ``from scratch''; all other spin representations we extract via simple algebra from this one. For instance the $\mathbf{16}$ and $\overline{\mathbf{16}}$ of $\mathrm{Spin}(1,9)$ are conveniently identified with the images $P(\mathbf{32})$ and $\overline{P}(\mathbf{32})$ of $\mathbf{32}$ under the projector
$P := \tfrac{1}{2}(1 + \Gamma_{10})$ and its adjoint, respectively cf. \eqref{16Rep} below.)

\medskip
\medskip

There exists an irreducible $\mathbb{R}$-linear representation $\mathbf{32}$ of $\mathrm{Pin}^+(1,10)$ with Clifford generators to be denoted
\begin{equation}
  \label{The11dMajoranaRepresentation}
  \Gamma_a 
  \;:\;
  \mathbf{32}
  \xrightarrow{\;\;}
  \mathbf{32}
\end{equation}
and equipped with a 
$\mathrm{Spin}(1,10)$-equivariant
skew-symmetric and non-degenerate
bilinear form
\begin{equation}
  \label{TheSpinorPairing}
  \big(\hspace{.8pt}
    \overline{(-)}
    (-)\,
  \big)
  \;:\;
  \mathbf{32}
  \otimes
  \mathbf{32}
  \xrightarrow{\quad}
  \mathbb{R}
\end{equation}
satisfying all of the following properties.

In stating these we use the following notation:

\begin{itemize}[
  leftmargin=.4cm,
  itemsep=2pt,
  topsep=2pt
]
\item We denote, as usual, the skew-symmetrized product of $k$ Clifford generators by
  \begin{equation}
    \label{CliffordBasisElements}
    \Gamma_{a_1 \cdots a_k}
    \;:=\;
    \tfrac{1}{k!}
    \underset{
      \sigma \in
      \mathrm{Sym}(k)
    }{\sum}
    \mathrm{sgn}(\sigma)
    \,
    \Gamma_{a_{\sigma(1)}}
    \cdot
    \Gamma_{a_{\sigma(2)}}
    \cdots
    \Gamma_{a_{\sigma(n)}}
    :
  \end{equation}

\item
The spinor pairing \eqref{TheSpinorPairing} serves as the {\it spinor metric} whose components -- being the odd partner of the Minkowski metric \eqref{MinkowskiMetric} --  we denote by
$\big(\eta_{\alpha\beta}\big)_{\alpha,\beta = 1}^{32}$:
\begin{equation}
  \label{TheSpinorMetric}
  \psi^\alpha
  \,
  \eta_{\alpha \beta}
  \,
  \phi^\beta
  \;\;
    :=
  \;\;
  \big(\hspace{1pt}
    \overline{\psi}
    \,
    \phi
  \big)
  \,.
\end{equation}
These are skew symmetric in their indices
\begin{equation}
  \label{SkewSymmetryOfSpinorMetric}
  \eta_{\alpha\beta}
  \;=\;
  -
  \eta_{\beta\alpha}
\end{equation}
which together with the inverse matrix $(\eta^{\alpha \beta})$ is
used to lower and raise spinor indices by contraction ``from the right'' (the position of the terms is irrelevant, since the components $\eta_{\alpha\beta}$ are commuting numbers, but the order of the indices matters due to the skew-symmetry):
\begin{equation}
  \label{LoweringOfSpinorIndices}
  \psi_\alpha
  \;:=\;
  \psi^{\alpha'}
  \eta_{\alpha' \alpha}
  \,,
  \;\;\;\;\;
  \psi^\alpha
  \;=\;
  \psi_{\alpha'}
  \eta^{\alpha' \alpha}
  \,,
  \;\;\;\;\;
  \psi_\alpha \phi^\alpha
  \;=\;
  -
  \psi^\alpha \phi_\alpha
  \,.
\end{equation}
\end{itemize}

\medskip

\noindent
Now, conventions may be chosen such that all of the following holds true:

\begin{itemize}[
  leftmargin=.4cm,
  itemsep=2pt,
  topsep=2pt
]
  \item
  The Clifford generators \eqref{The11dMajoranaRepresentation}
  square to the mostly plus Minkowski metric \eqref{MinkowskiMetric}
  \begin{equation}
    \label{CliffordDefiningRelation}
    \Gamma_a
    \Gamma_b
    +
    \Gamma_b
    \Gamma_a
    \;\;=\;\;
    +2 \, \eta_{a b}
    \,
    \mathrm{id}_{\mathbf{32}}
    \,.
  \end{equation}

  \item The Clifford product is given on the basis elements \eqref{CliffordBasisElements}
  as
\begin{equation}
  \label{GeneralCliffordProduct}
  \Gamma^{a_j \cdots a_1}
  \,
  \Gamma_{b_1 \cdots b_k}
  \;=\;
  \sum_{l = 0}^{
    \mathrm{min}(j,k)
  }
  \pm
  l!
\binom{j}{l}
 \binom{k}{l}
  \,
  \delta
   ^{[a_1 \cdots a_l}
   _{[b_1 \cdots b_l}
  \Gamma^{a_j \cdots a_{l+1}]}
  {}_{b_{l+1} \cdots b_k]}
  \,.
\end{equation}
  
  \item 
  The Clifford volume form equals the Levi-Civita symbol 
  \eqref{transversalizationOfLeviCivitaSymbol}:
  \begin{equation}
    \label{CliffordVolumeFormIn11d}
    \Gamma_{a_1 \cdots a_{11}}
    \;=\;
    \epsilon_{a_1 \cdots a_{11}}
    \mathrm{id}_{\mathbf{32}}
    \,.
  \end{equation}
\item The trace of all positive index Clifford basis elements vanishes:
  \begin{equation}
    \label{VanishingTraceOfCliffordElements}
    \mathrm{Tr}(
      \Gamma_{a_1 \cdots a_p}
    )
    \;=\;
    \left\{
    \def\arraystretch{1.3}
    \begin{array}{ccc}
      32 & \vert & p = 0
      \\
      0 & \vert & p > 0 
      \mathrlap{\,.}
    \end{array}
    \right.
    \,.
  \end{equation}
\item The Hodge duality relation on Clifford elements is:
\begin{equation}
  \label{HodgeDualityOnCliffordAlgebra}
  \Gamma^{a_1 \cdots a_p}
  \;=\;
  \tfrac{
    (-1)^{
      (p+1)(p-2)/2
    }
  }{
    (11-p)!
  }
  \,
  \epsilon^{ 
    a_1 \cdots a_p
    \,
    b_1 \cdots a_{11-p}
  }
  \,
  \Gamma_{b_1 \cdots b_{11-p}}
  \,.
\end{equation}
For instance:
\begin{equation}
  \label{ExamplesOfHodgeDualCliffordElements}
  \def\arraystretch{1.6}
  \def\arraycolsep{10pt}
  \begin{array}{l}
    \Gamma^{
      a_1 \cdots a_{11}
    }
    =
    \epsilon^{a_1 \cdots a_{11}}
    \mathrm{Id}_{\mathbf{32}}
    \,,
    \;\;\;\;\;\;\;
    \Gamma^{a_1 \cdots a_6}
    \;=\;
    +
    \tfrac{
      1
    }{
      5!
    }
    \,
    \epsilon^{
      a_1 \cdots a_6
      \,
      \color{darkblue}
      b_1 \cdots b_5
    }
    \,
    \Gamma_{
      \color{darkblue}
      b_1 \cdots b_5
    }
    \,,
    \\
    \Gamma^{a_1 \cdots a_{10}}
    =
    \epsilon^{
      a_1 \cdots a_{10} 
      \color{darkblue}
      b
    }
    \,
    \Gamma_{
      \color{darkblue}
      b
    }
    \,,
    \;\;\;\;\;\;\;\;
    \Gamma^{a_1 \cdots a_5}
    \;=\;
    -
    \tfrac{
      1
    }{
      6!
    }
    \,
    \epsilon^{
      a_1 \cdots a_5
      \,
      \color{darkblue}
      b_1 \cdots b_6
    }
    \,
    \Gamma_{
      \color{darkblue}
      b_1 \cdots b_6
    }
    \,.
  \end{array}
\end{equation}

  \item The Clifford generators are skew self-adjoint with respect to the pairing \eqref{TheSpinorPairing}
  \vspace{1mm} 
  \begin{equation}
    \label{SkewSelfAdjointnessOfCliffordGenerators}
    \overline{\Gamma_a}
    \;=\;
    - \Gamma_a
    \;\;\;\;\;\;
    \mbox{in that}
    \;\;\;\;\;\;
    \underset{
      \phi,\psi \in \mathbf{32}
    }{\forall}
    \;\;
    \big(\,
      \overline{(\Gamma_a \phi)}
      \,
      \psi
    \big)
    \;=\;
    -
    \big(\,
      \overline{\phi}
      \,
      (\Gamma_a \psi)
    \big)
    \,,
  \end{equation}
  so that generally
  \vspace{1mm} 
  \begin{equation}
    \label{AdjointnessOfCliffordBasisElements}
    \overline{\Gamma_{a_1 \cdots a_p}}
    \;=\;
    (-1)^{
      p + p(p-1)/2
    }
    \,
    \Gamma_{a_1 \cdots a_p}
    \,.
  \end{equation}

  \item
  The $\mathbb{R}$-vector space of $\mathbb{R}$-linear  endomorphisms of $\mathbf{32}$ has a linear basis given by the $\leq 5$-index Clifford elements 
  \vspace{1mm} 
  \begin{equation}
    \label{CliffordElementsSpanningLinearMaps}
    \mathrm{End}_{\mathbb{R}}\big(
      \mathbf{32}
    \big)
    \;\;
    =
    \;\;
    \big\langle
      1
      ,\,
      \Gamma_{a_1}
      ,\,
      \Gamma_{a_1 a_2}
      ,\,
      \Gamma_{a_1 a_2 a_3}
      ,\,
      \Gamma_{a_1 \cdots a_4}
      ,\,
      \Gamma_{a_1 \cdots a_5}
    \big\rangle_{
      a_i = 0, 1, \cdots
    }
    \,,
    \end{equation}

  \item
  The $\mathbb{R}$-vector space of {\it symmetric} bilinear forms on $\mathbf{32}$
  has a linear basis given by the expectation values with respect to \eqref{TheSpinorPairing} of the 1-, 2-, and 5-index Clifford basis elements:
  \begin{equation}
    \label{SymmetricSpinorPairings}
    \mathrm{Hom}_{\mathbb{R}}
    \Big(
    (\mathbf{32}\otimes \mathbf{32})_{\mathrm{sym}}
    ,\,
    \mathbb{R}
    \Big)
    \;\;
    \simeq
    \;\;
    \Big\langle
    \big(
      (\overline{-})
      \Gamma_a
      (-)
    \big)
    \,,\;\;
    \big(
      (\overline{-})
      \Gamma_{a_1 a_2}
      (-)
    \big)
    \,,\;\;
    \big(
      (\overline{-})
      \Gamma_{a_1 \cdots a_5}
      (-)
    \big)
    \Big\rangle_{
      a_i = 0, 1, \cdots
      \,,
    }
  \end{equation}
    which means in components that these Clifford generators are symmetric in their lowered indices \eqref{LoweringOfSpinorIndices}:
    \begin{equation}
      \label{SymmetryOfCliffordBasisElements}
      \Gamma^a_{\alpha \beta}
      \;=\;
      \Gamma^a_{\beta\alpha}
      \,,\;\;\;\;
      \Gamma^{a_1 a_2}_{\alpha \beta}
      \;=\;
      \Gamma^{a_1 a_2}_{\beta\alpha}
      \,,\;\;\;\;
      \Gamma^{a_1 \cdots a_5}_{\alpha \beta}
      \;=\;
      \Gamma^{a_1 \cdots a_5}_{\beta\alpha}
      \,,
    \end{equation}
  while a basis for the {\it skew-symmetric} bilinear forms is given by
  \begin{equation}
    \label{SkewSpinorPairings}
    \mathrm{Hom}_{\mathbb{R}}
    \Big(
    (\mathbf{32}\otimes \mathbf{32})_{\mathrm{skew}}
    ,\,
    \mathbb{R}
    \Big)
    \;\;
    \simeq
    \;\;
    \Big\langle
    \big(
      (\overline{-})
      (-)
    \big)
    \,,\;\;
    \big(
      (\overline{-})
      \Gamma_{a_1 a_2 a_3}
      (-)
    \big)
    \,,\;\;
    \big(
      (\overline{-})
      \Gamma_{a_1 \cdots a_4}
      (-)
    \big)
    \Big\rangle_{
      a_i = 0, 1, \cdots
      \,,
    }
  \end{equation}
    which means in components that these Clifford generators are skew-symmetric in their lowered indices \eqref{LoweringOfSpinorIndices}:
  \begin{equation}
    \eta_{\alpha\beta}
    \;=\;
    -
    \eta_{\beta\alpha}
    \,,\;\;\;\;
    \Gamma^{a_1 a_2 a_3}_{\alpha\beta}
    \;=\;
    -
    \Gamma^{a_1 a_2 a_3}_{\beta\alpha}
    \,,\;\;\;\;
    \Gamma^{a_1 \cdots a_5}_{\alpha\beta}
    \;=\;
    -
    \Gamma^{a_1 \cdots a_5}_{\beta\alpha}
  \end{equation}
\item
  Any linear endomorphism $\phi \in \mathrm{End}_{\mathbb{R}}(\mathbf{32})$ is uniquely a linear combination of Clifford elements as:
  \begin{equation}
    \label{CliffordExpansionOfEndomorphismOf32}
    \phi
      \;=\;
    \tfrac{1}{32}
    \sum_{p = 0}^5
    \;
    \tfrac{
      (-1)^{p(p-1)/2}
    }{ p! }
    \mathrm{Tr}\big(
      \phi \circ 
      \Gamma_{a_1 \cdots a_p}
    \big)
    \Gamma^{a_1 \cdots a_p}
    \,.
  \end{equation}

\item which implies in particular the Fierz expansion
\begin{equation}
  \label{FierzDecomposition}
  \hspace{-3mm} 
  \big(\,
  \overline{\phi}_1
  \,
  \psi
  \big)
  \big(\,
  \overline{\psi}
  \,
  \phi_2
  \big)
  \;
  =
  \;
  \tfrac{1}{32}\Big(
    \big(\,
      \overline{\psi}
      \,\Gamma^a\,
      \psi
    \big)
    \big(\,
      \overline{\phi}_1
      \,\Gamma_a\,
      \phi_2
    \big)
    -
    \tfrac{1}{2}
    \big(\,
      \overline{\psi}
      \,\Gamma^{a_1 a_2}\,
      \psi
    \big)
    \big(\,
      \overline{\phi}_1
      \,\Gamma_{a_1 a_2}\,
      \phi_2
    \big)
    +
    \tfrac{1}{5!}
    \big(\,
      \overline{\psi}
      \,\Gamma^{a_1 \cdots a_5}\,
      \psi
    \big)
    \big(\,
      \overline{\phi}_1
      \,\Gamma_{a_1 \cdots a_5}\,
      \phi_2
    \big)
 \Big).
\end{equation}

\end{itemize}

\vspace{-2mm} 
\begin{proposition}[{\bf The general Fierz identities} {\cite[(3.1-3) \& Table 2]{DF82}\cite[(II.8.69) \& Table II.8.XI]{CDF91}}]
 
  \noindent {\bf (i)}  The $\mathrm{Spin}(1,10)$-irrep decomposition of the first few symmetric tensor powers of $\mathbf{32}$ is:
  \begin{equation}
    \label{IrrepsInSymmetricPowersOf32}
    \def\arraystretch{1.3}
    \begin{array}{rcl}
       \big(
        \mathbf{32} \otimes \mathbf{32}
      \big)_{\mathrm{sym}}
      &\cong&
      \mathbf{11}
      \,\oplus\,
      \mathbf{55}
      \,\oplus\,
      \mathbf{462}
      \\
       \big(
        \mathbf{32}
          \otimes
        \mathbf{32} 
          \otimes 
        \mathbf{32}
      \big)_{\mathrm{sym}}
      &\cong&
      \mathbf{32}
      \,\oplus\,
      \mathbf{320}
      \,\oplus\,
      \mathbf{1408}
      \,\oplus\,
      \mathbf{4424}
      \\
       \big(
        \mathbf{32}
          \otimes
        \mathbf{32} 
          \otimes 
        \mathbf{32}
          \otimes 
        \mathbf{32}
      \big)_{\mathrm{sym}}
      &\cong&
      \mathbf{1}
      \,\oplus\,
      \mathbf{165}
      \,\oplus\,
      \mathbf{330}
      \,\oplus\,
      \mathbf{462}
      \,\oplus\,
      \mathbf{65}
      \,\oplus\,
      \mathbf{429}
      \,\oplus\,
      \mathbf{1144}
      \,\oplus\,
      \mathbf{17160}
      \,\oplus\,
      \mathbf{32604}\,.
    \end{array}
  \end{equation}
  \noindent 
  {\bf (ii)} In more detail, the irreps appearing on the right are tensor-spinors spanned by basis elements
  \begin{equation}
    \label{TheHigherTensorSpinors}
    \def\arraystretch{1.6}
    \begin{array}{l}
    \big\langle
      \Xi^\alpha_{a_1 \cdots a_p}
      \;=\;
      \Xi^\alpha_{[a_1 \cdots a_p]}
    \big\rangle_{
      a_i \in \{0,\cdots, 10\}, 
      \alpha \in \{1, \cdots 32\}
    }
    \;\;\;
    \in
    \;\;
    \mathrm{Rep}_{\mathbb{R}}\big(
      \mathrm{Spin}(1,10)
    \big)
    \\
    \mbox{\rm with}
    \;\;\;
    \Gamma^{a_1} \Xi_{a_1 a_2 \cdots a_p}
    \;=\;
    0
    \end{array}
  \end{equation}
  {\rm (jointly to be denoted $\Xi^{(N)}$ for the case of the irrep $\mathbf{N}$)}
  such that:
  \begin{equation}
    \label{GeneralCubicFierzIdentities}
    \def\arraycolsep{3pt}
    \def\arraystretch{1.5}
    \begin{array}{rcrrrr}
      \psi
      \big(\,
        \overline{\psi}
        \,\Gamma_a\,
        \psi
      \big)
      &=&
      \tfrac{1}{11}
      \,\Gamma_a\,
      \Xi^{(32)}
      &+\;
      \Xi^{(320)}_a
      \mathrlap{\,,}
      \\
      \psi 
      \big(\,
        \overline{\psi}
        \,\Gamma_{a_1 a_2}\,
        \psi
      \big)
      &=&
      \tfrac{1}{11}
      \,\Gamma_{a_1 a_2}\,
      \Xi^{(32)}
      &-\;
      \tfrac{2}{9}
      \,\Gamma_{[a_1}\,
      \Xi^{(320)}_{a_2]}
      &+\;
      \Xi^{(1408)}_{a_1 a_2}
      \mathrlap{\,,}
      \\
      \psi
      \big(\,
        \overline{\psi}
        \,\Gamma_{a_1 \cdots a_5}\,
        \psi
      \big)
      &=&
      -\tfrac{1}{77}
      \Gamma_{a_1 \cdots a_5}
      \Xi^{(32)}
      &+\;
      \tfrac{5}{9}
      \Gamma_{[a_1 \cdots a_4}
      \Xi^{(320)}_{a_5]}
      &+\;
      2
      \,\Gamma_{[a_1 a_2 a_3}\,
      \Xi^{(1408)}_{a_4 a_5]}
      &+\;
      \Xi^{(4224)}_{a_1 \cdots a_5}
      \mathrlap{\,.}
    \end{array}
  \end{equation}
\end{proposition}

\end{document}